\documentclass[acmsmall]{acmart}

\newif\ifappx
\appxtrue
\newcommand{\appx}[2]{\ifappx{#1}\else{#2}\fi}

\AtBeginDocument{%
  }

\setcopyright{rightsretained}
\acmDOI{10.1145/3649849}
\acmYear{2024}
\copyrightyear{2024}
\acmSubmissionID{oopslaa24main-p142-p}
\acmJournal{PACMPL}
\acmVolume{8}
\acmNumber{OOPSLA1}
\acmArticle{132}
\acmMonth{4}
\received{2023-10-21}
\received[accepted]{2024-02-24}

\usepackage{booktabs}
\usepackage{subcaption}
\usepackage{algorithm}
\usepackage[noend]{algpseudocode}
\usepackage{amsmath}
\usepackage{stmaryrd}
\usepackage{url}
\usepackage{listings}
\usepackage{graphicx}
\usepackage{wrapfig}
\usepackage{fancyvrb}
\usepackage{paralist}
\usepackage{caption}
\usepackage{comment}
\usepackage{xspace}
\usepackage{multirow}
\usepackage{booktabs}
\usepackage{enumerate}
\usepackage{enumitem}
\usepackage{caption}
\usepackage{graphicx}
\usepackage{listing}
\usepackage{tikz}
\usepackage{ifsym}
\usepackage{circledsteps}
\usepackage{mdframed}
\usepackage[
    type={CC},
    modifier={by},
    version={4.0},
]{doclicense}

\newcommand{\revision}[1]{#1}

\newcommand{\ywtodo}[1]{{\color{orange}{#1}}}

\newcommand{\irule}[2]%
   {\mkern-2mu\displaystyle\frac{#1}{\vphantom{,}#2}\mkern-2mu}
\newcommand{\irulelabel}[3]
{
\mkern-2mu
\begin{array}{ll}
\displaystyle\frac{#1}{\vphantom{,}#2} & \!\!\!\!\!\! #3
\end{array}
\mkern-2mu
}

\newcommand{\set}[1]{\{ #1 \}}
\newcommand*\circled[1]{\tikz[baseline=(char.base)]{
            \node[shape=circle,draw,inner sep=0.8pt] (char) {\scriptsize{#1}};}}

\newcommand{\proj}{\Pi}
\newcommand{\filter}{\sigma}
\newcommand{\rename}{\rho}
\newcommand{\ijoin}{\bowtie}
\newcommand{\ljoin}{\leftouterjoin}
\newcommand{\rjoin}{\rightouterjoin}
\newcommand{\fjoin}{\fullouterjoin}
\newcommand{\cplus}{\mathrel{\reflectbox{\rotatebox[origin=c]{180}{$\uplus$}}}}

\def\ojoin{\setbox0=\hbox{$\bowtie$}%
  \rule[.17ex]{.25em}{.6pt}\llap{\rule[.88ex]{.25em}{.6pt}}}
\def\leftouterjoin{\mathbin{\ojoin\mkern-5.8mu\bowtie}}
\def\rightouterjoin{\mathbin{\bowtie\mkern-5.8mu\ojoin}}
\def\fullouterjoin{\mathbin{\ojoin\mkern-5.8mu\bowtie\mkern-5.8mu\ojoin}}

\newcommand{\lljoin}{\scalebox{0.5}[1]{$\sqsupset$}\mkern-3.2mu\lower0em\hbox{$\Join$}}

\newcommand{\rosette}{\textsc{Rosette}\xspace}
\newcommand{\udp}{\textsc{UDP}\xspace}
\newcommand{\tool}{\textsc{VeriEQL}\xspace}

\newcommand{\mediator}{\textsc{Mediator}\xspace}
\newcommand{\datafiller}{\textsc{DataFiller}\xspace}
\newcommand{\xdata}{\textsc{XData}\xspace}

\newcommand{\spes}{\textsc{SPES}\xspace}
\newcommand{\qex}{\textsc{Qex}\xspace}
\newcommand{\hottsql}{\textsc{HoTTSQL}\xspace}
\newcommand{\cosette}{\textsc{Cosette}\xspace}

\newcommand{\calcite}{\texttt{Calcite}\xspace}
\newcommand{\leetcode}{\texttt{LeetCode}\xspace}

\newcommand{\literature}{\texttt{Literature}\xspace}

\newcommand{\schema}{\mathcal{S}}
\newcommand{\constraint}{\mathcal{C}}
\newcommand{\db}{\mathcal{D}}
\newcommand{\query}{Q}
\newcommand{\formula}{\Phi}
\newcommand{\bound}{N}
\newcommand{\context}{\Gamma}
\newcommand{\attributes}{\mathcal{A}}
\newcommand{\tuples}{\mathcal{T}}
\newcommand{\del}{\text{Del}\xspace}
\newcommand{\dom}{\textsf{Dom}\xspace}
\newcommand{\vars}{\textsf{Vars}\xspace}
\newcommand{\fresh}{\sigma}
\newcommand{\nullv}{\text{Null}\xspace}

\newcommand{\ite}{\text{ite}\xspace}

\newcommand{\alljoin}{\otimes}
\newcommand{\allcoll}{\oplus}
\newcommand{\allarith}{\diamond}
\newcommand{\alllogic}{\odot}
\newcommand{\paired}{\texttt{Paired}\xspace}
\newcommand{\eqtext}{\texttt{Eq}\xspace}
\newcommand{\aggr}{\mathcal{G}}
\newcommand{\interpretation}{\mathcal{I}}
\newcommand{\extends}{\sqsupseteq}
\newcommand{\conforms}{::}

\newcommand{\doubleplus}{+\!+~}
\newcommand{\doublebrackets}[1]{\llbracket #1 \rrbracket}
\newcommand{\denot}[1]{\doublebrackets{#1}}
\newcommand{\defeq}{\overset{\mathrm{def}}{=\joinrel=}}
\newcommand{\indicator}[1]{\mathbb{I}[#1]}

\newcommand{\sqlwhere}{\texttt{WHERE}\xspace}
\newcommand{\sqlgroupby}{\texttt{GROUP}\:\texttt{BY}\xspace}
\newcommand{\sqlhaving}{\texttt{HAVING}\xspace}
\newcommand{\sqlwith}{\texttt{WITH}\xspace}
\newcommand{\sqlin}{\texttt{IN}\xspace}
\newcommand{\sqlnotin}{\texttt{NOT}\:\texttt{IN}\xspace}
\newcommand{\sqlorderby}{\texttt{ORDER}\:\texttt{BY}\xspace}
\newcommand{\sqllimit}{\texttt{LIMIT}\xspace}

\newcommand{\sqlintersect}{\texttt{INTERSECT}\xspace}
\newcommand{\sqlexcept}{\texttt{EXCEPT}\xspace}

\newcommand{\sqland}{\texttt{AND}\xspace}
\newcommand{\sqlor}{\texttt{OR}\xspace}

\newcommand{\sqlexceptall}{\texttt{EXCEPT}\: \texttt{ALL}\xspace}
\newcommand{\sqlleftjoin}{\texttt{LEFT}\: \texttt{JOIN}\xspace}
\newcommand{\sqlinnerjoin}{\texttt{INNER}\: \texttt{JOIN}\xspace}
\newcommand{\sqljoin}{\texttt{JOIN}\xspace}
\newcommand{\sqlon}{\texttt{ON}\xspace}
\newcommand{\sqlnull}{\texttt{NULL}\xspace}
\newcommand{\sqlnotnull}{\texttt{NOT}\:\texttt{NULL}\xspace}
\newcommand{\sqlcount}{\texttt{COUNT}\xspace}
\newcommand{\sqlmax}{\texttt{MAX}\xspace}

\newcommand{\sqlsum}{\texttt{SUM}\xspace}

\newcommand{\sqlif}{\texttt{IF}\xspace}

\newcommand{\sqlcasewhen}{\texttt{CASE}\:\texttt{WHEN}\xspace}

\newcommand{\sqlspace}{\;}

\newcommand{\sqldistinctcolor}{{\color{black}{{\textbf{\texttt{DISTINCT}}\xspace\xspace}}}}
\newcommand{\sqlselectcolor}{{\color{black}{{\textbf{\texttt{SELECT}}\xspace\xspace}}}}
\newcommand{\sqlfromcolor}{{\color{black}{{\textbf{\texttt{FROM}}\xspace\xspace}}}}
\newcommand{\sqlwherecolor}{{\color{black}{{\textbf{\texttt{WHERE}}\xspace\xspace}}}}

\newcommand{\sqlgroupbycolor}{{\color{black}{{\textbf{\texttt{GROUP}}\sqlspace\textbf{\texttt{BY}}}\xspace}}}
\newcommand{\sqlhavingcolor}{{\color{black}{{\textbf{\texttt{HAVING}}\xspace\xspace}}}}

\newcommand{\sqljoincolor}{{\color{black}{{\textbf{\texttt{JOIN}}\xspace\xspace}}}}

\newcommand{\sqlleftjoincolor}{{\color{black}{{\textbf{\texttt{LEFT}}\sqlspace\textbf{\texttt{JOIN}}\xspace\xspace}}}}

\newcommand{\sqloncolor}{{\color{black}{{\textbf{\texttt{ON}}\xspace\xspace}}}}
\newcommand{\sqlascolor}{{\color{black}{{\textbf{\texttt{AS}}\xspace\xspace}}}}
\newcommand{\sqlandcolor}{{\color{black}{{\textbf{\texttt{AND}}\xspace\xspace}}}}

\newcommand{\sqlcountcolor}{{\color{black}{{\textbf{\texttt{COUNT}}\xspace\xspace}}}}
\newcommand{\sqlsumcolor}{{\color{black}{{\textbf{\texttt{SUM}}\xspace\xspace}}}}
\newcommand{\sqlavgcolor}{{\color{black}{{\textbf{\texttt{AVG}}\xspace\xspace}}}}

\newcommand{\sqlorcolor}{{\color{black}{{\textbf{\texttt{OR}}\xspace\xspace}}}}
\newcommand{\sqlnotcolor}{{\color{black}{{\textbf{\texttt{NOT}}\xspace\xspace}}}}
\newcommand{\sqlincolor}{{\color{black}{{\textbf{\texttt{IN}}\xspace\xspace}}}}
\newcommand{\sqlcasecolor}{{\color{black}{{\textbf{\texttt{CASE}}\xspace\xspace}}}}
\newcommand{\sqlwhencolor}{{\color{black}{{\textbf{\texttt{WHEN}}\xspace\xspace}}}}
\newcommand{\sqlthencolor}{{\color{black}{{\textbf{\texttt{THEN}}\xspace\xspace}}}}
\newcommand{\sqlelsecolor}{{\color{black}{{\textbf{\texttt{ELSE}}\xspace\xspace}}}}
\newcommand{\sqlendcolor}{{\color{black}{{\textbf{\texttt{END}}\xspace\xspace}}}}
\newcommand{\newpara}[1]{{{\vspace{1mm} \noindent \emph{\textbf{#1}}}}}
\newcommand{\evalfinding}[1]{
\vspace{1pt}
\begin{center}
\fcolorbox{black}{white}{\parbox{.98\linewidth}{
{#1}
}}
\vspace{-5pt}
\end{center}
}

\newcommand{\sfoldl}{\textsf{foldl}\xspace}
\newcommand{\sfoldr}{\textsf{foldr}\xspace}
\newcommand{\smap}{\textsf{map}\xspace}
\newcommand{\sfilter}{\textsf{filter}\xspace}
\newcommand{\sappend}{\textsf{append}\xspace}
\newcommand{\scons}{\textsf{cons}\xspace}
\newcommand{\smerge}{\textsf{merge}\xspace}
\newcommand{\srename}{\textsf{rename}\xspace}

\newcommand{\site}{\textsf{ite}\xspace}
\newcommand{\shead}{\textsf{head}\xspace}

\newcommand{\ptr}{\textsf{ptr}\xspace}

\begin{document}






\title{VeriEQL: Bounded Equivalence Verification for Complex SQL Queries with Integrity Constraints}

\author{Yang He}
\authornote{Both authors contributed equally to the paper.}
\orcid{0009-0007-7755-3112}
\affiliation{%
  \institution{Simon Fraser University}
  \city{Burnaby}
  \country{Canada}
}
\email{yha244@sfu.ca}

\author{Pinhan Zhao}
\authornotemark[1]
\orcid{0009-0002-1149-0706}
\affiliation{%
  \institution{University of Michigan}
  \city{Ann Arbor}
  \country{USA}
}
\email{pinhan@umich.edu}

\author{Xinyu Wang}
\orcid{0000-0002-1836-0202}
\affiliation{%
  \institution{University of Michigan}
  \city{Ann Arbor}
  \country{USA}
}
\email{xwangsd@umich.edu}

\author{Yuepeng Wang}
\orcid{0000-0003-3370-2431}
\affiliation{%
  \institution{Simon Fraser University}
  \city{Burnaby}
  \country{Canada}
}
\email{yuepeng@sfu.ca}

\renewcommand{\shortauthors}{Yang He, Pinhan Zhao, Xinyu Wang, and Yuepeng Wang}

\begin{abstract}
The task of SQL query equivalence checking is important in various real-world applications (including query rewriting and automated grading) that involve \emph{complex queries with integrity constraints}; yet, state-of-the-art techniques are very limited in their capability of reasoning about complex features (e.g., those that involve sorting, case statement, rich integrity constraints, etc.) in real-life queries. 
To the best of our knowledge, we propose the first SMT-based approach and its implementation, $\tool$, capable of proving and disproving bounded equivalence of \emph{complex} SQL queries. 
$\tool$ is based on a new logical encoding that models query semantics over symbolic tuples using the theory of integers with uninterpreted functions. 
It is \emph{simple yet highly practical} --- our comprehensive evaluation on over 20,000 benchmarks shows that $\tool$ outperforms \emph{all} state-of-the-art techniques by \emph{more than one order of magnitude} in terms of the number of benchmarks that can be proved or disproved. 
$\tool$ can also generate counterexamples that facilitate many downstream tasks (such as finding serious bugs in systems like MySQL and Apache Calcite).
\end{abstract}

\ccsdesc[500]{Theory of computation~Program verification}
\ccsdesc[500]{Software and its engineering~Software verification}
\ccsdesc[500]{Software and its engineering~Formal software verification}


\keywords{Program Verification, Equivalence Checking, Relational Databases.}

\maketitle

\section{Introduction} \label{sec:intro}


\noindent
Equivalence checking of SQL queries is an important problem with various real-world applications, including validating source-level query rewriting~\cite{graefe1995cascades,chu2017cosette} and automated grading of SQL queries~\cite{chandra2019automated}. 
A useful SQL equivalence checker should be able to (i) provide formal guarantee on query equivalence (either fully or in a bounded manner), (ii) generate counterexamples to witness query non-equivalence, and (iii) support an expressive query language. 
For example, in the context of query rewriting where a slow query $\query_1$ is rewritten to a faster query $\query_2$ using rewrite rules, one may want to ensure the rewrite is correct by showing $\query_1$ and $\query_2$ are semantically equivalent with a certain level of formal guarantee, and in case of non-equivalence, obtain a counterexample input database to help fix the incorrect rule. 
Equivalence checking is also useful for automated query grading. In particular, it can provide feedback to users by checking their submitted queries against a ground-truth query, where the feedback could be a counterexample (i.e., a concrete database) that illustrates why the user query is wrong. 
Moreover, these counterexamples can also serve as additional test cases to augment an existing test suite (such as the one maintained by LeetCode~\cite{leetcode}, the world's most popular online programming platform).

While prior work~\cite{chu2017cosette,chu2017demonstration,chu2017hottsql,chu2018axiomatic,wang2018speeding,qex} has made some advances in both proving and disproving query equivalence, there remain a number of challenges that significantly limit the practical usage of existing techniques in real-world applications. 
The gap, as we will also show in our evaluation section later, is in fact \emph{extremely large}: for example, existing work supports less than {2\%} of the SQL queries from LeetCode. The reasons are threefold.
\begin{itemize}[leftmargin=*]
\setlength\itemsep{3pt}
\item 
First, real-world queries are complex. In addition to the simplest select-project-join queries using common aggregate functions such as $\sqlsum, \sqlcount, \sqlmax$ --- that existing techniques support fairly well --- queries in the wild frequently use advanced SQL features with much more complex logic (such as sorting, case statement, common table expressions, $\sqlin$ and $\sqlnotin$ operators, etc.) which, to our best knowledge, existing SQL equivalence checkers rarely support. 
For example, among more than 20,000 real-life queries we studied, 
more than {20\%} of them involve $\sqlorderby$, over {30\%} have $\sqlcasewhen$, over {15\%} require $\sqlin$ or $\sqlnotin$, and {15\%} use $\sqlwith$, among others. This brings up new challenges in how to precisely model the semantics of these advanced SQL operators. 
\item 
Furthermore, the interleaving between these advanced operators and other SQL features (such as three-valued logic, aggregate functions, grouping, etc.) makes the problem even more challenging. For example, the three-valued logic uses a special value $\sqlnull$ that many existing techniques (such as $\hottsql$~\cite{chu2017hottsql}) do not consider.
{We must take into account all these additional SQL features to properly model the semantics of SQL operators in order to fully support complex real-world queries and reason about query equivalence.}
\item 
Finally, most real-world queries involve integrity constraints, but prior work barely supports them. 
For instance, over 95\% of our benchmarks require integrity constraints, yet all existing techniques \emph{cumulatively} support under 2\%, to our best knowledge.
These constraints are rich: in addition to the primary key and foreign key constraints that stipulate uniqueness and value references, there are many others including $\sqlnotnull$ (used by \emph{all} queries in our $\calcite$ benchmark suite) and various constraints that restrict attribute values (>70\% across all our benchmarks)
such as requiring an attribute to be only positive integers. 
This richness poses significant challenges: for example, to witness the non-equivalence of two queries, a valid counterexample must not only yield different outputs \emph{but also} meet the integrity constraints.


\end{itemize}



In this paper, we propose a \emph{simple yet practical} approach to SQL query equivalence checking --- we can prove equivalence (in a bounded fashion) and non-equivalence (by generating counterexamples) for \emph{a complex query language with rich integrity constraints.}
\revision{
Our key contribution is a \emph{new SMT encoding} tailored towards bounded equivalence verification, based on a \emph{new semantics formalization} for a practical fragment of SQL (which is \emph{significantly larger} than those considered in prior work). 
First, we formalize our SQL semantics using list and higher-order functions, different from the K-relations in \hottsql~\cite{chu2017hottsql} or U-semiring in \udp~\cite{chu2018axiomatic}. Our formalization is inspired by \mediator~\cite{mediator}. 
However, \mediator considers unbounded equivalence verification for a small set of SQL queries, 
while our formalization considers a much larger language. 
Second, building upon the standard approach of using SMT formulas to relate the program inputs and outputs, we develop a \emph{new SMT encoding}, for our formalized query semantics, in the theory of integers and uninterpreted functions that is based on symbolic tuples. 
This new SMT encoding allows us to handle complex SQL features (e.g., \sqlorderby, \sqlnull, etc) and rich integrity constraints without unnecessarily heavy theories such as theory of lists in \qex~\cite{qex}. It also does not require an indirect encoding by translating the SQL queries to an intermediate representation in solver-aided programming language (e.g., \rosette~\cite{rosette-pldi14}), unlike the previous \cosette line of work~\cite{chu2017cosette,chu2017demonstration,wang2018speeding}.
Last but not least, we provide detailed correctness proofs of our SMT encoding with respect to our formal semantics.
}

\newpara{$\tool$.}
\revision{
We have implemented our approach in \tool, which is described schematically in Figure~\ref{fig:workflow}.
}
At a high-level, $\tool$ takes as input a pair of queries $(\query_1, \query_2)$, the database schema $\schema$ and its integrity constraint $\constraint$, and a bound $\bound$ defining the input space. 
It constructs an SMT formula $\formula$ such that: 
(i) if $\formula$ is unsatisfiable, then $\query_1$ and $\query_2$ are guaranteed to be equivalent for all database relations with at most $\bound$ tuples, and 
(ii) if $\formula$ is satisfiable, then $\query_1$ and $\query_2$ are provably non-equivalent, and we generate a counterexample (i.e., a database that meets $\constraint$ but leads to different query outputs) from $\formula$'s satisfying assignments. 
Internally, {\small \Circled{1}} $\tool$ first creates a symbolic representation $\context$ of all input databases with up to $\bound$ tuples in their relations, and encodes the integrity constraint $\constraint$ over $\context$ into an SMT formula $\formula_{\constraint}$. 
Then, we process both queries and encode their equivalence into an SMT formula $\formula_{R}$: 
{\small \Circled{2}} we traverse $\query_i$ in a forward fashion, encode how each operator in $\query_i$ transforms its input to output, obtain the final output $R_i$ of $\query_i$ for all input databases under consideration, and {\small \Circled{3}} generate an SMT formula $\formula_{R}$ that asserts $R_1 \neq R_2$ (we support bag and list semantics). 
While this overall approach is standard, our encoding scheme of each operator's semantics is new and has some important advantages. 
First, our encoding is based on the theory of integers with uninterpreted functions: it is simple yet sufficient to precisely encode all SQL features in our language (such as complex aggregate functions with grouping and previously unsupported operators like $\sqlorderby$). Crucially, our approach can support these advanced SQL features without needing additional axioms, which prior work like $\qex$ would otherwise require. 
Second, our encoding follows the three-valued logic and supports $\sqlnull$ for all of our operators, whereas prior work supports $\sqlnull$ for significantly fewer cases. 
In the final step {\small \Circled{4}}, we construct $\formula = \formula_{\constraint} \land \formula_{R}$ and use an off-the-shelf SMT solver to solve $\formula$. Notably, $\formula$ takes into account both query semantics and integrity constraints in a much simpler and more unified manner than some prior work that has separate schemes to handle queries and integrity constraints. For example, SPES~\cite{zhou2022spes} performs symbolic encoding for queries but uses rewrite rules to deal with integrity constraints --- it is fundamentally hard for such approaches to support complex integrity constraints such as those that restrict the value range.

\begin{figure}[!t]
\centering
\includegraphics[scale=0.16]{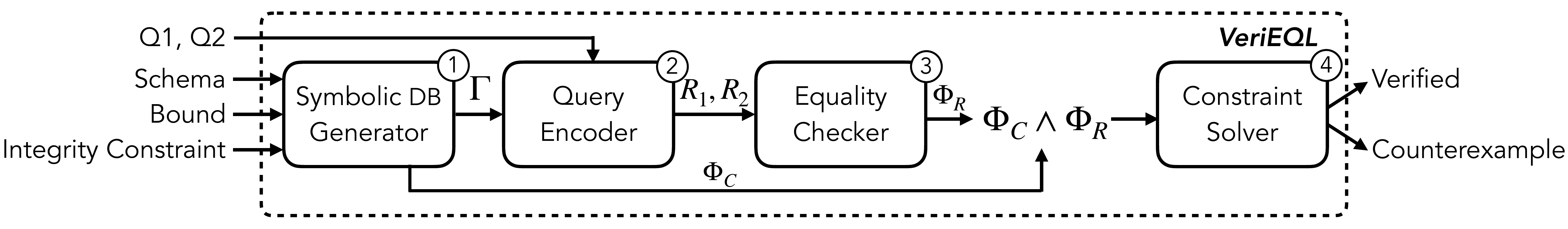}
\vspace{-5pt}
\caption{Schematic workflow of $\tool$.}
\label{fig:workflow}
\vspace{-10pt}
\end{figure}

We have evaluated $\tool$ on an extremely large number of benchmarks, consisting of 24,455 query pairs collected from three different workloads (including all benchmarks from the literature, standard benchmarks from Calcite, and over 20,000 new benchmarks from LeetCode).
Our evaluation results show that $\tool$ can prove the bounded equivalence for significantly more benchmarks than all state-of-the-art techniques, disprove and find counterexamples for two orders of magnitude more benchmarks, and uncover serious bugs in real-world codebases (including MySQL and Calcite).

\vspace{5pt}
\newpara{Contributions.}
This paper makes the following contributions. 
\begin{itemize}[leftmargin=*]
\item 
We formulate the problem of bounded SQL equivalence verification \emph{modulo integrity constraints}. 
\item 
\revision{
We formalize the semantics of a practical fragment of SQL queries through list and higher-order functions. Our SQL fragment is significantly larger than those in prior work.
}
\item
\revision{
We propose a novel SMT encoding tailored towards bounded equivalence verification of SQL queries, including previously unsupported SQL operators (e.g., \sqlorderby, \sqlcasewhen) and rich integrity constraints. To our best knowledge, this is the first approach that supports complex SQL queries with rich integrity constraints.
}
\item
\revision{
We prove our SMT encoding of SQL queries is correct with respect to the formal semantics.
}
\item
\revision{
We implement our approach in \tool.
Our comprehensive evaluation on a total of 24,455 benchmarks --- including a new benchmark suite with over 20,000 real-life queries --- shows that, $\tool$ can solve (i.e., prove or disprove equivalence) 77\% of these benchmarks, while state-of-the-art bounded verifier can solve <2\% and testing tools can disprove <1\%.
}
\end{itemize}

\section{Overview} \label{sec:overview}




\begin{figure}[!h]
\centering
\begin{subfigure}{\linewidth}
\includegraphics[width=\linewidth]{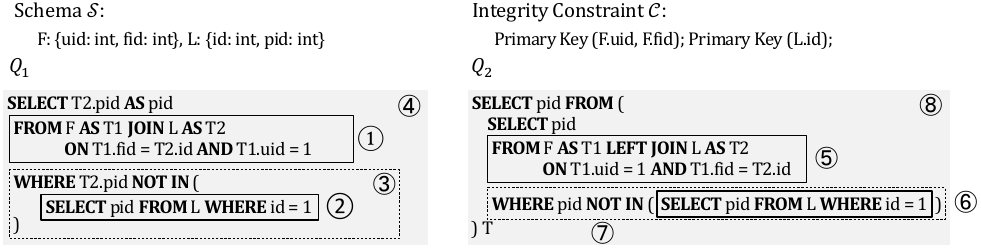}
\caption{Schema, integrity constraint, and queries.}
\label{fig:overview-query}
\vspace{3pt}
\end{subfigure}
\hfill
\begin{subfigure}{.4\linewidth}
\centering
\includegraphics[width=0.6\linewidth]{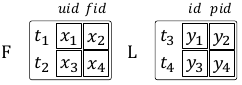}
\caption{Generating symbolic database.}
\label{fig:overview-db}
\end{subfigure}
\hfill
\begin{subfigure}{.59\linewidth}
\centering
\includegraphics[width=0.7\linewidth]{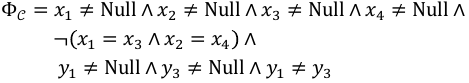}
\caption{Encoding integrity constraint.}
\label{fig:overview-ic}
\end{subfigure}
\hfill
\begin{subfigure}{\linewidth}
\vspace{3pt}
\includegraphics[width=\linewidth]{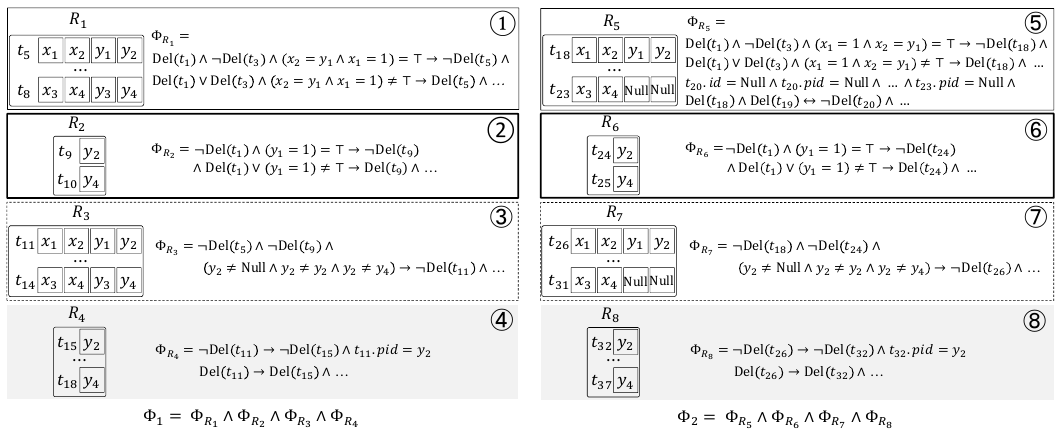}
\caption{Encoding query semantics.}
\label{fig:overview-query-semantics}
\vspace{5pt}
\end{subfigure}
\hfill
\begin{subfigure}{0.85\linewidth}
\includegraphics[width=\linewidth]{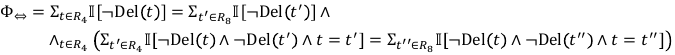}
\caption{Checking equality.}
\label{fig:overview-checking-equality}
\end{subfigure}
\hfill
\begin{subfigure}{\linewidth}
\includegraphics[width=\linewidth]{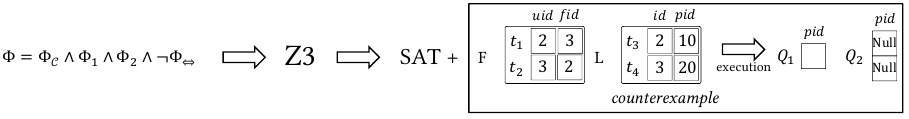}
\caption{Constraint solving.}
\label{fig:overview-results}
\end{subfigure}
\vspace{-20pt}
\caption{Illustration of how $\tool$ works on a simple LeetCode task.}
\label{fig:overview-example}
\vspace{-15pt}
\end{figure}

In this section, we further illustrate $\tool$'s workflow using a simple example from LeetCode\footnote{https://leetcode.com/problems/page-recommendations}. 
Note that for illustration purposes, we significantly simplify the database schema and queries in this task, but $\tool$ can handle the original queries\footnote{The original queries are available in \appx{Appendix~\ref{sec:full-semantics}}{the Appendix of the extended version~\cite{extended}}.}
 that are much more complex. 

Specifically, this task involves a database with two relations, \texttt{Friendship} (or $F$) and \texttt{Likes} (or $L$), where $F$ records users' friends and $L$ stores users' preferred pages. The task is to write a query that, given a user, returns the recommended pages which their friends prefer but are not preferred by the given user.
In what follows, we explain this task in more detail. 

\newpara{Schema and integrity constraint.}
Figure~\ref{fig:overview-query} shows the database schema $\schema$ with the two relations: (i) \texttt{Friendship} relation $F$ has two attributes $uid$ and $fid$, denoting the ID of each user and their friends, and (ii) \texttt{Likes} relation $L$ with two attributes $id$ and $pid$ which denote the user ID and the preferred page ID.
The integrity constraint $\constraint$ for this task 
specifies that the pair $(uid, fid)$ is the primary key of $F$, and $id$ is the primary key of $L$.



\newpara{Queries.}
Consider the user with $id = 1$. 
Figure~\ref{fig:overview-query} shows a query $\query_1$ that can solve the task for this given user. 
In particular, $\query_1$ uses an {\circled{1}} $\sqlinnerjoin$  to find the given user's friends, and uses a {\circled{2}} nested query followed by a {\circled{3}} $\sqlnotin$ filter to rule out the given user's preferred pages. 
Figure~\ref{fig:overview-query} also shows another query $\query_2$ which is similar to $\query_1$ but uses a {\circled{5}} $\sqlleftjoin$ instead.
Note that $\query_2$ is \emph{not} equivalent to $\query_1$ --- consider the case where the given user does not have any friends (i.e., $F$ does not have a tuple with $uid=1$). 
In this case, the $\sqlinnerjoin$ in $\query_1$ returns an empty result as the join condition $F.uid=1$ never holds, whereas the $\sqlleftjoin$ from $\query_2$ gives a product tuple $t'$ for each tuple $t$ in $L$ and $t'.pid=\sqlnull$. The \sqlnull values on $pid$ will eventually be projected out by $\query_2$, leading to (incorrect) \sqlnull tuples in the result.

\vspace{5pt}
In what follows, let us explain how $\tool$ proves the non-equivalence of $(\query_1, \query_2)$. 
Specifically, we will explain the key concepts in Figure~\ref{fig:workflow}: 
(1) what does a symbolic database look like, 
(2) how to encode the integrity constraint, 
(3) how to encode the query semantics, 
(4) how to check two queries always produce identical outputs, and
(5) how to obtain verification results.

\newpara{Generating symbolic database.}
Figure~\ref{fig:overview-db} shows the symbolic DB generated by $\tool$ that has \emph{up to} 2 rows in each relation (i.e., the input bound is 2). Each row is called a \emph{symbolic tuple}. 
For instance, we denote the tuples in $F$ by $t_1, t_2$ and denote the attribute values by $x_1 = t_1.uid$ and $x_2 = t_1.fid$ and etc.
Similarly, the tuples in $L$ are $t_3, t_4$ and the values for $id$ and $pid$ are $y_1, y_2$ and $y_3, y_4$, respectively.
In general, we use an \emph{uninterpreted} predicate $\del$ over tuples to indicate whether a tuple is deleted after an operation. Since the value of $\del(t_i)$ is unspecified, whether $t_i$ is deleted or not is non-deterministic. Thus, the initial symbolic database encodes all databases where each relation has \emph{at most} two tuples.


\newpara{Encoding integrity constraint.}
Figure~\ref{fig:overview-ic} shows $\tool$'s SMT encoding $\formula_\constraint$ of the integrity constraint $\constraint$. 
It has two parts. 
The first part 
\setlength{\abovedisplayskip}{3pt}
\setlength{\belowdisplayskip}{3pt}
\[
x_1 \neq \nullv \land x_2 \neq \nullv \land x_3 \neq \nullv \land x_4 \neq \nullv \land \neg (x_1 = x_3 \land x_2 = x_4)
\]
specifies that tuples $t_1, t_2$ are unique and all attributes are not null, since $(uid, fid)$ is a primary key. 
The second part $y_1 \neq \nullv \land y_3 \neq \nullv \land  y_1 \neq y_3$
encodes that $id$ is a primary key of relation $L$.

\vspace{-3pt}
\newpara{Encoding query semantics.}
To encode the semantics of a query, $\tool$ encodes the semantics of each operator which are then composed to form the encoding of the entire query. 
For example, the formula $\formula_1$ to encode $\query_1$ is the conjunction of $\formula_{R_1}$, $\formula_{R_2}$, $\formula_{R_3}$ and $\formula_{R_4}$.
Specifically, to encode the inner join $T_1 \ijoin T_2$ in \circled{1}, \tool considers the Cartesian product of tuples in $T_1$ and $T_2$, followed by a filter to set all resulted tuples that do not satisfy $T_1.fid = T_2.id$ and $T_1.id = 1$ as deleted. 
For the \sqlnotin in $\query_1$, it analyzes the nested query and generates the formula $\formula_{R_2}$ (as shown in \circled{2}), and checks the membership of $T_1.pid$ with the obtained $pid$, which is encoded as $\formula_{R_3}$ in \circled{3}.
Finally, the $pid$ is projected out by \circled{4}, resulting in the output $R_4$ and a formula $\formula_{R_4}$.
The obtained formula $\formula_1 = \formula_{R_1} \land \formula_{R_2} \land \formula_{R_3} \land \formula_{R_4}$ precisely encodes the query semantics of $\query_1$.
Similarly, $\query_2$'s output is $R_8$ and the corresponding formula $\formula_2$ consists of $\formula_{R_5}$, $\formula_{R_6}$, $\formula_{R_7}$ and $\formula_{R_8}$ from \circled{5} \circled{6} \circled{7} \circled{8}, respectively.

\noindent \textbf{\emph{Checking equality.}}
After encoding the semantics of queries $\query_1$ and $\query_2$, \tool needs to compare their outputs $R_4$ and $R_8$. Although $R_4$ and $R_8$ do not have the same number of symbolic tuples, they can still be equal because some tuples may have been deleted. To check the equality, \tool generates a formula $\formula_{\Leftrightarrow}$ (in Figure~\ref{fig:overview-checking-equality}) asserting $R_4$ is equal to $R_8$.
$\tool$ supports both list and bag semantics. In this example, it uses bag semantics because neither of the queries involves sorting operations.
$\formula_{\Leftrightarrow}$ encodes two properties: (1) $R_4$ and $R_8$ have the same number of non-deleted tuples and (2) for each non-deleted tuple in $R_4$, its multiplicity in $R_4$ is the same as that in $R_8$.
\revision{
Since properties (1) and (2) imply that (3) for each non-deleted tuple in $R_8$, its multiplicity in $R_8$ is the same as that in $R_4$, we do not need to include a formula for property (3) in $\formula_{\Leftrightarrow}$.
}

\newpara{Verification result.}
Finally, \tool builds a formula $\formula = \formula_{\constraint} \land \formula_1 \land \formula_2 \land \neg \formula_{\Leftrightarrow}$ encoding the existence of a database $\db$ such that (1) $\db$ satisfies the integrity constraint $\constraint$ and (2) $\query_1$ and $\query_2$ have different outputs on $\db$.
$\tool$ invokes Z3 and finds $\formula$ is satisfiable, so it concludes that $\query_1$ and $\query_2$ are not equivalent and generates a counterexample database (shown in Figure~\ref{fig:overview-results}) where $\query_1$ and $\query_2$ indeed yield different results.

\section{Problem Formulation} \label{sec:prelim}

\revision{
We first describe the preliminaries of relational schema and database, and then present our query language and integrity constraints, followed by a formal problem statement. 
}

\subsection{Relational Schema and Database}

\begin{figure}
\small
\centering

\begin{subfigure}[b]{0.45\textwidth}
\centering
\[
\begin{array}{rcl}
\text{Schema} & :: & \text{RelName} \to \text{RelSchema} \\
\text{RelSchema} & :: & [(\text{AttrName}, \text{Type})] \\
\text{Type} & \in & \set{\text{Int}, \text{Bool}} \\
\end{array}
\]
\vspace{-10pt}
\caption{Schema}
\label{fig:schema}
\end{subfigure}
    \hfill
\begin{subfigure}[b]{0.45\textwidth}
\centering
\[
\begin{array}{rcl}
\text{Database} & :: & \text{RelName} \to \text{Relation} \\
\text{Relation} & :: & [\text{Tuple}] \\
\text{Tuple} & :: & [(\text{AttrName}, \text{Value})] \\
\text{Value} & \in & \textbf{Int} \cup \textbf{Bool} \cup \set{\text{Null}} \\
\end{array}
\]
\vspace{-10pt}
\caption{Database}
\label{fig:db}
\end{subfigure}

\vspace{-10pt}
\caption{Relational schema and database.}
\label{fig:schema-db}
\vspace{-15pt}
\end{figure}

\vspace{-3pt}
\newpara{Schema.}
\revision{
As shown in Figure~\ref{fig:schema}, a relational database schema is a mapping from relation names to their relation schemas, where each relation schema is a list of attributes. Each attribute is represented by a pair of attribute name and type. All attribute names are assumed to be globally unique, which can be easily enforced by using fully qualified attribute names that include relation names, e.g., \texttt{EMP.age}. We only consider two primitive types, namely Int and Bool, in this paper.
Other types (e.g., strings and dates) can be treated as Int and functions involving those types can be treated as uninterpreted functions over Int.
}

\newpara{Database.}
\revision{
Similarly, as shown in Figure~\ref{fig:db}, a relational database is a mapping from relation names to their corresponding relations, where each relation consists of a list of tuples.\footnote{\revision{
If a relation is created by a non-\sqlorderby query, then we interpret this list as a bag. }} 
Each tuple contains a list of values with corresponding attribute names, and a value can be an integer, bool, or \sqlnull. 
}

\subsection{Syntax of SQL Queries}

\begin{figure}[!t]
\footnotesize 
\centering
\[
\hspace{-15pt}
\begin{array}{rcl}
\text{Query}~Q_r & ::= & \revision{Q ~|~ \text{OrderBy}(Q, \vec{E}, b)} \\
\text{Subquery}~Q & ::= & R ~|~ \proj_L(Q) ~|~ \filter_\phi(Q) ~|~ \revision{\rename_R(Q)} ~|~ Q \allcoll Q ~|~ \text{Distinct}(Q) ~|~ Q \alljoin Q
                 ~|~ \text{GroupBy}(Q, \vec{E}, L, \phi)  \\
                & | & \revision{\text{With}(\vec{Q}, \vec{R}, Q)}  \\
\text{Attr List}~L & ::= & \revision{id(A)} ~|~ \revision{\rename_a(A)} ~|~ L, L \\
\text{Attr}~A & ::= & {\text{Cast}(\phi)} ~|~ E ~|~ \mathcal{G}(E) ~|~ {A \allarith A} \\
\text{Pred}~\phi & ::= & b ~|~ \nullv ~|~ {A \alllogic A} ~|~ \text{IsNull}(E) ~|~ \vec{E} \in \vec{v} ~|~ \vec{E} \in Q
                 ~|~  \phi \land \phi ~|~ \phi \lor \phi ~|~ \neg \phi \\
\text{Expr}~E & ::= & a ~|~ v ~|~ E \allarith E ~|~ \text{ITE}(\phi, E, E) ~|~ \text{Case}(\vec{\phi}, \vec{E}, E) \\ 
\text{Join Op}~\alljoin & ::= & \times ~|~ \ijoin_\phi ~|~ \ljoin_\phi ~|~ \rjoin_\phi ~|~ \fjoin_\phi \\
\text{Collection Op}~\allcoll & ::= & \cup ~|~ \cap ~|~ \setminus ~|~ \uplus ~|~ \cplus ~|~ - \\
\text{Arith Op}~\allarith & ::= & + ~|~ - ~|~ \times ~|~ / ~|~ \% \\
\text{Logic Op}~\alllogic & ::= & \leq ~|~ < ~|~ = ~|~ \neq ~|~ > ~|~ \geq \\
\end{array}
\]
\[
\begin{array}{c}
\revision{R \in \textbf{Relation Names} \quad a \in \textbf{Attribute Names} \quad v \in \textbf{Values}} \quad b \in \textbf{Bools} \quad 
\revision{\mathcal{G} \in \{\text{Count}, \text{Min}, \text{Max}, \text{Sum}, \text{Avg}\}}
\end{array}
\]
\vspace{-10pt}
\caption{Syntax of SQL Queries. \revision{\textbf{Values} include integers, bools, and \nullv. $id(A)$ is a construct denoting attribute $A$ itself occurs in the attribute list. If the context is clear, we may also omit the $id$ constructor for brevity. {The $\text{Cast}$ function takes as input a predicate $\phi$ and returns \nullv if $\phi$ is $\nullv$, 1 if $\phi$ is $\top$, and 0 otherwise.}} }
\label{fig:syntax-sql}
\vspace{-10pt}
\end{figure}

The syntax of our query language is shown in Figure~\ref{fig:syntax-sql}, which covers various practical SQL operators,
\revision{
including projection $\proj$, selection $\filter$, renaming $\rename$, set union $\cup$, intersection $\cap$, minus $\setminus$, bag union $\uplus$, intersection $\cplus$, minus $-$, Distinct, Cartesian product $\times$, inner join $\ijoin_\phi$, left outer join $\ljoin_\phi$, right outer join $\rjoin_\phi$, full outer join $\fjoin_\phi$, GroupBy with Having clauses, With clauses, and OrderBy.
In addition, the query language also supports various attribute expressions such as arithmetic expressions $E_1 \allarith E_2$, aggregate functions $\mathcal{G}(E)$, if-then-else expressions $\text{ITE}(\phi, E_1, E_2)$, and case expressions $\text{Case}(\vec{\phi}, \vec{E}, E')$, as well as predicates such as logical comparison $A_1 \alllogic A_2$, null checks $\text{IsNull}(E)$, and membership check $\vec{E} \in Q$.
Many of these constructs are not supported by prior work, such as OrderBy, With, Case, and so on.
In what follows, we discuss the syntax of aggregate functions, GroupBy, and OrderBy in more detail.
}


\newpara{Aggregate functions and GroupBy.}
\revision{
Our language includes five common aggregate functions $\aggr$, and as is standard, it does not permit nested aggregate functions such as \texttt{SUM(SUM(a))}.
The language also has a construct $\text{GroupBy}(Q, \vec{E}, L, \phi)$ for grouping tuples. Intuitively, $\text{GroupBy}(Q, \vec{E}, L, \phi)$ groups the tuples of subquery $Q$ based on attribute expressions $\vec{E}$ and for each group satisfying condition $\phi$, it computes an aggregate tuple according to the attribute list $L$.
}

\begin{example}
The GroupBy operation is closely related to the \sqlgroupby and \sqlhaving clauses in standard SQL. For instance, consider a relation \texttt{EMP(id, gender, age, sal)} and a SQL query:
\begin{small}
\[
\sqlselectcolor \sqlspace \sqlavgcolor(\texttt{age}), \sqlspace \sqlavgcolor(\texttt{sal}) \sqlspace \sqlfromcolor \sqlspace \texttt{EMP} \sqlspace \sqlwherecolor \sqlspace \texttt{age} > 20 \sqlspace \sqlgroupbycolor \sqlspace \texttt{gender} \sqlspace \sqlhavingcolor \sqlspace \sqlavgcolor(\texttt{sal}) > 30000
\]
\end{small}
It can be represented by the following query in our language
\[
\text{GroupBy}(\filter_{\texttt{age}>20}(\texttt{EMP}), [\texttt{gender}], [\text{Avg}(\texttt{age}), \text{Avg}(\texttt{sal})], \text{Avg}(\texttt{sal})>30000)
\]
\end{example}

\revision{
It is worthwhile to point out that the grammar in Figure~\ref{fig:syntax-sql} does not precisely capture all syntactic requirements on valid queries. In particular, a query with certain syntax errors may also be accepted by the language, because those syntactic requirements are difficult to describe at the grammar level. We instead perform static analysis to check if a query is well-formed and throw errors if the query contains invalid expressions. For example, consider the following SQL query
}
\begin{small}
\[
\sqlselectcolor \sqlspace \sqlsumcolor(\texttt{sal}) \sqlspace \sqlfromcolor \sqlspace \texttt{EMP} \sqlspace \sqlgroupbycolor \sqlspace \texttt{age} \sqlspace \sqlhavingcolor \sqlspace \texttt{sal} > 10000
\]
\end{small}
where the \sqlhaving clause uses a non-aggregated attribute \texttt{sal} that is not in the \sqlgroupby list. Such a query is not permitted in standard SQL, because it may contain ambiguity where a group of tuples with the same \texttt{age} may have different \texttt{sal} values. Our query language also disallows such GroupBy queries. Although a simple syntactic check can reveal this error, describing the check in grammar is non-trivial, so we perform static analysis to reject such queries.

\newpara{Sorting and OrderBy.}
\revision{
Our language supports sorting the query result. Specifically, $\text{OrderBy}(Q, \vec{E}, b)$ can sort the result of subquery $Q$ according to a list of attribute expressions $\vec{E}$. The result is in ascending order if $b = \top$ and in descending order otherwise.
Note that as shown in Figure~\ref{fig:syntax-sql}, OrderBy is only allowed to be used at the topmost level of a query.
}

\subsection{Semantics of SQL Queries}

\begin{figure}[!t]
\scriptsize

\begin{mdframed}
\[
\denot{Q} :: \text{Database } \db \to \text{Relation}
\]
\end{mdframed}

\[
\begin{array}{lcl}

\denot{R}_\db & = & \db(R) \\
\denot{\proj_L(Q)}_\db & = & \site(\textsf{hasAgg}(L), [\denot{L}_{\db, \denot{Q}_\db}], \smap(\denot{Q}_\db, \lambda x. \denot{L}_{\db, x})) \\
\denot{\filter_\phi(Q)}_\db & = & \sfilter(\denot{Q}_\db, \lambda x. \denot{\phi}_{\db, [x]} = \top) \\
\denot{\rename_R{(Q)}}_\db & = & \smap(\denot{Q}_\db, \lambda x. \smap(x, \lambda (n, v). (\srename(R, n), v))) \\
\denot{Q_1 \cap Q_2}_\db & = & \sfilter(\denot{\text{Distinct}(Q_1)}_\db, \lambda x. x \in \denot{Q_2}_\db) \\
\denot{Q_1 \cup  Q_2}_\db & = & \denot{\text{Distinct}(Q_1 \uplus Q_2)}_\db \\
\denot{Q_1 \setminus Q_2}_\db & = & \sfilter(\denot{\text{Distinct}(Q_1)}_\db, \lambda x. x \not \in \denot{Q_2}_\db) \\
\denot{Q_1 \cplus Q_2}_\db & = & \denot{Q_1 - (Q_1 - Q_2)}_\db \\
\denot{Q_1 \uplus Q_2}_\db & = & \sappend(\denot{Q_1}_{\db}, \denot{Q_2}_{\db}) \\
\denot{Q_1 - Q_2}_\db & = & \sfoldl(\lambda xs. \lambda x. \site(x \in xs, xs - x, xs), \denot{Q_1}_\db, \denot{Q_2}_\db) \\
\denot{\text{Distinct}(Q)}_\db & = & \sfoldr(\lambda x. \lambda xs. \scons(x, \sfilter(xs, \lambda y. x \neq y)), [], \denot{Q}_\db) \\

\denot{Q_1 \times Q_2}_\db & = & \sfoldl(\lambda xs. \lambda x. \sappend(xs, \smap(\denot{Q_2}_\db, \lambda y. \smerge(x, y))), [], \denot{Q_1}_\db) \\
\denot{Q_1 \ijoin_\phi Q_2}_\db & = & \denot{\filter_\phi(Q_1 \times Q_2)}_\db \\
\denot{Q_1 \ljoin_\phi Q_2}_\db & = & \sfoldl(\lambda xs. \lambda x. \sappend(xs, \site(|v_1(x)| = 0, v_2(x), v_1(x)), [],\denot{Q_1}_\db) \\
& & \quad \text{where}~ v_1(x) = \denot{[x] \ijoin_\phi Q_2}_\db ~\text{and}~ v_2(x) = [\smerge(x, T_{\nullv})] \\
\denot{Q_1 \rjoin_\phi Q_2}_\db & = & \sfoldl(\lambda xs. \lambda x. \sappend(xs, \site(|v_1(x)| = 0, v_2(x), v_1(x)), [],\denot{Q_2}_\db) \\
& & \quad \text{where}~ v_1(x) = \denot{Q_1 \ijoin_\phi [x]}_\db ~\text{and}~ v_2(x) = [\smerge(T_{\nullv}, x)] \\
\denot{Q_1 \fjoin_\phi Q_2}_\db & = & \sappend(\denot{Q_1 \ljoin_\phi Q_2}_\db, \smap(xs, \lambda x. \smerge(T_{\nullv}, x)) ) \\
& & \quad \text{where}~ xs = \sfilter(\denot{Q_2}_\db, \lambda y. |\denot{Q_1 \ijoin_\phi [y]}_\db| = 0 ) \\

\denot{\text{GroupBy}(Q, \vec{E}, L, \phi)}_\db \!\!\!\! & = & 
\smap(
\sfilter(Gs, \lambda xs. \denot{\phi}_{\db, xs} = \top),
\lambda xs. \denot{L}_{\db, xs}
)
~\text{where} \\
& & \quad Gs = \smap(\text{Dedup}(Q, \vec{E}), \lambda y. \sfilter(\denot{Q}_\db, \lambda z. \text{Eval}(\vec{E}, [z]) = y)), \\
& & \quad \text{Dedup}(Q, \vec{E}) = \sfoldr(\lambda x. \lambda xs. \scons(x, \sfilter(xs, \lambda y. x \neq y)), [],
\smap(\denot{Q}_\db, \lambda z. \text{Eval}(\vec{E}, [z]))),\\
& & \quad \text{Eval}(\vec{E}, xs) = \smap(\vec{E}, \lambda e. \denot{e}_{\db, xs}) \\
\denot{\text{With}(\vec{Q}, \vec{R}, Q)}_\db & = & 
\denot{Q}_{\db'} ~\text{where}~ \db' = \db[R_i \mapsto \denot{Q_i}_{\db} ~|~ R_i \in \vec{R}] \\
\denot{\text{OrderBy}(Q, \vec{E}, b)}_\db & = & \sfoldl(\lambda xs. \lambda \_. (\sappend(xs, [\text{MinTuple}(\vec{E}, b, \denot{Q}_\db - xs)])), [], \denot{Q}_\db) ~\text{where} \\
& & \quad \text{MinTuple}(\vec{E}, b, xs) = \sfoldl(\lambda x. \lambda y. \site(\text{Cmp}(\vec{E}, b, x, y), y, x), \shead(xs), xs), \\
& & \quad \text{Cmp}(\vec{E}, b, x_1, x_2) = b \neq \sfoldr(\lambda E_i. \lambda y. \site(v(x_1, E_i) < v(x_2, E_i), \top, \\
& & \quad \quad \quad \quad \quad \quad \quad \quad \quad \quad \quad \quad \quad \quad \quad \quad \quad \quad \quad \site(v(x_1, E_i) > v(x_2, E_i), \bot, y)), \top, \vec{E}),\\
& & \quad v(x, E) = \site(\denot{E}_{\db, [x]} = \nullv, -\infty, \denot{E}_{\db, [x]}) \\
\end{array}
\]
\vspace{-10pt}
\caption{Formal semantics of SQL queries. \site is the standard if-then-else function. $\textsf{hasAgg}(L)$ checks if the attribute list $L$ has an attribute computing an aggregation. $\smerge(x, y)$ merges two tuples $x$ and $y$ into one tuple. $\srename(R, n)$ replaces all occurrences of the original table name $R$ with the new globally unique name $n$. $T_{\nullv}$ represents a tuple of \nullv's, whose length is determined as appropriate by the context for brevity.
{Note that $xs - x$ on the right-hand side of $\denot{Q_1 - Q_2}_\db$ denotes deleting one tuple $x$ in $xs$ iff $xs \in x$.}
Full formal semantics of all constructs in our query language, including predicates $\denot{\phi}_{\db, xs}$, attribute lists $\denot{L}_{\db, xs}$, and expressions $\denot{e}_{\db, xs}$ is available in \appx{Appendix~\ref{sec:full-semantics}}{the Appendix of the extended version~\cite{extended}} .}
\label{fig:semantics-query}
\vspace{-10pt}
\end{figure}

\revision{
The denotational semantics of our SQL queries is presented in Figure~\ref{fig:semantics-query}, where $\denot{Q}$ takes as an input a database $\db$ and produces as output a relation.
}

\newpara{Bag and list semantics.}
\revision{
As is standard in SQL, we conceptually view a relation as a bag (multiset) of tuples and use the bag semantics for queries. However, our denotational semantics in Figure~\ref{fig:semantics-query} uses lists to implement bags and define query operators through higher-order combinators such as \smap, \sfilter, and \sfoldl. In this way, we can easily define complex query operators and switch to list semantics when the query needs to sort the result using \sqlorderby.
}

\newpara{Three-valued semantics.}
\revision{
Similar to standard SQL, our semantics also uses three-valued logic.
}
More specifically, a relation may use \sqlnull values to represent unknown information, and all query operators should behave correctly with respect to the \sqlnull value.
\revision{
For example, our semantics considers a predicate can be evaluated to three possible values, namely $\top$ (\textsf{true}), $\bot$ (\textsf{false}), and Null. A notation like $\denot{\phi}_{\db, xs} = \top$ means the predicate $\phi$ evaluates to \textsf{true} (not \textsf{false} nor Null) given database $\db$ and tuples $xs$.
\footnote{Detailed definition of predicate semantics is available in \appx{Appendix~\ref{sec:full-semantics}}{the Appendix of the extended version~\cite{extended}}.}
}

\newpara{Basics.}
\revision{
At a high level, our denotational semantics can be viewed as a function or a functional program $\denot{Q}_{\db}$ that pattern matches different constructors of query $Q$ and evaluates them to relations given a concrete database $\db$. If the query is a relation name $R$, then $\denot{R}_{\db}$ simply looks up the name $R$ in $\db$. If the query is a projection, $\denot{\proj_L(Q)}_{\db}$ first checks if the attribute list $L$ includes aggregation functions. If so, it evaluates query $Q$ and invokes $\denot{L}_{\db, \denot{Q}_{\db}}$ to compute the aggregate values. Otherwise, it projects each tuple in the result of $\denot{Q}_{\db}$ according to $L$ through a standard \smap~ combinator. If the query is renaming, $\denot{\rename_R(Q)}_{\db}$ first evaluates $\denot{Q}_{\db}$. Then for each tuple in the result, it updates the attribute name $n$ to be a new name $\textsf{rename}(R, n)$ that is related to relation name $R$. Following similar ideas of functional programming, we can define semantics for filtering $\filter$, set operations $\cap, \cup, \setminus$, bag operations $\cplus, \uplus, -$, and Distinct for removing duplicated tuples.
}

\newpara{Cartesian product and joins.}
Apart from the standard Cartesian product $Q_1 \times Q_2$, there are four different joins in our language.
\revision{
An \emph{inner join} $Q_1 \ijoin_\phi Q_2$ can be viewed as a syntactic sugar of $\filter_{\phi}(Q_1 \times Q_2)$.
}
A \emph{left outer join} $Q_1 \ljoin_\phi Q_2$ is more involved.
\revision{
$\denot{Q_1 \ljoin Q_2}_{\db}$ iterates tuples in the result of $\denot{Q_1}_{\db}$ and for each tuple $x$, it computes the inner join of $[x]$ and $Q_2$ and denotes the result by $v_1(x)$. If $v_1(x)$ is not empty, it is appended to the final result of the left outer join. Otherwise, $\denot{Q_1 \ljoin Q_2}_{\db}$ computes a \emph{null extension} $v_2(x)$ of $x$ by taking the Cartesian product between $x$ and $T_{\textsf{Null}}$ (a tuple of Null values) and appends $v_2(x)$ to the result.
Similarly, a \emph{right outer join} $Q_1 \rjoin_\phi Q_2$ can be defined in the same way.
For \emph{full outer join}, $\denot{Q_1 \fjoin_\phi Q_2}_{\db}$ starts with the left outer join $Q_1 \ljoin_\phi Q_2$. Then for each tuple $y$ in $\denot{Q_2}_{\db}$, if the inner join between $Q_1$ and $[y]$ returns empty, $\denot{Q_1 \fjoin_\phi Q_2}_{\db}$ adds a null extension $\smerge(T_{\text{Null}}, y)$ to the result.
}

\begin{example}
Consider the following two relations \texttt{EMP} and \texttt{DEPT}. What follows illustrates the difference in the results of various joins and their Cartesian product.
\begin{figure}[H]
\tiny
\begin{subfigure}[b]{0.2\textwidth}
\centering
\begin{tabular}{|c|c|c|}
\hline
eid & ename & did \\
\hline
1 & A & 11 \\
\hline
2 & B & 12 \\
\hline
\end{tabular}
\vspace{-5pt}
\caption*{\texttt{EMP}}
\end{subfigure}
~
\begin{subfigure}[b]{0.15\textwidth}
\centering
\begin{tabular}{|c|c|}
\hline
id & dname \\
\hline
10 & C \\
\hline
11 & D \\
\hline
\end{tabular}
\vspace{-5pt}
\caption*{\texttt{DEPT}}
\end{subfigure}
~
\begin{subfigure}[b]{0.3\textwidth}
\centering
\begin{tabular}{|c|c|c|c|c|}
\hline
eid & ename & did & id & dname \\
\hline
1 & A & 11 & 11 & D \\
\hline
\end{tabular}
\vspace{-5pt}
\caption*{$\texttt{EMP} \ijoin_{\texttt{did=id}} \texttt{DEPT}$}
\end{subfigure}
~
\begin{subfigure}[b]{0.3\textwidth}
\centering
\begin{tabular}{|c|c|c|c|c|}
\hline
eid & ename & did & id & dname \\
\hline
1 & A & 11 & 11 & D \\
\hline
2 & B & 12 & \sqlnull & \sqlnull \\
\hline
\end{tabular}
\vspace{-5pt}
\caption*{$\texttt{EMP} \ljoin_{\texttt{did=id}} \texttt{DEPT}$}
\end{subfigure}

\vspace{5pt}

\begin{subfigure}[b]{0.3\textwidth}
\centering
\begin{tabular}{|c|c|c|c|c|}
\hline
eid & ename & did & id & dname \\
\hline
\sqlnull & \sqlnull & \sqlnull & 10 & C \\
\hline
1 & A & 11 & 11 & D \\
\hline
\end{tabular}
\vspace{-5pt}
\caption*{$\texttt{EMP} \rjoin_{\texttt{did=id}} \texttt{DEPT}$}
\end{subfigure}
~
\begin{subfigure}[b]{0.35\textwidth}
\centering
\begin{tabular}{|c|c|c|c|c|}
\hline
eid & ename & did & id & dname \\
\hline
1 & A & 11 & 11 & D \\
\hline
2 & B & 12 & \sqlnull & \sqlnull \\
\hline
\sqlnull & \sqlnull & \sqlnull & 10 & C \\
\hline
\end{tabular}
\vspace{-5pt}
\caption*{$\texttt{EMP} \fjoin_{\texttt{did=id}} \texttt{DEPT}$}
\end{subfigure}
~
\begin{subfigure}[b]{0.35\textwidth}
\centering
\begin{tabular}{|c|c|c|c|c|}
\hline
eid & ename & did & id & dname \\
\hline
1 & A & 11 & 10 & C \\
\hline
1 & A & 11 & 11 & D \\
\hline
2 & B & 12 & 10 & C \\
\hline
2 & B & 12 & 11 & D \\
\hline
\end{tabular}
\vspace{-5pt}
\caption*{$\texttt{EMP} \times \texttt{DEPT}$}
\end{subfigure}

\end{figure}

\end{example}

\newpara{GroupBy.}
\revision{
To define the semantics of $\text{GroupBy}(Q, \vec{E}, L, \phi)$, we use several auxiliary functions. Specifically, $\text{Eval}(\vec{E}, xs)$ computes the values of expressions $\vec{E}$ to be grouped over tuple list $xs$, and $\text{Dedup}(Q, \vec{E})$ invokes $\text{Eval}$ to retain only unique values of $\vec{E}$. Then GroupBy maps each unique $\vec{E}$ value to a group of tuples sharing the same value over $\vec{E}$ and obtains all groups $Gs$. It evaluates the \sqlhaving condition $\phi$ and only retains those groups satisfying $\phi$ by the \sfilter combinator. Finally, for each retained group $xs$, GroupBy computes an aggregate value $\denot{L}_{\db, xs}$ and adds it to the result. 
}

\newpara{With clause.}
\revision{
To evaluate $\text{With}(\vec{Q}, \vec{R}, Q)$, our semantics first evaluates all subqueries $Q_i \in \vec{Q}$, then creates a new database $\db'$ by adding mappings $R_i \mapsto \denot{Q_i}_{\db}$, and finally evaluates $Q$ under $\db'$. Intuitively, the With clause creates local bindings for subqueries $Q_i$ which can be used in query $Q$.
}

\newpara{OrderBy.}
\revision{
We define a deterministic semantics for the OrderBy construct in our language, i.e., two runs of OrderBy with the same arguments always return the same result. Specifically, $\text{OrderBy}(Q, \vec{E}, \\$$b)$ performs a selection sort on the results of $\denot{Q}_{\db}$ by expressions $\vec{E}$. The sorted results are in ascending order if $b = \top$ and in descending order if $b = \bot$. At a high level, OrderBy maintains a list of sorted tuples in ascending (resp. descending) order and repeatedly selects the minimum (resp. maximum) tuple from the unsorted tuples using the MinTuple function. In particular, MinTuple invokes the Cmp function for pair-wise comparison of two tuples $x_1$ and $x_2$ over expression list $\vec{E}$. Null is viewed as a special value $-\infty$ that is smaller than any other values, so Null's can also be ordered correctly by our semantics.
}

\newpara{Our language vs. prior work's.}
\revision{
Our formal semantics is inspired by prior work such as \mediator~\cite{mediator}. While \mediator defines semantics for simple update and query operators, our semantics defines more complex query operators, including outer joins, \sqlgroupby, \sqlwith, and \sqlorderby. Furthermore, to our best knowledge, our semantics considers many operators that are important in practice but not considered in any prior work such as $\sqlorderby$, $\sqlwith$, $\sqlif$, $\sqlintersect$, $\sqlexcept$, etc.
}
In addition, our work supports complex attribute expressions and predicates, such as \texttt{Avg(age\:+\:10)}, whereas prior work has very limited support for such features. 

\subsection{Integrity Constraints}

Integrity constraints are fundamental to data integrity and therefore must be followed by a database. In this paper, we support the constraints shown in Figure~\ref{fig:syntax-integrity-constraints}. In particular, an integrity constraint consists of a set of primitive constraints, detailed as follows.

\newpara{Primary keys.}
$\text{PK}(R, \vec{a})$ says that a list of attributes $\vec{a}$ is the primary key of relation $R$. In particular, it requires that all values of any attribute in $\vec{a}$ are not \sqlnull. Furthermore, given two different tuples $t_1, t_2 \in R$, $t_1$ and $t_2$ must have different values on at least one attribute in $\vec{a}$.

\newpara{Foreign keys.}
$\text{FK}(R_1, a_1, R_2, a_2)$ means that the attribute $a_1$ of relation $R_1$ is a foreign key referencing the attribute $a_2$ of relation $R_2$. Specifically, it requires that, for each tuple $t_1 \in R_1$, there exists a tuple $t_2 \in R_2$ such that $t_1.a_1 = t_2.a_2$.

\begin{figure}[!t]
\centering
\[
\begin{array}{rcl}
\text{Constraint}~\constraint & ::= & \text{PK}(R, \vec{a}) ~|~ \text{FK}(R, a, R, a) ~|~ \text{NotNull}(R, a) ~|~ \text{Check}(R, \psi) ~|~ \text{Inc}(R, a, v) ~|~ \constraint \land \constraint \\
\text{Pred}~\psi & ::= & a \odot v ~|~ a \odot a ~|~ a \in \vec{v} ~|~ \psi \land \psi ~|~ \psi \lor \psi ~|~ \neg \psi \\
\text{Logic Op}~\odot & ::= & \leq ~|~ < ~|~ = ~|~ \neq ~|~ > ~|~ \geq
\end{array}
\]
\[
R \in \textbf{Relations} \qquad a \in \textbf{Attributes} \qquad v \in \textbf{Values}
\]
\vspace{-15pt}
\caption{Syntax of Integrity Constraints.}
\label{fig:syntax-integrity-constraints}
\vspace{-10pt}
\end{figure}

\newpara{Not-null constraints.}
NotNull(R, a) imposes a \sqlnotnull constraint on the attribute $a$ of relation $R$. It stipulates that for each tuple $t \in R$, $t.a$ is not \sqlnull.

\newpara{Check constraints.}
$\text{Check}(R, \psi)$ requires that every tuple $t \in R$ must satisfy the predicate $\psi$, where $\psi$ can be a boolean combination of atomic predicates. Each atomic predicate can be logical comparisons between two attributes of $R$, between an attribute and a constant value. or of the form $a \in \vec{v}$ meaning the value of attribute $a$ is in a list of provided constants $\vec{v}$.

\newpara{Auto increment.}
$\text{Inc}(R, a, v)$ means that the value of attribute $a$ in relation $R$ starts with $v$ and strictly increases by one for each new tuple added to $R$.
\revision{
We require the value to strictly increase by one to obtain a deterministic semantics of the auto increment constraint. However, this is not a fundamental limitation. We can easily generalize to the scenarios where values are not continuous.
}

\subsection{Problem Statement}



\revision{
To describe our problem statement, we start with a notion of conformance between a relational database and its schema.
Specifically, an attribute value $(a, v)$ conforms to an attribute schema $(b, \tau)$, denoted $(a, v) \conforms_A (b, \tau)$, if $a = b$ and $v$ is of type $\tau$.
\footnote{\revision{Null is viewed as a polymorphic value of both type Int and Bool.}}
A tuple $t$ conforms to a relation schema $RS$, denoted $t \conforms_T RS$, if (1) they have the same size, i.e. $|t| = |RS|$ and (2) each attribute value $(a_i, v_i)$ in $T$ is conforming to the corresponding attribute schema $(b_i, \tau_i)$ in $RS$.
}

\begin{definition}[Conformance between Database and Schema]
\label{def:conformance}
\revision{
A database $\db$ conforms to a schema $\schema$, denoted $\db \conforms \schema$, if (1) $\db$ and $\schema$ have the same domain and (2) for each relation name $R \in \dom(\db)$, all tuples in $\db(R)$ conform to their corresponding relation schema $\schema(R)$, i.e.,
\[
\db \conforms \schema \defeq \dom(\db) = \dom(\schema) \land (\forall R \in \dom(\db). \forall t \in R.~ t \conforms_T \schema(R))
\]
}
\end{definition}

\revision{
Recall from Figure~\ref{fig:schema-db} that we only consider Int and Bool types in this paper. Other types of values (e.g., strings and dates) can be treated as integers and functions involving those types can be treated as uninterpreted functions over integers.
}

\begin{example}\label{ex:conformance}
\revision{
Let schema $\schema = [\texttt{EMP} \mapsto [(\texttt{eid}, Int), (\texttt{cid}, Int)], \texttt{CAR} \mapsto [(\texttt{id}, Int), (\texttt{used}, Bool)]]$ and database
\[
\begin{array}{rl}
\db = [ & \texttt{EMP} \mapsto [[(\texttt{eid}, 1), (\texttt{cid}, 10)], [(\texttt{eid}, 2), (\texttt{cid}, \sqlnull)]], \\
        & \texttt{CAR} \mapsto [(\texttt{id}, 10), (\texttt{used}, \top)] \quad] \\
\end{array}
\]
Here, $\db \conforms \schema$ because $[(\texttt{eid}, 1), (\texttt{cid}, 10)] \conforms_T [(\texttt{eid}, Int), (\texttt{cid}, Int)]$ and $[(\texttt{eid}, 2), (\texttt{cid}, \sqlnull)] \conforms_T [(\texttt{eid}, Int), (\texttt{cid}, Int)]$ hold for relation name \texttt{EMP} and $[(\texttt{id}, 10), (\texttt{used}, \top)] \conforms_T [(\texttt{id}, Int), (\texttt{used}, \\$$Bool)]$ holds for relation name $\texttt{CAR}$.
}
\end{example}

\begin{definition}[Bounded Equivalence modulo Integrity Constraint]
\label{def:BEMIC}
\revision{
Given two queries $(Q_1, Q_2)$ under schema $\schema$ and a positive integer bound $\bound$, $Q_1$ and $Q_2$
}
are said to be \emph{bounded equivalent modulo integrity constraint} $\constraint$, denoted by \revision{$\query_1 \simeq_{\schema, \constraint, \bound} \query_2$}, if for any database $\db$ that conforms to schema $\schema$ satisfies $\constraint$ and each relation has at most $\bound$ tuples, the execution result of $\query_1$ is the same as the execution result of $\query_2$ on $\db$, i.e.,
\revision{
\[
\query_1 \simeq_{\schema, \constraint, \bound} \query_2 \defeq \forall \db.~ \db \conforms \schema \land (\forall R \in \db.~ |R| \leq \bound) \land \constraint(\db) \Rightarrow \denot{Q_1}_{\db} = \denot{Q_2}_{\db}
\]
}
\end{definition}

\section{Bounded Equivalence Verification Modulo Integrity Constraints} \label{sec:checking}

This section presents our algorithm for bounded equivalence verification of two queries modulo integrity constraints.

\subsection{Algorithm Overview}

The top-level algorithm of our verification technique is shown in Algorithm~\ref{algo:verify}. The \textsc{Verify} procedure takes as input two queries $\query_1$ and $\query_2$ over schema $\schema$, an integrity constraint $\constraint$, and a bound $\bound$ on the size of all relations in the database. It returns $\top$ indicating $\query_1$ and $\query_2$ are bounded equivalent modulo integrity constraint $\constraint$, i.e., \revision{$\query_1 \simeq_{\schema, \constraint, \bound} \query_2$}. Otherwise, it returns a counterexample database satisfying $\constraint$ where $\query_1$ and $\query_2$ yield different results.

\begin{figure}[!t]
\small 
\vspace{-5pt}
\begin{algorithm}[H]
\caption{Equivalence Verification}
\label{algo:verify}
\begin{algorithmic}[1]
\Procedure{\textsc{Verify}}{$\query_1, \query_2, \schema, \constraint, \bound$}
\vspace{2pt}
\Statex \textbf{Input:} Queries $\query_1, \query_2$, schema $\schema$, integrity constraint $\constraint$, and bound size $\bound$
\Statex \textbf{Output:} $\top$ for equivalence, otherwise a counterexample 
\vspace{2pt}

\State $\context \gets \textsf{BuildSymbolicDB}(\schema, \bound)$;
\State $\formula_{\constraint} \gets \textsc{EncodeConstraint}(\context, \constraint)$;
\State $\formula_{R_1}, R_1 \gets \textsc{EncodeQuery}(\schema, \context, \query_1)$;
\State $\formula_{R_2}, R_2 \gets \textsc{EncodeQuery}(\schema, \context, \query_2)$;
\State $\formula \gets \formula_{\constraint} \land \formula_{R_1} \land \formula_{R_2} \land \neg \textsc{Equal}(R_1, R_2)$;
\If{\textsf{UNSAT}($\formula$)} ~ \Return ~ $\top$;
\Else ~ \Return \textsf{BuildExample}($\schema$, \textsf{Model}($\formula$));
\EndIf

\EndProcedure
\end{algorithmic}
\end{algorithm}
\vspace{-20pt}
\end{figure}

At a high-level, our technique reduces the verification problem into a constraint-solving problem and generates an SMT formula through the encoding of integrity constraints and SQL operations.
Specifically, we first create a symbolic representation of the database $\context$, where each relation has at most $\bound$ tuples (Line 2).
Then we encode the integrity constraint $\constraint$ over $\context$ and obtain an SMT formula $\formula_{\constraint}$ (Line 3).
Next, we analyze queries $\query_1, \query_2$ to obtain their outputs $R_1, R_2$ given input $\context$ and two SMT formulas $\formula_{R_1}, \formula_{R_2}$ encoding how $R_1, R_2$ are computed from $\context$ (Lines 4 -- 5).
Finally, we build a formula $\formula$ asserting the existence of a database satisfying the integrity constraint $\constraint$ such that $R_1$ is different from $R_2$ (Line 6). If no such database exists (i.e., formula $\Phi$ is unsatisfiable), then $\query_1$ and $\query_2$ are bounded equivalent modulo integrity constraint $\constraint$, i.e., \revision{$\query_1 \simeq_{\schema, \constraint, \bound} \query_2$} (Line 7). Otherwise, if such a database exists, $\query_1$ and $\query_2$ are not equivalent, so we build a counterexample database from a model of $\formula$ to disprove the equivalence (Line 8).

\subsection{Schema and Symbolic Database}

Since our verification technique is centered around a symbolic representation of the database, we first describe how to build the symbolic database.



\newpara{Symbolic database.}
Given a schema $\schema$ and a bound $\bound$ for the size of relations, we build a symbolic database containing all relations in $\dom(\schema)$ and each relation has $\bound$ symbolic tuples. We denote the symbolic database by $\context$ where \revision{$\context: \text{RelName} \to [\text{SymTuple}]$ maps relation names} to their corresponding lists of symbolic tuples.
In general, we introduce an uninterpreted predicate $\del(t)$ for each tuple $t$ to indicate whether or not $t$ is deleted by subquery.
\footnote{\revision{
Alternatively, we can introduce a predicate $\textsf{Present}(t)$ to indicate a tuple is indeed present, i.e., $\textsf{Present}(t) \Leftrightarrow \neg \del(t)$.
}}
Since the $\del$ predicates hold non-deterministic values in the symbolic database, we can use $\context$ to encode all possible databases where each relation has \emph{at most} $\bound$ tuples.

\begin{example}
Consider again the schema $\schema$ in Example~\ref{ex:conformance}. Given a bound $\bound = 2$, we can build a symbolic database $\context = \set{ \texttt{EMP} \mapsto [t_1, t_2], \texttt{CAR} \mapsto [t_3, t_4] }$, where $t_1, t_2, t_3, t_4$ are symbolic tuples.
\end{example}

\newpara{Encoding attributes.}
\revision{
As is standard, we use uninterpreted functions to encode attributes. Specifically, for each attribute $attr$ in the database schema, we introduce an uninterpreted function called $attr$ that takes as input a symbolic tuple and produces as output a symbolic value. For example, $t.name$ for getting the $name$ attribute of tuple $t$ should be encoded as $name(t)$ where $name$ is an uninterpreted function.
}

\newpara{Encoding \sqlnull.}
\revision{
To support \sqlnull and three-valued semantics, we encode each symbolic value as a pair $(b, v)$ in the SMT formula, where $b$ is a boolean variable indicating whether the value is \sqlnull, and $v$ is a non-\sqlnull value. In particular, if $b$ is $\top$ (i.e., \textsf{true}), then the value $(b, v)$ is \sqlnull. Otherwise, if $b$ is $\bot$ (i.e., \textsf{false}), the value is $v$. For example, constant 1 is represented by $(\bot, 1)$. In this way, \sqlnull is not equal to any legitimate non-\sqlnull values.
}
Furthermore, we consider all the null values are equal. Although we allow different representations of null values like $(\top, 1)$ and $(\top, 2)$, they are considered equal because both of them are considered to be \sqlnull.

\subsection{Encoding Integrity Constraints}

\begin{figure}[!t]
\footnotesize
\[
\begin{array}{c}

\irulelabel
{\begin{array}{c}
\context(R) = [t_1, \ldots, t_n] \quad
\formula_1 = \revision{\land_{i=1}^{n}} \land_{k=1}^{m} t_i.a_k \neq \nullv \\
m = |\vec{a}| \quad
\formula_2 = \land_{i=1}^{n} \land_{j=i+1}^{n} \neg (\land_{k=1}^{m} t_i.a_k = t_j.a_k) \\
\end{array}}
{\context \vdash \text{PK}(R, \vec{a}) \rightsquigarrow \formula_1 \land \formula_2}
{\textrm{(IC-PK)}}

\irulelabel
{\begin{array}{c}
\context(R_1) = [t_1, \ldots, t_n] \\
\context(R_2) = [t'_1, \ldots, t'_m] \\
\formula = \land_{i=1}^{n} \lor_{j=1}^{m} t_i.a_1 = t'_j.a_2 \\
\end{array}}
{\context \vdash \text{FK}(R_1, a_1, R_2, a_2) \rightsquigarrow \formula}
{\textrm{(IC-FK)}}

\\ \ \\

\irulelabel
{\begin{array}{c}
\context(R) = [t_1, \ldots, t_n] \quad
\formula = \land_{i=1}^{n} t_i.a \neq \nullv \\
\end{array}}
{\context \vdash \text{NotNull}(R, a) \rightsquigarrow \formula}
{\textrm{(IC-NN)}}

\irulelabel
{\begin{array}{c}
\context(R) = [t_1, \ldots, t_n] \quad
\formula = \land_{i=1}^{n} \denot{\psi}_{t_i} \\
\end{array}}
{\context \vdash \text{Check}(R, \psi) \rightsquigarrow \formula}
{\textrm{(IC-Check)}}

\\ \ \\

\irulelabel
{\begin{array}{c}
\context(R) = [t_1, \ldots, t_n] \\
\formula_1 = \land_{i=1}^{n} t_i.a = v+i-1 \quad
\formula_2 = \land_{i=1}^{n} t_i.a \neq \nullv \\
\end{array}}
{\context \vdash \text{Inc}(R, a, v) \rightsquigarrow \formula_1 \land \formula_2}
{\textrm{(IC-Inc)}}

\irulelabel
{\begin{array}{c}
\context \vdash \constraint_1 \rightsquigarrow \formula_1 \\
\context \vdash \constraint_2 \rightsquigarrow \formula_2 \\
\end{array}}
{\context \vdash \constraint_1 \land \constraint_2 \rightsquigarrow \formula_1 \land \formula_2}
{\textrm{(IC-Comp)}}

\end{array}
\]
\vspace{-10pt}
\caption{Symbolic encoding of integrity constraints.}
\label{fig:rules-constraints}
\vspace{-10pt}
\end{figure}

\begin{figure}[!t]
\small







\begin{tabular}{rcl c rcl}
$\denot{a \alllogic v}_{t}$ & = & $t.a \neq \nullv \land t.a \alllogic v$  & \hspace{50pt}  & $\denot{\psi_1 \land \psi_2}_{t}$  & =  & $\denot{\psi_1}_{t} \land \denot{\psi_2}_{t}$ \\
$\denot{a_1 \in \vec{v}}_{t}$ & = & $t.a \neq \nullv \land \bigvee_{i=1}^{|\vec{v}|} t.a = v_i$  &   & $\denot{\psi_1 \lor \psi_2}_{t}$  & =  & $\denot{\psi_1}_{t} \lor \denot{\psi_2}_{t}$ \\
$\denot{a_1 \alllogic a_2}_{t}$ & = & $t.a_1 \alllogic t.a_2 ~\text{where}~ \alllogic \in \{=, \neq\}$  &   & $\denot{\neg \psi}_{t}$  & =  & $\neg \denot{\psi}_{t}$ \\
$\denot{a_1 \alllogic a_2}_{t}$ & = & \multicolumn{5}{l}{$t.a_1 \neq \nullv \land t.a_2 \neq \nullv \land t.a_1 \alllogic t.a_2 ~\text{where}~ \alllogic \in \{\leq, <, >, \geq \}$}
\end{tabular}

\vspace{-10pt}
\caption{Symbolic encoding of predicates in integrity constraints.}
\label{fig:rules-ic-preds}
\vspace{-10pt}
\end{figure}

Given a symbolic database $\context$ and an integrity constraint $\constraint$, we follow the semantics of $\constraint$ to encode $\constraint$ as an SMT formula over $\context$. The encoding procedure is summarized as inference rules in Figure~\ref{fig:rules-constraints}, where judgments of the form $\context \vdash \constraint \rightsquigarrow \formula$ represent that the encoding of integrity constraint $\constraint$ is $\formula$ given a symbolic database $\context$.

In a nutshell, we encode each atomic integrity constraint in $\constraint$ and conjoin the formulas together according to the IC-Comp rule.
Specifically, the IC-PK rule specifies that the encoding of a primary key constraint $\text{PK}(R, a)$ consists of two parts: $\formula_1$ asserting all attributes in the primary key have no \sqlnull values, and $\formula_2$ stating for any pair of tuples $t_1$ and $t_2$ where $t_1 \neq t_2$, they do not agree on all attributes in $\vec{a}$.
For $\text{FK}(R_1, a_1, R_2, a_2)$ where $R_1.a_1$ is a foreign key referencing $R_2.a_2$, the IC-FK rule looks up the symbolic database $\context$ and finds the tuples for $R_1$ are $t_1, \ldots, t_n$ and the tuples for $R_2$ are $t'_1, \ldots, t'_m$. The formula asserts that for each tuple $t_i$ in $R_1$, there exists a tuple $t'_j$ in $R_2$ such that the value $t_i.a_1$ is equal to $t'_j.a_2$.
The IC-NN rule simply encodes that $\text{NotNull}(R, a)$ requires that all tuples in R must have a non-\sqlnull value on attribute a.
For the constraint $\text{Check}(R, \psi)$ that specifies ranges of values, the IC-Check rule uses an auxiliary function $\denot{\psi}_t$ (shown in Figure~\ref{fig:rules-ic-preds}) to compute the formula of predicate $\psi$ on each tuple $t_i$ in $R$ and obtain the formula by conjoining the $\denot{\psi}_t$ formulas together.
For the auto-increment constraint $\text{Inc}(R, a, v)$, the IC-Inc rule also looks up the symbolic database $\context$ and finds all tuples $t_1, \ldots, t_n$ of $R$. It then enforces the value of $t_i.a$ is not \sqlnull and that $t_i.a = v+i-1$.

\begin{lemma}\label{lem:ic}
\footnote{The proof of all lemmas and theorems can be found in \appx{Appendix~\ref{sec:proof}}{the Appendix of the extended version~\cite{extended}}.}
Given a symbolic database $\context$ and an integrity constraint $\constraint$, consider a formula $\formula$ such that $\context \vdash \constraint \rightsquigarrow \formula$. If $\formula$ is satisfiable, then the model of $\formula$ corresponds to a database consistent with $\context$ that satisfies $\constraint$. If $\formula$ is unsatisfiable, then no database consistent with $\context$ satisfies $\constraint$.
\end{lemma}

\subsection{Encoding SQL Queries} \label{sec:encode-sql}

To encode the semantics of a query, we traverse the query and encode how each operator transforms its input to output. Since this process requires the attributes and symbolic tuples of each subquery, we first describe how to compute attributes and tuples of arbitrary intermediate subqueries, followed by the encoding of all query operators.

\subsubsection{Attributes of Intermediate Subqueries} \hfill
\vspace{8pt}

\begin{figure}[!t]
\footnotesize
\[
\hspace{-10pt}
\begin{array}{c}




\irulelabel
{R \in dom(\schema)}
{\schema \vdash R : \schema(R)}
{\textrm{(A-Rel)}}

\irulelabel
{\begin{array}{c}
\schema \vdash \query : \attributes \\
\forall e \in L. \forall a \in \text{Attrs}(e).~ a \in \attributes \\
\end{array}}
{\schema \vdash \Pi_L(\query) : L}
{\textrm{(A-Proj)}}

\irulelabel
{\begin{array}{c}
\schema \vdash \query_1 : \attributes_1 \quad
\schema \vdash \query_2 : \attributes_2 \\
\alljoin \in \set{\times, \bowtie_\phi, \leftouterjoin_\phi, \rightouterjoin_\phi, \fullouterjoin_\phi} \\
\end{array}}
{\schema \vdash \query_1 \alljoin \query_2 : \attributes_1 \doubleplus \attributes_2}
{\textrm{(A-Join)}}

\\ \ \\

\irulelabel
{\begin{array}{c}
\schema \vdash \query : \attributes \quad
\attributes' = [\revision{\text{rename}(R, a)} ~|~ a \in \attributes] \\
\end{array}}
{\schema \vdash \revision{\rho_{R}}(\query) : \attributes'}
{\textrm{(A-Rename)}}

\irulelabel
{\begin{array}{c}
\schema \vdash \query : \attributes \\
\end{array}}
{\schema \vdash \filter_\phi(\query) : \attributes}
{\textrm{(A-Filter)}}

\\ \ \\

\irulelabel
{\begin{array}{c}
\schema \vdash \query : \attributes \\
\end{array}}
{\schema \vdash \text{GroupBy}(\query, \vec{E}, L, \phi) : \attributes}
{\textrm{(A-Group)}}

\irulelabel
{\begin{array}{c}
\schema \vdash \query : \attributes \\
\end{array}}
{\schema \vdash \text{OrderBy}(\query, \vec{E}, b) : \attributes}
{\textrm{(A-Order)}}

\irulelabel
{\begin{array}{c}
\schema \vdash \query : \attributes \\
\end{array}}
{\schema \vdash \text{Distinct}(\query) : \attributes}
{\textrm{(A-Dist)}}

\\ \ \\

\irulelabel
{\begin{array}{c}
\allcoll \in \set{\cup, \cap, \setminus, \uplus, \cplus, -} \\
\schema \vdash \query_1 : \attributes \quad
\schema \vdash \query_2 : \attributes \\
\end{array}}
{\schema \vdash \query_1 \allcoll \query_2 : \attributes}
{\textrm{(A-Coll)}}

\irulelabel
{\begin{array}{c}
\schema \vdash \query_i : \attributes_i \quad
1 \leq i \leq n \quad
n = |\vec{\query}| = \revision{|\vec{R}|}\\
\revision{\schema[R_1 \mapsto \attributes_1, \ldots, R_n \mapsto \attributes_n]} \vdash \query' : \attributes'
\end{array}}
{\schema \vdash \text{With}(\vec{\query}, \revision{\vec{R}}, \query') : \attributes'}
{\textrm{(A-With)}}



\end{array}
\]
\vspace{-10pt}
\caption{Rules for inferring the attributes of a query result. As is standard in SQL, the name of an attribute expression is assumed to be a function over all attributes involved in the expression. For example, the name of attribute expression $\text{Avg}(a)$ is a string \texttt{Avg\_a}.}
\label{fig:rules-schema}
\vspace{-10pt}
\end{figure}

\noindent
While the database schema describes the attributes of each relation in the \emph{initial} database, it does not directly give the attributes of those intermediate query results. To infer the attributes of each intermediate subquery, we develop an algorithm for all query operations in our SQL language. The algorithm is summarized as a set of inference rules, shown in Figure~\ref{fig:rules-schema}. Intuitively, our rules for inferring attributes are similar to the traditional typing rules. The judgments of the form $\schema \vdash \query : \attributes$ mean the attributes of query $\query$ is $\attributes$ under schema $\schema$.

Specifically, the attributes of a relation in the initial database can be obtained by looking up the schema $\schema$ directly (A-Rel).
To find all attributes of a projection $\proj_L(\query)$, the A-Proj rule first computes the attribute list $\attributes$ of query $\query$ and then checks if all attributes occurring in the attribute expression list $L$ belong to $\attributes$. If so, the attributes of $\proj_L(\query)$ are $L$ with all attribute expressions converted to the corresponding attribute names, e.g., \textsf{Avg(a)} to \texttt{Avg\_a}. Otherwise, there is a type error in the query.
For Cartesian product or any join operator $\alljoin$, the attributes of $\query_1 \alljoin \query_2$ are obtained by concatenating the attributes of $\query_1$ and those of $\query_2$ (A-Join).
The renaming operation \revision{$\rename_R(\query)$} first computes the attributes of $\query$ and then renames each of them according to the new relation name $R$ in the result (A-Rename).
Based on the A-Filter, A-Group, A-Order, and A-Dist rules, the filtering, GroupBy, OrderBy, and Distinct operations do not change the attribute list. 
The A-Coll rule specifies that for any collection operator $\allcoll$, the attributes of its two operands $\query_1$ and $\query_2$ must be the same, which are also identical to the attributes of $\query_1 \allcoll \query_2$.
The inference rule for a \sqlwith clause \revision{$\text{With}(\vec{\query}, \vec{R}, \query')$} is slightly more involved. Since the query $\query'$ can use \revision{$R_i$} to refer to the results of a subquery $\query_i$, the A-With rule first infers the attributes $\attributes_i$ of each subquery $\query_i$ and then augments the schema $\schema$ to a temporary new schema $\schema'$ by adding entries \revision{$R_i \mapsto \attributes_i$}. Finally, the attributes of query $\query'$ (and the whole \sqlwith clause) are inferred based on the schema $\schema'$.

\subsubsection{Symbolic Tuples of Intermediate Subqueries.} \hfill
\vspace{8pt}

\begin{figure}[!t]
\footnotesize
\[
\begin{array}{c}





\irulelabel
{\begin{array}{c}
R \in dom(\context) \\
\end{array}}
{\context \vdash R \hookrightarrow \context(R)}
{\textrm{(T-Rel)}}

\irulelabel
{\begin{array}{c}
    \text{hasAgg}(L) \quad
\end{array}}
{\context \vdash \Pi_L(\query) \hookrightarrow \fresh(1)}
{\textrm{(T-Agg)}}

\irulelabel
{\begin{array}{c}
    \neg \text{hasAgg}(L) \quad
    \context \vdash Q \hookrightarrow \tuples \\
\end{array}}
{\context \vdash \Pi_L(\query) \hookrightarrow \fresh(|\tuples|)}
{\textrm{(T-Proj)}}

\\ \ \\

\irulelabel
{\begin{array}{c}
\context \vdash \query \hookrightarrow \tuples \\
\end{array}}
{\context \vdash \filter_\phi(\query) \hookrightarrow \fresh(|\tuples|)}
{\textrm{(T-Filter)}}

\irulelabel
{\begin{array}{c}
\context \vdash \query \hookrightarrow \tuples \\
\end{array}}
{\context \vdash \revision{\rho_R(Q)} \hookrightarrow \fresh(|\tuples|)}
{\textrm{(T-Rename)}}

\irulelabel
{\begin{array}{c}
\context \vdash \query_1 \hookrightarrow \tuples_1 \quad
\context \vdash \query_2 \hookrightarrow \tuples_2 \\
\end{array}}
{\context \vdash \query_1 \times \query_2 \hookrightarrow \fresh(|\tuples_1| \cdot |\tuples_2|)}
{\textrm{(T-Prod)}}

\\ \ \\

\irulelabel
{\begin{array}{c}
\context \vdash \filter_{\phi}(\query_1 \times \query_2) \hookrightarrow \tuples \\
\end{array}}
{\context \vdash \query_1 \bowtie_\phi \query_2 \hookrightarrow \tuples}
{\textrm{(T-IJoin)}}

\irulelabel
{\begin{array}{c}
\context \vdash \query_1 \bowtie_\phi \query_2 \hookrightarrow \tuples_1 \\
\context \vdash \query_1 \hookrightarrow \tuples_2 \\
\tuples' = \tuples_1 \doubleplus \fresh(|\tuples_2|) \\
\end{array}}
{\context \vdash \query_1 \leftouterjoin_\phi \query_2 \hookrightarrow \tuples'}
{\textrm{(T-LJoin)}}

\irulelabel
{\begin{array}{c}
\context \vdash \query_1 \bowtie_\phi \query_2 \hookrightarrow \tuples_1 \\
\context \vdash \query_2 \hookrightarrow \tuples_2 \\
\tuples' = \tuples_1 \doubleplus \fresh(|\tuples_2|) \\
\end{array}}
{\context \vdash \query_1 \rightouterjoin_\phi \query_2 \hookrightarrow \tuples'}
{\textrm{(T-RJoin)}}

\\ \ \\

\irulelabel
{\begin{array}{c}
\context \vdash \query_1 \bowtie_\phi \query_2 \hookrightarrow \tuples_1 \quad
\context \vdash \query_1 \hookrightarrow \tuples_2 \quad
\context \vdash \query_2 \hookrightarrow \tuples_3 \\
\end{array}}
{\context \vdash \query_1 \fullouterjoin_\phi \query_2 \hookrightarrow 
\tuples_1 \doubleplus \fresh(|\tuples_2| + |\tuples_3|)}
{\textrm{(T-FJoin)}}

\irulelabel
{\begin{array}{c}
\context \vdash \query \hookrightarrow \tuples \\
\end{array}}
{\context \vdash \text{Distinct}(\query) \hookrightarrow \fresh(|\tuples|)}
{\textrm{(T-Dist)}}

\\ \ \\

\irulelabel
{\begin{array}{c}
\context \vdash \query \hookrightarrow \tuples \\
\end{array}}
{\context \vdash \text{GroupBy}(\query, \vec{E}, L, \phi) \hookrightarrow \fresh(|\tuples|)}
{\textrm{(T-GroupBy)}}

\irulelabel
{\begin{array}{c}
\context \vdash \query \hookrightarrow \tuples \\
\end{array}}
{\context \vdash \text{OrderBy}(\query, \vec{E}, b) \hookrightarrow \fresh(|\tuples|)}
{\textrm{(T-OrderBy)}}

\\ \ \\

\irulelabel
{\begin{array}{c}
\allcoll \in \set{\cap, \cplus} \quad
\context \vdash \query_1 \hookrightarrow \tuples_1 \\
\end{array}}
{\context \vdash \query_1 \allcoll \query_2 \hookrightarrow 
\fresh(|\tuples_1|)}
{\textrm{(T-Intx)}}

\irulelabel
{\begin{array}{c}
\allcoll \in \set{\setminus, -} \quad
\context \vdash \query_1 \hookrightarrow \tuples_1 \\
\end{array}}
{\context \vdash \query_1 \allcoll \query_2 \hookrightarrow \fresh(|\tuples_1|)}
{\textrm{(T-Ex)}}

\\ \ \\

\irulelabel
{\begin{array}{c}
\allcoll \in \set{\cup, \uplus} \\
\context \vdash \query_1 \hookrightarrow \tuples_1 \quad
\context \vdash \query_2 \hookrightarrow \tuples_2 \\
\end{array}}
{\context \vdash \query_1 \allcoll \query_2 \hookrightarrow \fresh(|\tuples_1| + |\tuples_2|)}
{\textrm{(T-Union)}}

\irulelabel
{\begin{array}{c}
\context \vdash \query_i \hookrightarrow \tuples_i \quad
1 \leq i \leq n \quad
n = |\vec{Q}| = \revision{|\vec{R}|} \\
\revision{\context[R_1 \mapsto \tuples_1, \ldots, R_n \mapsto \tuples_n]} \vdash \query' \hookrightarrow \tuples' \\
\end{array}}
{\context \vdash \text{With}(\vec{\query}, \revision{\vec{R}}, \query') \hookrightarrow \tuples'}
{\textrm{(T-With)}}

\\ \ \\

\end{array}
\]
\vspace{-10pt}
\caption{Rules for inferring the tuples of a query. $\fresh(n)$ generates a list of $n$ fresh tuples.}
\label{fig:rules-tuples}
\vspace{-10pt}
\end{figure}

\noindent
Similar to the database schema, a symbolic database $\context$ only describes those tuples in each relation of the \emph{initial} database, but it does not present what tuples are in the result of intermediate subqueries. To obtain the tuples in each intermediate subquery, we develop a set of inference rules as shown in Figure~\ref{fig:rules-tuples}, where judgments of the form $\context \vdash \query \hookrightarrow \tuples$ mean the result of query $\query$ has a list of symbolic tuples $\tuples$ given the initial database $\context$.

The goal of these rules is to compute the number (at most) of symbolic tuples in the result of each intermediate subquery and what are their names, so we just need to look up the tuples for relations in $\context$ (T-Rel) and generate fresh tuples for other queries.
Specifically, if a projection $\proj_L(\query)$ has aggregate functions in the attribute expression list $L$, there is only one tuple in the result (T-Agg). But if the projection $\proj_L(\query)$ does not have aggregate functions, the number of tuples in the result is the same as that in $\query$.
For filtering $\filter_\phi(\query)$, the T-Filter rule generates the same number of tuples as $\query$, because the predicate $\phi$ may not filter out any tuple from $\query$. The T-Rename rule specifies that the renaming operation \revision{$\rho_R(\query)$} produces the same number of tuples as $\query$. 
Given the tuples of $\query_1$ are $\tuples_1$ and the tuples of $\query_2$ are $\tuples_2$, the numbers of fresh tuples for the Cartesian product, inner join, left outer join, right outer join, and full outer join are $|\tuples_1| \cdot |\tuples_2|$, $|\tuples_1| \cdot |\tuples_2|$, $|\tuples_1| \cdot (|\tuples_2|+1)$, $(|\tuples_1|+1) \cdot |\tuples_2|$, and $|\tuples_1| \cdot |\tuples_2| + |\tuples_1| + |\tuples_2|$, respectively.
In addition, the numbers of fresh tuples for union, intersect, and except operations of $\query_1$ and $\query_2$ are $|\tuples_1| + |\tuples_2|$, $|\tuples_1|$, and $|\tuples_1|$, respectively. 
According to the T-Dist, T-GroupBy, and T-OrderBy rules, the Distinct, GroupBy, and OrderBy operation preserves the number of tuples in their subqueries.
Finally, to infer the tuples for \revision{$\text{With}(\vec{\query}, \vec{R}, \query')$}, we first need to obtain the tuples $\tuples_i$ for its subquery $\query_i$. Then we create a new symbolic database $\context'$ by adding mappings from tuples $\tuples_i$ to relation \revision{$R_i$} to $\context$. Finally, we infer the tuples $\tuples'$ for $\query'$ given the new symbolic database $\context'$ and get the result of \revision{$\text{With}(\vec{\query}, \vec{R}, \query')$} based on the T-With rule.

\subsubsection{Symbolic Encoding of Query Operators} \hfill
\vspace{8pt}

\begin{figure}[!t]
\[
\hspace{-12pt}
\scriptsize
\begin{array}{c}


\irulelabel
{\begin{array}{c}
    R \in \dom(\context) \\
\end{array}}
{\schema, \context \vdash R \leadsto \top}
{\textrm{(E-Rel)}}

\irulelabel
{\begin{array}{c}
    \neg \text{hasAgg}(L) \quad
    \schema, \context \vdash \query \leadsto \formula_1 \\

    \context \vdash \query \hookrightarrow [t_1, \ldots, t_n] \\

    \schema \vdash \Pi_{L}(\query) : [a_1', \ldots, a_l'] \quad 
    \context \vdash \Pi_{L}(\query) \hookrightarrow [t_1', \ldots, t_n'] \\


    \formula_2 = \land_{i=1}^{n} (\land_{j=1}^{l} 
    \denot{a_j'}_{\schema, \context, [t_{i}']} = 
    \denot{\revision{a_j'}}_{\schema, \context, [t_{i}]} \land \del(t_i') \leftrightarrow \del(t_{i})) \\
\end{array}}
{\schema, \context \vdash \Pi_{L}(\query) \leadsto \formula_1 \land \formula_2}
{\hspace{-1pt}\textrm{(E-Proj)}}







\\ \ \\

\irulelabel
{\begin{array}{c}
    \context \vdash \query \hookrightarrow [t_1, \ldots, t_n] \quad
    \context \vdash \filter_\phi(\query) \hookrightarrow [t_1', \ldots, t_n'] \quad
    \schema, \context \vdash \query \leadsto \formula_1 \\
    \formula_2 = \land_{i=1}^{n} (
    (\neg \del(t_i) \land \denot{\phi}_{\schema, \context, [t_i]} = \top) \to t_{i}' =  t_{i} 
    \land
    (\del(t_i) \lor \denot{\phi}_{\schema, \context, [t_i]} \neq \top) \to \del(t_{i}'))\\
\end{array}}
{\schema, \context \vdash \filter_\phi(\query) \leadsto \formula_1 \land \formula_2}
{\textrm{(E-Filter)}}

\\ \ \\

\irulelabel
{\begin{array}{l}
    \schema \vdash \query_1 : [a_1, \ldots, a_{p}] \quad
    \context \vdash \query_1 \hookrightarrow [t_1, \ldots, t_{n}] \quad
    \schema, \context \vdash \query_1 \leadsto \formula_1 \\
    
    \schema \vdash \query_2 : [a_1', \ldots, a_{q}'] \quad
    \context \vdash \query_2 \hookrightarrow [t_1', \ldots, t_{m}'] \quad
    \schema, \context \vdash \query_2 \leadsto \formula_2 \quad
    
    \context \vdash \query_1 \times \query_2 \hookrightarrow [t_{1,1}'', \ldots, t_{n,m}''] 
    \\
    
    \formula_3 = \land_{i=1}^{n} \land_{j=1}^{m} (\neg \del(t_i) \land \neg \del(t_j')) \to \land_{k=1}^{p} \denot{a_{k}}_{\schema, \context, [t_{i,j}'']} = \denot{a_{k}}_{\schema, \context, [t_{i}]} \land \land_{k=1}^{q} \denot{a_{k}'}_{\schema, \context, [t_{i,j}'']} = \denot{a_{k}'}_{\schema,\context,[t_{j}']} \land \\
    \qquad (\del(t_i) \lor \del(t_j')) \to \del(t_{i,j}'')\\
\end{array}}
{\schema, \context \vdash \query_1 \times \query_2 \leadsto \formula_1 \land \formula_2 \land \formula_3}
{\textrm{(E-Prod)}}
\\ \ \\



    
    

\end{array}
\]
\vspace{-10pt}
\caption{
Sample inference rules for encoding SQL queries. $\del(t)$ indicates that a tuple $t$ is deleted.
}
\label{fig:rules-encode}
\vspace{-10pt}
\end{figure}

\noindent
Since our query language supports various complex attribute expressions and predicates, we introduce two auxiliary functions $\denot{e}_{\schema, \context, \tuples}$ and $\denot{\phi}_{\schema, \context, \tuples}$ to encode attribute expressions and predicates in a query, respectively\footnote{We precisely define these auxiliary functions in \appx{Appendix~\ref{sec:full-semantics}}{the appendix of the extended version~\cite{extended}}.}. Intuitively, $\denot{\cdot}_{\schema, \context, \tuples}$ evaluates an attribute expression or a predicate in a recursive fashion given the schema $\schema$, database $\context$, and a list of tuples $\tuples$. For example, suppose $a$ is an attribute of an relation and $t$ is a tuple of that relation, $\denot{a + 1}_{\schema, \context, [t]} = \denot{a}_{\schema, \context, [t]} + \denot{1}_{\schema, \context, [t]} = t.a + 1$. {Furthermore, it is worthwhile to point out that the tuple list $\tuples$ has more than one tuple in several cases when evaluating \sqlgroupby and aggregate queries. For instance, $\denot{\text{AVG}(\text{EMP.age})}_{\schema, \context, [t_1, t_2, ..., t_n]}$ is used to compute the average age of EMP where the EMP table has $n$ symbolic tuples.}

Our encoding algorithm for query operators is represented as a set of inference rules. Judgments are of the form $\schema, \context \vdash \query \leadsto \formula$, meaning the encoding of query $\query$ is formula $\formula$ given schema $\schema$ and database $\context$.
Figure~\ref{fig:rules-encode} presents a sample set of such inference rules.

\noindent \textbf{\emph{Relation.}}
The encoding for a simple relation query is trivially $\top$, because the result can be obtained from the database $\context$ directly, i.e., the output is the same as input.

\newpara{Filtering.}
The E-Filter rule specifies how to encode filtering operations. In particular, given a query $\filter_{\phi}(\query)$, we can first determine the input (result of $\query$) is $[t_1, \ldots, t_n]$ and the output should be $[t'_1, \ldots, t'_n]$ based on the rules in Figure~\ref{fig:rules-tuples}. Then we can generate a formula $\formula_2$ to describe the relationship between $[t'_1, \ldots, t'_n]$ and $[t_1, \ldots, t_n]$. Specifically, if a tuple $t_i$ is not deleted and the predicate $\phi$ evaluates to be $\top$ on $t_i$, then it is retained in the result, i.e., $t'_i = t_i$. Otherwise, the corresponding output tuple $t'_i$ is deleted. The final formula encoding $\filter_{\phi}(\query)$ is the conjunction of $\formula_2$ and the formula $\formula_1$ that encodes $\query$.

\newpara{Projection.}
According to the E-Proj rule, to encode $\proj_L(\query)$ where $L$ does not contain aggregate functions, we can first infer that its input is $[t_1, \ldots, t_n]$ with attributes $[a_1, \ldots, a_m]$ and output is $[t'_1, \ldots, t'_n]$ with attributes $[a'_1, \ldots, a'_l]$ based on the rules in Figure~\ref{fig:rules-tuples} and Figure~\ref{fig:rules-schema}. Then for each $a'_k$, we find its corresponding index $c_k$ in the attributes of $\query$ and generate a formula $\formula_2$ that asserts for each tuple and its output $(t_i, t'_i)$, they have the same \del status and they agree on attribute $a'_k$.
The E-Agg rule for encoding projection with aggregate functions is similar to E-Proj. The main difference is that it only generates one tuple $t'_1$ in the output and sets the aggregated value based on all input tuples, i.e., $\denot{a_j'}_{\schema,\context,[t_1']} = \denot{\revision{a_j'}}_{\schema,\context,\vec{t}}$.

\newpara{Cartesian product.}
Based on the E-Prod rule to encode $\query_1 \times \query_2$, we just need to generate a formula $\formula_3$ that encodes an output tuple $t''_{i,j}$ is obtained by concatenating a tuple $t_i$ from $\query_1$ and a tuple $t'_j$ from $\query_2$. Specifically, $\formula_3$ describes the attributes of $t''_{i, j}$ agree with those of $t_i$ and $t'_j$, and $t''_{i,j}$ is deleted if either $t_i$ or $t'_j$ is deleted.

\newpara{Left outer join.}
The E-LJoin rule specifies the encoding of a left outer join $\query_1 \ljoin_\phi \query_2$ is based on the formula $\formula_1$ of inner join $\query_1 \ijoin_\phi \query_2$ (which is a syntactic sugar of $\filter_\phi(\query_1 \times \query_2)$). In addition to $\formula_1$, the encoding also includes $\formula_2$, which describes the output tuples from $\query_1$'s null extension.

\newpara{Other operations.}
The rules for encoding other query operations are in similar flavor of the above rules. Specifically, the high-level idea is to first obtain the schema and tuples for its input and output based on the rules in Figure~\ref{fig:rules-schema} and Figure~\ref{fig:rules-tuples}. Then we can encode the relationship between the input and output tuples and generate a formula that encodes the query semantics. Due to page limit, the complete set of our inference rules are presented in \appx{Appendix~\ref{sec:appendix-encoding}}{the Appendix of the extended version~\cite{extended}}.
\revision{
These rules closely follow the standard three-valued logic when evaluating expressions and predicates that involve \sqlnull's. For example, the predicate \sqlnull = \sqlnull evaluates to \sqlnull, which is consistent with the three-valued logic.
}

\begin{definition}[Interpretation]
\label{def:interp}
\revision{
An interpretation of a formula is a mapping from variables, function symbols, and predicate symbols in the formula to values, functions, and predicates, respectively.
}
\end{definition}

\begin{definition}[Interpretation Extension]
\label{def:interp_ext}
\revision{
Interpretation $\interpretation_1$ is an \emph{extension} of interpretation $\interpretation_2$, denoted $\interpretation_1 \extends \interpretation_2$, if (1) $\dom(\interpretation_1) \supseteq \dom(\interpretation_2)$, (2) $\forall v \in \vars(\interpretation_2).~ \interpretation_1(v) = \interpretation_2(v)$, and (3) for the $\del$ predicate, $\forall v \in \vars(\interpretation_1).~ \interpretation_2(\del)(v) \leftrightarrow \interpretation_1(\del)(v)$.
}
\end{definition}

\revision{
Intuitively, interpretation $\interpretation_1$ is an extension of $\interpretation_2$ if (1) all variables, functions, and predicates defined by $\interpretation_2$ are also defined by $\interpretation_1$, (2) $\interpretation_1$ preserves the values of variables occurring in $\interpretation_2$ while it assigns values to new variables, and (3) $\interpretation_1$ preserves the definition of $\del$ predicate on variables occurring in $\interpretation_2$.
}

\begin{theorem}\label{lem:query1}
\revision{
Let $\db$ be a database over schema $\schema$ and $o_R$ be the output relation of running query $\query$ over $\db$.
Consider a symbolic database $\context$ over $\schema$, a symbolic relation $R$ and a formula $\formula$ such that $\context \vdash \query \hookrightarrow R$ and $\schema, \context \vdash \query \leadsto \formula$. For any interpretation $\interpretation$ such that $\interpretation(\context) = \db$, there exists an extension $\interpretation'$ of $\interpretation$ such that (1) the result of query $\query$ over $\db$ is $\interpretation'(R)$, and (2) the database $\db$ and $\interpretation'$ jointly satisfy $\formula$, i.e.,
\footnote{\revision{
We slightly abuse the notation $\interpretation(x)$ over lists to apply $\interpretation$ to each element in list $x$. We also abuse the notation $\interpretation(x)$ over maps to apply $\interpretation$ to each value in the range of map $x$.
}}
\[
\begin{array}{l}
(\context \vdash \query \hookrightarrow R) \land (\schema, \context \vdash \query \leadsto \formula ) \land (\interpretation(\context) = \db) \Rightarrow \\
\qquad \qquad \qquad \qquad \qquad \qquad \qquad \qquad \exists \interpretation' \extends \interpretation.~ \interpretation'(R) = \denot{Q}_{\db} \land (\interpretation', \db[o_R \mapsto \interpretation'(R)] \models \formula) 
\end{array}
\]
}
\end{theorem}

\vspace{5pt}

\begin{theorem}\label{lem:query2}
\revision{
Let $\context$ be a symbolic database over schema $\schema$ and $\query$ be a query. Consider a symbolic relation $R$ and a formula $\formula$ such that $\context \vdash \query \hookrightarrow R$ and $\schema, \context \vdash \query \leadsto \formula$. If $\formula$ is satisfiable, then for any satisfying interpretation $\interpretation$ of $\formula$, running $\query$ over the concrete database $\interpretation(\context)$ yields the relation $\interpretation(R)$, i.e.,
\[
(\context \vdash \query \hookrightarrow R) \land (\schema, \context \vdash \query \leadsto \formula ) \land  (\interpretation \models \formula) \Rightarrow \denot{Q}_{\interpretation(\context)} = \interpretation(R)
\]
}
\end{theorem}

\vspace{3pt}

\subsection{Equivalence under Bag and List Semantics}

We first discuss how to check the equality of symbolic tuples and then present how to check the equality of query outputs under \emph{both} bag and list semantics. We use bag semantics by default, but for queries that involve sorting, we switch to list semantics to preserve the order of sorted tuples.

\newpara{Equality of symbolic tuples.}
Recall from Figure~\ref{fig:db} that a tuple is represented by a list of pairs, where each pair consists of an attribute name and its corresponding value.
Two symbolic tuples $t$ and $t'$ are considered equal iff (1) they have the same number of pairs and (2) their corresponding symbolic values at the same index are equal.
Specifically, suppose $t = [(a_1, v_1), (a_2, v_2), \ldots, (a_n, v_n)]$ and $t' = [(a'_1, v'_1), (a'_2, v'_2), \ldots, (a'_n, v'_n)]$ where $a_i, a'_i$ are the attribute names and $v_i, v'_i$ are their corresponding values. We say $t = t'$ if $v_i = v'_i$ holds for all $1 \leq i \leq n$. Note that the attribute names are ignored for equality comparison to support the renaming operations in SQL.

\newpara{Equivalence under bag semantics.}
Given two query results with symbolic tuples $R_1 = [t_1, \ldots, t_n]$ and $R_2 = [r_1, \ldots, r_m]$, we encode the following two properties to ensure $R_1$ and $R_2$ are equivalent under bag semantics:

\begin{itemize}[leftmargin=*]
\item
$R_1$ and $R_2$ have the same number of non-deleted tuples.
\begin{equation} \label{eq:bag1}
\vspace{-2pt}
\sum\limits_{i=1}^{n} \indicator{\neg \del(t_i)} = \sum\limits_{j=1}^{m} \indicator{\neg \del(r_j)}
\vspace{-2pt}
\end{equation}

where $\indicator{b}$ is an indicator function that evaluates to 1 when $b$ is true or 0 when $b$ is false.
\item
For each non-deleted symbolic tuple in $R_1$, its multiplicity in $R_1$ is equal to its multiplicity in $R_2$.
\begin{equation} \label{eq:bag2}
\vspace{-2pt}
\bigwedge\limits_{i=1}^{n} \Big( \sum\limits_{j=1}^{n} \indicator{\neg \del(t_i) \land \neg \del(t_j) \land t_i = t_j} = \sum\limits_{k=1}^{m} \indicator{\neg \del(t_i) \land \neg \del(r_k) \land t_i = r_k} \Big)
\vspace{-2pt}
\end{equation}
Here, we include the predicates $\neg \del(t_i)$, $\neg \del(t_j)$, and $\neg \del(r_k)$ to only count those non-deleted symbolic tuples in $R_1$ and $R_2$. 
\end{itemize}

\newpara{Equivalence under list semantics.}
Since the equivalence of two queries is only checked when they involve sorting operations, the corresponding query results $R_1 = [t_1, \ldots, t_n]$ and $R_2 = [r_1, \ldots, r_m]$ are sorted using the OrderBy operator.
\revision{
Recall from Figure~\ref{fig:syntax-sql} that OrderBy must be the last operation of a query. Furthermore, our encoding of OrderBy ensures that all deleted tuples are placed at the end of the list}
and non-deleted tuples are sorted in ascending or descending order, so we just need to encode the following two properties to assert $R_1$ and $R_2$ are equivalent under list semantics.

\begin{itemize}[leftmargin=*]
\item
$R_1$ and $R_2$ have the same number of non-deleted tuples.
\begin{equation} \label{eq:list1}
\vspace{-2pt}
\sum\limits_{i=1}^{n} \indicator{\neg \del(t_i)} = \sum\limits_{j=1}^{m} \indicator{\neg \del(r_j)}
\vspace{-2pt}
\end{equation}
\item
The tuples with the same index in $R_1$ and $R_2$ are equal.
\begin{equation} \label{eq:list2}
\vspace{-2pt}
\bigwedge\limits_{i=1}^{\text{min}\set{m,n}} t_i = r_i
\vspace{-2pt}
\end{equation}
\end{itemize}

Intuitively, we just need to check the pair-wise equality of tuples until the end of $R_1$ or $R_2$, because the OrderBy operator has moved the deleted tuples to the end of the list.
\revision{
This OrderBy encoding greatly simplifies the equality check. Alternatively, one can encode the check without the assumption of OrderBy moving deleted tuples to the end. The key idea is to maintain two pointers, one on each list, moving from the beginning to the end of the lists. Every time a pointer moves to a new location, it checks whether or not the corresponding tuple is deleted. If the tuple is not deleted, then it checks if the tuple is equal to the tuple pointed by the other pointer. Otherwise, if the tuple is deleted, it moves to the next location.
}

\begin{lemma}\label{lem:equal}
Given two relations $R_1 = [t_1, \ldots, t_n]$ and $R_2 = [r_1, \ldots, r_m]$, if the formula $(\ref{eq:bag1}) \land (\ref{eq:bag2})$ is valid, then $R_1$ is equal to $R_2$ under bag semantics. If the formula $(\ref{eq:list1}) \land (\ref{eq:list2})$ is valid, then $R_1$ is equal to $R_2$ under list semantics.
\end{lemma}

\begin{theorem}\label{thm:verify}
Given two queries $Q_1, Q_2$ under schema $\schema$, an integrity constraint $\constraint$, a bound $\bound$, if $\textsc{Verify}(\query_1, \query_2, \schema, \constraint, \bound)$ returns $\top$, then \revision{$\query_1 \simeq_{\schema, \constraint, \bound} \query_2$}. Otherwise, if $\textsc{Verify}(\query_1, \query_2, \schema, \constraint, \bound)$ returns a database $\db$, then \revision{$\db \conforms \schema$ and $\denot{Q_1}_{\db} \neq \denot{Q_2}_{\db}$}. 
\end{theorem}

\section{Implementation} \label{sec:impl}

Based on the techniques in Section~\ref{sec:checking}, we have implemented a tool called \tool for verifying the bounded equivalence of SQL queries modulo integrity constraints. \tool uses the Z3 SMT solver~\cite{z3-tacas08} for constraint solving and model generation. In this section, we discuss several details that are important to the implementation of \tool.

\newpara{Attribute renaming.}
We pre-process the queries to resolve attribute renaming issues before starting the verification. Specifically, we design a unique identifier generator and a name cache for pre-processing. The unique identifier generator allows \tool to efficiently track attributes while the name cache pool stores temporary aliases of an attribute under different scopes and automatically updates them for scope change.

\newpara{Sorting symbolic tuples.}
To sort symbolic tuples based on an \sqlorderby operation, we implement an encoding for the standard bubble sort algorithm. More specifically, we first move those deleted tuples to the end of the list, and then sort non-deleted tuples based on the sorting criteria specified by \sqlorderby. As part of the bubble sort algorithm, two adjacent tuples $t_1, t_2$ in the list are swapped under two conditions: (1) $t_1 \neq \nullv$ whereas $t_2 = \nullv$, or (2) $t_1, t_2$ are not Null but $t_1 < t_2$ for ascending order ($t_1 > t_2$ for descending order).



\section{Evaluation} \label{sec:eval}

This section presents a series of experiments that are designed to answer the following questions. 
\begin{itemize}[leftmargin=*]
\setlength\itemsep{3pt}
\item
\textbf{RQ1}: 
How does $\tool$ compare with state-of-the-art SQL equivalence checking techniques in terms of \emph{coverage}? 
That is, does $\tool$ support more complex queries? 
(Section~\ref{sec:comparison})
\item 
\textbf{RQ2:}
How effective is $\tool$ at \emph{disproving} query equivalence? 
Does $\tool$ generate useful counterexamples to facilitate downstream tasks? 
(Section~\ref{sec:counterexample})
\item 
\textbf{RQ3}: 
\revision{
How large input bounds can \tool reach during verification?
}
(Section~\ref{sec:exp-bound})
\end{itemize}

\subsection{Experimental Setup}


\newpara{Benchmarks.}
We collected benchmarks --- each of which is a pair of queries --- from {three} different workloads. To the best of our knowledge, we have incorporated \emph{all} benchmarks publicly available from recent work on query equivalence checking~\cite{chu2017cosette,chu2017demonstration,chu2017hottsql,chu2018axiomatic,zhou2019automated,zhou2022spes,wang2018speeding}. Additionally, we also curated a very large collection of new queries that, to the best of our knowledge, has not been used in any prior work. 

\begin{itemize}[leftmargin=*]
\setlength\itemsep{3pt}
\item 
$\leetcode$. 
This is a new dataset curated by us containing a total of {23,994} query pairs $(\query_1, \query_2)$. 
To curate this dataset, we first crawled \emph{all} publicly available queries accepted by LeetCode~\cite{leetcode} that are syntactically distinct. 
Then, we grouped them based on the LeetCode problem --- 
queries in the same group are 
written and submitted by actual users that solve the same problem. 
For each problem, we have also manually written a ``ground-truth'' query that is \emph{guaranteed} to be correct. 
Finally, each \emph{user query} $\query_1$ is paired with its corresponding \emph{ground-truth query} $\query_2$, yielding a total of {23,994} query \emph{pairs}.
In each pair, $\query_1$ is supposed to be equivalent to $\query_2$, which, however, is not always the case due to missing test cases on LeetCode. 
Indeed, as we will see shortly, $\tool$ has identified many wrong queries confirmed by LeetCode. 
As a final note: queries in our dataset naturally exhibit diverse patterns since they were written by different people with diverse backgrounds around the world, which makes this particular dataset especially valuable for evaluating query equivalence checkers. 
We also manually formalize the schema and integrity constraint for each LeetCode problem. 
\item 
$\calcite$. 
Our second benchmark set is constructed from the Calcite's optimization rules test suite~\cite{calcite-tests}. 
This is a standard workload used extensively in prior work~\cite{zhou2022spes,zhou2019automated,wang2018speeding,chu2017cosette,chu2017demonstration,chu2017hottsql,chu2018axiomatic}. 
In particular, each Calcite test case has a pair of queries $(\query_1, \query_2)$, where $\query_1$ can be rewritten to $\query_2$ using the optimization rule under test. 
In other words, $\query_1$ and $\query_2$ should be equivalent to each other, since all optimization rules are supposed to preserve the semantics. 
We included query pairs from \emph{all} test cases and ended up with a total of 397 query pairs. 
This number is noticeably higher than 232 pairs reported in prior work~\cite{zhou2022spes}, because the Calcite project is under active development and prior work used test cases from the version in year {2018}.  
The Calcite project also includes necessary schema and integrity constraints, which were all incorporated into our benchmark suite. 
\item 
$\literature$. 
Finally, we used \emph{all} other benchmarks from recent work that are publicly available. 
In particular, the $\cosette$ series~\cite{cosette-website,chu2017cosette,chu2017demonstration,chu2017hottsql,chu2018axiomatic} of work has {38} benchmarks --- all of them are included in our suite (8 involve integrity constraints).
We also used all \revision{26} benchmarks from~\cite{wang2018speeding}. 
In total, our $\literature$ dataset has {64} benchmarks. 
\end{itemize}
Across all our workloads, the average query size is 84.54 (measured by the number of AST nodes), where the max, min, and median are 679, 5, and 73, respectively.
In addition to the large query size, our benchmarks are complex in terms of query nesting. In particular, 57.7\% have sub-queries (e.g., \texttt{SELECT\:*\:FROM\:T1\:WHERE\:T1.id\:IN\:(SELECT\:uid\:FROM\:T2)}), and a number of them have four levels of query-nesting. 
In addition, 95.7\% of the queries involve joining multiple tables, and 33.2\% use the outer-join operation, which is generally quite challenging to verify. 
Finally, our queries also use other SQL features, such as group-by (58.7\%), aggregate functions (65.3\%), and order-by (21.6\%).


%

\newpara{Baselines.}
We compare $\tool$ against \emph{all} state-of-the-art \revision{\emph{bounded}} equivalence checking techniques that have corresponding tools publicly available. 
These tools all accept a smaller subset of SQL than $\tool$. 

\begin{itemize}[leftmargin=*]
\setlength\itemsep{1pt}
\item 
An extended version of $\cosette$~\cite{chu2017cosette,chu2017demonstration} which uses the provenance-base pruning technique from~\cite{wang2018speeding} to speed up the original $\cosette$ technique. 
This baseline represents the state-of-the-art along the $\cosette$ line of work. 
It does not support any integrity constraints.\footnote{$\cosette$ assumes all columns are NOT NULL --- which can be viewed as a ``built-in integrity constraint'' --- even if a column can in fact be NULL. In other words, $\cosette$ may not consider all valid inputs and, therefore, may be incomplete.}
For brevity, we call this baseline $\cosette$. 
\item 
An extension of $\qex$~\cite{qex} that incorporates the provenance-base pruning technique from~\cite{wang2018speeding} and 
significantly outperforms the original $\qex$.
For brevity, we call this baseline $\qex$. 
This baseline does not consider integrity constraints.
\end{itemize}

\revision{
In addition, we compare \tool against two SQL testing tools that cannot verify equivalence of two queries (neither in a bounded nor unbounded way) but can potentially generate counterexamples to disprove their equivalence.
}

\begin{itemize}[leftmargin=*]
\setlength\itemsep{1pt}
\item
\revision{
\datafiller~\cite{guagliardo2017formalsemantics, datafiller-website} is a random database generator that produces databases based on the schema and integrity constraints. We use \datafiller to generate random databases and run two queries over the databases to check if the outputs are different.
}
\item
\revision{
\xdata~\cite{chandra2015data} is a constraint-based mutation testing tool that can generate databases to disprove equivalence of two queries.
}
\end{itemize}

While not apples-to-apples, we also compare~$\tool$ with state-of-the-art \emph{full} (i.e., unbounded) verification techniques that can prove query equivalence. None of these tools are able to generate counterexamples and all of them support a smaller subset of SQL than $\tool$. 

\begin{itemize}[leftmargin=*]
\setlength\itemsep{1pt}
\item 
$\spes$~\cite{zhou2022spes} is a state-of-the-art verifier from the databases community. 
It 
does not support query operations beyond select-project-join (e.g., $\sqlorderby$ and set operations such as $\sqlintersect$), and has very limited support for integrity constraints\footnote{While the $\spes$ paper claims to accept simple primary key constraints, the artifact (https://github.com/georgia-tech-db/spes) is not parameterized with any constraints. {Instead, it is specialized to Calcite's tables, schema, and integrity constraints.}}.
\item 
$\hottsql$~\cite{chu2017hottsql} and $\udp$~\cite{chu2018axiomatic} use a proof assistant (in particular, Coq and Lean respectively) to prove query equivalence. 
They are quite different from the aforementioned techniques that utilize an automated theorem prover (e.g., Z3) and require more manual effort. 
\revision{
Since $\udp$ does not have a publicly available artifact that is usable, we include \hottsql as a baseline to facilitate a complete and thorough evaluation.
}
\end{itemize}


\subsection{RQ1: Coverage and Comparison against State-of-the-Art Techniques}
\label{sec:comparison}

In this section, we report the number of (1) query pairs whose \emph{bounded} equivalence (against a space of bounded-size inputs) can be successfully verified by $\tool$ and (2) those that can be proved non-equivalent (i.e., counterexamples are generated). 
We also compare $\tool$ with baselines.

\newpara{Setup.}
Given a benchmark consisting of queries $(\query_1, \query_2)$ with their schema and integrity constraint, we run $\tool$ using a {10-minute} timeout. 
There are three possible outcomes for each benchmark: (1) unsupported, (2) checked, or (3) \revision{refuted}. 
``unsupported'' means $\tool$ is not applicable to the benchmark (e.g., due to unsupported integrity constraints or SQL features). 
Otherwise, $\tool$ would report either ``checked'' (meaning bounded equivalence) or ``\revision{refuted}'' (with a counterexample witnessing the non-equivalence of $\query_1$ and $\query_2$). 
The bound is incrementally increased from 1 until the timeout is reached or a counterexample is identified. 
If no counterexample is found before timeout \revision{and the bounded equivalence is verified for at least bound 1}, we report ``checked''. 

\revision{
The same setup is used for bounded verification baselines, namely \cosette and \qex. For testing baselines (i.e., \datafiller and \xdata), there are three possible outcomes: (1) unsupported, (2) not-refuted, and (3) refuted. Since \datafiller is a random database generator, we use it to generate $1,000$ random databases where each relation has $100$ tuples for each supported benchmark. If any of these databases leads to different execution results on the two queries, we report it as ``refuted''; otherwise, we report ``not-refuted''. By contrast, \xdata performs mutation testing, so we only run the tool once on each supported benchmark. We report ``refuted'' if it finds a counterexample of equivalence and report ``not-refuted'' otherwise. For unbounded verification baselines (i.e., \spes and \hottsql), there are three possible outcomes: (1) unsupported, (2) verified, and (3) not-verified. ``Verified'' means the full equivalence of two queries is verified, and ``not-verified'' means the result is unknown (as none of them can generate counterexamples). 
}
This leads to an apples-to-oranges comparison; nevertheless, we include the results as well for a complete evaluation.
\revision{
All baseline tools use a 10-minute timeout, which is the same as \tool.
}



\begin{figure}[!t]
\centering
\begin{subfigure}{\linewidth}
\includegraphics[width=.9\linewidth]{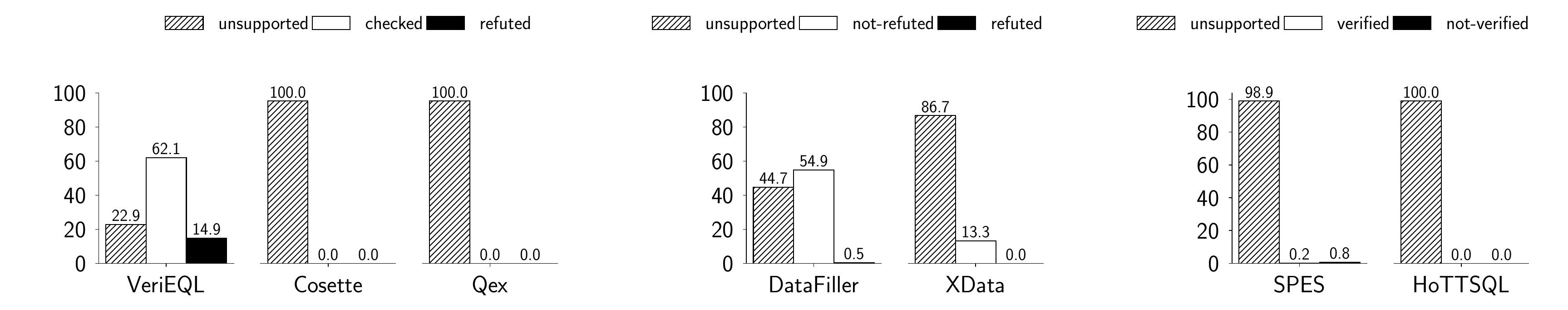}
\centering
\caption{\emph{LeetCode.}}
\label{fig:coverage-leetcode}
\end{subfigure}
\hfill
\begin{subfigure}{\linewidth}
\includegraphics[width=.9\linewidth]{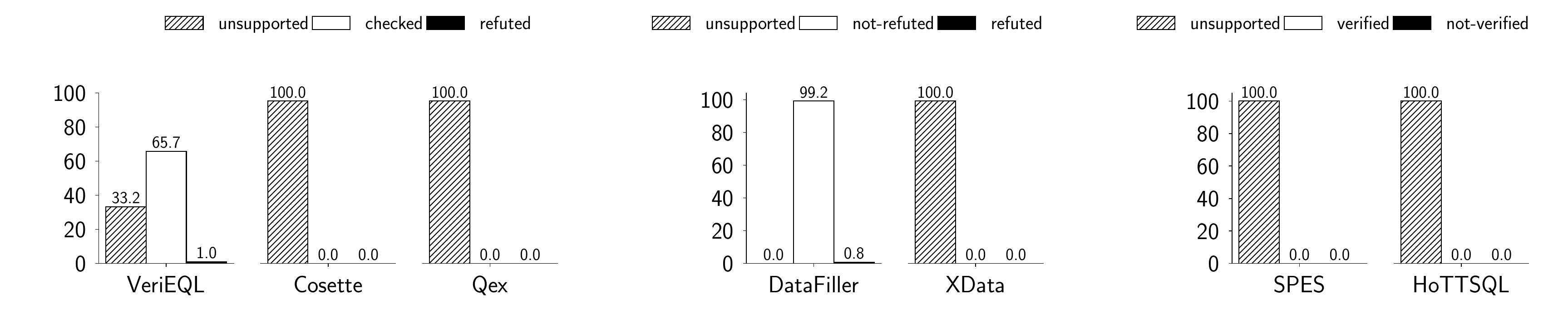}
\centering
\caption{\emph{Calcite.}}
\label{fig:coverage-calcite}
\end{subfigure}
\hfill
\begin{subfigure}{\linewidth}
\includegraphics[width=.9\linewidth]{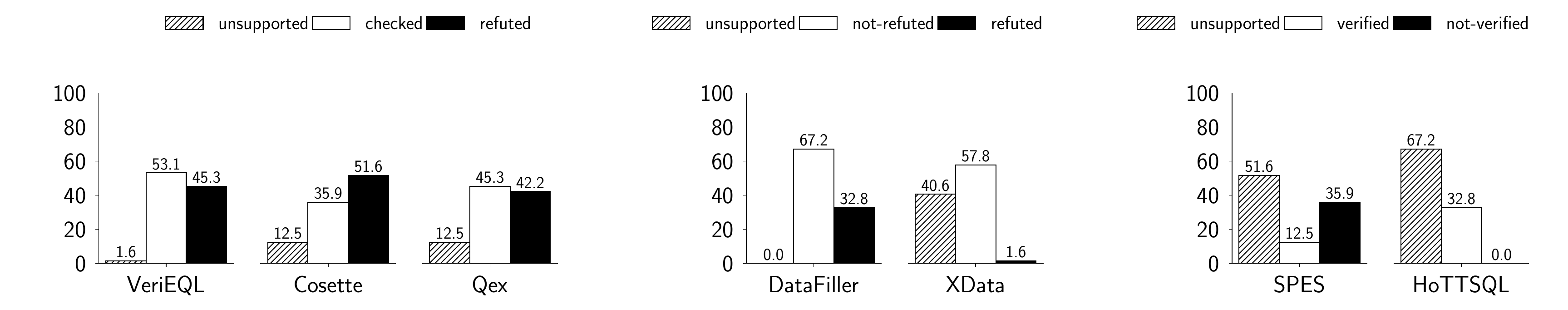}
\centering
\caption{\emph{Literature.}}
\label{fig:coverage-literature}
\end{subfigure}
\vspace{-20pt}
\caption{\textbf{RQ1}: Each workload has a sub-figure that shows, for each tool, the percentage of benchmarks (i.e., query pairs) in each category:
\revision{
unsupported/checked/refuted for bounded verification, unsupported/not-refuted/refuted for testing, unsupported/verified/not-verified for unbounded verification.
}
}
\label{fig:coverage-result}
\vspace{-10pt}
\end{figure}

\evalfinding{\textbf{RQ1 take-away:} $\tool$ supports over {75\%} of the benchmarks, \revision{which is more than all} baselines. $\tool$ also significantly outperforms all baselines across all workloads in terms of both disproving equivalence and proving (bounded) equivalence of query pairs.}

\newpara{Results.}
Our main result for each workload is shown in Figures~\ref{fig:coverage-leetcode},~\ref{fig:coverage-calcite}, and~\ref{fig:coverage-literature}.
\revision{
For the \leetcode dataset, \tool supports 77.1\% of benchmarks, which is higher than all baselines. This very large gap is due to baselines' poor support for integrity constraints and complex SQL features such as $\sqlorderby$.
}
Furthermore, \tool is able to disprove equivalence for {14.9\%} of the $\leetcode$ benchmarks, whereas none of the \emph{bounded} verification baselines can disprove any. \revision{
Unbounded verification tools cannot disprove equivalence by nature, and only the testing tool \datafiller can disprove 0.5\%.
For \calcite (see Figure~\ref{fig:coverage-calcite}), the result is similar. The only difference is that \datafiller can support all benchmarks, but it can only refute 0.8\% of benchmarks, smaller than the 1.0\% of \tool.
For \literature (see Figure~\ref{fig:coverage-literature}), \tool supports 98.4\% of benchmarks, which outperforms most of baselines other than \datafiller. It is comparable with bounded verification baselines, hypothetically because baselines were better-engineered for  $\literature$ queries, but it still outperforms testing baseline on disproving equivalence.
}

\newpara{Qualitative Analysis on Failure to Disprove Equivalence.}
\revision{
To understand whether two queries are indeed equivalent if \tool reports ``checked'', we perform manual inspections during the evaluation. Given that we have over $24,000$ benchmarks in total and exhaustive inspection is not practically feasible, we perform a non-exhaustive manual inspection.
Specifically, we sampled 50 benchmarks from \leetcode where \tool reported ``checked'' and confirmed all of them are indeed equivalent. In addition, for the benchmarks that can be proved equivalent by other tools (i.e., SPES), \tool also agrees on the result. This serves as some evidence on which we believe the equivalences checked by \tool are valid.
For \calcite dataset, we found that one benchmark where \tool reported ``checked'' but one baseline tool (namely \cosette) reported ``refuted''. After further inspection, we believe this benchmark should be equivalent but \cosette gave a spurious counterexample input on which both queries produce the same output table. We suspect \cosette may have an implementation bug that led to spurious counterexamples.
For \literature dataset, we manually inspected all of the 34 ``checked'' benchmarks and found that 27 are indeed equivalent but 7 are not equivalent. Among these 7 non-equivalent benchmarks, baselines generate spurious counterexamples for 5 of them; in other words, baselines also cannot disprove these 5. The remaining two non-equivalent benchmarks need large input databases (e.g., with $>$1000 tuples) in order to be differentiated; within our 10-minute timeout, \tool cannot disprove them but \cosette can. The reason is that \cosette has some internal heuristics that consider specific inputs with more tuples before exhausting the space of smaller inputs, whereas \tool simply increments the size of the input database from one.\footnote{Even with this simple size-increasing strategy, $\tool$ can disprove 5 $\literature$ benchmarks for which $\cosette$ reports ``checked.''}
}

\newpara{Small-Scope Hypothesis.}
\revision{
Our evaluation echos the small-scope hypothesis discussed in prior work~\cite{miao2019explaining}: ``mistakes in most of the queries may be explained by only a small number of tuples.'' There are only a small number of sources in SQL queries that could potentially break the small scope hypothesis, such as \sqllimit clause~\footnote{\revision{Recall from Figure~\ref{fig:syntax-sql} that we do not support \sqllimit clauses.}} or aggregation function \sqlcount with a large constant. Empirically, we observe that the small scope hypothesis holds for most of all our benchmarks. We only find two benchmarks that require an input relation with more than 1000 tuples to disprove equivalence; both of them have a clause like \texttt{COUNT(a) > 1000}.
}

\newpara{Limitations of Bounded Equivalence Verification.}
\revision{
While in principle the bounded equivalence verification approach taken by \tool can disprove or prove equivalence of two SQL queries for all input relations up to a finite size, it may not be able to disprove equivalence within a practical amount of time if the counterexample can only be large input relations. For example, during our evaluation, \tool failed to disprove equivalence for the following two queries from the \literature dataset:
}
\begin{tabular}{l}
\sqlselectcolor C.name \sqlfromcolor Carriers \sqlascolor C \sqljoincolor Flights \sqlascolor F \sqloncolor C.cid = F.carrier\_id \\
\hspace{3em}    \sqlgroupbycolor \ C.name, F.year, F.month\_id, F.day\_of\_month \sqlhavingcolor \sqlcountcolor(F.fid) > 1000 \\
\end{tabular}

\begin{tabular}{l}
\sqlselectcolor C.name \sqlfromcolor Carriers \sqlascolor C \sqljoincolor Flights \sqlascolor F \sqloncolor C.cid = F.carrier\_id \\
\hspace{3em}    \sqlgroupbycolor \ C.name, F.day\_of\_month \sqlhavingcolor \sqlcountcolor(F.fid) > 1000 \\
\end{tabular}

\revision{
In fact, a counterexample refuting their equivalence has more than $1000$ tuples, which \tool cannot find within the 10-minute time limit.
}

\newpara{Discussion.}
Interested readers might wonder why baselines perform poorly on $\calcite$ given it was used extensively as a benchmark suite in prior work~\cite{wang2018speeding,chu2017cosette}. 
The reason is that $\cosette$, $\qex$, and $\hottsql$ do not support the $\sqlnotnull$ constraint (among others) required in all $\calcite$ benchmarks. As a result, we directly report ``unsupported'' in \textbf{RQ1} setup. 
Results reported in prior work were obtained by running the tools \emph{without considering} constraints (like $\sqlnotnull$).
One may also wonder why $\tool$ has {1\%} (i.e., {4} benchmarks) ``not checked'' for $\calcite$, where all Calcite's query pairs are supposed to be equivalent. {We manually inspected these benchmarks and believe some of Calcite's rewrite rules are not equivalence preserving (i.e., wrong), leading to non-equivalent query pairs --- we will expand on this in \revision{\textbf{RQ2}}.}

\vspace{-5pt}
\subsection{RQ2: Effectiveness at Generating Counterexamples to Facilitate Downstream Tasks}
\label{sec:counterexample}

So far, we have seen the overall result of $\tool$ and how that compares against baselines.
In \revision{\textbf{RQ2}}, we will concentrate on those ``\revision{refuted}'' benchmarks and evaluate $\tool$'s capability at generating counterexamples. 
Notably, $\tool$'s counterexamples revealed serious bugs in MySQL and Calcite, and suggested new test inputs to augment LeetCode's existing test suite. 


\newpara{Setup.}
\revision{
Since the unbounded verifiers do not generate counterexamples, let us answer the following question: given each workload, for each \emph{bounded verification} and \emph{testing}
}
approach (including $\tool$ and baselines), how many benchmarks did the tool identify a counterexample for? How many of these counterexamples are genuine (i.e., not spurious)? 
Here, a counterexample is genuine if (i) it meets the integrity constraint associated with the benchmark and (ii) the two queries yield different outputs. 
\revision{
To study the impact of integrity constraints, we run two versions of each tool: (1) the original version as in \textbf{RQ1}, and (2) a version denoted by a suffix ``-noIC'' that drops all integrity constraints when running the tool (but we still consider integrity constraints when checking if the counterexample is genuine).
}


\evalfinding{\textbf{RQ2 take-away:} $\tool$ can \revision{disprove equivalence} for significantly more benchmarks than all bounded baselines across all workloads. 
$\tool$ can also generate counterexamples to help identify bugs in real-world systems and augment existing test suites.}

\begin{figure}[!t]
\centering
\begin{subfigure}{.32\linewidth}
\includegraphics[width=\linewidth]{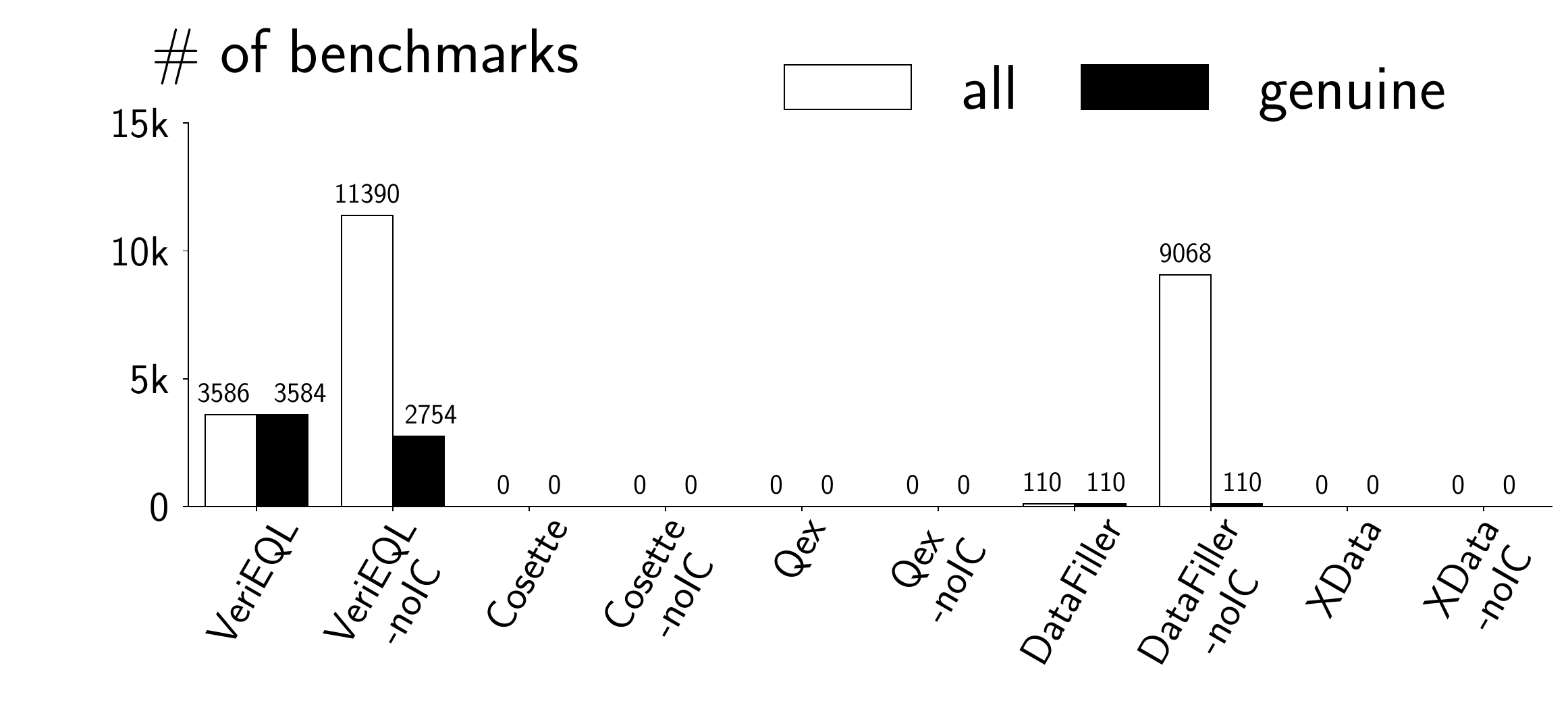}
\label{fig:counterexamples-leetcode}
\vspace{-15pt}
\caption{\emph{LeetCode.}}
\end{subfigure}
\hfill
\begin{subfigure}{.32\linewidth}
\includegraphics[width=\linewidth]{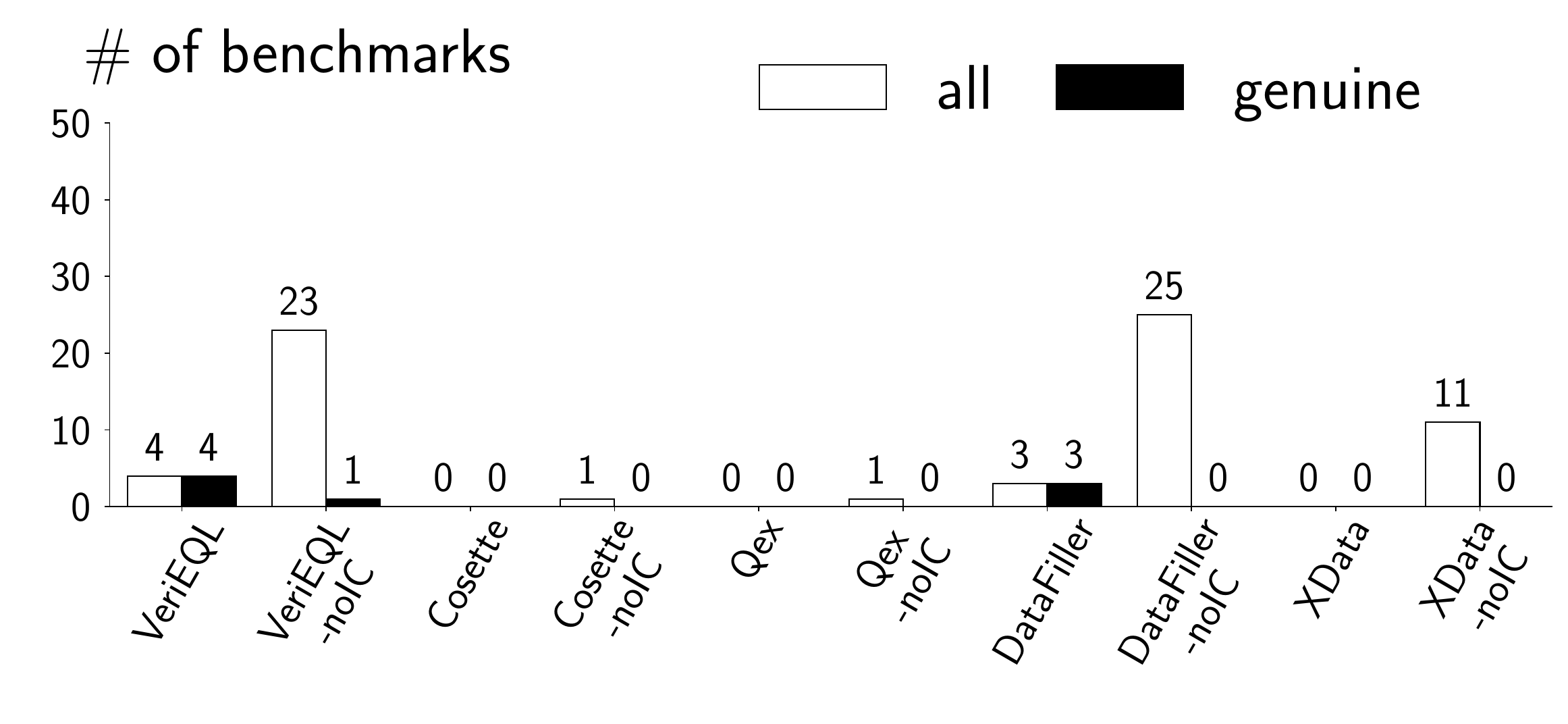}
\label{fig:counterexamples-calcite}
\vspace{-15pt}
\caption{\emph{Calcite.}}
\end{subfigure}
\hfill
\begin{subfigure}{.32\linewidth}
\includegraphics[width=\linewidth]{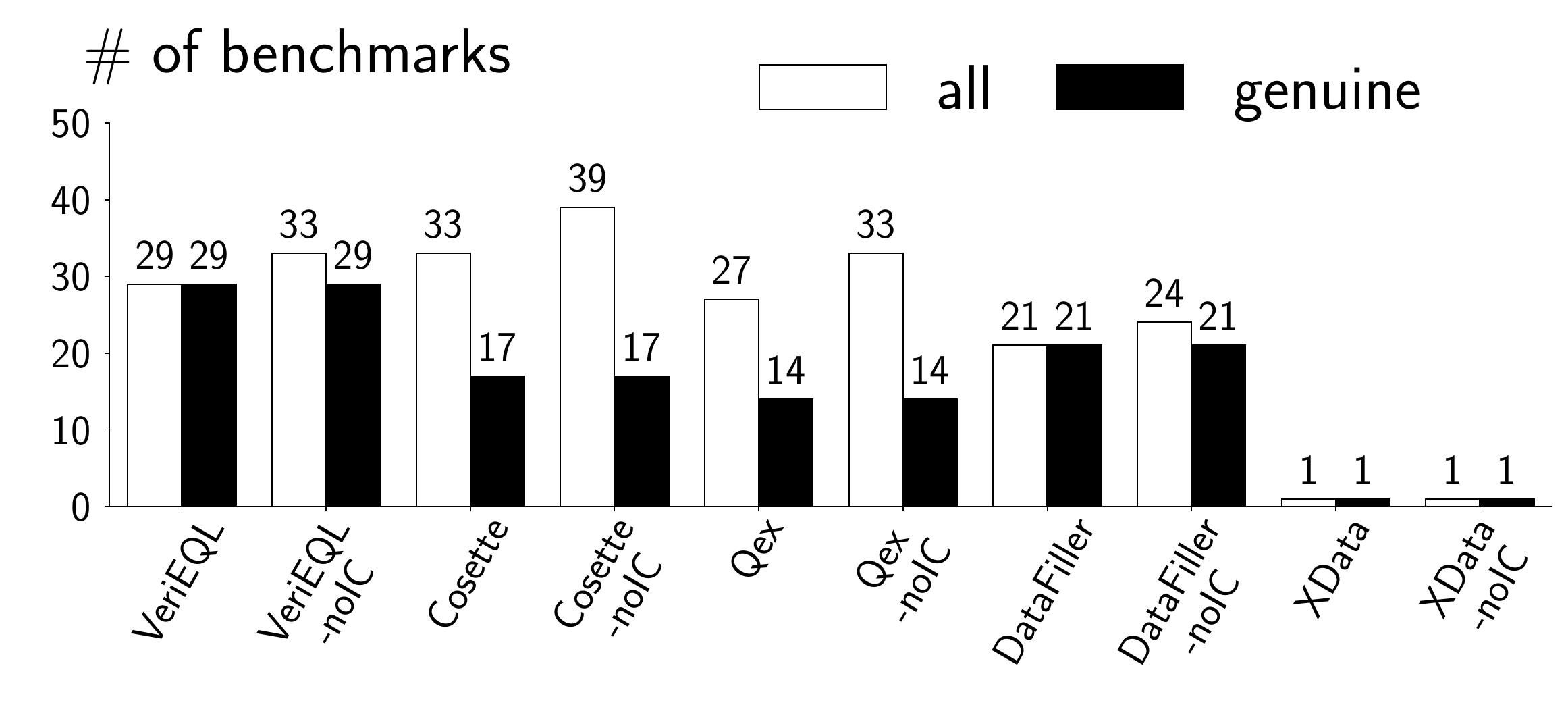}
\label{fig:counterexamples-literature}
\vspace{-15pt}
\caption{\emph{Literature.}}
\end{subfigure}
\vspace{-5pt}
\caption{\revision{\textbf{RQ2}: We have a sub-figure for each workload that shows (1) the number of \emph{all} benchmarks for which a tool generates a counterexample and (2) the number of benchmarks where the counterexamples are \emph{genuine}.}}
\label{fig:counterexamples}
\vspace{-15pt}
\end{figure}


\newpara{Results.}
Our main results are presented in Figure~\ref{fig:counterexamples}. 
In a nutshell, $\tool$ can prove much more non-equivalent benchmarks than all baselines across all workloads. 
For example, $\tool$ found counterexamples for {3,586} $\leetcode$ benchmarks; {3,584} of them are confirmed to be genuine. 
The counterexample inputs for the remaining two benchmarks conform to the integrity constraints but they produce the same output when executed on the input \emph{using MySQL.} 
It turns out these two seemingly ``spurious'' counterexamples are due to a previously unknown bug in MySQL --- they should produce different outputs. 
When dropping integrity constraints, $\tool$-noIC found {11,390} counterexamples --- but only {2,754} are genuine, which is even lower than $\tool$'s. This suggests that the false alarm rate will be significantly higher without considering integrity constraints.
For the other two workloads, we observe a very similar trend.

\newpara{Finding bugs in MySQL.}
As mentioned earlier, for two $\leetcode$ benchmarks, $\tool$ generates ``spurious'' counterexamples --- that are actually genuine --- due to a bug
\revision{
in MySQL's latest release version 8.0.32. The MySQL verification team had confirmed and classified this bug with \emph{serious} severity. Details of the bug are provided in Appendix D under supplementary materials.
}

\newpara{Detecting missing test cases for LeetCode.}
Recall that some of the user-submitted queries from our $\leetcode$  workload may be wrong (i.e., not equivalent to the ground-truth queries), which suggests that new test cases are needed.
In particular, as of the submission date (October 20, 2023), we have manually inspected 17 LeetCode problems for which $\tool$ reported non-equivalent benchmarks, and filed issues to the LeetCode team. Notably, our issue reports also have meaningful counterexamples generated by $\tool$.  
To date, 13 of them have been confirmed and fixed, while the remaining ones have all been acknowledged. 
A common pattern is that existing tests miss the $\sqlnull$ case --- this again highlights the importance of modeling the three-valued logic in SQL.


\newpara{Identifying buggy Calcite rewrite rules.}
$\tool$ found 4 ``not checked'' (i.e., non-equivalent) $\calcite$ benchmarks with valid counterexamples. The Calcite team has confirmed that two of them are due to bugs in two rewrite rules: one is a duplicate of an existing bug report, while the other is a new bug. In particular, for the new bug, the rule would rewrite a query $\query_1$ to a non-equivalent $\query_2$ when $\query_1$ has $\sqlsum$ applied to an empty table. 
The third is due to a bug in an internal translation step in Calcite, which has also been confirmed by the Calcite team. 
The fourth one is still being worked on by the team currently.
This result again demonstrates that $\tool$ is able to uncover very tricky bugs from well-maintained open-source projects.


\subsection{\revision{\textbf{RQ3}: Distribution of Bounds on Checked Benchmarks}}
\label{sec:exp-bound}

\begin{figure}
\centering
\begin{subfigure}[b]{0.3\textwidth}
 \centering
 \includegraphics[width=\textwidth]{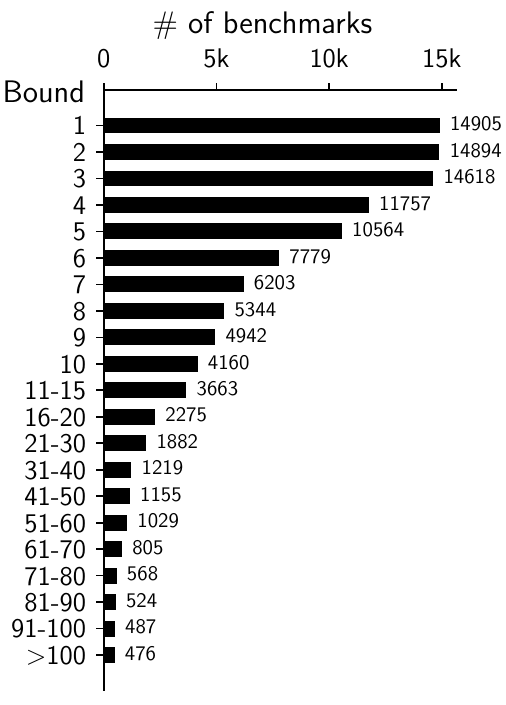}
 \caption{\emph{LeetCode}}
 \label{fig:varying_bound_leetcode}
\end{subfigure}
\hfill
\begin{subfigure}[b]{0.283\textwidth}
 \centering
 \includegraphics[width=\textwidth]{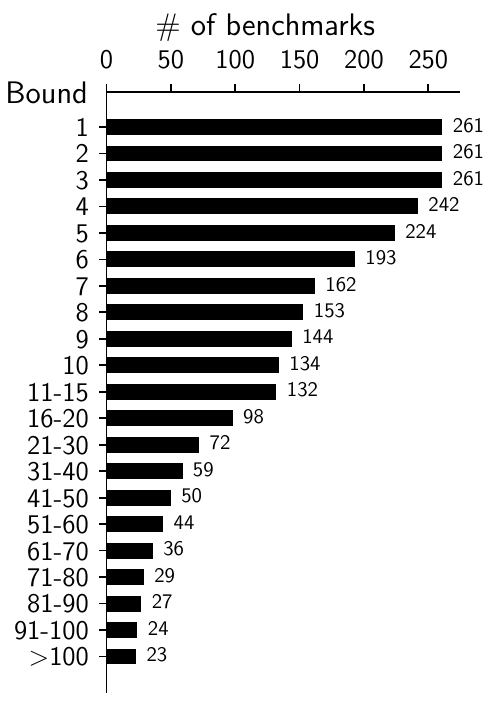}
 \caption{\emph{Calcite}}
 \label{fig:varying_bound_calcite}
\end{subfigure}
\hfill
\begin{subfigure}[b]{0.275\textwidth}
 \centering
 \includegraphics[width=\textwidth]{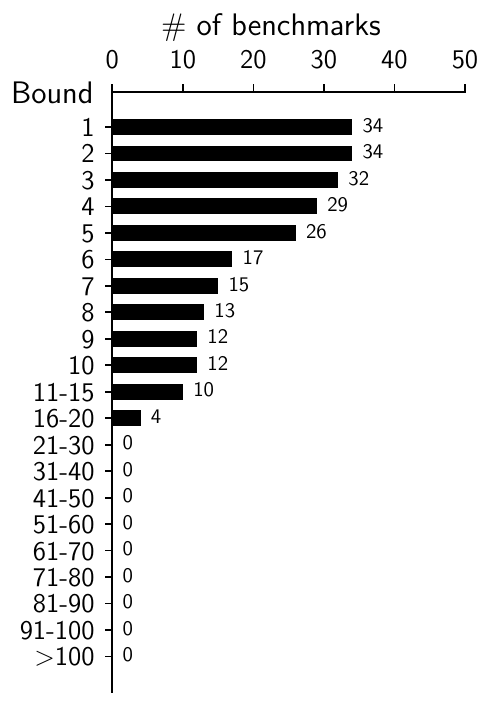}
 \caption{\emph{Literature}}
 \label{fig:varying_bound_literature}
\end{subfigure}
\vspace{-10pt}
\caption{\revision{\textbf{RQ3}: distribution of bounds on checked benchmarks.}}
\label{fig:varying_bound}
\vspace{-10pt}
\end{figure}

\revision{
As an indicator of \tool's scalability, we study how large input bounds \tool can reach within the 10-minute time limit when it reports a benchmark is checked.
}

\newpara{Setup.}
\hspace{-5pt}
\revision{
We collected all 14,905 (62.1\% of 23,994) \leetcode benchmarks that \tool reports checked in \textbf{RQ1}. Similarly, we also collected all 261 (65.7\% of 397) checked benchmarks from \calcite and all 34 (53.1\% of 64) checked benchmarks from \literature.
For each benchmark, we run \tool by gradually increasing the bound from 1 until it gets timeout and keep records of the bound.
}

\evalfinding{\revision{\textbf{RQ3 take-away:} \tool can verify over 70\% of the checked benchmarks within 10 minutes given an input bound 5. For 3\% of the checked benchmarks, \tool can verify the query equivalence for an input bound >100.}}

\newpara{Results.}
\revision{
The evaluation results are presented in Figure~\ref{fig:varying_bound}, where the bar shows the number of benchmarks \tool can verify for a given input bound. For example, as shown in Figure~\ref{fig:varying_bound_leetcode}, \tool can verify equivalence for bound 1 on all 14,905 \leetcode benchmarks. Among them, it can verify equivalence for bound 2 on 14,894 benchmarks. In a nutshell, on 71\% of checked \leetcode benchmarks, \tool can reach bound 5, 28\% for bound 10, 8\% for bound 50, and 3\% for bound 100. It shows similar patterns on \calcite. For \literature, \tool only reaches bound 20 within the time limit, hypothetically because many benchmarks combine several complex SQL operators (e.g., Distinct and GroupBy) that are challenging for symbolic reasoning.
}
\section{Related Work} \label{sec:related}


\newpara{Equivalence checking for SQL queries.}
SQL equivalence checking is an important problem in both programming languages and database communities. There is a line of work on automated verification of query equivalence, including full verification~\cite{zhou2022spes,zhou2019automated,chu2017hottsql,chu2018axiomatic,Green11containment,Chandra77cq,Aho79equivalence} and bounded verification~\cite{chu2017demonstration,chu2017cosette,qex,wang2018speeding}.
For example, SPES~\cite{zhou2022spes} takes a symbolic approach for proving the existence of a bijective identity map between tuples returned by two queries.
HoTTSQL~\cite{chu2017hottsql} and  UDP~\cite{chu2018axiomatic} encode SQL queries in an algebraic structure to syntactically canonicalize two queries and then verify the equivalence of corresponding algebraic expressions using interactive theorem provers.
Unlike prior work that aims for full equivalence verification, \tool targets to solve the bounded equivalence verification problem modulo integrity constraints and can generate counterexamples for disproving equivalence.
Compared to prior work for bounded verification such as Cosette~\cite{chu2017cosette,wang2018speeding}, \tool supports a larger fragment of SQL queries and a richer set of integrity constraints and significantly outperforms all baselines as demonstrated by our comprehensive evaluation.
\revision{
One limitation of \tool is that it does not support correlated subqueries in its current status, which can be overcome by incorporating unnesting techniques proposed by~\citet{seshadri1996complex}.
}


\newpara{Testing SQL queries.}
Another line of related work is about testing the correctness of SQL queries. For instance, \revision{\xdata}~\cite{chandra2015data} uses mutation-based testing to detect common mistake patterns in SQL queries. \citet{shah2011generating} also uses mutation testing with hard-coded rules to kill as many query mutants as possible. In addition, EvoSQL~\cite{castelein2018search} uses an evolutionary search algorithm to generate test data in an offline manner, guided by predicate coverage~\citet{tuya2010full}.  RATest~\cite{miao2019explaining} uses a provenance-based algorithm to find a minimal counterexample database that can explain incorrect queries.
\revision{\datafiller~\cite{guagliardo2017formalsemantics, datafiller-website} uses random database generation to produce test databases based on the schema and integrity constraints.}
Different from these testing approaches, \tool provides formal guarantees on query equivalence (in a bounded fashion).

\newpara{Formal methods for databases.}
Formal methods have been used to facilitate the development and maintenance of database systems and applications, such as optimizing database applications~\cite{Cheung13optimizing,Delaware15fiat} and data storage~\cite{Feser20deductive}, model checking database applications~\cite{Deutsch05verifier,Deutsch07spec,gligoric2013model}, providing foundations to queries~\cite{Ricciotti19mixing,Cheney21null,Ricciotti22null}, verifying correctness of schema refactoring~\cite{mediator}, and synthesizing programs for schema evolution~\cite{wang2019synthesizing,Wang20dynamite}.
Among various work, Mediator~\cite{mediator} is the most related to our work. However, Mediator studies the full equivalence verification problem of two database applications over different schemas, whereas \tool performs bounded equivalence verification of two SQL queries under the same schema, but it generates counterexamples for disproving query equivalence.

\newpara{Semantics of SQL queries.}
Different formal semantics of SQL queries have been proposed in the literature, such as set semantics, bag semantics, list semantics, and their combinations~\cite{negri1991formal,benzaken2019coq,Ricciotti19mixing,chu2017hottsql,chu2018axiomatic,mediator}.
We base the SQL semantics of \tool on list operations and quotient by bag equivalence at the end mainly because lists, by their ordered nature, facilitate the definition of \sqlorderby in the semantics. We have not explored the formal difference between our semantics and state-of-the-art semantics proposals.
In fact, we believe that proving the equivalence between our semantics and existing semantics or characterizing the differences between semantics with demonstrating examples is an interesting direction for future work.

\newpara{Bounded verification.}
Bounded verification has been used to find various kinds of bugs in software systems while providing formal guarantees on the result.
Prior work such as CBMC~\cite{clarke2004tool}, JBMC~\cite{cordeiro2018jbmc}, and Corral~\cite{lal2012solver} has been quite successful to this end for imperative languages like C, C++, or Java.
However, \tool performs bounded verification to check equivalence of declarative SQL queries. Unlike prior work that unrolls loops, \tool bounds the size of relations in the database and can guarantee the equivalence of two queries for all databases where relations are up to a given finite size.

\newpara{Equivalence verification.}
Equivalence verification has been studied extensively in the literature. Researchers have developed many approaches for checking equivalence of various applications, such as translation validation for compilers~\cite{Pnueli98tv,Necula00tv,Stepp11tv}, product program~\cite{Zaks08product,Barthe11product}, relational Hoare logic~\cite{Benton04rhl}, Cartesian Hoare Logic~\cite{Sousa16chl}, and so on.
\tool is along the line of equivalence verification, but it is tailored towards a different application, SQL queries, and presents an SMT-based approach for verifying bounded equivalence of queries modulo integrity constraints.

\section{Conclusion} \label{sec:concl}

In this paper, we presented $\tool$, a simple yet highly practical approach that can both prove and disprove the bounded equivalence of \emph{complex SQL queries with integrity constraints}.
For a total of 24,455 benchmarks, $\tool$ can prove and disprove over 70\% of them, significantly outperforming all state-of-the-art SQL equivalence checking techniques.

\section*{Acknowledgments} \label{sec:ack}

We would like to thank the anonymous reviewers of OOPSLA for their detailed and helpful comments on an earlier version of this paper. This research is supported by the Natural Sciences and Engineering Research Council of Canada (NSERC) Discovery Grant and the National Science Foundation (NSF) Grant No. CCF-2210832.

\section*{Data-Availability Statement}\label{sec:das}

The software that implements the techniques described in Section~\ref{sec:checking} and supports the evaluation results reported in Section~\ref{sec:eval} is available on Zenodo~\cite{artifact}.

\bibliography{main}

\appx{
\newpage
\appendix
\section{Semantics of Query Language} \label{sec:full-semantics}

\vspace{-300pt}
\begin{figure}[H]
\footnotesize


\begin{mdframed}
\[
\denot{Q_r} :: \text{Database } \db \to \text{Relation}
\]
\end{mdframed}

\vspace{-3pt}

\[
\begin{array}{lcl}
\denot{\text{OrderBy}(Q, \vec{E}, b)}_\db & = & \sfoldl(\lambda xs. \lambda \_. (\sappend(xs, [\text{MinTuple}(\vec{E}, b, \denot{Q}_\db - xs)])), [], \denot{Q}_\db) ~\text{where}\\
& & \quad \text{MinTuple}(\vec{E}, b, xs) = \sfoldl(\lambda x. \lambda y. \site(\text{Cmp}(\vec{E}, b, x, y), y, x), \shead(xs), xs), \\
& & \quad \text{Cmp}(\vec{E}, b, x_1, x_2) = b \neq \sfoldr(\lambda E_i, \lambda y. \site(v(x_1, E_i) < v(x_2, E_i), \top, \\
& & \quad \quad \quad \quad \quad \quad \quad \quad \quad \quad \quad \quad \quad \quad \quad \quad \quad \quad \site(v(x_1, E_i) > v(x_2, E_i), \bot, y)), \top, \vec{E}), \\
& & \quad v(x, E) = \site(\denot{E}_{\db, [x]} = \nullv, -\infty, \denot{E}_{\db, [x]}) \\
\end{array}
\]


\begin{mdframed}
\[
\denot{Q} :: \text{Database } \db \to \text{Relation}
\]
\end{mdframed}

\vspace{-3pt}

\[
\begin{array}{lcl}

\denot{R}_\db & = & \db(R) \\
\denot{\proj_L(Q)}_\db & = & \site(\text{hasAgg(L)}, [\denot{L}_{\db, \denot{Q}_\db}], \smap(\denot{Q}_\db, \lambda x. \denot{L}_{\db, x})) \\
\denot{\filter_\phi(Q)}_\db & = & \sfilter(\denot{Q}_\db, \lambda x. \denot{\phi}_{\db, [x]} = \top) \\
\denot{\rename_R{(Q)}}_\db & = & \smap(\denot{Q}_\db, \lambda x. \smap(x, \lambda (n,v). (\text{rename}(R, n), v))) \\

\denot{Q_1 \cap Q_2}_\db & = & \sfilter(\denot{\text{Distinct}(Q_1)}_\db, \lambda x. x \in \denot{Q_2}_\db) \\
\denot{Q_1 \cup  Q_2}_\db & = & \denot{\text{Distinct}(Q_1 \uplus Q_2)}_\db \\
\denot{Q_1 \setminus Q_2}_\db & = & \sfilter(\denot{\text{Distinct}(Q_1)}_\db, \lambda x. x \not \in \denot{Q_2}_\db) \\
\denot{Q_1 \cplus Q_2}_\db & = & \denot{Q_1 - (Q_1 - Q_2)}_\db \\
\denot{Q_1 \uplus Q_2}_\db & = & \sappend(\denot{Q_1}_\db, \denot{Q_2}_\db) \\
\denot{Q_1 - Q_2}_\db & = & \sfoldl(\lambda xs. \lambda x. \site(x \in xs, xs - x, xs), \denot{Q_1}_\db, \denot{Q_2}_\db) \\
\denot{\text{Distinct}(Q)}_\db & = & \sfoldr(\lambda x. \lambda xs. \scons(x, \sfilter(xs, \lambda y. x \neq y)), [], \denot{Q}_\db) \\

\denot{Q_1 \times Q_2}_\db & = & \sfoldl(\lambda xs. \lambda x. \sappend(xs, \smap(\denot{Q_2}_\db, \lambda y. \smerge(x, y))), [], \denot{Q_1}_\db) \\
\denot{Q_1 \ijoin_\phi Q_2}_\db & = & \denot{\filter_\phi(Q_1 \times Q_2)}_\db \\
\denot{Q_1 \ljoin_\phi Q_2}_\db & = & \sfoldl(\lambda xs. \lambda x. \sappend(xs, \site(|v_1(x)| = 0, v_2(x), v_1(x)), [],\denot{Q_1}_\db) \\
& & \quad \text{where}~ v_1(x) = \denot{[x] \ijoin_\phi Q_2}_\db ~\text{and}~ v_2(x) = [\smerge(x, T_{\nullv})] \\
\denot{Q_1 \rjoin_\phi Q_2}_\db & = & \sfoldl(\lambda xs. \lambda x. \sappend(xs, \site(|v_1(x)| = 0, v_2(x), v_1(x)), [],\denot{Q_2}_\db) \\
& & \quad \text{where}~ v_1(x) = \denot{Q_1 \ijoin_\phi [x]}_\db ~\text{and}~ v_2(x) = [\smerge(T_{\nullv}, x)] \\
\denot{Q_1 \fjoin_\phi Q_2}_\db & = & \sappend(\denot{Q_1 \ljoin_\phi Q_2}_\db, \smap(v, \lambda x. \smerge(T_{\nullv}, x)) ) \\
& & \quad \text{where}~ v = \sfilter(\denot{Q_2}_\db, \lambda y. |\denot{Q_1 \ijoin_\phi [y]}_\db| = 0 ) \\

\denot{\text{GroupBy}(Q, \vec{E}, L, \phi)}_\db \!\!\!\!& = & 
\smap(\sfilter(Gs, \lambda xs. \denot{\phi}_{\db, xs} = \top), \lambda xs. \denot{L}_{\db, xs}) ~\text{where} \\
& & \quad Gs = \smap(\text{Dedup}(Q, \vec{E}), \lambda y. \smap(\denot{Q}_\db, \lambda z. \text{Eval}(\vec{E}, [z]) = y)), \\
& & \quad \text{Dedup}(Q, \vec{E}) = \sfoldr(\lambda x. \lambda xs. \scons(x, \sfilter(xs, \lambda y. x \neq y)), [], \\
& & \quad \qquad \qquad \qquad \qquad \qquad  \smap(\denot{Q}_\db, \lambda z. \text{Eval}(\vec{E}, [z]))),\\
& & \quad \text{Eval}(\vec{E}, xs) = \smap(\vec{E}, \lambda e. \denot{e}_{\db, xs}) \\
\denot{\text{With}(\vec{Q}, \vec{R}, Q)}_\db & = & 
\denot{Q}_{\db'} ~\text{where}~ \db' = \db[R_i \mapsto \denot{\query_i}_\db ~|~ R_i \in \vec{R}] \\
\end{array}
\]

\vspace{10pt}

\vspace{-10pt}
\caption{Formal semantics for queries.}
\label{fig:proof-sematics1}
\vspace{-10pt}
\end{figure}

\begin{figure}[!h]
\footnotesize

\vspace{20pt}

\begin{mdframed}
\[
\denot{L} :: \text{Database } \db \to \text{Relation} \to [(\text{AttributeName}, \text{Value})]
\]
\end{mdframed}

\[
\begin{array}{lcl}

\denot{L}_{\db, xs} & = & \smap(L, \lambda y. \denot{y}_{\db, xs}) \\ 
\denot{id(A)}_{\db, xs} & = & [(\text{ToString}(A), \denot{A}_{\db, xs})] \\ 
\denot{\rho_a(A)}_{\db, xs} & = & [(a, \denot{A}_{\db, xs})] \\ 
\denot{L_1, L_2}_{\db, xs} & = & \smap(\sappend(L_1, L_2), \lambda y. \denot{y}_{\db, xs}) \\ 

\end{array}
\]

\vspace{10pt}

\begin{mdframed}
\[
\denot{E} :: \text{Database } \db \to \text{Relation} \to \text{Value}
\]
\end{mdframed}

\[
\begin{array}{lcl}
\denot{a}_{\db, xs} & = & \textsf{lookup}(\shead(xs), a) \\ 
\denot{v}_{\db, xs} & = & v \\
\denot{E_1 \allarith E_2}_{\db, xs} & = & \text{let}~ v_1 = \denot{E_1}_{\db, xs} ~\text{and}~ v_2 = \denot{E_2}_{\db, xs} ~\text{in}~ \site(v_1 = \nullv \lor v_2 = \nullv, \nullv, v_1 \allarith v_2) \\

\denot{\text{ITE}(\phi, E_1, E_2)}_{\db, xs} & = & \site(\denot{\phi}_{\db, xs} = \top, \denot{E_1}_{\db, xs}, \denot{E_2}_{\db, xs}) \\
\denot{\text{Case}(\vec{\phi}, \vec{E}, E')}_{\db, xs} & = & 
\sfoldr(\lambda y. \lambda (\phi_i, E_i). \denot{\text{\site}(\phi_i, E_i, y)}_{\db, xs}, E', \textsf{reverse}(zip(\vec{\phi}, \vec{E})))
~\text{where}~ n = |\vec{\phi}| = |\vec{E}| \\

\end{array}
\]

\vspace{10pt}

\begin{mdframed}
\[
\denot{\aggr(E)} :: \text{Database } \db \to \text{Relation} \to \text{Value}
\]
\end{mdframed}

\[
\begin{array}{lcl}
\denot{\text{Count}(E)}_{\db, xs} & = & \site(\text{AllNull}(E, \db, xs), \nullv, \sfoldl(+, 0, \smap(xs, \lambda y. \site(\denot{E}_{\db, [y]} = \nullv, 0, 1))) \\
& & \quad \text{where}~ \text{AllNull}(E, \db, xs) = \land_{x \in xs} \denot{E}_{\db, [x]} = \nullv \\
\denot{\text{Sum}(E)}_{\db, xs} & = & \site(\text{AllNull}(E, \db, xs), \nullv, \sfoldl(+, 0, \smap(xs, \lambda y. \site(\denot{E}_{\db, [y]} = \nullv, 0, \denot{E}_{\db, [y]}))) \\
\denot{\text{Avg}(E)}_{\db, xs} & = & \denot{\text{Sum}(E)}_{\db, xs} / \denot{\text{Count}(E)}_{\db, xs} \\
\denot{\text{Min}(E)}_{\db, xs} & = & \site(\text{AllNull}(E, \db, xs), \nullv, \sfoldl(\text{min}, +\infty, \smap(xs, \lambda y. \site(\denot{E}_{\db, [y]} = \nullv, +\infty, \denot{E}_{\db, [y]})))) \\
\denot{\text{Max}(E)}_{\db, xs} & = & \site(\text{AllNull}(E, \db, xs) \nullv, \sfoldl(\text{max}, -\infty, \smap(xs, \lambda y. \site(\denot{E}_{\db, [y]} = \nullv, -\infty, \denot{E}_{\db, [y]})))) \\
\end{array}
\]

\vspace{10pt}

\begin{mdframed}
\[
\denot{\phi} :: \text{Database } \db \to \text{Relation} \to \text{Bool} \cup {\nullv}
\]
\end{mdframed}

\[
\begin{array}{lcl}

\denot{\top}_{\db, xs} & = & \top \\
\denot{\bot}_{\db, xs} & = & \bot \\
\denot{\nullv}_{\db, xs} & = & \nullv \\
\denot{A_1 \alllogic A_2}_{\db, xs} & = & \text{let}~ v_1 = \denot{A_1}_{\db, xs}, v_2 = \denot{A_2}_{\db, xs} ~\text{in}~ \site(v_1 = \nullv \lor v_2 = \nullv, \bot, v_1 \alllogic v_2) \\
\denot{\text{IsNull}(E)}_{\db, xs} & = & \site(\denot{E}_{\db, xs} = \nullv, \top, \bot) \\
\denot{\vec{E} \in \vec{v}}_{\db, xs} & = & \lor_{i = 1}^{n} \land_{j = 1}^{m} \denot{E_j}_{\db, xs} = \denot{v_{i,j}}_{\db, xs} ~\text{where}~ n = |\vec{v}| ~\text{and}~ m = |\vec{E}| = |\vec{v}_i| \\
\denot{\vec{E} \in Q}_{\db, xs} & = & 
\sfoldl(\lambda ys. \lambda y. ys \lor \denot{\vec{E} \in y}_{\db, xs}, \bot, \denot{Q}_\db)
\\
\denot{\phi_1 \land \phi_2}_{\db, xs} & = & \text{let}~ v_1 = \denot{E_1}_{\db, xs} ~\text{and}~ v_2 = \denot{E_2}_{\db, xs} ~\text{in}~ \\
&  & \quad \site((v_1 = v_2 = \nullv) \lor (v_1 = \nullv \land v_2 = \top) \lor (v_1 = \top \land v_2 = \nullv), \nullv, v_1 = \top \land v_2 = \top) \\
\denot{\phi_1 \lor \phi_2}_{\db, xs} & = & \text{let}~ v_1 = \denot{E_1}_{\db, xs} ~\text{and}~ v_2 = \denot{E_2}_{\db, xs} ~\text{in}~ \\
&  & \quad \site((v_1 = v_2 = \nullv) \lor (v_1 = \nullv \land v_2 = \top) \lor (v_1 = \top \land v_2 = \nullv), \nullv, v_1 = \top \lor v_2 = \top) \\
\denot{\neg \phi}_{\db, xs} & = & \text{let}~ v = \denot{\phi}_{\db, xs} ~\text{in}~ \site(v = \nullv, \nullv, \neg v) \\
\end{array}
\]
\caption{Formal semantics for expressions and predicates.}
\label{fig:proof-sematics2}
\end{figure}

\section{Encoding of SQL Queries} \label{sec:appendix-encoding}

\begin{figure}[H]
\[
\hspace{-12pt}
\scriptsize
\begin{array}{c}

\irulelabel
{\begin{array}{c}
    R \in \dom(\context) \\
\end{array}}
{\schema, \context \vdash R \leadsto \top}
{\textrm{(E-Rel)}}

\irulelabel
{\begin{array}{c}
\context \vdash \query \hookrightarrow [t_1, \ldots, t_n] \quad
\schema, \context \vdash \query \leadsto \formula_1 \\
\context \vdash \rho_X(\query) \hookrightarrow [t_1', \ldots, t_n'] \quad
\formula_2 = \land_{i=1}^{n} t_{i}' = t_{i} \\
\end{array}}
{\schema, \context \vdash \rho_X(\query) \leadsto \formula_1 \land \formula_2}
{\textrm{(E-Rename)}}

\\ \ \\

\irulelabel
{\begin{array}{c}
    \neg \text{hasAgg}(L) \quad
    \schema, \context \vdash \query \leadsto \formula_1 \quad

    \context \vdash \query \hookrightarrow [t_1, \ldots, t_n] \quad

    \schema \vdash \Pi_{L}(\query) : [a_1', \ldots, a_l'] \quad 
    \context \vdash \Pi_{L}(\query) \hookrightarrow [t_1', \ldots, t_n'] \\


    \formula_2 = \land_{i=1}^{n} (\land_{j=1}^{l} 
    \denot{a_j'}_{\schema, \context, [t_{i}']} = 
    \denot{\revision{a_j'}}_{\schema, \context, [t_{i}]} \land \del(t_i') \leftrightarrow \del(t_{i})) \\
\end{array}}
{\schema, \context \vdash \Pi_{L}(\query) \leadsto \formula_1 \land \formula_2}
{\textrm{(E-Proj)}}







\\ \ \\

\irulelabel
{\begin{array}{c}
    \text{hasAgg}(L) \quad
    \schema, \context \vdash \query \leadsto \formula_1 \quad
    \context \vdash \query \hookrightarrow [t_1, \ldots, t_n]  \quad

    \schema \vdash \Pi_{L}(\query) : [a_1', \ldots, a_l'] \quad
    \context \vdash \Pi_{L}(\query) \hookrightarrow [t_1'] \\

    \formula_2 = \land_{j=1}^{l} \denot{a_j'}_{\schema,\context,[t_1']} = \denot{\revision{a_j'}}_{\schema,\context,\vec{t}} \land {\neg \del(t_1')} \\
\end{array}}
{\schema, \context \vdash \Pi_{L}(\query) \leadsto \formula_1 \land \formula_2}
{\textrm{(E-Agg)}}

\\ \ \\

\irulelabel
{\begin{array}{c}
    \context \vdash \query \hookrightarrow [t_1, \ldots, t_n] \quad
    \context \vdash \filter_\phi(\query) \hookrightarrow [t_1', \ldots, t_n'] \quad
    \schema, \context \vdash \query \leadsto \formula_1 \\
    \formula_2 = \land_{i=1}^{n} (
    (\neg \del(t_i) \land \denot{\phi}_{\schema, \context, [t_i]} = \top) \to t_{i}' =  t_{i} 
    \land
    (\del(t_i) \lor \denot{\phi}_{\schema, \context, [t_i]} \neq \top) \to \del(t_{i}'))\\
\end{array}}
{\schema, \context \vdash \filter_\phi(\query) \leadsto \formula_1 \land \formula_2}
{\textrm{(E-Filter)}}

\\ \ \\

\irulelabel
{\begin{array}{l}
    \schema \vdash \query_1 : [a_1, \ldots, a_{p}] \quad
    \context \vdash \query_1 \hookrightarrow [t_1, \ldots, t_{n}] \quad
    \schema, \context \vdash \query_1 \leadsto \formula_1 \\
    
    \schema \vdash \query_2 : [a_1', \ldots, a_{q}'] \quad
    \context \vdash \query_2 \hookrightarrow [t_1', \ldots, t_{m}'] \quad
    \schema, \context \vdash \query_2 \leadsto \formula_2 \quad
    
    \context \vdash \query_1 \times \query_2 \hookrightarrow [t_{1,1}'', \ldots, t_{nm}''] 
    \\
    
    \formula = \land_{i=1}^{n} \land_{j=1}^{m} (\neg \del(t_i) \land \neg \del(t_j')) \to \land_{k=1}^{p} \denot{a_{k}}_{\schema, \context, [t_{i,j}'']} = \denot{a_{k}}_{\schema, \context, [t_{i}]} \land \land_{k=1}^{q} \denot{a_{k}'}_{\schema, \context, [t_{i,j}'']} = \denot{a_{k}'}_{\schema,\context,[t_{j}']} \land \\
    \qquad (\del(t_i) \lor \del(t_j')) \to \del(t_{i,j}'')\\
\end{array}}
{\schema, \context \vdash \query_1 \times \query_2 \leadsto \formula_1 \land \formula_2 \land \formula_3}
{\textrm{(E-Prod)}}
\\ \ \\

\irulelabel
{\begin{array}{l}
    \schema, \context \vdash \filter_{\phi}(\query_1 \times \query_2) \leadsto \formula \\
\end{array}}
{\schema, \context \vdash \query_1 \bowtie_\phi \query_2 \leadsto \formula}
{\textrm{(E-IJoin)}}

\\ \ \\

\irulelabel
{\begin{array}{c}
    \schema \vdash \query_1 : [a_1, \ldots, a_{p}] \quad 
    \context \vdash \query_1 \hookrightarrow [t_1, \ldots, t_n] \quad
    \schema \vdash \query_2 : [a_1', \ldots, a_{q}'] \quad
    \context \vdash \query_2 \hookrightarrow [t_1', \ldots, t_m']
    \\
    
    \schema, \context \vdash \query_1 \bowtie_\phi \query_2 \leadsto \formula_1 
    \quad
    \context \vdash \query_1 \leftouterjoin_\phi \query_2 \hookrightarrow [t_{1,1}'', \ldots, t_{n, m}'', t_{1, m+1}'', \ldots, t_{n, m+1}'']
    \\
    \formula_2 = \land_{i=1}^{n} (\land_{k=1}^{p} \denot{a_{k}}_{\schema,\context,[t_{i,m+1}'']} = \denot{a_{k}}_{\schema,\context,[t_{i}]} \land \land_{k=1}^{q} \denot{a_{k}'}_{\schema,\context,[t_{i,m+1}]} = \text{Null} \land (\revision{\neg \del(t_i) \land} \land_{j=1}^{m} \del(t_{i,j}'')) \leftrightarrow \neg \del(t_{i,m+1}'')  )\\
\end{array}}
{\schema, \context \vdash \query_1 \leftouterjoin_\phi \query_2 \leadsto \formula_1 \land \formula_2}
{\textrm{(E-LJoin)}}
\\ \ \\

\irulelabel
{\begin{array}{c}
    \schema \vdash \query_1 : [a_1, \ldots, a_{p}] \quad 
    \context \vdash \query_1 \hookrightarrow [t_1, \ldots, t_n] \quad
    \schema \vdash \query_2 : [a_1', \ldots, a_{q}'] \quad
    \context \vdash \query_2 \hookrightarrow [t_1', \ldots, t_m'] \\
    
    \schema, \context \vdash \query_1 \bowtie_\phi \query_2 \leadsto \formula_1 
    \quad
    \context \vdash \query_1 \rightouterjoin_\phi \query_2 \hookrightarrow [t_{1,1}'', \ldots, t_{n, m}'', t_{n+1, 1}'', \ldots, t_{n+1, m}'']
    \\
    \formula_2 = \land_{j=1}^{m} (\land_{k=1}^{p} \denot{a_{k}}_{\schema,\context,[t_{n+1,j}'']} = \text{Null} \land \land_{k=1}^{q} \denot{a_{k}'}_{\schema,\context,[t_{n+1,j}'']} = \denot{a_k'}_{\schema,\context,[t_j]} \land (\revision{\neg \del(t_j') \land} \land_{i=1}^{n} \del(t_{ij}'') \leftrightarrow \neg \del(t_{n+1,j}''))  )\\
\end{array}}
{\schema, \context \vdash \query_1 \rightouterjoin_\phi \query_2 \leadsto \formula_1 \land \formula_2}
{\textrm{(E-RJoin)}}
\\ \ \\

\irulelabel
{\begin{array}{c}
    \schema \vdash \query_1 : [a_1, \ldots, a_{p}] \quad 
    \context \vdash \query_1 \hookrightarrow [t_1, \ldots, t_n] \quad
    \schema \vdash \query_2 : [a_1', \ldots, a_{q}']  \quad
    \context \vdash \query_2 \hookrightarrow [t_1', \ldots, t_m']
    \\

    \revision{\schema, \context \vdash \query_1 \ljoin_\phi \query_2 \leadsto \formula } \quad
    \context \vdash \query_1 \fullouterjoin_\phi \query_2 \hookrightarrow [t_{1,1}'', \ldots, t_{n, m}'', t_{n+1,1}'', \ldots, t_{n+1,m}'', t_{1,m+1}'', \ldots, t_{n,m+1}''] \\

    \formula_1 = \land_{j=1}^{m} (\land_{k=1}^{p} \denot{a_{k}}_{\schema,\context,[t_{n+1,j}'']} = \text{Null}
    \land \land_{k=1}^{q} \denot{a_{k}'}_{\schema,\context,[t_{n+1,j}'']} = \denot{a_{k}'}_{\schema,\context,[t_{j}']} ) \land \revision{(\neg \del(t_j') \land  \land_{i=1}^{n} \del(t_{ij}'') \leftrightarrow \neg \del(t_{n+1,j}''))}) \\

\end{array}}
{\context \vdash \query_1 \fullouterjoin_\phi \query_2 \leadsto \formula \land \formula_1}
{\textrm{(E-FJoin)}}
\\ \ \\

\irulelabel
{\begin{array}{c}
\context \vdash \query \hookrightarrow [t_1, \ldots, t_n] \quad
\schema, \context \vdash \query \leadsto \formula_1 \quad
\context \vdash \text{GroupBy}(\query, \vec{E},\phi) \hookrightarrow [t_1', \ldots, t_n'] \\

\formula_2 = \land_{i=1}^{n} (
\sum_{j=1}^{i} \text{If}(g(t_{i}, j), 1, 0) = \text{If}(\del(t_i), 0, 1) \land \revision{\land_{j=1}^{i-1}} g(t_{i}, j) = (\neg \del(t_i) \land g(t_j, j) \land \denot{\vec{E}}_{\schema, \context, [t_i]} = \denot{\vec{E}}_{\schema, \context, [t_j]})
)
\\

\formula_3 = \land_{i=1}^{n} (
 (
g(t_i, i) \land 
 \revision{\denot{\phi}_{\schema, \context, g^{-1}(i)} = \top} \to 
 \neg \del(t_i') \land \land_{k=1}^{l} \denot{a_{k}'}_{\schema, \context, [t_i']} = \denot{\revision{a_k'}}_{\schema, \context, g^{-1}(i)} 
 )
 \land 
 \\
 \qquad
 (
 \neg g(t_i, i) \lor
 \revision{\denot{\phi}_{\schema, \context, g^{-1}(i)} \neq \top} \to \del(t_i') 
 )
 )\\
\end{array}}
{\schema, \context \vdash \text{GroupBy}(\query, \vec{E}, [a_1', ..., a_l'],\phi) \leadsto \formula_1 \land \formula_2 \land \formula_3}
{\textrm{(E-GroupBy)}}

\\ \ \\



\end{array}
\]

\caption{
Symbolic encoding of SQL queries. 
The $\del(t)$ function represents whether a symbolic tuple $t$ is deleted.
The $g(t, j)$ is an uninterpreted function determining whether the tuple $t$ belongs to the group $j$; and $g^{-1}(j)$ is a reversible function of $g$ which returns the tuples $\vec{t}$ belonging to the group $j$.
}
\label{fig:rules-encode-b}
\end{figure}

Figure~\ref{fig:rules-encode-b} and ~\ref{fig:rules-encode-c} present symbolic encoding for SQL queries in \tool.
Specifically, Figure~\ref{fig:rules-encode-b} mainly focuses on formulating \sqljoin, \sqlgroupby and \sqlorderby while the set operations are fully encoded in Figure~\ref{fig:rules-encode-c}.
Among these encoding rules, most of which are straightforward to understand, there are two interesting operations (i.e., \sqlgroupby and \sqlexceptall) on which we further elaborate.

We decompose \sqlgroupby into the \textit{grouping} and \textit{reducing} procedures.
While grouping, our top priority focuses on the partition over a symbolic table based on the evaluation of the metrics $\vec{E}$.
The tuples with the same $\vec{E}$ values will be partitioned into a group.
Furthermore, a tuple $t_i$ cannot appear in the groups with greater indices, i.e., $\exists j. 1 \leq j \leq i \to t_i \in G_j$.
In Figure~\ref{fig:groupby}, for instance, $t_1$ and $t_3$ are partitioned into $G_1$ since they share the same grouping expressions, whereas $t_2$ solely constructs $G_2$.
In the reducing process, we exploit those groups and their shared information to determine the existence of new tuples $[t_1', t_2', t_3']$ and to compute projected expressions $L$.
For instance, $t_1'$ and $t_2'$ exist because $G_1$ and $G_2$ have at least a non-deleted tuple, while $t_3'$ is deleted because no tuple is categorized under this group $G_3$.

\begin{figure}[!t]
\[
\hspace{-12pt}
\scriptsize
\begin{array}{c}

\irulelabel
{\begin{array}{c}
\context \vdash \query \hookrightarrow [t_1, \ldots, t_n] \quad
\schema, \context \vdash \query \leadsto \formula_1 \quad
\context \vdash \text{OrderBy}(\query, \vec{E}, b) \hookrightarrow [t_1', \ldots, t_n'] \\
\revision{
\text{fresh } t''_i \quad
1 \leq i \leq n \quad
\formula_2 =  \text{moveDelToEnd}(\vec{t}, \vec{t}'') \quad
\formula_3 = \land_{i=1}^{n} t_i' = t_{\text{find}(i, \vec{E}, b)}''
}

\end{array}}
{\schema, \context \vdash \text{OrderBy}(\query, \vec{E}, \revision{b}) \leadsto \formula_1\revision{ \land \formula_2 \land \formula_3}}
{\textrm{(E-OrderBy)}}
\\ \ \\

\irulelabel
{\begin{array}{l}
    \context \vdash \query \hookrightarrow [t_1, \ldots, t_n] \quad
    \schema, \context \vdash \query \hookrightarrow \formula_1 \quad
    \context \vdash \text{Distinct}(\query) \hookrightarrow [t_1', \ldots, t_n'] \quad
    \formula_2 = \text{Dedup}(\vec{t'}, \vec{t})


    
\end{array}}
{\schema, \context \vdash \text{Distinct}(\query) \leadsto \formula_1 \land \formula_2}
{\textrm{(E-Dist)}}
\\ \ \\

\irulelabel
{\begin{array}{l}
\schema, \context \vdash \text{Distinct}(\query_1 \uplus \query_2) \leadsto \formula 
\end{array}}
{\schema, \context \vdash \query_1 \cup \query_2 \leadsto \formula}
{\textrm{(E-Union)}}
\\ \ \\

\irulelabel
{\begin{array}{c}
    \context \vdash \query_1 \hookrightarrow [t_1, \ldots, t_n] \quad
    \schema, \context \vdash \query_1 \leadsto \formula_1 \quad
    \context \vdash \query_2 \hookrightarrow [t_1', \ldots, t_m'] \quad
    \schema, \context \vdash \query_2 \leadsto \formula_2 \\
    \context \vdash \query_1 \cup \query_2 \hookrightarrow [t_1'', \ldots, t_{n + m}'']  \quad
    \formula_3 = \land_{i=1}^{n} t_{i}'' = t_{i} \land \land_{j=n+1}^{n+m} t_{j}'' = t_{j-n}' \\
\end{array}}
{\schema, \context \vdash \query_1 \uplus \query_2 \leadsto \formula_1 \land \formula_2 \land \formula_3}
{\textrm{(E-UnionAll)}}
\\ \ \\

\irulelabel
{\begin{array}{c}
    \revision{\context \vdash \text{Distinct}(\query_1) \hookrightarrow [t_1, \ldots, t_n] \quad 
    \schema, \context \vdash \text{Distinct}(\query_1) \leadsto \formula_1} \quad
    \context \vdash \query_2 \hookrightarrow [t_1', \ldots, t_m'] \quad
    \schema, \context \vdash \query_2 \leadsto \formula_2 \\
    \context \vdash \query_1 \cap \query_2 \hookrightarrow [t_1'', \ldots, t_n'']  \quad
    \formula_3 = \land_{i=1}^{n} (t_i \in \vec{t'} \to t_{i}'' = t_{i} \land t_i \not \in \vec{t'} \to \del(t_{i}''))
    \\
\end{array}}
{\schema, \context \vdash \query_1 \cap \query_2 \leadsto \formula_1 \land \formula_2 \land \formula_3}
{\textrm{(E-Intx)}}
\\ \ \\

\irulelabel
{\begin{array}{l}
 \schema, \context \vdash \query_1 - (\query_1 - \query_2) \leadsto \formula   
\end{array}}
{\schema, \context \vdash \query_1 \cplus \query_2 \leadsto \formula}
{\textrm{(E-IntxAll)}}
\\ \ \\

\irulelabel
{\begin{array}{c}
    \revision{\context \vdash \text{Distinct}(\query_1) \hookrightarrow [t_1, \ldots, t_n] \quad 
    \schema, \context \vdash \text{Distinct}(\query_1) \leadsto \formula_1} \\
    
    \context \vdash \query_2 \hookrightarrow [t_1', \ldots, t_m'] \quad
    \schema, \context \vdash \query_2 \leadsto \formula_2 \\
    
    \context \vdash \query_1 \setminus \query_2 \hookrightarrow [t_1'', \ldots, t_n'']  \quad
    \formula_3 = \land_{i=1}^{n}  (t_i \not \in \vec{t'} \to t_{i}'' = t_{i} \land t_i \in \vec{t'} \to \del(t_{i}'')) 
\end{array}}
{\schema, \context \vdash \query_1 \setminus \query_2 \leadsto \formula_1 \land \formula_2 \land \formula_3}
{\textrm{(E-Ex)}}
\\ \ \\

\irulelabel
{\begin{array}{l}
    \context \vdash \query_1 \hookrightarrow [t_1, \ldots, t_n] \quad 
    \context \vdash \query_2 \hookrightarrow [t_1', \ldots, t_m'] \quad
    \schema, \context \vdash \query_1 \leadsto \formula_1 \quad
    \schema, \context \vdash \query_2 \leadsto \formula_2 
    \\

    \schema \vdash \query : \attributes \quad 
    \context \vdash \query_1 - \query_2 \hookrightarrow [t_1'', \ldots, t_n''] 
    \\

    \formula_3 = \paired(1, 1) \leftrightarrow \eqtext(t_1, t_1', \attributes) \\
    \formula_4 = \land_{j=2}^{m} \paired(1, j) \leftrightarrow (\eqtext(t_1, t_j', \attributes) \land \land_{k=1}^{j-1} \neg \paired(1, k)) \\ 
    \formula_5 = \land_{i=2}^{n} \paired(i, 1) \leftrightarrow (\eqtext(t_i, t_1', \attributes) \land \land_{k=1}^{i-1} \neg \paired(k, 1)) \\
    \formula_6 = \land_{i=2}^{n} \land_{j=2}^{m} \paired(i, j) =  \eqtext(t_i, t_j, \attributes) \land \land_{k=1}^{i-1} \neg \paired(k, j) \land \land_{k=1}^{j-1} \neg \paired(i, k)
    \\
    
    \formula_7 = \revision{\land_{i=1}^{n} (\neg \del(t_i) \land \land_{j=1}^{m} \neg \paired(i,j) \to t_{i}'' = t_{i} \land \del(t_i) \lor \lor_{j=1}^{m} \paired(i,j) \to \del(t_{i}'') )}
    
\end{array}}
{\schema, \context \vdash \query_1 - \query_2 \leadsto \formula_1 \land \formula_2 \land \formula_3 \land \formula_4 \land \formula_5 \land \formula_6 \land \formula_7}
{\textrm{(E-ExAll)}}
\\ \ \\

\irulelabel
{\begin{array}{l}
\schema \vdash \query : \attributes_i \quad
\context \vdash \query \hookrightarrow \tuples_i \quad
\schema, \context \vdash \query_i \leadsto \formula_i \quad
1 \leq i \leq n \quad
n = |\vec{Q}| \\
\revision{\schema[R_1 \mapsto \attributes_1, \ldots, R_n \mapsto \attributes_n], \context[R_1 \mapsto \tuples_1, \ldots, R_n \mapsto \tuples_n]} \vdash \query' \leadsto \formula' \\
\end{array}}
{\schema, \context \vdash \text{With}(\vec{\query}, \revision{\vec{R}}, \query') \leadsto \land_{i=1}^{n} \formula_i \land \formula'}
{\textrm{(E-With)}}
\\ \ \\

\end{array}
\]

\vspace{-10pt}

\caption{
Symbolic encoding of SQL queries. 
\revision{
In $\text{OrderBy}$, we first ``remove'' those deleted symbolic tuples (i.e., $\del(t_i) = \top$) by moving them to the end of the list s.t. $\exists m. 1 \leq m \leq n \land \land_{i=1}^m \del(t_i) = \bot \land \land_{i=m+1}^n \del(t_i) = \top$; and the $\text{find}(i, \vec{E}, b)$ function will find the $i$-th smallest/largest non-deleted tuple $t_i''$ w.r.t. the expression list $\vec{E}$ and the ascending order $b$.
}
$\text{Dedup}(\tuples', \tuples)$ serves as de-duplication on a set of tuples $\text{Dedup}(\vec{t'}, \vec{t}) = \tuples$, 
i.e., $t_1' = t_1 \land \land_{i=2}^{n} (t_i \in \vec{t}_{1:i-1} \to \del(t_i') \land t_i \notin \vec{t}_{1:i-1} \to t_i' = t_i)$ where $|\vec{t'}| = |\vec{t}| = n$. 
}
\label{fig:rules-encode-c}
\end{figure}

Unlike the \sqlexcept operation, which rules out all occurrences of duplicate tuples, \sqlexceptall only eliminates occurrences of duplicate data with limited times.
We utilize dynamic programming to encode the \sqlexceptall operation.
Let us define two tuples $t_i \in R_1$ and $t_j' \in R_2$ are \textit{paired} iff they are (1) \textit{equal} (i.e., $t_i$ and $t_j'$ are not deleted and share the same values), and (2) \textit{not paired} with the previous tuples from the other relation.
Formally, \textit{paired} is defined as follows:
\begin{equation}
\paired(i, j) = \left\{
\begin{aligned}
& \eqtext(t_i, t_j', \attributes) & i = 1, j = 1, \\
& \eqtext(t_i, t_j', \attributes) \land \land_{k=1}^{j-1} \neg \paired(1,k) & i = 1,j \neq 1, \\
& \eqtext(t_i, t_j', \attributes) \land \land_{k=1}^{i-1} \neg \paired(k,1) & i \neq 1,j = 1, \\
& \eqtext(t_i, t_j', \attributes) \land \land_{k=1}^{i-1} \neg \paired(k,j)  \land \land_{k=1}^{j-1} \neg \paired(i,k) & \text{otherwise}.
\end{aligned}
\right.
\nonumber
\end{equation}
where the $\eqtext(t_i, t_j', \attributes)$ predicate asserts $t_i \in R_1$ and $t_j' \in R_2$ are equal, i.e., $\eqtext(t_i, t_j', \attributes) \Leftrightarrow \neg \del(t_i) \land \neg \del(t_j') \land \land_{a \in \attributes} t_i.a = t_j'.a$.
For instance, $t_3$ and $t_2'$ are paired because they share the same values and are unable to be paired with other tuples (such as $t_1$ or $t_1'$), as shown in Figure~\ref{fig:exceptall}.

\begin{figure}[!t]
\centering
\begin{subfigure}{.4\linewidth}
\includegraphics[width=\linewidth]{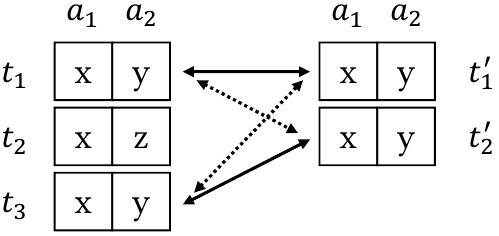}
\caption{The \textit{paired} mechanism of \sqlexceptall.}
\label{fig:exceptall}
\end{subfigure}
\hfill
\begin{subfigure}{.55\linewidth}
\includegraphics[width=\linewidth]{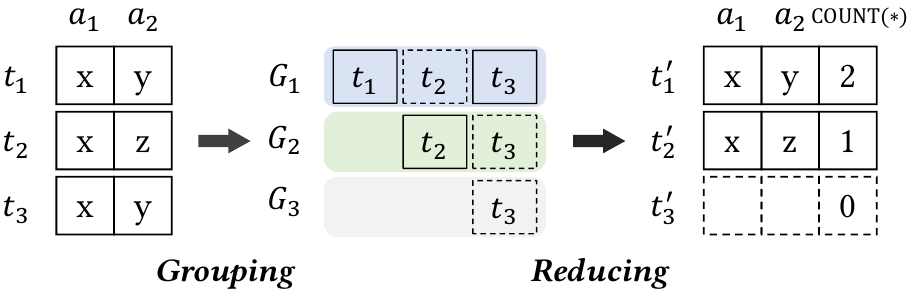}
\caption{The decomposition of \sqlgroupby.}
\label{fig:groupby}
\end{subfigure}
\caption{An illustration on \sqlexceptall and \sqlgroupby.}
\label{fig:exceptall_and_groupby}
\vspace{-10pt}
\end{figure}

Moreover, Figure~\ref{fig:rules-predicates} encodes rules for evaluating predicates.
Note that since we consider the three-valued logic, the logical operations $\alllogic$ are able to return \sqlnull values if the operands have \sqlnull, and the situations for \sqland and \sqlor become more complicated.
Furthermore, Figure~\ref{fig:rules-expressions} evaluates symbolic expressions in SQL queries.
Due to \sqlnull values, aggregation functions might also be \sqlnull if the predicate $\text{AllNull}(E, \schema, \context, \tuples)$ holds, and return numeric values otherwise.

\begin{figure}[!h]
\footnotesize
\[
\begin{array}{lcl}

\denot{\top}_{\schema, \context, \tuples} &=& \top \\
\denot{\bot}_{\schema, \context, \tuples} &=& \bot \\
\denot{\nullv}_{\schema, \context, \tuples} &=& \nullv \\
\denot{A_1 \odot A_2}_{\schema, \context, \tuples} &=& \text{let}~ v_1 = \denot{A_1}_{\schema, \context, \tuples}, v_2 = \denot{A_2}_{\schema, \context, \tuples} ~\text{in}~ \ite(v_1 = \nullv \lor v_2 = \nullv, \nullv, v_1 \odot v_2) \\
\denot{\text{IsNull}(E)}_{\schema, \context, \tuples} & = & \denot{E}_{\schema, \context, \tuples} = \nullv \\

\denot{\vec{E} \in \vec{v}}_{\schema, \context, \tuples} & = & 
\text{let} \\
& & \quad \vec{v} = [v_{1,1}, \ldots, v_{n,m}] \quad
\text{where}~ m = |\vec{E}| = |\vec{v}_i| \\
& & \text{in}~ \denot{\lor_{i=1}^{n} \land_{j=1}^{m} E_j = v_{i,j} }_{\schema, \context, \tuples}
\\

\denot{\vec{E} \in \query}_{\schema, \context, \tuples} & = & \text{let}~ \\
& & \quad 
\schema \vdash \query : [a_1, \ldots, a_m], \context \vdash \query \hookrightarrow [t_1, \ldots, t_n], \schema, \context \vdash \query \leadsto \formula \quad \text{where}~ m = |\vec{E}| \\
& & \quad v_{i,j} = \denot{a_j}_{\schema, \context, [t_i]} \quad 1 \leq i \leq n, 1 \leq j \leq m \\
& & ~\text{in}~ \formula \land \denot{\vec{E} \in \vec{v}}_{\schema, \context, \tuples}
\\

\denot{\phi_1 \land \phi_2}_{\schema, \context, \tuples} & = & \text{let}~ v_1 = \denot{\phi_1}_{\schema, \context, \tuples},
v_2 = \denot{\phi_2}_{\schema, \context, \tuples}
~\text{in}~ \\
& & \quad 
\ite((v_1 = v_2 = \nullv) \lor (v_1 = \nullv \land v_2 = \top) \lor (v_1 = \top \land v_2 = \nullv), \nullv, 
v_1  = \top \land v_2 = \top)
\\

\denot{\phi_1 \lor \phi_2}_{\schema, \context, \tuples} & = & 
\text{let}~ v_1 = \denot{\phi_1}_{\schema, \context, \tuples},
v_2 = \denot{\phi_2}_{\schema, \context, \tuples}
~\text{in}~ \\
& & \quad 
\ite((v_1 = v_2 = \nullv) \lor (v_1 = \nullv \land v_2 = \bot) \lor (v_1 = \bot \land v_2 = \nullv), \nullv, v_1  = \top \lor v_2 = \top)
\\

\denot{\neg \phi}_{\schema, \context, \tuples} & = & 
\text{let}~ v = \denot{\phi}_{\schema, \context, \tuples} ~\text{in}~
\ite(v = \nullv, \nullv, \neg v)
\\

\end{array}
\]
\caption{Rules for evaluating predicates.}
\label{fig:rules-predicates}
\end{figure}

\begin{figure}[!h]
\footnotesize
\[
\begin{array}{lcl}

\denot{a}_{\schema, \context, \tuples} & = & t.a \quad \text{where } t = \text{First}(\tuples) \\
\denot{v}_{\schema, \context, \tuples} & = & v \\
\denot{E_1 \allarith E_2}_{\schema, \context, \tuples} & = & \text{let}~ v_1 = \denot{E_1}_{\schema, \context, \tuples},~ v_2 = \denot{E_2}_{\schema, \context, \tuples} ~\text{in}~ \ite(v_1 = \nullv \lor v_2 = \nullv, \nullv, v_1 \allarith v_2) \\
\denot{\text{Count}(E)}_{\schema, \context, \tuples} & = & \ite(\text{AllNull}(E, \schema, \context, \tuples), \nullv, \Sigma_{t \in \tuples} \ite(\del(t) \lor \denot{E}_{\schema, \context, [t]} = \nullv, 0, 1)) \\
\denot{\text{Sum}(E)}_{\schema, \context, \tuples} & = & \ite(\text{AllNull}(E, \schema, \context, \tuples), \nullv, \Sigma_{t \in \tuples} \ite(\del(t) \lor \denot{E}_{\schema, \context, [t]} = \nullv, 0, \denot{E}_{\schema, \context, [t]})) \\
\denot{\text{Avg}(E)}_{\schema, \context, \tuples} & = & \denot{\text{Sum}(E)}_{\schema, \context, \tuples} /~ \denot{\text{Count}(E)}_{\schema, \context, \tuples} \\
\denot{\text{Min}(E)}_{\schema, \context, \tuples} & = & \ite(\text{AllNull}(E, \schema, \context, \tuples), \nullv, \text{min}_{t \in \tuples} \ite(\del(t) \lor \denot{E}_{\schema, \context, [t]} = \nullv, +\infty, \denot{E}_{\schema, \context, [t]})) \\
\denot{\text{Max}(E)}_{\schema, \context, \tuples} & = & \ite(\text{AllNull}(E, \schema, \context, \tuples), \nullv, \text{max}_{t \in \tuples} \ite(\del(t) \lor \denot{E}_{\schema, \context, [t]} = \nullv, -\infty, \denot{E}_{\schema, \context, [t]})) \\
\denot{\text{ITE}(\phi, E_1, E_2)}_{\schema, \context, \tuples} & = & \ite(\denot{\phi}_{\schema, \context, \tuples} = \top, \denot{E_1}_{\schema, \context, \tuples}, \denot{E_2}_{\schema, \context, \tuples}) \\
\denot{\text{Case}(\vec{\phi}, \vec{E}, E')}_{\schema, \context, \tuples} & = & \denot{\text{ITE}(\phi_1, E_1, \text{ITE}(\phi_2, E_2, \text{ITE}(\ldots, \text{ITE}(\phi_n, E_n, E'))))}_{\schema, \context, \tuples} \text{ where } n = |\vec{\phi}| = |\vec{E}| \\

\end{array}
\]
\caption{Rules for evaluating expressions. $f$ is an uninterpreted function. The $\text{First}(\tuples)$ function finds the first tuple from $\tuples$. The $\text{AllNull}(E, \schema, \context, \tuples)$ predicate asserts the attribute expression $E$ of all non-deleted tuples in $\tuples$ is \nullv, i.e., $\text{AllNull}(E, \schema, \context, \tuples) \leftrightarrow \land_{t \in \tuples} (\neg \del(t) \to \denot{E}_{\schema, \context, [t]} = \nullv)$.}
\label{fig:rules-expressions}
\end{figure}

\section{Additional Experimental Results}

\subsection{RQ: Importance of Supporting Rich Integrity Constraints}
\label{sec:exp-ic}

As mentioned in \textbf{RQ1} of Section~\ref{sec:comparison}, one reason that leads to a large number of ``unsupported'' benchmarks is due to baselines' poor support for integrity constraints. 
A natural question arises --- what if we \emph{do not} consider integrity constraints?
That is, how important is it to support integrity constraints?
This section aims to address exactly this.

\begin{figure}[!t]
\centering
\begin{subfigure}{\linewidth}
\includegraphics[width=.9\linewidth]{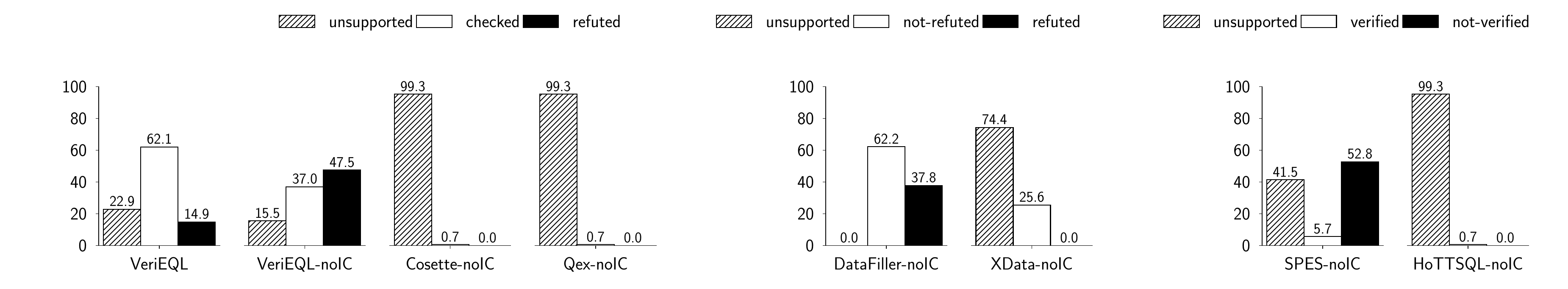}
\centering
\caption{\emph{LeetCode.}}
\label{fig:coverage-no-IC-leetcode}
\end{subfigure}
\hfill
\begin{subfigure}{\linewidth}
\includegraphics[width=.9\linewidth]{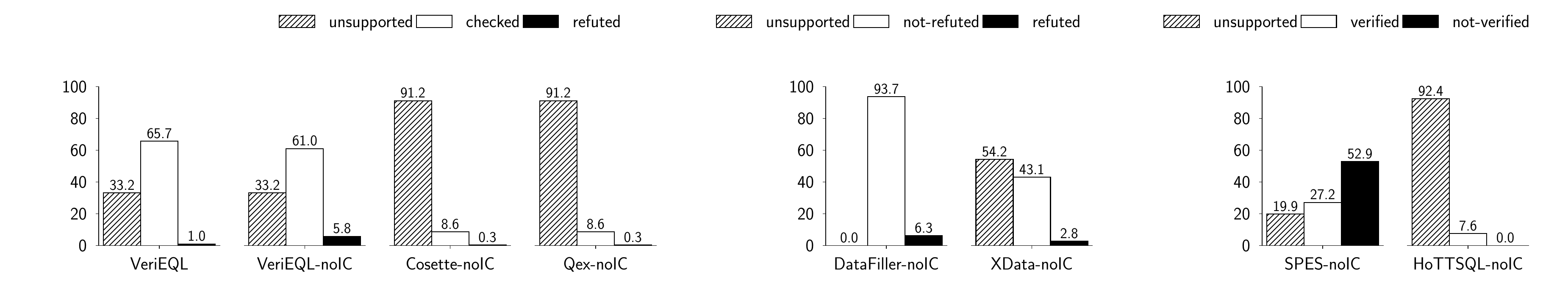}
\centering
\caption{\emph{Calcite.}}
\label{fig:coverage-no-IC-calcite}
\end{subfigure}
\hfill
\begin{subfigure}{\linewidth}
\includegraphics[width=.9\linewidth]{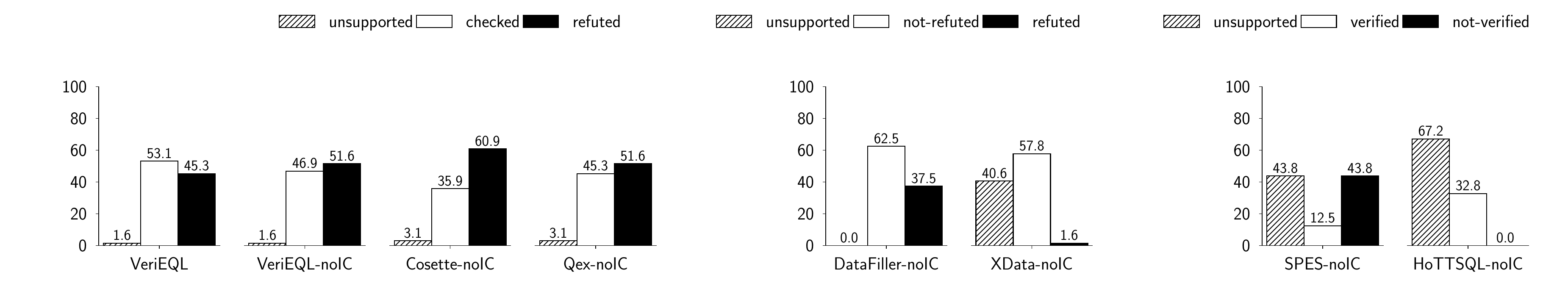}
\centering
\caption{\emph{Literature.}}
\label{fig:coverage-no-IC-literature}
\end{subfigure}
\vspace{-20pt}
\caption{Each sub-figure corresponds to one workload, where it shows the percentage of benchmarks in each outcome category. A tool name that is appended with ``-RQ2'' means the tool is run under \textbf{RQ2} setup. We also include $\tool$-RQ1 in this figure to facilitate an easy comparison of $\tool$'s results under two setups. To compare a baseline's results under two setups, please also refer to Figure~\ref{fig:coverage-result}.}
\label{fig:coverage-no-IC-result}
\vspace{-10pt}
\end{figure}

\newpara{Setup.}
We use the same setup as in \textbf{RQ1} except for one difference --- we discard integrity constraints. 
That is, we drop the integrity constraint $\constraint$ associated with a benchmark, regardless of whether or not the tool supports $\constraint$. 
This effectively means 
we are checking equivalence with respect to a \emph{superset of inputs} considered in \textbf{RQ1} which now also includes databases that violate $\constraint$. 

In what follows, we report the number of benchmarks in each category under this new setup,
\revision{annotated with a suffix ``-noIC''.}
The same setup is used for $\tool$ and baselines. 
We will also compare the results with \textbf{RQ1}'s from Section~\ref{sec:comparison}.

\evalfinding{\textbf{Take-away}: it is critical to support rich integrity constraints in order to avoid generating spurious counterexamples. For instance, without considering integrity constraints, the majority of  ``\revision{refuted}'' benchmarks --- many of which should be equivalent --- reported by $\tool$ are mistakenly proved non-equivalent, due to (spurious) inputs that violate the integrity constraints.}

\newpara{Results.}
Our main results are presented in Figure~\ref{fig:coverage-no-IC-result}. 
Let us first focus on $\tool$. 
For $\leetcode$ benchmarks, comparing $\tool$ and $\tool$-noIC, the most noticeable change is the \emph{significantly higher} ``\revision{refuted}'' bar --- from {14.9\%} (with IC) to {47.5\%} (without IC). 
Notably, more than {75\%} of \revision{noIC's ``refuted''} benchmarks are caused by \emph{spurious counterexamples}. 
Some of \revision{noIC's ``refuted''} benchmarks were previously ``checked'' in \textbf{RQ1} --- the ``checked'' bar lowers from {62.1\%} to {37\%} --- whereas the rest were previously ``unsupported'' in \textbf{RQ1}. 
A very small number (i.e., {3.3\%}) of benchmarks previously ``unsupported'' in \textbf{RQ1} now become ``checked'' in noIC --- which suggests that discarding integrity constraints does \emph{not} help prove many more benchmarks. 
For the other two workloads, we observe a similar trend: ``\revision{refuted}'' bars go up (due to spurious counterexamples) and ``checked'' bars go down, indicating it is important to support integrity constraints for $\calcite$ and $\literature$ as well. 

Let us now switch gears and briefly go over the baselines' results. 
On $\leetcode$, all baselines \revision{except \datafiller-noIC} still have high ``unsupported'' percentages, due to their limited support of SQL operations. 
\revision{
\datafiller is a random database generation tool, which does not reply on SQL queries to generate database instances, so its noIC ``unsupported'' percentages drop to zero.
However, as shown in Figure~\ref{fig:counterexamples} of Section~\ref{sec:counterexample}, most of the counterexamples generated by \datafiller-noIC are spurious.
}
$\spes$-noIC also has a high ``\revision{not-verified}'' bar: it can now successfully run on {58.5\%} but {52.8\%} are ``\revision{not-verified}''. 
But unfortunately, there is no easy way to validate its result in the entirety, as $\spes$ does not provide counterexamples. 
Among these {52.8\%}, $\tool$ can find valid counterexamples for {26.2\%}. 
We manually inspected {50} samples from the remaining {26.6\%}, and believe all of them are equivalent. 
The situation for $\calcite$ is similar, where $\spes$-noIC has the highest ``\revision{not-verified}'' bar ({52.9\%}), though we believe they are false alarms since most  $\calcite$ query pairs (if not all) should be equivalent. 
Such anecdotal evidence again indicates the importance of considering integrity constraints. 
Finally, for $\literature$, we also see noticeably higher ``\revision{refuted}'' --- we confirmed for $\cosette$, $\qex$, \revision{and \datafiller}, it is due to spurious counterexamples.

\newpara{Discussion.}
Careful readers may wonder why $\tool$-noIC's ``unsupported'' bar is still non-zero, given we discard integrity constraints. 
This is because $\tool$ currently does not support some SQL features such as {correlated subqueries (i.e., subqueries that use values from the outer query) and window functions.}

\subsection{Finding Bugs in MySQL}
\label{sec:mysql-bug}

\begin{figure}[!t]
\centering
\begin{minipage}{.62\linewidth}
\centering
\begin{subfigure}{\linewidth}
\footnotesize
\begin{tabular}{|l|l|}
\hline
$Q_1$ & 
\begin{tabular}[c]{@{}l@{}}
\sqlselectcolor \sqldistinctcolor page\_id \sqlascolor recommended\_page \\
\sqlfromcolor (\sqlselectcolor \sqlcasecolor \\
\hspace{5.5em}           \sqlwhencolor user1\_id~ =~ 1 \sqlthencolor user2\_id \\
\hspace{5.5em}           \sqlwhencolor user2\_id~ =~ 1 \sqlthencolor user1\_id \\
\hspace{5.5em}           \sqlelsecolor \sqlnull \\
\hspace{5.5em}         \sqlendcolor \sqlascolor user\_id \\
\hspace{2.3em}  \sqlfromcolor friendship) \sqlascolor tb1 \sqljoincolor likes \sqlascolor tb2 \\
\sqloncolor tb1.user\_id~ =~ tb2.user\_id \\
\sqlwherecolor page\_id \sqlnotcolor \sqlincolor (\sqlselectcolor page\_id \\ \hspace{9.3em} \sqlfromcolor likes \sqlwherecolor user\_id~ =~ 1) \\

\end{tabular} \\ \hline
$Q_2$  & 
\begin{tabular}[c]{@{}l@{}}
\sqlselectcolor \sqldistinctcolor page\_id \sqlascolor recommended\_page \\
\sqlfromcolor (\sqlselectcolor b.user\_id, b.page\_id \\
\hspace{2.3em}  \sqlfromcolor friendship a \sqlleftjoincolor likes b \\
\hspace{2.3em}  \sqloncolor \hspace{1.5em} (a.user2\_id~ =~ b.user\_id \sqlorcolor a.user1\_id=b.user\_id) \\ 
\hspace{3.5em} \sqlandcolor (a.user1\_id~ =~ 1 \sqlorcolor a.user2\_id~ =~ 1) \\
\hspace{2.3em}  \sqlwherecolor b.page\_id \sqlnotcolor \sqlincolor (\sqlselectcolor \sqldistinctcolor(page\_id) \\
\hspace{12.5em} \sqlfromcolor likes \sqlwherecolor user\_id=1) ) T
\end{tabular} 
\\ \hline
\end{tabular}
\caption{Two non-equivalent queries $(\query_1, \query_2)$ from $\leetcode$.}
\label{fig:mysql-bug-queries}
\end{subfigure}
\end{minipage}
\ \  
\begin{minipage}{.36\linewidth}
\small 
\centering
\begin{subfigure}{\linewidth}
\texttt{friendship} table: 
\\[10pt]
\begin{tabular}{|l|l|}
\hline
user1\_id & user2\_id \\ \hline
0         & 1         \\ \hline
\end{tabular}
\\[10pt] 
\texttt{likes} table: 
\\[10pt]
\begin{tabular}{|l|l|}
\hline
user\_id & page\_id \\ \hline
-1       & 0        \\ \hline
\end{tabular}
\caption{Counterexample input database.}
\label{fig:mysql-bug-cex}
\end{subfigure}
\\[30pt] 
\begin{subfigure}{\linewidth}
\begin{tabular}{|l|l|}
\hline
page\_id \\ \hline
\sqlnull \\ \hline
\end{tabular}
\caption{$\query_2$'s expected output.}
\label{fig:mysql-bug-q2-expected-output}
\end{subfigure}
\end{minipage}
\vspace{-5pt}
\caption{Queries from Figure~\ref{fig:mysql-bug-queries}, when executed using MySQL on the database from Figure~\ref{fig:mysql-bug-cex}, both return the empty table, but they should return different output tables according to SQL semantics. In particular, $\query_1$ should return the empty table and $\query_2$ should return the table in Figure~\ref{fig:mysql-bug-q2-expected-output}.}
\label{fig:mysql-bug} 
\vspace{-10pt}
\end{figure}

As mentioned in Section~\ref{sec:counterexample}, for two $\leetcode$ benchmarks, $\tool$ generates ``spurious'' counterexamples --- that are actually genuine --- due to a bug in MySQL.
Figure~\ref{fig:mysql-bug} shows one such benchmark\footnote{\href{https://leetcode.com/problems/page-recommendations/}{https://leetcode.com/problems/page-recommendations/}} $(\query_1, \query_2)$ (see Figure~\ref{fig:mysql-bug-queries}) and the corresponding counterexample database $\db$ generated by $\tool$ (see Figure~\ref{fig:mysql-bug-cex}). 
This benchmark has an integrity constraint $\constraint$ requiring: 
(i) (\texttt{user1\_id}, \texttt{user2\_id}) is the primary key of the \texttt{friendship} table, 
(ii) (\texttt{user\_id}, \texttt{page\_id}) is the primary key of the \texttt{likes} table, 
and (iii) the \texttt{friendship} table does not have any rows with $\texttt{user1\_id} = \texttt{user2\_id}$. 
As we can see, $\db$ satisfies $\constraint$. Furthermore, we believe $\query_1$'s output on $\db$ should be the empty table, because the $\sqlon$ predicate associated with the $\sqljoin$ operation in $Q_1$ is always false. 
On the other hand, $\query_2$ should return the table shown in Figure~\ref{fig:mysql-bug-q2-expected-output}: the $\sqlleftjoin$ in $Q_2$ will retain all records from the \texttt{friendship} table in the output, and the $\sqlwhere$ predicate will be true since \texttt{b.page\_id} (i.e., \texttt{likes.page\_id} which is 0) will not be in an empty set; therefore, every tuple after the join will be eventually preserved in $Q_2$'s output.
However, both queries return the empty table on MySQL's latest release version 8.0.32. The MySQL verification team had confirmed and classified this bug with \emph{serious} severity.

\subsection{Detailed Analysis}
\label{sec:RQ4}

This section reports detailed performance statistics of $\tool$.

\newpara{Time breakdown.}
$\tool$ essentially performs two steps: constraint generation and constraint solving, where the latter takes 67.7\% of the total time on average across all benchmarks.


\newpara{Efficiency at disproving query equivalence.}
$\tool$ is very efficient at disproving query pairs. 
For example, among the {3,586} non-equivalent $\leetcode$ benchmarks that $\tool$ can generate counterexamples for in \textbf{RQ1}, {3,108} of them can be disproved within 1 second, {3,528} can be disproved within 5 seconds, and the slowest one takes {251} seconds. 
All $\calcite$'s non-equivalent benchmarks can be proved within 2 seconds. 
Finally, all {29} non-equivalent $\literature$ benchmarks are identified within 2 minutes where {26} of them took less than 10 seconds.

\newpara{Efficiency of bounded verification.}
This experiment concerns the number of benchmarks whose \emph{bounded equivalence} can be proved by $\tool$ in a given timeout, with respect to different input bounds. 
In particular, we vary the bound from 1 to {10}, and for each bound we record the number of benchmarks whose equivalence can be \emph{completely} verified using a 10-minute timeout. That is, if $\tool$ does not terminate on a benchmark within the given timeout, we do not count this benchmark as ``proved'' for the given bound. 
The main result can be found in Table~\ref{tab:varying-bound}. 
The key message is that $\tool$ can efficiently prove most of the benchmarks up to bound 5. 
While the number drops for larger bounds, $\tool$ can still prove a significant number of benchmarks quite efficiently. 

\begin{table}[t]
\caption{For each workload, we vary the input database size from 1 to 10 and show the number of benchmarks that can be \emph{completely} verified by $\tool$ for the given bound using a 10-minute timeout. If $\tool$ does not finish the bounded verification within the timeout, we do not count it as ``checked.'' We also show the average and median verification times for those benchmarks that terminate. }
\vspace{-10pt}
\footnotesize 
\centering
\setlength{\tabcolsep}{3pt}
\begin{tabular}{ c  ||  r | r | r || r  | r |  r ||  r | r | r }
\toprule
\multirow{2}{*}{\rotatebox{0}{\emph{bound}}} 
& \multicolumn{3}{c||}{$\leetcode$} 
& \multicolumn{3}{c||}{$\calcite$} 
& \multicolumn{3}{c}{$\literature$}
\\[2pt]
& \emph{\# proved} & \emph{mean (sec)} & \emph{median (sec)} 
& \emph{\# proved} & \emph{mean (sec)} & \emph{median (sec)} 
& \emph{\# proved} & \emph{mean (sec)} & \emph{median (sec)} 
\\ 
\midrule
%
%

1 &14905	&0.4	&0.3	&261	&0.7	&0.7	&34	&0.7	&0.7 \\
2 &14894	&0.7	&0.4	&261	&0.4	&0.3	&34	&2.6	&0.3 \\
3 &14618	&16.1	&0.8	&261	&3.1	&0.4	&32	&4.6	&0.6 \\
4 &11757	&24.5	&1.6	&242	&9.7	&0.6	&29	&17.0	&1.3 \\
5 &10564	&73.0	&5.1	&224	&33.6	&0.7	&26	&71.9	&3.4 \\
6 &7779	&63.0	&4.5	&193	&33.0	&0.8	&17	&19.3	&0.8 \\
7 &6203	&41.9	&6.6	&162	&15.5	&0.8	&15	&43.3	&1.2 \\
8 &5344	&34.4	&7.7	&153	&20.0	&1.0	&13	&12.6	&2.8 \\
9 &4942	&53.5	&12.0	&144	&26.5	&1.3	&12	&37.0	&5.7 \\
10 &4160	&49.7	&4.6	&134	&14.4	&1.7	&12	&65.2	&12.4 \\
\bottomrule
\end{tabular}
\label{tab:varying-bound}
\end{table}


\section{Proof} \label{sec:proof}
In this section, we provide the proof of lemmas and theorems in the main paper.

\begin{proof}[Proof of Theorem~\ref{lem:ic}:]
\revision{
Given a symbolic database $\context$ and an integrity constraint $\constraint$, consider a formula $\formula$ such that $\context \vdash \constraint \rightsquigarrow \formula$. If $\formula$ is satisfiable, then the model of $\formula$ corresponds to a database consistent with $\context$ that satisfies $\constraint$. If $\formula$ is unsatisfiable, then no database consistent with $\context$ satisfies $\constraint$.
}
\end{proof}

\begin{proof}[Proof.]
By structural induction on $\constraint$. 

\begin{enumerate}

\item Base case: $\constraint = \text{PK}(R, \vec{a})$.

Suppose $\context(R) = [t_1, \ldots, t_n]$ and $m = |\vec{a}|$ by Figure~\ref{fig:rules-ic-preds} and, for every tuple $t_i \in$ $\context(R)$, it has $t_i.a_k \neq \nullv$ where $a_k$ denotes the attribute and $1 \leq k \leq m$. 
Also, $\formula_2$ in Figure~\ref{fig:rules-ic-preds} ensures no two tuples $t_i$ and $t_j$ where $t_i \in \context(R)$ and $t_j \in \context(R)$ such that $\land_{k=1}^{m} t_i.a_k = t_j.a_k$ is true.  
Therefore, it holds that if $\formula$ is satisfiable, we have a model of $\formula$ corresponds to a database consistent with $\context$ that satisfies $\text{PK}(R, \vec{a})$ as $\formula_1 \land \formula_2$ satisfies $\forall t_i \in \context(R) . t_i.a \neq \nullv  \land \forall t_i \in \context(R). \forall t_j \in \context(R) . i \neq j \to \neg (\land_{k=1}^{m} t_i.a_k = t_j.a_k)$, thus ensuring the model conforms to primary key properties that 1) attributes $\vec{a}$ of every tuple in $\context(R)$ are non-\sqlnull, and 2) there are no two tuples in $\context(R)$ such that their values of $\vec{a}$ are identical.  And on the contrary, if $\formula$ is unsatisfiable, then it is necessarily true that there exists no database consistent with $\context$ that can satisfy $\text{PK}(R, \vec{a})$.

\item Base case: $\constraint = \text{FK}(R_1, a_1, R_2, a_2)$.

Suppose $\context(R_1) = [t_1, \ldots, t_n]$ and $\context(R_2) = [t_1', \ldots, t_m']$, by Figure~\ref{fig:rules-ic-preds}, $\formula$ ensures for any tuple $t_i \in \context(R_1)$, $t_i.a_1 \in [t_1'.a_2, \ldots, t_m'.a_2]$.
Therefore, it holds that if $\formula$ is satisfiable, we have a model of $\formula$ corresponds to a database consistent with $\context$ that satisfies $\text{FK}(R_1, a_1, R_2, a_2)$ as $\formula$ satisfies $\forall t_i \in \context(R_1). t_i.a_1 \in [t_1'.a_2, \ldots, t_m'.a_2]$ (the value of $a_1$ of every tuple in $\context(R_1)$ references the value $a_2$ of any tuple in $\context(R_2)$); also, if $\formula$ is unsatisfiable, then it is necessarily true that there exists no database consistent with $\context$ satisfies $\text{FK}(R_1, a_1, R_2, a_2)$.

\item Base case: $\constraint = \text{NotNull}(R, a)$.

Suppose $\context(R) = [t_1, \ldots, t_n]$, by Figure~\ref{fig:rules-ic-preds}, $\formula$ ensures for any tuple $t_i$ in $\context(R)$, $t_i.a \neq \nullv$.
Therefore, it holds that if $\formula$ is satisfiable, we have a model of $\formula$ corresponding to a database consistent with $\context$ where $\forall t_i \in \context(R) . t_i.a \neq \nullv$, which satisfies that the value of $a$ of every tuple in $\context(R)$ is non-\sqlnull, and also if $\formula$ is unsatisfiable, then it is necessarily true that there exists no database consistent with $\context$ satisfies $\text{NotNull}(R, a)$.

\item Base case: $\constraint = \text{Check}(R, \psi)$.

Suppose $\context(R) = [t_1, \ldots, t_n]$, by Figure~\ref{fig:rules-ic-preds}, $\formula$ ensures for any tuple $t_i$, the predicate $\denot{\psi}_{t_i}$ evaluates to $\top$.
Therefore, it holds that if $\formula$ is satisfiable, the model of $\formula$ corresponds to a database consistent with $\context$ that satisfies $\forall t_i \in \context(R) . \denot{\psi}_{t_i} = \top$ that conforms to the \text{Check} constraint that every tuple in $\context(R)$ satisfies $\psi$, and if $\formula$ is unsatisfiable, then it is necessarily true that there exists no database consistent with $\context$ satisfies $\text{Check}(R, \psi)$.

\item Base case: $\constraint = \text{Inc}(R, a, v)$.

Suppose $\context(R) = [t_1, \ldots, t_n]$, by Figure~\ref{fig:rules-ic-preds}, $\formula_2$ foremost constrains the attribute $t_i.a$ is non-\sqlnull for every tuple $t_i \in \context(R)$.  $\formula_1$ then ensures for any tuple $t_i$, the attribute $t_i.a = v + i - 1$, where $v$ is the initial value.
Therefore, it holds that if $\formula$ is satisfiable, we have a model of $\formula$ corresponding to a database consistent with $\context$ that satisfies $\forall t_i \in \context(R) . t_i.a \neq \sqlnull \land t_i.a = v + i - 1$ and hence ensures the value of $a$ of every tuple in $\context(R)$ is incrementally increased by 1 from the initial value $v$; if $\formula$ is unsatisfiable, then it is necessarily true that there exists no database consistent with $\context$ satisfies $\text{Inc}(R, a, v)$.

\item Inductive case: $\constraint'' = \constraint \land \constraint'$.

By the base cases shown above, the Theorem for the inductive case $\constraint'' = \constraint \land \constraint'$ can be proved.

\end{enumerate}
\end{proof}

\begin{proof}[Proof of Theorem~\ref{lem:query1}:]
\revision{
Let $\db$ be a database over schema $\schema$ and $o_{R}$ be the output relation of running query $\query$ over $\db$. Consider a symbolic database $\context$ over $\schema$, a symbolic relation $R$ and a formula $\formula$ such that $\context \vdash \query \hookrightarrow R$ and $\schema, \context \vdash \query \leadsto \formula$. For any interpretation $\interpretation$ such that $\interpretation(\context) = \db$, there exists an extension $\interpretation'$ of $\interpretation$ such that (1) the result of query $\query$ over $\db$ is $\interpretation'(R)$, and (2) the database $\db$ and $\interpretation'$ jointly satisfy $\formula$, i.e., 
\[
\begin{array}{l}
(\context \vdash \query \hookrightarrow R) \land (\schema, \context \vdash \query \leadsto \formula) \land (\interpretation(\context) = \db) \Rightarrow \\
\qquad \qquad \qquad \qquad \qquad \qquad \qquad \qquad \exists \interpretation' \sqsupseteq \interpretation. \interpretation'(R) = \denot{\query}_{\db} \land (\interpretation', \db[o_R \mapsto \interpretation'(R)] \models \formula)
\end{array}
\]
}

\vspace{10pt}

Prove by structural induction on $\query$. ~\footnote{\revision{To be concise, we widely use a function $\ptr$ in this proof section. The function $\ptr$ takes as input an index $i$  and returns an index $j$ s.t. $\interpretation(t_j) = x_i$ where $x_i \in \denot{\query}_\db$ and $t_j \in R$.
}}

\begin{enumerate}
\item Base case: $\query = T$. 

$R = \context(T)$ by Figure~\ref{fig:rules-tuples}.
$\formula = \top$ by Figure~\ref{fig:rules-encode}.

Take $\interpretation' = \interpretation$ and $\db' = \db[o_R \mapsto \interpretation(R)]$.
Obviously, we have (a) $\interpretation'(R) = \interpretation(R) = \interpretation(\context(T)) = \db(T) = \denot{T}_\db$ because $\denot{T}_\db = \db(T)$ by Figure~\ref{fig:rules-encode-b}, (b) $\interpretation', \db' \models \top$.

Thus, Theorem~\ref{lem:query1} for the inductive case $\query = T$ is proved.

\item Inductive case: $\query' = \proj_{L}(\query)$. 

Suppose that $\schema \vdash \query : [a_1, \ldots, a_m]$ by Figure~\ref{fig:rules-schema}; 
$\schema \vdash \filter_L(\query) : [a_1', \ldots, a_l']$ where $l = |L|$ by Figure~\ref{fig:rules-schema};
$\context \vdash \query \hookrightarrow R$ where $R$ has symbolic tuples $[t_1, \ldots, t_n]$ by Figure~\ref{fig:rules-tuples};
$\schema, \context \vdash \query \leadsto \formula$ by Figure~\ref{fig:rules-encode-b};
$\denot{\proj_L(\query)}_{\db} = \site(\text{hasAggr}(L), [\denot{L}_{\db, \denot{\query}_{\db}}], \smap(\denot{Q}_\db, \lambda x. \denot{L}_{\db, x}))$ by Figure~\ref{fig:proof-sematics2}.

By the inductive hypothesis, we have $\interpretation(R) = \denot{\query}_\db$ and $\interpretation, \db \models \formula$.

Let us discuss $\text{hasAggr}(L)$ in two cases.
\begin{enumerate}[label=(\alph*)]
\item \label{proj_L_query_a}
If $\text{hasAggr}(L)$ is true, then $\context \vdash \proj_{L}(\query) \hookrightarrow R'$ where $R'$ only has one symbolic tuple $[t_1']$ by Figure~\ref{fig:rules-tuples}, $\denot{\proj_L(\query)}_{\db} = [\denot{L}_{\db, \denot{\query}_{\db}}]$ by Figure~\ref{fig:proof-sematics2}, $\schema, \context \vdash \proj_{L}(\query) \leadsto \formula \land \land_{i=1}^{l} \formula_{i}' \land \neg \del(t_1')$ where $\formula_{i}' = \denot{a_j'}_{\schema, \context, [t_1']} = \denot{a_j'}_{\schema, \context, \vec{t}}$ by Figure~\ref{fig:rules-encode}.

Take $\interpretation'$ s.t. $\interpretation'(\del)(t_i) \leftrightarrow \interpretation(\del)(t_i)$, $\interpretation'(t_i) = \interpretation(t_i)$, $\interpretation'(\del)(t_1') = \bot$ and $\interpretation'(t_1') = [(\text{ToString}(a_1'), \denot{a_1'}_{\db, \interpretation(R)}), \ldots, (\text{ToString}(a_l'), \denot{a_l'}_{\db, \interpretation(R)})]$.

By the definition of $\interpretation'$ and the inductive hypothesis $\interpretation(R) = \denot{\query}_{\db}$, we know 
\[
\begin{array}{rcl}
\interpretation'(\proj_L(\query)) 
&=& [(\text{ToString}(a_1'), \denot{a_1'}_{\db, \interpretation(R)}), \ldots, (\text{ToString}(a_l'), \denot{a_l'}_{\db, \interpretation(R)})] \\
&=& [(\text{ToString}(a_1'), \denot{a_1'}_{\db, \denot{\query}_\db}), \ldots, (\text{ToString}(a_l'), \denot{a_l'}_{\db, \denot{\query}_\db})] \\
&=& [\denot{L}_{\db, \denot{\query}_\db}] \\
&=& \denot{\proj_L(\query)}_{\db}
\end{array}
\]

Let $\db' = \db[o_{\proj_L(\query)} \mapsto \denot{\proj_L(\query)}_{\db}]$.
By Lemma~\ref{lem:attribute}, we know $\interpretation(\denot{a}_{\schema, \context, \tuples}) = \denot{a}_{\interpretation(\context), \interpretation(\tuples)}$ and , therefore, $\interpretation, \db \models \denot{a'}_{\schema, \context, \tuples} = \denot{a}_{\schema, \context, \tuples} \Leftrightarrow \interpretation, \db \models \denot{a'}_{\interpretation(\context) , \interpretation(\tuples)} = \denot{a}_{\interpretation(\context) , \interpretation(\tuples)}$.
\[
\begin{array}{rcl}
\interpretation', \db' &\models& \land_{j=1}^{l} \denot{a_j'}_{\interpretation(\context), [\interpretation(t_1')]} = \denot{a_j'}_{\interpretation(\context), \interpretation([t_1, \ldots, t_n])} \\
\interpretation', \db' &\models& \land_{j=1}^{l} \denot{a_j'}_{\schema, \context, [t_1']} = \denot{a_j'}_{\schema, \context, [t_1, \ldots, t_n]} \\
\end{array}
\]
Thus, we have $\interpretation', \db' \models \formula \land \land_{i=1}^{l} \formula_i' \land \neg \del(t_1')$

\item \label{proj_L_query_b}
If $\text{hasAggr}(L)$ is false, then $\context \vdash \proj_{L}(\query) \hookrightarrow R'$ where $R'$ has symbolic tuples $[t_1', \ldots, t_n']$ by Figure~\ref{fig:rules-tuples}, $\denot{\proj_L(\query)}_{\db} = \smap(\denot{\query}_\db, \lambda x. \denot{L}_{\db ,x})$ by Figure~\ref{fig:proof-sematics2}, $\schema, \context \vdash \proj_{L}(\query) \leadsto \formula \land \land_{i=1}^{n} \formula_{i}'$ where $\formula_{i}' = \land_{j=1}^{l} \denot{a_j'}_{\schema, \context, [t_i']} = \denot{a_j'}_{\schema, \context, [t_i]} \land \del(t_i') \leftrightarrow \del(t_i)$ by Figure~\ref{fig:rules-encode}.

Take $\interpretation'$ s.t. $\interpretation'(\del)(t_i) \leftrightarrow \interpretation(\del)(t_i)$, $\interpretation'(t_i) = \interpretation(t_i)$, $\interpretation'(\del)(t_i') \leftrightarrow \interpretation(\del)(t_i)$ and $\interpretation'(t_i') = [(\text{ToString}(a_1'), \denot{a_1'}_{\db, [\interpretation(t_i)]}), \ldots, (\text{ToString}(a_l'), \denot{a_l'}_{\db, [\interpretation(t_i)]})] = \denot{L}_{\db, [\interpretation(t_i)]}$.

Let $x_i \in \denot{\query}_\db$ where $1 \leq i \leq c$ and $c = |\denot{\query}_\db|$, then we know $x_i \in \interpretation(R)$ because $\interpretation(R) = \denot{\query}_\db$, $\interpretation(\del)(x_i) = \bot$.
Further, for any tuple $x_i \in \denot{\query}_\db$, we have two symbolic tuples $t_{\ptr(i)} \in \query$ and $t_{\ptr(i)}' \in \proj_L(\query)$ s.t. $\interpretation'(\del)(t_{\ptr(i)}') = \interpretation(\del)(t_{\ptr(i)}) = \interpretation(\del)(x_i) = \bot$ and $\denot{L}_{\db, [x_i]} = \denot{L}_{\db, [\interpretation(t_{\ptr(i)})]} = \interpretation'(t_{\ptr(i)}')$.

By the semantics of $\smap$, we have
\[
\begin{array}{rcl}
\denot{\proj_L(\query)}_{\db}
&=& \smap(\denot{\query}_\db, \lambda x. \denot{L}_{\db ,x}) \\
&=& [\denot{L}_{\db, [x_1]}, \ldots, \denot{L}_{\db, [x_c]}] \\
&=& [\denot{L}_{\db, [\interpretation(t_{\ptr(1)})]}, \ldots, \denot{L}_{\db, [\interpretation(t_{\ptr(c)})]}] \\
&=& [\interpretation'(t_{\ptr(1)}'), \ldots, \interpretation'(t_{\ptr(c)}')] \\
&=& \interpretation'([t_{\ptr(1)}', \ldots, t_{\ptr(c)}']) \\
&=& \interpretation'(\proj_L(\query)) \\
\end{array}
\]

Let $\db' = \db[o_{} \mapsto \denot{\proj_\phi(\query)}_\db]$.
Since $\denot{a_j'}_{\db, [\interpretation(t_i')]} = \denot{a_j'}_{\db, [\interpretation(t_i)]} \Leftrightarrow \denot{a_j'}_{\schema, \context, [t_i']} = \denot{a_j'}_{\schema, \context, [t_i]}$ and $\interpretation(\del)(t_i') = \interpretation(\del)(t_i)$, we know $\interpretation', \db' \models \formula_i'$ and, therefore, $\interpretation', \db' \models \formula \land \land_{i=1}^{n} \formula_i'$.
\end{enumerate}

Hence, $\interpretation', \db' \models \formula \land \land_{i=1}^{n} \formula_i'$ by \ref{proj_L_query_a} and \ref{proj_L_query_b} and Theorem~\ref{lem:query1} for the inductive case $\query' = \proj_L(\query)$ is proved.

\item Inductive case: $\query' = \filter_{\phi}(\query)$. 

Suppose that $\context \vdash \query \hookrightarrow R$ where $R$ has symbolic tuples $[t_1, \ldots, t_n]$, $\context \vdash \filter_{\phi}(\query) \hookrightarrow R'$ where $R'$ has symbolic tuples $[t_1', \ldots, t_n']$ by Figure~\ref{fig:rules-tuples}; $\schema, \context \vdash \query \leadsto \formula$, $\schema, \context \vdash \filter_\phi(\query) \leadsto \formula \land \land_{i=1}^n \formula_i'$ where $\formula_i' = (\neg \del(t_i) \land \denot{\phi}_{\schema, \context, [t_i]} = \top) \to t_i' = t \land (\del(t_i) \lor \denot{\phi}_{\schema, \context, [t_i]} \neq \top) \to \del(t_i')$ by Figure~\ref{fig:rules-encode-b}.

By the inductive hypothesis, we have $\interpretation(R) = \denot{\query}_{\db}$ and $\interpretation, \db \models \formula$.

Take $\interpretation'$ s.t. $\interpretation'(\del)(t_i) \leftrightarrow \interpretation(\del)(t_i)$, $\interpretation'(t_i') = \interpretation'(t_i) = \interpretation(t_i)$ and $\interpretation'(\del)(t_i') \leftrightarrow \interpretation(\del)(t_i) \lor \denot{\phi}_{\db, [\interpretation(t_i)]} \neq \top$.

For any tuple $x_i \in \denot{\query}_\db$, we have a symbolic tuple $t_{\ptr(i)} \in R$ s.t. $\interpretation(\del)(t_{\ptr(i)}) = \interpretation(\del)(x_i) \\$ $= \bot$ and $\interpretation(t_{\ptr(i)}) = \interpretation(x_i)$.
If $\denot{\phi}_{\db, [x_i]} = \top$, then $\interpretation'(\del)(t_{\ptr(i)}') =  \interpretation(\del)(t_{\ptr(i)}) \lor  \\$ $\denot{\phi}_{\db, [\interpretation(t_{\ptr(i)})]} \neq \top = \bot$ and $\interpretation'(t_{\ptr(i)}') = \interpretation(t_{\ptr(i)}) = x_i$.
Otherwise, $\interpretation'(\del)(t_{\ptr(i)}') = \interpretation(\del)(t_{\ptr(i)}) \lor (\denot{\phi}_{\db, [\interpretation(t_{\ptr(i)})]} \neq \top) = \top$ and $\interpretation'(t_{\ptr(i)}') = \interpretation(t_{\ptr(i)}) = x_i$.
Therefore, for any $x_i'' \in \denot{\filter_\phi(\query)}_\db$, we have a symbolic tuple $t_{ptr(i)}' \in R'$.
$\denot{\filter_\phi(\query)}_\db = \interpretation(R')$ is proved.

Let $\db' = \db[o_{R'} \mapsto \denot{\filter_\phi(\query)}_\db]$. 
By Lemma~\ref{lem:predicate}, we know $\interpretation(\denot{\phi}_{\schema, \context, [x]}) = \denot{\phi}_{\interpretation(\context), \interpretation([x])} = \denot{\phi}_{\db, [\interpretation(x)]}$ and, therefore, $\interpretation, \db \models \denot{\phi}_{\schema, \context, [x]} \Leftrightarrow \interpretation, \db \models \denot{\phi}_{\db, [\interpretation(x)]}$ and $\interpretation, \db \models \formula_i' \Leftrightarrow \interpretation, \db \models (\neg \del(t_i) \land \denot{\phi}_{\db, [\interpretation(t_i)]} = \top) \to t_{i}' =  t_{i} \land (\del(t_i) \lor \denot{\phi}_{\db, [\interpretation(t_i)]} \neq \top) \to \del(t_i')$.

For any  $t_i \in R$, let us discuss it in two cases.
\begin{enumerate}[label=(\alph*)]
\item \label{filter_phi_Q_a} If $\interpretation(\del)(t_i) = \top$, then, by the definition of $\interpretation'$, we have $\interpretation'(\del)(t_i) = \interpretation(\del)(t_i) = \top$, $\interpretation'(\del)(t_i') = \interpretation(\del)(t_i) \lor (\denot{\phi}_{\db, [\interpretation(t_i)]} \neq \top) = \top$, $\interpretation'(t_i') = \interpretation'(t_i) = \interpretation(t_i)$,
\[
\interpretation', \db' \models (\neg \del(t_i) \land \denot{\phi}_{\db', [\interpretation'(t_i)]} = \top) \to t_{i}' = t_{i}
\]
and
\[
\interpretation', \db' \models (\del(t_i) \lor \denot{\phi}_{\db', [\interpretation'(t_i)]} \neq \top) \to \del(t_i')
\]
Thus, $\interpretation', \db' \models \formula_i'$.

\item \label{filter_phi_Q_b} If $\interpretation(\del)(t_i) = \bot$, then, by the definition of $\interpretation'$, we have $\interpretation'(\del)(t_i) = \interpretation(\del)(t_i) = \bot$, $\interpretation'(\del)(t_i') = \interpretation(\del)(t_i) \lor (\denot{\phi}_{\db, [\interpretation(t_i)]} \neq \top) =  \denot{\phi}_{\db, [\interpretation(t_i)]} \neq \top$ and $\interpretation'(t_i') = \interpretation'(t_i) = \interpretation(t_i)$. 

Let us further discuss $\denot{\phi}_{\db, [\interpretation(t_i)]}$ in two cases.
\begin{enumerate}[label=(b\arabic*)]
\item \label{filter_phi_Q_b1} If $\denot{\phi}_{\db, [\interpretation(t_i)]} = \top$, then $\interpretation'(\del)(t_i') = \denot{\phi}_{\db, [\interpretation(t_i)]} \neq \top = \bot$, we have
\[
\interpretation', \db' \models (\neg \del(t_i) \land \denot{\phi}_{\db, [\interpretation(t_i)]} = \top) \to t_{i}' =  t_{i}
\]
and
\[
\interpretation', \db' \models (\del(t_i) \lor \denot{\phi}_{\db, [\interpretation(t_i)]} \neq \top) \to \del(t_i')
\]

\item \label{filter_phi_Q_b2} If $\denot{\phi}_{\db, [\interpretation(t_i)]} \neq \top$, then $\interpretation'(\del)(t_i') = \denot{\phi}_{\db, [\interpretation(t_i)]} \neq \top = \top$, we have
\[
\interpretation', \db' \models (\neg \del(t_i) \land \denot{\phi}_{\db, [\interpretation(t_i)]} = \top) \to t_{i}' =  t_{i}
\]
and
\[
\interpretation', \db' \models (\del(t_i) \lor \denot{\phi}_{\db, [\interpretation(t_i)]} \neq \top) \to \del(t_i')
\]

\end{enumerate}

Hence, $\interpretation', \db' \models \formula_i'$ by \ref{filter_phi_Q_b1} and \ref{filter_phi_Q_b2}.

\end{enumerate}

Thus, $\interpretation', \db' \models \formula \land \land_{i=1}^{n} \formula_i'$ by \ref{filter_phi_Q_a} and \ref{filter_phi_Q_b} and Theorem~\ref{lem:query1} for the inductive case $\query' = \filter_\phi(Q)$ is proved.

\item Inductive case: $\query' = \rename_T(\query)$.

Suppose that $\schema \vdash \query : [a_1, \ldots, a_p]$ by Figure~\ref{fig:rules-schema}; $\context \vdash \query \hookrightarrow R$ where $R$ has symbolic tuples $[t_1, \ldots, t_n]$, $\context \vdash \rename_T(\query) \hookrightarrow R'$ where $R'$ has symbolic tuples $[t_1', \ldots, t_n']$ by Figure~\ref{fig:rules-tuples}; $\schema, \context \vdash \query \leadsto \formula$, $\schema, \context \vdash \formula \land \land_{i=1}^{n} t_i' = t_i$ by Figure~\ref{fig:rules-encode-b}.

By the inductive hypothesis, we have $\interpretation(R) = \denot{\query}_{\db}$ and $\interpretation, \db \models \formula$.

Take $\interpretation'$ s.t. $\interpretation'(\del)(t_i') \leftrightarrow \interpretation'(\del)(t_i) \leftrightarrow \interpretation(\del)(t_i)$, $\interpretation'(t_i) = \interpretation(t_i)$, and $\interpretation'(t_i') = \\$ $ [(\text{rename}(R, \text{ToString}(a_1)), \denot{a_1}_{\db, [\interpretation(t_i)]}), \ldots, (\text{rename}(R, \text{ToString}(a_p)), \denot{a_p}_{\db, [\interpretation(t_i)]})]$.

Let $x_i \in \denot{\query}_\db$ where $1 \leq i \leq n$, then we have $t_{\ptr(i)} \in \interpretation(R)$ because $\interpretation(R) = \denot{\query}_\db$.
Also, by the semantics of $\smap$, we know
\[
\begin{array}{rcl}
\denot{\rho_T(\query)}_\db
&=& \smap(\denot{\query}_\db, \lambda x. \smap(x, \lambda (n, v). (v, \text{rename}(T, n)))) \\
&=& \smap([x_1, \ldots, x_n], \lambda x. \smap(x, \lambda (n, v). (v, \text{rename}(T, n)))) \\
&=& [ \\
&& \quad [(\text{rename}(T, \text{ToString}(a_1)), \denot{a_1}_{\db, [x_1]}), \ldots, \\
&& \qquad \qquad (\text{rename}(T, \text{ToString}(a_p)), \denot{a_p}_{\db, [x_1]})] \\
&& \quad \ldots, \\
&& \quad [(\text{rename}(T, \text{ToString}(a_1)), \denot{a_1}_{\db, [x_n]}), \ldots, \\
&& \qquad \qquad (\text{rename}(T, \text{ToString}(a_p)), \denot{a_p}_{\db, [x_n]})] \\
&& ] \\
\end{array}
\]
By the definition of $\ptr$, we know, for any $x_i' \in \denot{\rho_T(\query)}_\db$, there exists a symbolic tuple $t_{\ptr(i)}' \in R'$ s.t. $\interpretation'(t_{\ptr(i)}') = x_i'$.
Further, $\denot{\rho_T(\query)}_\db = [x_1', \ldots, x_n'] = [\interpretation'(t_{\ptr(1)}'), \ldots, $ $\interpretation'(t_{\ptr(n)}')] = \interpretation'([t_{\ptr(1)}', \ldots, t_{\ptr(n)}']) = \interpretation'(R')$.

Let $\db' = \db[T \mapsto \denot{\rho_T(\query)}_\db]$.
For any symbolic tuple $t_i \in R$ and $t_i \in R$, we know $\interpretation'(\del)(t_i') = \interpretation(\del)(t_i)$ and $t_i$ is equal to $t_i'$ in terms of every attribute.
Therefore, $\interpretation', \db' \models t_i' = t_i$ and $\interpretation', \db' \models \formula \land \land_{i=1}^{n} t_i' = t_i$.

Thus, Theorem~\ref{lem:query1} for the inductive case $\query' = \rename_T(Q)$ is proved.

\item Inductive case: $\query'' = \query \times \query'$. 

Suppose that $\schema \vdash \query : [a_1, \ldots, a_p]$, $\schema \vdash \query' : [a_1', \ldots, a_q']$ by Figure~\ref{fig:rules-schema}; $\context \vdash \query \hookrightarrow R$ where $R$ has symbolic tuples $[t_1, \ldots, t_n]$, $\context \vdash \query' \hookrightarrow R'$ where $R'$ has symbolic tuples $[t_1', \ldots, t_m']$, $\context \vdash \query \times \query' \hookrightarrow R''$ where $R''$ has symbolic tuples $[t_{1,1}'', \ldots, t_{n,m}'']$ by Figure~\ref{fig:rules-tuples}; $\schema, \context \vdash \query \leadsto \formula$, $\schema, \context \vdash \query' \leadsto \formula'$, $\schema, \context \vdash \query \times \query' \leadsto \formula \land \formula' \land \land_{i=1}^{n} \land_{j=1}^{m} \formula''_{i,j}$ where $\formula''_{i,j} = (\neg \del(t_i) \land \neg \del(t_j')) \to (\neg \del(t_{i,j}'') \land \land_{k=1}^{p} \denot{a_k}_{\schema, \context, [t_{i,j}'']} = \denot{a_k}_{\schema, \context, [t_i]} \land \land_{k=1}^{q} \denot{a_k}_{\schema, \context, [t_{i,j}'']} = \denot{a_k'}_{\schema, \context, [t_j']}) \land (\del(t_i) \lor \del(t_j') \to \del(t_{i,j}''))$ by Figure~\ref{fig:rules-encode-b}.

By the inductive hypothesis, we have $\interpretation(R) = \denot{\query}_{\db}$, $\interpretation(R') = \denot{\query'}_{\db}$ and $\interpretation, \db \models \formula \land \formula'$.

Take $\interpretation'$ s.t. $\interpretation'(\del)(t_i) \leftrightarrow \interpretation(\del)(t_i)$, $\interpretation'(t_i) = \interpretation(t_i)$, $\interpretation'(\del)(t_j') \leftrightarrow \interpretation(\del)(t_j')$, $\interpretation'(t_j') = \interpretation(t_j')$, $\interpretation'(\del)(t_{i,j}'') = \interpretation(\del)(t_i) \lor \interpretation(\del)(t_j')$ and $\interpretation'(t_{i,j}'') = \smerge(\interpretation(t_i), \interpretation(t_j'))$. 


Let $x_i \in \denot{\query}_\db$, $x_j' \in \denot{\query'}_\db$ where $1 \leq i \leq c \leq n$, $1 \leq j \leq d \leq m$, $c = |\denot{\query}_\db|$ and $d = |\denot{\query'}_\db|$, then we have $x_i \in \interpretation(R)$ and $x_j' \in \interpretation(R')$ because $\interpretation(R) = \denot{\query}_\db$ and $\interpretation(R') = \denot{\query'}_\db$, $\interpretation(\del)(x_i) = \interpretation(\del)(x_j') = \bot$.
Also, by the semantics of $\sfoldl$, we know
\[
\begin{array}{rcl}
\denot{\query \times \query'}_{\db} 
&=& \sfoldl(\lambda xs. \lambda x. \sappend(xs, \smap(\denot{\query'}_{\db}, \lambda y. \smerge(x, y))), [], \denot{\query}_{\db}) \\
&=& \sfoldl(\lambda xs. \lambda x. \sappend(xs, \\
&& \qquad [\smerge(x, x_1'), \ldots, \smerge(x, x_q')]), [], \denot{\query}_{\db}) \\
&=& [ \\
&& \qquad [\smerge(x_1, x_1'), \ldots, \smerge(x_1, x_q')], \\
&& \qquad \ldots, \\
&& \qquad [\smerge(x_p, x_1'), \ldots, \smerge(x_p, x_q')] \\
&& ] \\
&=& [ [x_{1,1}, \ldots, x_{1,d}], \ldots, [x_{c,1}, \ldots, x_{c,d}]] \\
&=& \interpretation'([[x_{1,1}, \ldots, x_{1,d}], \ldots, [x_{c,1}, \ldots, x_{c,d}]]) \\
\end{array}
\]
By the definition of $\ptr$, we know, for any $x_{i,j}'' \in \denot{\query \times \query'}_\db$, there exists a symbolic tuple $t_{\ptr(i),\ptr(j)}'' \in R''$ as $x_{i,j}$ is a concatenation of $x_i \in \query$ and $x_j \in \query'$.
Conversely, for any symbolic tuple $t_{i,j}'' \in R'' - \denot{\query \times \query'}_\db$, $\interpretation'(\del)(t_{i,j}'') = \bot$ because $\interpretation(\del)(t_i) \lor \interpretation(\del)(t_j') = \bot$.
Therefore, $\interpretation'(R'') = \denot{\query \times \query'}_\db$.

Let $\db' = \db[o_{\query''} \mapsto \denot{\query \times \query'}_\db]$.
For any $t_i \in R$ and $t_j' \in R'$, let us discuss them in four cases.
\begin{enumerate}[label=(\alph*)]

\item If $\interpretation(\del)(t_i) = \interpretation(\del)(t_j') = \top$, then, by the definition of $\interpretation'$, we have $\interpretation'(\del)(t_i) = \interpretation(\del)(t_i) = \top$ and $\interpretation'(\del)(t_j') = \interpretation(\del)(t_j') = \top$, $\interpretation'(\del)(t_{i,j}'') = \interpretation(\del)(t_i)  \lor \interpretation(\del)(t_j') \\$$= \top$, $\denot{a_k}_{\db', [\interpretation'(t_{i,j}'')]} = \denot{a_k}_{\db', [\interpretation(t_i)]}$ and $\denot{a_k'}_{\db', [\interpretation'(t_{i,j}'')]} = \denot{a_k'}_{\db', [\interpretation(t_j')]}$,
\[
\begin{array}{rcl}
\interpretation', \db' & \models & (\neg \del(t_i) \land \neg \del(t_j')) \to (\neg \del(t_{i,j}'') \land \land_{k=1}^{p} \denot{a_k}_{\db', [\interpretation'(t_{i,j}'')]} = \denot{a_k}_{\db', [\interpretation'(t_i)]} \\
&& \qquad \land \land_{k=1}^{q} \denot{a_k}_{\db', [\interpretation'(t_{i,j}'')]} = \denot{a_k'}_{\db', [\interpretation'(t_j')]})
\end{array}
\]
and
\[
\interpretation', \db' \models \del(t_i) \lor \del(t_j') \to \del(t_{i,j}'')
\]
Thus, $\interpretation', \db' \models \formula_{i,j}''$.

\item If $\interpretation(\del)(t_i) = \bot$ and $\interpretation(\del)(t_j') = \top$, then, by the definition of $\interpretation'$, we have $\interpretation'(\del)(t_i) = \interpretation(\del)(t_i) = \bot$ and $\interpretation'(\del)(t_j') = \interpretation(\del)(t_j') = \top$, $\interpretation'(\del)(t_{i,j}'') = \interpretation(\del)(t_i) \lor \interpretation(\del)(t_j') \\$$= \top$, $\denot{a_k}_{\db', [\interpretation'(t_{i,j}'')]} = \denot{a_k}_{\db', [\interpretation(t_i)]}$ and $\denot{a_k'}_{\db', [\interpretation'(t_{i,j}'')]} = \denot{a_k'}_{\db', [\interpretation(t_j')]}$,
\[
\begin{array}{rcl}
\interpretation', \db' & \models & (\neg \del(t_i) \land \neg \del(t_j')) \to (\neg \del(t_{i,j}'') \land \land_{k=1}^{p} \denot{a_k}_{\db', [\interpretation'(t_{i,j}'')]} = \denot{a_k}_{\db', [\interpretation'(t_i)]} \\
&& \qquad \land \land_{k=1}^{q} \denot{a_k}_{\db', [\interpretation'(t_{i,j}'')]} = \denot{a_k'}_{\db', [\interpretation'(t_j')]})
\end{array}
\]
and
\[
\interpretation', \db' \models \del(t_i) \lor \del(t_j') \to \del(t_{i,j}'')
\]
Thus, $\interpretation', \db' \models \formula_{i,j}''$.

\item If $\interpretation(\del)(t_i) = \top$ and $\interpretation(\del)(t_j') = \bot$, then, by the definition of $\interpretation'$, we have $\interpretation'(\del)(t_i) = \interpretation(\del)(t_i) = \top$ and $\interpretation'(\del)(t_j') = \interpretation(\del)(t_j') = \bot$, $\interpretation'(\del)(t_{i,j}'') = \interpretation(\del)(t_i) \lor \interpretation(\del)(t_j') \\$$= \top$, $\denot{a_k}_{\db', [\interpretation'(t_{i,j}'')]} = \denot{a_k}_{\db', [\interpretation(t_i)]}$ and $\denot{a_k'}_{\db', [\interpretation'(t_{i,j}'')]} = \denot{a_k'}_{\db', [\interpretation(t_j')]}$,
\[
\begin{array}{rcl}
\interpretation', \db' & \models & (\neg \del(t_i) \land \neg \del(t_j')) \to (\neg \del(t_{i,j}'') \land \land_{k=1}^{p} \denot{a_k}_{\db', [\interpretation'(t_{i,j}'')]} = \denot{a_k}_{\db', [\interpretation'(t_i)]} \\
&& \qquad \land \land_{k=1}^{q} \denot{a_k}_{\db', [\interpretation'(t_{i,j}'')]} = \denot{a_k'}_{\db', [\interpretation'(t_j')]})
\end{array}
\]
and
\[
\interpretation', \db' \models \del(t_i) \lor \del(t_j') \to \del(t_{i,j}'')
\]
Thus, $\interpretation', \db' \models \formula_{i,j}''$.

\item If $\interpretation(\del)(t_i) = \bot$ and $\interpretation(\del)(t_j') = \bot$, then, by the definition of $\interpretation'$, we have $\interpretation'(\del)(t_i) = \interpretation(\del)(t_i) = \bot$ and $\interpretation'(\del)(t_j') = \interpretation(\del)(t_j') = \bot$, $\interpretation'(\del)(t_{i,j}'') = \interpretation(\del)(t_i) \lor \interpretation(\del)(t_j') \\$$= \bot$, $\denot{a_k}_{\db', [\interpretation'(t_{i,j}'')]} = \denot{a_k}_{\db', [\interpretation(t_i)]}$ and $\denot{a_k'}_{\db', [\interpretation'(t_{i,j}'')]} = \denot{a_k'}_{\db', [\interpretation(t_j')]}$,
\[
\begin{array}{rcl}
\interpretation', \db' & \models & (\neg \del(t_i) \land \neg \del(t_j')) \to (\neg \del(t_{i,j}'') \land \land_{k=1}^{p} \denot{a_k}_{\db', [\interpretation'(t_{i,j}'')]} = \denot{a_k}_{\db', [\interpretation'(t_i)]} \\
&& \qquad \land \land_{k=1}^{q} \denot{a_k}_{\db', [\interpretation'(t_{i,j}'')]} = \denot{a_k'}_{\db', [\interpretation'(t_j')]})
\end{array}
\]
and
\[
\interpretation', \db' \models \del(t_i) \lor \del(t_j') \to \del(t_{i,j}'')
\]
Thus, $\interpretation', \db' \models \formula_{i,j}''$.

\end{enumerate}

Thus, $\interpretation', \db' \models \formula \land \formula' \land \land_{i=1}^{n} \land_{j=1}^{m} \formula_{i,j}''$ and Theorem~\ref{lem:query1} for the inductive case $\query'' = \query \times \query'$ is proved.

\item Inductive case: $\query'' = \query \ijoin_\phi \query'$. 

This case can be proved by the inductive cases $\query'' = \query \times \query'$ and $\query' = \filter_\phi(\query)$.

\item Inductive case: $\query'' = \query \ljoin_\phi \query'$. 

Suppose that $\schema \vdash \query : [a_1, \ldots, a_p]$, $\schema \vdash \query' : [a_1', \ldots, a_q']$ by Figure~\ref{fig:rules-schema}; $\schema \vdash \query \hookrightarrow R$ where $R$ has symbolic tuples $[t_1, \ldots, t_n]$, $\schema \vdash \query' \hookrightarrow R'$ where $R'$ has symbolic tuples $[t_1', \ldots, t_m']$ by Figure~\ref{fig:rules-schema}; $\schema, \context \vdash \query \ijoin_\phi \query' \leadsto \formula$, $\schema, \context \vdash \query \ljoin_\phi \query' \leadsto \formula \land \land_{i=1}^{n} \formula_i'$ where $\formula_i' = \land_{k=1}^{p} \denot{a_k}_{\schema, \context, [t_{i,m+1}'']} = \denot{a_k}_{\schema, \context, [t_i]} \land \land_{k=1}^{q} \denot{a_k'}_{\schema, \context, [t_{i,m+1}'']} = \nullv \land (\neg \del(t_i) \land \land_{j=1}^{m} \del(t_{i,j}'') \leftrightarrow \neg \del(t_{i,m+1}''))$ by Figure~\ref{fig:rules-encode-b}.

By the inductive hypothesis, we have $\interpretation(R) = \denot{\query}_\db$, $\interpretation(R') = \denot{\query'}_\db$ and $\interpretation, \db \models \formula$.

Take $\interpretation'$ s.t. $\interpretation'(\del)(t_i) \leftrightarrow \interpretation(\del)(t_i)$, $\interpretation'(t_i) = \interpretation(t_i)$, $\interpretation'(\del)(t_j') \leftrightarrow \interpretation(\del)(t_j')$, $\interpretation'(t_j') = \interpretation(t_j')$, 
$\interpretation'(t_{i,j}'') = \smerge(t_i, t_j')$ for $1 \leq i \leq n$ and $1 \leq j \leq m$,
$\interpretation'(t_{i,m+1}'') = \smerge(t_i, T_{\nullv})$ for $1 \leq i \leq n$,
$\interpretation'(\del)(t_{i,j}'') = \interpretation(\del)(t_i) \lor \interpretation(\del)(t_j') \lor \denot{\phi}_{\db, [\interpretation'(t_{i,j}'')]}$ for $1 \leq i \leq n$ and $\interpretation'(\del)(t_{i, m+1}'') = \interpretation'(\del)(t_i) \lor \neg \land_{j=1}^{m} \interpretation'(\del)(t_{i,j}'')$ for $1 \leq i \leq n$. 

Let $x_i \in \denot{\query}_\db$, $x_j' \in \denot{\query'}_\db$ where $1 \leq i \leq c \leq n$, $1 \leq j \leq d \leq m$, $c = |\denot{\query}_\db|$ and $d = |\denot{\query'}_\db|$, then 
$\interpretation(\del)(x_i) = \interpretation(\del)(x_j') = \bot$. Also, by the semantics of $\sfoldl$ and $\sappend$, we know
\[
\begin{array}{rcl}
\denot{\query \ljoin_\phi \query'}_\db
&=& \sfoldl(\lambda xs. \lambda x. \sappend(xs, \site(|v_1(x)| = 0, v_2(x), v_1(x)), [], \denot{\query}_\db)) \\
&=& \sfoldl(\lambda xs. \lambda x. \sappend(xs, \site(|v_1(x)| = 0, v_2(x), v_1(x)), [], \vec{x}_{1:c})) \\
&=& \sfoldl(\lambda xs. \lambda x. \sappend(xs, \site(|v_1(x)| = 0, v_2(x), v_1(x)), \\
&& \quad [\site(|v_1(x_1)| = 0, v_2(x_1), v_1(x_1))], \vec{x}_{2:c})) \\
&=& \ldots \\
&=& [ \\
&& \quad \site(|v_1(x_1)| = 0, v_2(x_1), v_1(x_1)) \\
&& \quad \ldots \\
&& \quad \site(|v_1(x_c)| = 0, v_2(x_c), v_1(x_c)) \\
&& ]
\end{array}
\]

For any entry $\site(|v_1(x_i)| = 0, v_2(x_i), v_1(x_i)) \in \denot{\query \ljoin_\phi \query'}$, let us discuss $|v_1(x_i)| = 0$ in two cases.
\begin{enumerate}[label=(\alph*)]
\item \label{ljoin_phi_query_times_query_a}
If $|v_1(x_i)| = 0$ is true, then $|\denot{[x_i] \ijoin_\phi \query'}_\db| = 0$ and $\site(|v_1(x_i)| = 0, v_2(x_i), v_1(x_i)) = v_2(x_i) = [\smerge(x_i, T_{\nullv})]$.
By the inductive case $\query'' = \query \times \query'$, we know $\denot{[x_i] \ijoin_\phi \query'}_\db = \filter_\phi(\denot{[x_i] \times \query'}_\db) = \filter_\phi(\denot{[x_i] \times [x_1', \ldots, x_d']}_\db) = \filter_\phi([x_{i,1}'', \ldots, x_{i,d}''])$ where $\interpretation(\del)(x_i) = \interpretation(\del)(x_j') = \bot$.
Therefore, for any $1 \leq j \leq d$, $|\denot{[x_i] \ijoin_\phi \query'}_\db| = 0 \Leftrightarrow \denot{\phi}_{\db, [x_{i,j}'']} \neq \top$ by Lemma~\ref{lem:predicate}, $\interpretation'(\del)(x_{i,j}'') = \top$ and $\interpretation'(\del)(x_{i, c+1}'') = \bot$,
Further, for $x_i \in \denot{\query}_\db$ and $x_j' \in \denot{\query'}_\db$, there exist two symbolic tuples $t_{\ptr(i)} \in R$ and $t_{\ptr(j)}' \in R'$ s.t. $\interpretation(t_{\ptr(i)}) = x_i$, $\interpretation(t_{\ptr(j)}') = x_j'$ and $\interpretation(\del)(x_i) = \interpretation(\del)(t_{\ptr(i)}) = \interpretation(\del)(x_j') = \interpretation(\del)(t_{\ptr(j)}') = \bot$.
Therefore, $\interpretation'(t_{\ptr(i),m+1}'') = \smerge(x_i, T_{\nullv})$ and all tuples $t_{\ptr(i),j}''$ are deleted where $1 \leq j \leq m$, i.e., $\land_{j=1}^m \interpretation'(\del)(t_{\ptr(i),j}'') = \top$ and $\interpretation'(\del)(t_{\ptr(i),m+1}'') = \interpretation'(\del)(t_{\ptr(i)}) \lor \neg \land_{j=1}^m \interpretation'(\del)(t_{\ptr(i),j}'') = \bot$.

\item \label{ljoin_phi_query_times_query_b}
If $|v_1(x_i)| = 0$ is false, then $|\denot{[x_i] \ijoin_\phi \query'}_\db| \neq 0$ and $\site(|v_1(x_i)| = 0, v_2(x_i), v_1(x_i)) = v_2(x_i) = [[x_{i,1}'', \ldots, x_{i,c}'']]$.
By the inductive case $\query'' = \query \times \query'$, we know $\denot{[x_i] \ijoin_\phi \query'}_\db = \filter_\phi(\denot{[x_i] \times \query'}_\db) = \filter_\phi([x_{i,1}'', \ldots, x_{i,d}''])$.
Therefore, $|\denot{[x_i] \ijoin_\phi \query'}_\db| \neq 0 \Leftrightarrow \denot{\phi}_{\db, [x_{i,j}'']} = \top$ for some $1 \leq j \leq d$ by Lemma~\ref{lem:predicate}, $\interpretation'(\del)(x_{i, c+1}'') = \top$.
Further, for $x_i \in \interpretation(R)$ and $x_j' \in \interpretation(R')$, there exist two symbolic tuples $t_{\ptr(i)} \in R$ and $t_{\ptr(j)}' \in R'$ s.t.  $\interpretation(t_{\ptr(i)}) = x_i$, $\interpretation(t_{\ptr(j)}') = x_j'$ and $\interpretation(\del)(t_{\ptr(i)}) = \interpretation(\del)(x_i) = \interpretation(\del)(t_{\ptr(j)}') = \interpretation(\del)(x_j') = \bot$.
Therefore, we have $\interpretation'(\del)(t_{\ptr(i),m+1}'') = \top$ and some tuples $t_{\ptr(i),j}''$ remain alive, i.e., $\interpretation'(\del)(t_{\ptr(i),\ptr(j)}'') = \bot$ for some $1 \leq j \leq m$ and $\interpretation'(t_{\ptr(i),\ptr(j)}'') = x_{i,j}''$.

\end{enumerate}

Besides \ref{ljoin_phi_query_times_query_a} and \ref{ljoin_phi_query_times_query_b}, we know, for any tuple $t \in R'' - \denot{\query \ljoin_\phi \query'}_{\db}$, $\interpretation'(\del)(t) = \top$. Thus, $\interpretation'(R'') = \denot{\query \ljoin_\phi \query'}_{\db}$.

Let $\db' = \db[o_{\query''} \mapsto \denot{\query \ljoin_\phi \query'}_\db]$.
$T_\nullv(A)$ denote a tuple that has the attributes from $A$ but all their values are \nullv. For instance, assume $A = [\text{"a"}, \text{"b"}]$, then $T_\nullv(A) = [(\text{"a"}, \nullv), $ $(\text{"b"}, \nullv)]$.
By the semantics of $\smerge$ and the definition of $T_\nullv$, $\interpretation', \db' \models \formula_j'$ and, therefore, $\interpretation', \db' \models \formula \land \land_{j=1}^m \formula_j'$

Hence, Theorem~\ref{lem:query1} for the inductive case $\query \ljoin_\phi \query'$ is proved.

\item Inductive case: $\query'' = \query \rjoin_\phi \query'$. 

Following a similar method in the inductive case $\query'' = \query \ljoin_\phi \query'$, we can find a $\interpretation'$  s.t. $\interpretation'(\query \ljoin_\phi \query') = \denot{\query \ljoin_\phi \query'}_\db$ and a $\db'$ s.t. $\interpretation', \db' \models \formula$.

\item Inductive case: $\query'' = \query \fjoin_\phi \query'$. 

Suppose that by $\schema \vdash \query : [a_1, \ldots, a_p]$, $\schema \vdash \query : [a_1', \ldots, a_q']$ by Figure~\ref{fig:rules-schema}; $\schema \vdash \query \hookrightarrow R$ where $R$ has symbolic tuples $[t_1, \ldots, t_n]$, $\schema \vdash \query' \hookrightarrow R'$ where $R'$ has symbolic tuples $[t_1', \ldots, t_m']$, $\schema \vdash \query \fjoin_\phi \query' \hookrightarrow R''$ where $R$ has symbolic tuples $[t_{1,1}'', \ldots, t_{n,m}'', t_{n+1,1}'', \ldots, t_{n+1,m}'', t_{1,m+1}'', \ldots, t_{n,m+1}'']$, by Figure~\ref{fig:rules-tuples}; $\schema, \context \vdash \query \ljoin_\phi \query' \leadsto \formula$, $\schema, \context \vdash \query \fjoin_\phi \query' \leadsto \formula \land \land_{j=1}^{m} \formula_j'$ where $\formula_j' = \land_{k=1}^p \denot{a_k}_{\schema, \context, [t_{n+1,j}'']} = \nullv \land \land_{k=1}^q \denot{a_k'}_{\schema, \context, [t_{n+1,j}'']} = \denot{a_k'}_{\schema, \context, [t_j']} \land (\neg \del(t_j') \land \land_{i=1}^{n} \del(t_{i,j}'') \leftrightarrow \neg \del(t_{n+1,j}''))$ by Figure~\ref{fig:rules-encode-b}.

By the inductive hypothesis, we have $\interpretation, \db \models \formula$, $\interpretation(R) = \denot{\query}_\db$, $\interpretation(R') = \denot{\query'}_\db$, $\interpretation([t_{1,1}'', \ldots,$ $ t_{n,m}'', t_{n+1,1}'', \ldots, t_{n+1,m}'']) = \denot{\query \ljoin_\phi \query'}_\db$.

Take $\interpretation'$ s.t. $\interpretation'(\del)(t_i) \leftrightarrow \interpretation(\del)(t_i)$, $\interpretation'(t_i) = \interpretation(t_i)$, $\interpretation'(\del)(t_j') \leftrightarrow \interpretation(\del)(t_j')$, $\interpretation'(t_i) = \interpretation(t_i)$, $\interpretation'(\del)(t_{i,j}'') \leftrightarrow \interpretation(\del)(t_{i,j}'')$ for $1 \leq i \leq n$ and $1 \leq j \leq m+1$, $\interpretation'(t_{i,j}'') = \interpretation(t_{i,j}'')$  for $1 \leq i \leq n$ and $1 \leq j \leq m+1$, $\interpretation'(\del)(t_{n+1,j}'') \leftrightarrow (\interpretation(\del)(t_j') \lor \neg \land_{i=1}^{n} \interpretation(\del)(t_{i,j}''))$ for $1 \leq j \leq m$, $\interpretation'(t_{n+1,j}'') = \smerge(T_\nullv, t_j')$ for $1 \leq j \leq m$.

Let $x_j' \in \denot{\query'}_\db$ where $1 \leq j \leq c \leq m$ and $c = |\denot{\query'}_\db|$, then we have $x_j' \in \interpretation(\query')$ because $\interpretation(R') = \denot{\query'}_\db$, $\interpretation(\del)(x_j') = \bot$.
For any tuple $x_j' \in \denot{\query'}_\db$, we have a symbolic tuple $t_{\ptr(j)}'$ s.t. $\interpretation(t_{\ptr(j)}') = \interpretation(x_j')$ and $\interpretation(\del)(t_{\ptr(j)}') = \bot$ by the definition of $\interpretation'$.
By the semantics of $\sfilter$, we know $|v| \leq |\denot{\query}_\db|$.
Also, for any tuple $x_j' \in \denot{\query'}_\db$, $\smerge(T_\nullv, x_j')$ exists iff $|\denot{\query \ijoin_\phi [x_j']}_\db| = 0$.
Therefore, $\smap(v, \lambda x. \smerge(T_\nullv, x)) = \interpretation'([t_{n+1, 1}'', \ldots, t_{n+1, m}''])$ by Figure~\ref{fig:rules-encode-b}.
Furthermore, we have 
\[
\begin{array}{rcl}
\denot{\query \fjoin_\phi \query'}_\db 
&=& \sappend(\denot{\query \ljoin_\phi \query'}_\db, \smap(v, \lambda x. \smerge(T_\nullv, x))) \\
&=& \sappend(\interpretation'([t_{1, 1}'', \ldots, t_{n, m+1}'']),\interpretation'([t_{n+1, 1}'', \ldots, t_{n+1, m}''])) \\
&=& \interpretation'(R'')
\end{array}
\] by the semantics of $\smap$.

Let $\db' = \db[o_{\query''} \mapsto \denot{\query \fjoin_\phi \query'}_\db]$. 
By the semantics of $\smerge$ and the definition of $T_\nullv$, $\interpretation', \db' \models \formula_j'$ and, therefore, $\interpretation', \db' \models \formula \land \land_{j=1}^m \formula_j'$

Hence, Theorem~\ref{lem:query1} for the inductive case $\query \fjoin_\phi \query'$ is proved.

\item Inductive case: $\query' = \text{Distinct}(\query)$. 

Suppose that $\schema, \context \vdash \query \hookrightarrow R$ where $R$ has symbolic tuples $[t_1, \ldots, t_n]$, $\schema, \context \vdash \text{Distinct}(\query) \hookrightarrow R'$ where $R'$ has symbolic tuples $[t_1', \ldots, t_n']$ by Figure~\ref{fig:rules-tuples}; $\schema, \context \vdash \query \leadsto \formula$, $\schema, \context \vdash \query \leadsto \formula \land \formula'$ where $\formula' = \text{Dedup}(\vec{t}, \vec{t'})$ by Figure~\ref{fig:rules-encode-b}.

By the inductive hypothesis, we have $\interpretation, \db \models \formula$ and $\interpretation(R) = \denot{\query}_\db$.

Take $\interpretation'$ s.t. $\interpretation'(\del)(t_i) \leftrightarrow \interpretation'(\del)(t_i)$, $\interpretation'(t_i) = \interpretation'(t_i)$, $\interpretation'(\del)(t_1) \leftrightarrow \interpretation'(\del)(t_1)$, $\interpretation'(\del)\\$$(t_i') \leftrightarrow \land_{j=1}^{i-1} (\interpretation(t_i) = \interpretation(t_j) \land \neg \interpretation(\del)(t_j))$ for $2 \leq j \leq i$, $\interpretation'(t_i') = \interpretation(t_i)$.

For any tuples $x_i'$ and $x_j'$ in $\denot{\text{Distinct}(\query)}_\db$ s.t. $x_i' \neq x_j'$, we have two symbolic tuples $t_{\ptr(i)}'$ and $t_{\ptr(j)}'$ in $R$ s.t. $x_i' = \interpretation(t_{\ptr(i)}')$ and $\interpretation(\del)(t_{\ptr(i)}') = \bot$, $x_j' = \interpretation(t_{\ptr(j)}')$ and $\interpretation(\del)(t_{\ptr(j)}') = \bot$.
Moreover, if $t_i$ is a unique symbolic tuple in $R$, then $\interpretation'(t_i') = \interpretation(t_i)$ and $\interpretation'(\del)(t_i') = \interpretation(\del)(t_i) = \bot$ where $t_i' \in R'$ by the definition of $\interpretation'$; otherwise, there exist some tuples $[t_{i_1}, \ldots, t_{i_l}]$ s.t. $\interpretation(\del)(t_{i_j}) = \bot$ and $t_{i_j} \in R$ for $1 \leq j \leq k$, and $\interpretation(t_{i_j}) = \interpretation(t_{i_k})$ for $1 \leq j \neq k \leq k$.
Further, by the definition of $\interpretation'(\del)(t_i') = \land_{j=1}^{i} \interpretation(t_i) = \interpretation(t_j)$, only the first symbolic tuple $t_{i_1}$ will be retained in $R'$ and the other tuples $[t_{i_2}, \ldots, t_{i_l}]$ will be removed, i.e., $\interpretation'(t_{i_1}) = \top$ and $\interpretation'(t_{i_j}) = \top$ for $2 \leq j \leq l$.
Therefore, $\interpretation'(R') = \denot{\text{Distinct}(\query)}_\db$.

Let $\db' = \db[o_{\query'} \mapsto \text{Distinct}(\query)]$.
By the definition of $\interpretation'$, we know $\interpretation', \db' \models t_1' = t_1 \land \land_{i=2}^{n} (t_i \in \vec{t}_{1:i-1} \to \del(t_i') \land t_i \not \in \vec{t}_{1:i-1} \to t_i' = t_i)$ and, therefore, $\interpretation', \db' \models \formula \land \formula'$.

Hence, Theorem~\ref{lem:query1} for the inductive case $\text{Distinct}(\query)$ is proved.

\item Inductive case: $\query'' = \query \cap \query'$. 

Suppose that $\schema \vdash \text{Distinct}(\query) \hookrightarrow R$ where $R$ ash symbolic tuples $[t_1, \ldots, t_n]$, $\schema \vdash \query' \hookrightarrow R'$ where $R'$ ash symbolic tuples $[t_1', \ldots, t_m']$, $\schema \vdash \query \cap \query' \hookrightarrow R''$ where $R''$ ash symbolic tuples $[t_1'', \ldots, t_n'']$ by Figure~\ref{fig:rules-tuples}; $\schema, \context \vdash \text{Distinct}(\query) \leadsto \formula$, $\schema, \context \vdash \query'' \leadsto \formula'$ , $\schema, \context \vdash \query \cap \query' \leadsto \formula \land \formula' \land \land_{i=1}^{n} \formula_i''$ where $\formula_i'' = t_i \in \vec{t'} \to t_i'' = t_i \land t_i \not \in \vec{t'} \to \del(t_i'')$ by Figure~\ref{fig:rules-encode-b}.

By the inductive hypothesis, we know $\interpretation, \db \models \formula \land \formula'$, $\interpretation(R) = \denot{\text{Distinct}(\query)}_\db$ and $\interpretation(R') = \denot{\query'}_\db$.

Take $\interpretation'$ s.t. $\interpretation'(t_i) = \interpretation(t_i)$, $\interpretation'(\del)(t_i) \leftrightarrow \interpretation(\del)(t_i)$, $\interpretation'(t_j') = \interpretation'(t_j)$, $\interpretation'(\del)(t_j') \leftrightarrow \interpretation'(\del)(t_j)$, $\interpretation'(t_i'') = \interpretation(t_i)$ and $\interpretation'(\del)(t_i'') \leftrightarrow \interpretation(\del)(t_i) \lor \neg \lor_{j=1}^m (\interpretation(t_i) = \interpretation(t_j') \land \neg \interpretation(\del)(t_j'))$.

For any tuple $x_i \in \denot{\text{Distinct}(\query)}_\db$, let us discuss $x_i \in \denot{\query'}_\db$ in two cases.
\begin{enumerate}[label=(\alph*)]
\item \label{query_cap_query_a}
If $x_i \in \denot{\query'}_\db$ is true, then we can find at least one tuple $x_j' \in \denot{\query'}_\db$ and its corresponding symbolic tuple $t_{\ptr(j)}' \in R'$ s.t. $\interpretation(\del)(x_i) = \interpretation(\del)(x_j') = \interpretation(\del)(t_{\ptr(j)}') = \bot$ and $\interpretation(x_i) = \interpretation(x_j') = \interpretation(t_{\ptr(j)}')$. 
By the definition of $\interpretation'$ and the semantics of $\sfilter$, we know there exist two symbolic tuples $t_{\ptr(i)}'' \in R''$ and $t_{\ptr(i)} \in R$ s.t. $\interpretation(\del)(t_{\ptr(i)}) = \bot$, $\interpretation'(t_{\ptr(i)}'') = \interpretation(t_{\ptr(i)}) = \interpretation(x_i)$, $\interpretation'(\del)(t_{\ptr(i)}'') = \interpretation(\del)(t_{\ptr(i)}) \lor \neg \lor_{j=1}^m (\interpretation(t_{\ptr(i)}) = \interpretation(t_j') \land \neg \interpretation(\del)(t_j')) = \bot$.

\item \label{query_cap_query_b}
If $x_i \in \denot{\query'}_\db$ is false, then $x_i$ will be removed by the semantics of $\sfilter$.
Assume the symbolic tuple $t_{\ptr(i)} \in R$ is corresponding to $x_i$, i.e., $\interpretation(t_{\ptr(i)}) = \interpretation(x_i)$ and $\interpretation(\del)(t_{\ptr(i)}) = \interpretation(\del)(x_i) =\bot$, then $\interpretation'(\del)(t_{\ptr(i)}'')= \interpretation(\del)(t_i) \lor \neg \lor_{j=1}^m (\interpretation(t_{\ptr(i)}) = \interpretation(t_j') \land  \neg \interpretation(\del)\\$$(t_j')) = \top$.

\end{enumerate}

By \ref{query_cap_query_a} and \ref{query_cap_query_b}, we have two symbolic tuples $t_{\ptr(i)} \in R$ and $t_{\ptr(i)}'' \in R''$ for any tuple $x_i \in \denot{\query \cap \query'}_\db$ s.t. $\interpretation'(\del)(t_{\ptr(i)}'') = \interpretation'(\del)(t_{\ptr(i)}) = \bot$ and $\interpretation'(t_{\ptr(i)}'') = \interpretation'(t_{\ptr(i)})$ iff $\interpretation(\del)(t_{\ptr(i)}) \lor \neg \lor_{j=1}^m (\interpretation(t_a) = \interpretation(t_j') \land \neg \interpretation(\del)(t_j'))$ is false.
For the other tuple $t'' \in R'' - \denot{\query \cap \query'}_\db$, $\interpretation'(\del)(t'') = \top$ because $\interpretation'(\del)(t) = \top$ or $\interpretation'(t) \not \in \denot{\query'}_\db$.
Therefore, $\interpretation'(R'') = \denot{\query \cap \query'}_\db$.

Let $\db' = \db[o_{\query''} \mapsto \denot{\query \cap \query'}_\db]$.
By the definition of $\interpretation'$, $\interpretation', \db' \models t_i \in \vec{t'} \to t_i'' = t_i \land t_i \not \in \vec{t'} \to \del(t_i'')$ and, therefore, $\interpretation', \db' \models \formula \land \formula' \land \land_{i=1}^{n} \formula''_i$

Hence, Theorem~\ref{lem:query1} for the inductive case $\query'' = \query \cap \query'$ is proved.

\item Inductive case: $\query'' = \query \setminus \query'$.

Suppose that $\schema \vdash \text{Distinct}(\query) \hookrightarrow R$ where $R$ ash symbolic tuples $[t_1, \ldots, t_n]$, $\schema \vdash \query' \hookrightarrow R'$ where $R'$ ash symbolic tuples $[t_1', \ldots, t_m']$, $\schema \vdash \query \setminus \query' \hookrightarrow R''$ where $R''$ ash symbolic tuples $[t_1'', \ldots, t_n'']$ by Figure~\ref{fig:rules-tuples}; $\schema, \context \vdash \text{Distinct}(\query) \leadsto \formula$, $\schema, \context \vdash \query'' \leadsto \formula'$ , $\schema, \context \vdash \query \setminus \query' \leadsto \formula \land \formula' \land \land_{i=1}^{n} \formula_i''$ where $\formula_i'' = t_i \not \in \vec{t'} \to t_i'' = t_i \land t_i \in \vec{t'} \to \del(t_i'')$ by Figure~\ref{fig:rules-encode-b}.

By the inductive hypothesis, we know $\interpretation(R) = \denot{\text{Distinct}(\query)}_\db$, $\interpretation(R') = \denot{\query'}_\db$ and $\interpretation, \db \models \formula \land \formula'$.

Take $\interpretation'$ s.t. $\interpretation'(t_i) = \interpretation(t_i)$, $\interpretation'(\del)(t_i) \leftrightarrow \interpretation(\del)(t_i)$, $\interpretation'(t_j') = \interpretation'(t_j)$, $\interpretation'(\del)(t_j') \leftrightarrow \interpretation'(\del)(t_j)$, $\interpretation'(t_i'') = \interpretation(t_i)$ and $\interpretation'(\del)(t_i'') \leftrightarrow \interpretation(\del)(t_i) \lor \lor_{j=1}^m (\interpretation(t_i) = \interpretation(t_j') \land \neg \interpretation(\del)\\$$(t_j'))$

For any tuple $x_i \in \denot{\text{Distinct}(\query)}_\db$, let us discuss the predicate  $x_i \in \denot{\query'}_\db$ in two cases.
\begin{enumerate}[label=(\alph*)]
\item \label{query_setminus_query_a}
If the predicate $x_i \in \denot{\query'}_\db$ is true, then we can find at least one tuple $x_j' \in \denot{\query'}_\db$ and its corresponding symbolic tuple $t_{\ptr(j)}'$ s.t. $\interpretation(\del)(x_i) = \interpretation(\del)(x_j') = \interpretation(\del)(x_{\ptr(j)}') = \bot$ and $x_i = x_j' = \interpretation(t_{\ptr(j)}')$. By the definition of $\interpretation'$ and the semantics of $\sfilter$, we have the symbolic tuples $t_{\ptr(i)}'' \in R''$ and $t_{\ptr(i)} \in R$ s.t. $\interpretation(\del)(t_{\ptr(i)}) = \bot$, $\interpretation'(t_{\ptr(i)}'') = \interpretation(t_{\ptr(i)}) = x_i$, $\interpretation'(\del)(t_{\ptr(i)}'') = \interpretation(\del)(t_{\ptr(i)}) \lor \lor_{j=1}^m (\interpretation(t_{\ptr(i)}) = \interpretation(t_j') \land \neg \interpretation(\del)(t_j')) = \top$.

\item \label{query_setminus_query_b}
If the predicate $x_i \in \denot{\query'}_\db$ is false, then $x_i$ will be retained by the semantics of $\sfilter$.
Assume the symbolic tuple $t_{\ptr(i)} \in R$ is corresponding to $x_i$, i.e., $\interpretation(t_{\ptr(i)}) = \interpretation(x_i)$ and $\interpretation(\del)(t_a) = \interpretation(\del)(x_i)$, then $\interpretation'(t_{\ptr(i)}'')= \neg \interpretation(t_i) \land \neg \lor_{j=1}^m (\interpretation(t_{\ptr(i)}) = \interpretation(t_j') \land \neg \interpretation(\del)(t_j')) = \bot$.

\end{enumerate}

By \ref{query_setminus_query_a} and \ref{query_setminus_query_b}, we have two symbolic tuples $t_i \in R$ and $t_i'' \in R''$ for any tuple $x \in \denot{\query \setminus \query'}_\db$ s.t. $\interpretation'(\del)(t_i'') = \interpretation'(\del)(t_i) = \bot$ and $\interpretation'(t_i'') = \interpretation'(t_i)$ iff $\interpretation(\del)(t_i) \lor \lor_{j=1}^m (\interpretation(t_i) = \interpretation(t_j') \land \neg \interpretation(\del)(t_j'))$ is false.
For the other tuple $t'' \in R'' \setminus \denot{\query \setminus \query'}_\db$, $\interpretation'(\del)(t'') = \top$ because $\interpretation'(\del)(t) = \top$ or $\interpretation'(t) \in \denot{\query'}_\db$.
Therefore, $\interpretation'(R'') = \denot{\query \setminus \query'}_\db$.

Let $\db' = \db[o_{\query''} \mapsto \denot{\query \setminus \query'}_\db]$.
By the definition of $\interpretation'$, $\interpretation', \db' \models t_i \not \in \vec{t'} \to t_i'' = t_i \land t_i \in \vec{t'} \to \del(t_i'')$ and, therefore, $\interpretation', \db' \models \formula \land \formula' \land \land_{i=1}^{n} \formula''_i$

Hence, Theorem~\ref{lem:query1} for the inductive case $\query'' = \query \setminus \query'$ is proved.

\item Inductive case: $\query'' = \query - \query'$.

Suppose that $\schema \vdash \query \hookrightarrow R$ where $R$ has symbolic tuples $[t_1, \ldots, t_n]$, $\schema \vdash \query' \hookrightarrow R'$ where $R'$ has symbolic tuples $[t_1', \ldots, t_m']$, $\schema \vdash \query - \query' \hookrightarrow R''$ where $R''$ has symbolic tuples $[t_1'', \ldots, t_n'']$ by Figure~\ref{fig:rules-tuples}; $\schema, \context \vdash \query \leadsto \formula$, $\schema, \context \vdash \query' \leadsto \formula'$, $\schema, \context \vdash \query - \query' \leadsto \formula \land \formula' \land \land_{i=1}^n \land_{j=1}^m \formula_{i,j}''$ where $\formula_{i,j}'' = \neg \del(t_i) \land \neg \text{Paired}(i,j) \to t_i'' = t_i \land \del(t_i) \lor \text{Paired}(i,j) \to \del(t_i'')$ by Figure~\ref{fig:rules-encode-c}.

Let the Paired function only register in $\interpretation$.~\footnote{We slightly abuse interpretation here since the Paired function only works in this case.}
Take $\interpretation'$ s.t. $\interpretation'(t_i) = \interpretation(t_i)$, $\interpretation'(\del)(t_i) \leftrightarrow \interpretation(\del)(t_i)$, $\interpretation'(t_j') = \interpretation(t_j')$, $\interpretation'(\del)(t_j') \leftrightarrow \interpretation(\del)(t_j')$, $\interpretation'(t_i'') = \interpretation(t_i)$, $\interpretation'(\del)(t_i'') \leftrightarrow \interpretation(\del)(t_i) \lor \lor_{j=1}^{m} \interpretation'(\text{Paired}(i,j))$.

By definition of $\denot{\query - \query'}_\db$, let $xs_k = \site(x' \in xs_{k-1}, xs_{k-1} - x, xs_{k-1})$ and $xs_0 = \denot{\query}_\db$. 
Obviously, the value of $xs_k$ will be updated iff $x_j' \in xs_{k-1}$.
Let us discuss the predicate $x_j' \in xs_{k-1}$ in two cases.
\begin{enumerate}[label=(\alph*)]
\item \label{query_except_all_query_a}
If the predicate $x_j' \in xs_{k-1}$ is true, then we have two corresponding symbolic tuples $t_i \in R$ and $t_{\ptr(j)}' \in R'$ s.t. $\interpretation(t_i) = \interpretation(t_{\ptr(j)}') = x_j' \in xs_{k-1}$ and $\interpretation(\del)(t_i) = \interpretation(\del)(t_{\ptr(j)}') = \interpretation(\del)(x_j') = \bot$.
Therefore, $\text{Paired}(i,j) = \top$ and $\text{Paired}(i,l) = \text{Paired}(k,j) = \bot$ for any $1 \leq k \neq i \leq n$ and $1 \leq l \neq j \leq m$ to follow the semantics of \sqlexceptall.
Further, we know $xs_{k} = xs_{k-1} - x$ and $\interpretation'(\del)(t_{\ptr(i)}'') = \top$.

\item \label{query_except_all_query_b}
If the predicate $x_j' \in xs_{k-1}$ is false, then we have a corresponding symbolic tuple $t_{\ptr(j)}' \in R'$ s.t. $\interpretation(t_{\ptr(j)}') = x_j' \not \in xs_{k-1}$, $\interpretation(\del)(t_{\ptr(j)}') = \interpretation(\del)(x_j) = \bot$.
Also, by the definition of the Paired function, we know $\text{Paired}(i,{\ptr(j)}) = \bot$ where $1 \leq i \leq n$ because $t_{\ptr(j)}'$ is never paired with previously paired tuples in $R$ and $\interpretation(x_i) \neq \interpretation(t_{\ptr(j)}')$.
Therefore, we have $xs_{k} = xs_{k-1}$.

\end{enumerate}

Let $xs' = [x_{i_1}, \ldots, x_{i_l}]$ be the deleted tuples from $\denot{\query}_\db$, $i_k$ the real index in $\denot{\query}_\db$ and $1 \leq k \leq l$.
By \ref{query_except_all_query_a} and \ref{query_except_all_query_b}, we have corresponding symbolic tuples $t_{\ptr(i_k)}'' \in R''$ s.t. $\interpretation'(\del)(t_{\ptr(i_k)}'') = \top$.
Further, for any symbolic tuple $t_i \in \denot{\query}_\db - xs'$, $\interpretation'(\del)(t_i'') = \bot$ and $\interpretation'(t_i'') = \interpretation(t_i)$; for the resting symbolic tuples $t \in R' - \denot{\query}_\db$, $\interpretation'(\del)(t'') = \top$ because $t \not \in \denot{\query}_\db$.
Therefore, $\interpretation'(R'') = \denot{\query - \query'}_\db$.


Let $\db' = \db[o_{\query''} \mapsto \denot{\query - \query'}_\db]$
For any tuple $xs_i' \in xs'$, there exist symbolic tuples $t_a \in R$ and $t_b' \in R'$ s.t. $\interpretation(xs_i') = \interpretation(t_a) = \interpretation(t_b')$, $\interpretation(\del)(xs_i') = \interpretation(\del)(t_a) =\interpretation(\del)(t_b') = \bot$ and $\text{Paired}(a, b) = \top$. 
Also, by the definition of the Paired function, we know $\text{Paired}(a, j) = \bot$ for any $j \neq b$ and $\text{Paired}(i, b) = \bot$ for any $i \neq a$.
Therefore, $\interpretation', \db \models \neg \del(t_a) \land \land_{j=1}^m \neg \text{Paired}(a,j) \to t_a'' = t_a \land \del(t_a) \lor \lor_{j=1}^m \neg \text{Paired}(a,j) \to \del(t_a'')$ and $\interpretation', \db' \models \formula \land \formula' \land \land_{i=1}^n \land_{j=1}^m \formula_{i,j}''$

Hence, Theorem~\ref{lem:query1} for the inductive case $\query'' = \query - \query'$ is proved.

\item Inductive case: $\query'' = \query \cup \query'$. 

Since $\query \cup \query' \Leftrightarrow \text{Distinct}(\query \cplus \query')$, this case can be proved by the inductive cases $\query''= \query \cplus \query'$ and $\query' = \text{Distinct}(\query \cplus \query')$.

\item Inductive case: $\query'' = \query \uplus \query'$. 

Suppose that $\schema \vdash \query \hookrightarrow R$ where $R$ has symbolic tuples $[t_1, \ldots, t_n]$, $\schema \vdash \query' \hookrightarrow R'$ where $R'$ has symbolic tuples $[t_1', \ldots, t_m']$, $\schema \vdash \query \uplus \query' \hookrightarrow R''$ where $R''$ has symbolic tuples $[t_1'', \ldots, t_{n+m}'']$ by Figure~\ref{fig:rules-tuples}; $\schema, \context \vdash \query \leadsto \formula$, $\schema, \context \vdash \query' \leadsto \formula'$, $\schema, \context \vdash \query \uplus \query' \leadsto \formula \land \formula' \land \formula''$ where $\formula'' = \land_{i=1}^n t_i'' = t_i \land \land_{j=n+1}^{n+m} t_j'' = t_{j-n}'$ by Figure~\ref{fig:rules-encode-b}.

By the inductive hypothesis, we know $\interpretation(R) = \denot{\query}_\db$, $\interpretation(R') = \denot{\query'}_\db$ and $\interpretation, \db \models \formula \land \formula'$.

Take $\interpretation'$ s.t. $\interpretation'(t_i) = \interpretation(t_i)$, $\interpretation'(\del)(t_i) \leftrightarrow \interpretation(\del)(t_i)$, $\interpretation'(t_j') = \interpretation(t_j')$, $\interpretation'(\del)(t_j') \leftrightarrow \interpretation(\del)(t_j')$, $\interpretation'(t_i'') = \interpretation(t_i)$ and $\interpretation'(\del)(t_i'') \leftrightarrow \interpretation'(\del)(t_i)$ for $1 \leq i \leq n$, $\interpretation'(t_j'') = \interpretation(t_{j-n}')$ and $\interpretation'(\del)(t_j'') \leftrightarrow \interpretation'(\del)(t_j')$ for $n+1 \leq i \leq n+m$.

Further, since $\interpretation(R) = \denot{\query}_\db$ and $\interpretation(R') = \denot{\query'}_\db$, 
\[
\begin{array}{rcl}
\denot{\query \uplus \query'}_\db 
&=& \sappend(\denot{\query}_\db, \denot{\query'}_\db) \\
&=& \sappend(\interpretation(R), \interpretation(R')) \\
&=& \sappend(\interpretation([t_1, \ldots, t_n]), \interpretation([t_1', \ldots, t_m'])) \\
&=& \sappend([\interpretation(t_1), \ldots, \interpretation(t_n)], [\interpretation(t_1'), \ldots, \interpretation(t_m')]) \\
&=& [\interpretation(t_1), \ldots, \interpretation(t_n), \interpretation(t_1'), \ldots, \interpretation(t_m')] \\
&=& [\interpretation'(t_1''), \ldots, \interpretation'(t_n''), \interpretation'(t_{n+1}''), \ldots, \interpretation'(t_{n+m}'')] \\
&=& \interpretation'([t_1'', \ldots, t_n'', t_{n+1}'', \ldots, t_{n+m}'']) \\
&=& \interpretation'(R'') 
\end{array}
\]

Let $\db' = \db[o_{\query''} \mapsto \denot{\query \uplus \query'}_\db ]$.
By the definition of $\interpretation'$, we have $\interpretation', \db' \models \land_{i=1}^n t_i'' = t_i \land \land_{j=n+1}^{n+m} t_j'' = t_{j-n}'$ and, therefore, $\interpretation', \db' \models \formula \land \formula' \land \formula''$.

Hence, Theorem~\ref{lem:query1} for the inductive case $\query'' = \query \uplus \query'$ is proved.

\item Inductive case: $\query'' = \query \cplus \query'$. 

Since $\query \cplus \query' \Leftrightarrow \query - (\query - \query')$, this case can be proved by the inductive case $\query''= \query - \query'$.

\item Inductive case: $\query' = \text{GroupBy}(\query, \vec{E}, L, \phi)$. 

Suppose that $\context \vdash \query \hookrightarrow R$ where $R$ has symbolic tuples $[t_1, \ldots, t_n]$, $\context \vdash \text{GroupBy}(\query, \vec{E}, L, \phi) \hookrightarrow R$ where $R$ has symbolic tuples $[t_1', \ldots, t_n']$ by Figure~\ref{fig:rules-tuples}; $\schema, \context \vdash \query \leadsto \formula$, $\schema, \context \vdash \text{GroupBy}(\query, \vec{E}, $ $L, \phi) \leadsto \formula \land \land_{i=1}^n \formula_i' \land \land_{i=1}^n \formula_i''$ where $\formula_i' = \Sigma_{j=1}^n \text{If}(g(t_i, j), 0, 1) = \text{If}(\del(t_i), 0, 1) \land \land_{j=1}^{i-1} g(t_i, j) = (\neg \del(t_i) \land g(t_i, j) \land \denot{\vec{E}})_{\schema, \context, [t_i]} = \denot{\vec{E}})_{\schema, \context, [t_ij}$ and $\formula_i'' = (g(t_i, j) \land \denot{\phi}_{\schema, \context, g^{-1}(i)} = \top \to \neg \del(t_i') \land \land_{k=1}^l \denot{a_k'}_{\schema, \context, [t_i']} =\denot{a_k'}_{\schema, \context, g^{-1}(i)}) \land (\neg g(t_i, j) \lor \denot{\phi}_{\schema, \context, g^{-1}(i)} \neq \top \to \del(t_i'))$ by Figure~\ref{fig:rules-encode-b}; $L = [a_1, \ldots a_l]$.

By the inductive hypothesis, we have $\interpretation(R) = \denot{\query}_\db$ and $\interpretation, \db \models \formula$.

Take $\interpretation'$ s.t. $\interpretation'(\del)(t_i) \leftrightarrow \interpretation(\del)(t_i)$, $\interpretation'(t_i) = \interpretation(t_i)$, $\interpretation'(\del)(t_i') \leftrightarrow \neg g(t_i, i)$, $\interpretation'(t_i) = \denot{L}_{\db, [g^{-1}(i)]}$ where $g$ checks whether $t_i$ is partitioned into the $i$-th group, and $g^{-1}(i)$ finds symbolic tuples from the $i$-th group.

For any tuple $x_i \in \denot{\query}_\db$, we have a symbolic tuple $t_{\ptr(i)}' \in R$ s.t. $\interpretation(\del)(t_{\ptr(i)}') = \interpretation(\del)(x_i) \\$$= \bot$ and $\interpretation(t_{\ptr(i)}') = \interpretation(x_i)$.
By the definition of $\sfoldl$, we know $\text{Dedup}(\query, \vec{E}) = \\$$\text{Distinct}([\denot{\vec{E}}_{\db, [x_1]}, \ldots, \denot{\vec{E}}_{\db, [x_c]}])$ where $c = |\denot{\query}_\db| \leq n$.
Assuming that $\text{Dedup}(\query, \vec{E}) = [ \denot{\vec{E}}_{\db, [x_{f(1)}]}, \ldots, \denot{\vec{E}}_{\db, [x_{f(d)}]} ]$ after executing Distinct operation where $d = |\text{Dedup}(\query, \vec{E})|$ $ \leq c$ and the function $f(i)$ is the original index for the $i$-th entry in $\text{Dedup}(\query, \vec{E})$.
Furthermore, by the definition of $ys = \smap(\text{Dedup}(\query, \vec{E}), \lambda y. \smap(\denot{\query}_\db, \lambda z. \text{Eval}(\vec{E}, [z]) = y))$, we have a set of tuples $G_i$ sharing the same values over $\vec{E}$ for any entry $e_{f(i)} \in \text{Dedup}(\query, \vec{E})$.

On the other hand, the definition of $g(t_i,j)$ guarantees two properties. The first property states that, for any existing symbolic tuple $t_i$, it only can be partitioned into one symbolic group $G_j'$ whose group index $j$ cannot be greater than $i$, i.e., $\sum_{j=1}^i \text{If}(g(t_i,j), 1, 0) = \text{If}(\del(t_i), 0, 1)$; otherwise, this deleted symbolic tuple $t_i$ never belongs to any group.
The second one ensures the symbolic tuple $t_i$ is partitioned into a previous symbolic group $G_j'$ iff $t_i$ and $t_j \in G_j'$ shares the same values over $\vec{E}$, i.e., $\land_{j=1}^{i-1} g(t_i, j) = (\neg \del(t_i) \land g(t_i, j) \land \denot{\vec{E}}_{\schema, \context, [t_i]} = \denot{\vec{E}}_{\schema, \context, [t_j]})$.
Therefore, we have symbolic groups $[G_1', \ldots, G_n']$ as every tuple might be not deleted and unique over $\vec{E}$.

Let us discuss the tuple $x_i \in \denot{\query}_\db$ in the two cases.
\begin{enumerate} [label=(\alph*)]
\item \label{groupby_query_a}
If no other tuple $x_j \in \denot{\query}_\db$ shares the same values over $\vec{E}$ with $x_i$, then only $x_i$ will be partitioned into a group $G_i$. Similarly, the symbolic tuple $t_{\ptr(i)}$ is separated into the symbolic group $G_i'$ because $g(t_{\ptr(i)}, i) = \top$ and $g(t_{\ptr(i)}, j) = \top$ where $1 \leq  j \leq i-1$.

\item \label{groupby_query_b}
If there exist some tuples in $\denot{\query}_\db$ share the same values over $\vec{E}$ with $x_i$. 
Let $x_j$ be the first tuple in the group $G_j$, then we have $\denot{E}_{\db, t_j} = \denot{E}_{\db, t_i}$ where $j < i$ and $G_j = \sappend(G_j, t_i)$.
In symbolic setting, we have $g(t_{\ptr(i)}, t_{\ptr(j)}) = \top$ because $g(t_j, j) = \top$ and $\denot{\vec{E}}_{\schema, \context, [t_i]} = \denot{\vec{E}}_{\schema, \context, [t_j]}$ is true.

\end{enumerate}

By \ref{groupby_query_a} and \ref{groupby_query_b}, we know if $\denot{\query}_\db$ is partitioned into $[G_1, \ldots, G_m]$ where $1 \leq m \leq |\denot{\query}_\db|$, there exists a symbolic group list s.t. $[G_{\ptr(1)}', \ldots, G_{\ptr(m)}']$ and $G_i = \interpretation(G_{\ptr(i)}')$ for $1 \leq i \leq m$ because $\denot{\query}_\db = \interpretation(R)$.
For any other group $G_j' \in [G_1', \ldots, G_n'] - [G_{\ptr(1)}', \ldots, G_{\ptr(m)}']$, there is no symbolic tuple inside, i.e., $\interpretation(g^{-1}(j)) = []$ by the definition of the function $g(t_i, j)$.
Therefore, $[G_1, \ldots, G_m] = \interpretation([G_{\ptr(1)}', \ldots, G_{\ptr(m)}'])$.
Furthermore, since we will perform the same operations (e.g., $\filter_\phi$ and $\proj$) over $[G_1, \ldots, G_m]$ and $[G_{\ptr(1)}', \ldots, G_{\ptr(m)}']$, 
\[
\begin{array}{rcl}
\denot{\text{GroupBy}(\query, \vec{E}, L, \phi)}_\db
&=& \smap(\sfilter([G_1, \ldots, G_m], \lambda xs. \denot{\phi}_{\db, xs} = \top), \lambda xs. \denot{L}_{\db, xs}) \\
&=& \smap(\sfilter(\interpretation([G_{\ptr(1)}', \ldots, G_{\ptr(m)}']), \lambda xs. \denot{\phi}_{\db, xs} = \top), \lambda xs. \denot{L}_{\db, xs}) \\
&=& \interpretation(\smap(\sfilter([G_{\ptr(1)}', \ldots, G_{\ptr(m)}'], \lambda xs. \denot{\phi}_{\db, xs} = \top), \lambda xs. \denot{L}_{\db, xs})) \\
&=& \interpretation(\text{GroupBy}(\query, \vec{E}, L, \phi)) \\
\end{array}
\]

Let $\db' = \db[\query' \mapsto \denot{\text{GroupBy}(\query, \vec{E}, L, \phi)}_\db]$, then we know $x_i' = \smap(\sfilter([G_i], \lambda xs. \denot{\phi}_{\db, xs}$ $=\top), \lambda xs. \denot{L}_{\db, xs})$ and $\interpretation(\del)(x_i)$ $ = \bot$ for any tuple $x_i' \in \denot{\text{GroupBy}(\query, \vec{E}, L, \phi)}_\db$.

If $x_i$ is the first tuple of $G_i$, then $t_{\ptr(i)}$ is the first symbolic tuple of $G_{\ptr(i)}'$, and $g^{-1}(\ptr(i)) = G_{\ptr(i)}'$.
Since $ys \in [G_1, \ldots, G_m]$, let us discuss the predicate $\denot{\phi}_{\db, G_i}$ in two cases.
\begin{enumerate}[label=(\alph*)]
\item 
If $\denot{\phi}_{\db, G_i} \neq \top$ is true, then the group $G_i$ will be removed and the tuple $\smap(\sfilter([G_i], $ $\lambda xs. \denot{\phi}_{\db, xs}=\top), \lambda xs. \denot{L}_{\db, xs})$ does not exist, $\interpretation'(\del)(t_{\ptr(i)}') = \top$,  $\denot{\phi}_{\db, \interpretation(G_{\ptr(i)}')} = \denot{\phi}_{\db, \interpretation(g^{-1}(\ptr(i)))} \neq \top$.
As $g(t_{\ptr(i)}, i)$ always holds for $G_i$, we know $g(t_{\ptr(i)}, i) \land\\$$ \denot{\phi}_{\db, \interpretation(g^{-1}(\ptr(i)))} = \top$ is false and $\neg g(t_{\ptr(i)}, i) \lor \denot{\phi}_{\db, \interpretation(g^{-1}(\ptr(i)))} \neq \top$ is true.
Therefore, $\interpretation', \db' \models \formula_{\ptr(i)}''$.

\item 
If $\denot{\phi}_{\db, G_i} \neq \top$ is false, then the group $G_i$ will projected out and the tuple $\smap(\sfilter([G_i], $ $\lambda xs. \denot{\phi}_{\db, xs}=\top), \lambda xs. \denot{L}_{\db, xs})$ does exist, $\denot{\phi}_{\db, \interpretation(G_{\ptr(i)}')} =\denot{\phi}_{\db, \interpretation(g^{-1}(\ptr(i)))} = \top$, \\ $\interpretation'(\del)(t_{\ptr(i)}') = \bot$.
As $g(t_{\ptr(i)}, i)$ always holds for $G_i$, we know $g(t_{\ptr(i)}, i) \land \\$$\denot{\phi}_{\db, \interpretation(g^{-1}(\ptr(i)))}  = \top$ is true and $\neg g(t_{\ptr(i)}, i) \lor \denot{\phi}_{\db, \interpretation(g^{-1}(\ptr(i)))} \neq \top$ is false.
Therefore, $\interpretation', \db' \models \formula_{\ptr(i)}''$.

\end{enumerate}
Also, by the definition of the function $g$, we know $\interpretation', \db' \models \land_{i=1}^n \formula_i'$.

Hence, $\interpretation', \db' \models \formula \land \land_{i=1}^n \formula_i' \land \land_{i=1}^n \formula_i''$ and Theorem~\ref{lem:query1} for the inductive case $\query' = \text{GroupBy}(\query, $ $\vec{E}, L, \phi)$ is proved.

\item Inductive case: $\query'' = \text{With}(\vec{\query}, \vec{T}, \query')$. 

Suppose that $\schema \vdash \query_i : \attributes_i$ by Figure~\ref{fig:rules-schema}; $\context \vdash \query_i \hookrightarrow R_i$ by Figure~\ref{fig:rules-tuples}; $\schema, \context \vdash \query_i \leadsto \formula_i$ by Figure~\ref{fig:rules-encode-b}.

By the inductive hypothesis, we have $\interpretation(R_i) = \denot{\query_i}_{\db}$ where $1 \leq i \leq n$ and $n = |\vec{\query}| = |\vec{T}|$; 
$\context' \vdash \query' \hookrightarrow R'$, $\schema', \context' \vdash \query' \leadsto \formula'$, $\interpretation(R') = \denot{\query'}_{\db_{n}}$ where $\schema' = \schema[T_1 \mapsto \attributes_1, \ldots, T_n \mapsto \attributes_n]$ and $\context' = \context[T_1 \mapsto \tuples_1, \ldots, T_n \mapsto \tuples_n]$.
Also, by Figure~\ref{fig:proof-sematics2}, we know $\db' = \db[R_i \mapsto \denot{\query}_\db ~|~ R_i \in R] = [R_1 \mapsto \denot{\query_1}_\db, \ldots, R_n \mapsto \denot{\query_n}_\db]$.

Let $\interpretation' = \interpretation$, $\db'' = \db'[o_{\query'} \mapsto \denot{\query'}_{\db}]$.
Obviously, we have $\denot{\text{With}(\vec{\query}, \vec{T}, \query')}_{\db''} = \denot{\query'}_{\db''} = \interpretation(R')$
 and $\interpretation', \db'' \models \land_{i=1}^{n} \formula_i \land \formula'$.

Hence, Theorem~\ref{lem:query1} for the inductive case $\query'' = \text{With}(\vec{\query}, \vec{T}, \query')$ is proved.

\item Inductive case: $\query' = \text{OrderBy}(\query, \vec{E}, b)$. 

Suppose that $\context \vdash \query \hookrightarrow R$ where $R$ has symbolic tuples $[t_1, \ldots, t_n]$, $\context \vdash \text{OrderBy}(\query, \vec{E}, b) \hookrightarrow R'$ where $R'$ has symbolic tuples $[t_1', \ldots, t_n']$ Figure~\ref{fig:rules-tuples};
$\schema, \context \vdash \query \leadsto \formula$ by Figure~\ref{fig:rules-encode-c};

By the inductive hypothesis, we have $\interpretation(R) = \denot{\query}_\db$ and $\interpretation, \db \models \formula$.

We create a fresh symbolic table $R''$ where $R''$ has symbolic tuples $[t_1'', \ldots, t_n'']$, and the function $\text{moveDelToEnd}(R, R'')$ can ''remove'' deleted symbolic tuples in $R$ by moving them to the end of a list $R$ and return the sorted list to $R''$. 
Further, $\text{moveDelToEnd}(\vec{t}, \vec{t}'')$ guarantees that there exists $m$ symbolic tuples s.t. $\forall i \leq m. \del(t_i) = \bot$ and $\forall m < i. \del(t_i) = \top$ and, for each tuple $t_i'' \in R''$, we can find its corresponding tuple $t_{\text{indexOf}(\vec{t}, i)}$ in $R$ where the $\text{indexOf}(R, i)$ function take as input a symbolic tuple list $R$ and an index $i$, and returns the corresponding index $\text{indexOf}(R, i)$ s.t. $\interpretation(t_i'') = \interpretation(t_{\text{indexOf}(\vec{t}, i)})$ and $\interpretation(\del)(t_i'') = \interpretation(\del)(t_{\text{indexOf}(\vec{t}, i)})$ because $t_i'' = t_{\text{indexOf}(\vec{t}, i)}$.

Take $\interpretation'$ s.t. $\interpretation'(\del)(t_i) \leftrightarrow \interpretation(\del)(t_i)$, $\interpretation'(t_i) = \interpretation(t_i)$, $\interpretation'(\del)(t_i'') = \interpretation(\del)(t_i'')$, $\interpretation'(t_i'') = \interpretation(t_{\text{indexOf}(\vec{t}, i)})$, $\interpretation'(t_i') = \interpretation(t_{\text{find}(i, \vec{E}, b)}'')$ and $\interpretation(\del)(t_i') \leftrightarrow \interpretation(\del)(t_{\text{find}(i, \vec{E}, b)}'')$ where the $\text{find}$ function seeks out the $i$-th smallest symbolic tuple $t_{\text{find}(i, \vec{E}, b)''} \in R''$ w.r.t. the expression list $\vec{E}$ and the ascending flag $b$.
Additionally, for those deleted tuples $[t_{m+1}, \ldots t_{n}]$ in $R''$, the $\text{find}$ function will not consider them because they are ''deleted''.
Hence, for any tuple $t_i'' \in [t_1'', \ldots, t_m'']$, we always can find a tuple $t_{\text{find}(i, \vec{E}, b)}''$ s.t. $\interpretation'(\del)(t_i') = \interpretation(\del)(t_{\text{find}(i, \vec{E}, b)}'') = \interpretation(\del)(t_{\text{find}(\text{indexOf}(\vec{t}, i), \vec{E}, b)}'') = \bot$ and $\interpretation'(t_i') = \interpretation(t_{\text{find}(i, \vec{E}, b)}'') = \interpretation(t_{\text{find}(\text{indexOf}(\vec{t}, i)}'', \vec{E}, b))$.

To compare two tuples (e.g., $x_1$ and $x_2$) w.r.t. the expression list $\vec{E}$ and the ascending flag $b$, let us discuss the value of $\text{Cmp}(\vec{E}, b, x_1, x_2)$ in three cases.
\begin{enumerate}[label=(\alph*)]
\item 
If $v(x_1, E_i) < v(x_2, E_i)$, then $\text{Cmp}(\vec{E}, b, x_1, x_2) = b \neq \top$ by the semantics of $\sfoldr$.
Furthermore, if $b = \top$ (i.e., in ascending order), then $\text{Cmp}(\vec{E}, b, x_1, x_2) = \bot$ and vice versa.

\item 
If $v(x_1, E_i) > v(x_2, E_i)$, then $\text{Cmp}(\vec{E}, b, x_1, x_2) = b \neq \bot$ by the semantics of $\sfoldr$.
Furthermore, if $b = \top$ (i.e., in ascending order), then $\text{Cmp}(\vec{E}, b, x_1, x_2) = \top$ and vice versa.

\item 
If $v(x_1, E_i) = v(x_2, E_i)$, then $\text{Cmp}(\vec{E}, b, x_1, x_2) = b \neq \sfoldr(\lambda E_{i}. \lambda y. \site(v(x_1, E_{i}) < v(x_2, E_{i}), \\$$\top, \site(v(x_1, E_{i}) > v(x_2, E_{i}), \bot, y)), \top, \vec{E}_{i+1:m})$ where $m=|\vec{E}|$, $y$ denotes the accumulator of $\sfoldr$ and $\text{Cmp}(\vec{E}, b, x_1, x_2)$ will be determined by the following expressions $[E_{i+1}, \ldots, E_m]$. Also, we initialize the accumulator $y$ to $\top$ which aims to find the first smallest (last biggest) tuple if $x_1$ and $x_2$ are equal w.r.t. $\vec{E}$ and b.
\end{enumerate}

Further, by the semantics of $\sfoldl$ and the value of $\text{Cmp}(\vec{E}, b, x_1, x_2)$, we know $\text{MinTuple}(\vec{E}, b,$ $ xs)$ is to find the first smallest (last biggest) tuple from $xs$ w.r.t. $\vec{E}$ and $b$.
\[
\begin{array}{rcl}
\text{MinTuple}(\vec{E}, b, xs)
&=& \sfoldl(\lambda x. \lambda y. \site(\text{Cmp}(\vec{E}, b, x_1, x_2), y, x), \shead(xs), xs) \\
&=& \sfoldl(\lambda x. \lambda y. \site(\text{Cmp}(\vec{E}, b, x_1, x_2), y, x), \\
&& \qquad \site(\text{Cmp}(\vec{E}, b, xs_1, xs_1), xs_1, xs_1), xs_{2:n}) \\
&=& \sfoldl(\lambda x. \lambda y. \site(\text{Cmp}(\vec{E}, b, x_1, x_2), y, x), xs_1, xs_{2:n}) \\
\end{array}
\]

Let us discuss $\text{MinTuple}(\vec{E}, b, xs)$ in two cases.
\begin{enumerate}[label=(\alph*)]
\item 
If $\text{Cmp}(\vec{E}, b, x_1, x_2)$ is true, the accumulator (i.e., the smallest or biggest tuple of $xs$) in $\text{MinTuple}(\vec{E}, b, xs)$ will be updated.

\item 
If $\text{Cmp}(\vec{E}, b, x_1, x_2)$ is false, the accumulator (i.e., the smallest or biggest tuple of $xs$) in $\text{MinTuple}(\vec{E}, b, xs)$ remains constant.

\end{enumerate}

Therefore, $\text{MinTuple}(\vec{E}, b, \denot{\query}_{\db} - xs)$ seeks out the smallest (biggest) tuple from $\denot{\query}_{\db} - xs$.
Let $a_i$ denote the $i$-th smallest (biggest) tuple from $\denot{\query}_{\db} - xs$ (i.e., $a_i = \text{MinTuple}(\vec{E}, b, \denot{\query}_{\db} - xs)$ where $xs = [a_i, \ldots, a_{i-1}]$) and $\vec{t}_{1:p} = \denot{\query}_\db$, then 
\[
\begin{array}{rcl}
\denot{\text{OrderBy}(\query, \vec{E}, b)}_{\db}
&=& \sfoldl(\lambda xs. \lambda \_. (\sappend(xs, \text{MinTuple}(\vec{E}, b, \denot{\query}_{\db} - xs))), [], \denot{\query}_{\db}) \\
&=& \sfoldl(\lambda xs. \lambda \_. (\sappend(xs, \text{MinTuple}(\vec{E}, b, \denot{\query}_{\db} - xs))), [], \vec{t}_{1:p}) \\
&=& \sfoldl(\lambda xs. \lambda \_. (\sappend(xs, \text{MinTuple}(\vec{E}, b, \denot{\query}_{\db} - xs))), [a_1], \vec{t}_{2:p}) \\
&& \ldots \\
&=& [a_1, \ldots, a_p] \\
\end{array}
\]

Since (i) $t_{\text{find}(i, \vec{E}, b)}''$ and $a_i$ are the $i$-th smallest (biggest) tuples from $\interpretation(R'')$ and $\denot{\query}_\db$, respectively,  (ii) $\interpretation(R'') = \denot{\query}_\db$, we have, for any $1 \leq i \leq p$, $\interpretation(t_{\text{find}(i, \vec{E}, b)}'') = \interpretation(a_i)$ and $\interpretation(\del)(t_{\text{find}(i, \vec{E}, b)}'') = \interpretation(\del)(a_i) = \bot$; and, for any $p \leq i \leq n$, $\interpretation(\del)(t_{\text{find}(i, \vec{E}, b)}'') = \top$.
Also by the definition of $\interpretation'$, we have $t_i' = t_{\text{find}(i, \vec{E}, b)}''$, $\interpretation(t_i') = \interpretation(a_i)$ and $\interpretation(\del)(t_i') = \interpretation(\del)(a_i) = \bot$ where $1 \leq i \leq p$. 
Thus, $\interpretation'(\text{OrderBy}(\query, \vec{E}, b)) = \denot{\text{OrderBy}(\query, \vec{E}, b)}_{\db}$.

Let $\db' = \db[o_{\text{OrderBy}(\query, \vec{E}, b)} \mapsto \denot{\text{OrderBy}(\query, \vec{E}, b)}_{\db}]$.
Obviously, $\interpretation', \db' \models \formula \land \formula' \land \formula''$ where $\formula' = \text{moveDelToEnd}(\vec{t}, \vec{t}'')$ and $\formula'' = \land_{i=1}^{n} t_i' = t_{\text{find}(i, \vec{E}, b)}''$ hold because $\interpretation'(t_i') = \interpretation'(t_{\text{find}(i, \vec{E}, b)}'') = \interpretation'(t_{\text{find}(\text{indexOf}(\vec{t}, i), \vec{E}, b)})$ and $\interpretation'(\del)(t_i') = \interpretation'(\del)(t_{\text{find}(i, \vec{E}, b)}) = \interpretation'(\del)$ $(t_{\text{find}(\text{indexOf}(\vec{t}, i), \vec{E}, b)})$.

Hence, Theorem~\ref{lem:query1} for the inductive case $\query' = \text{OrderBy}(\query, \vec{E}, b)$ is proved.

\end{enumerate}
\end{proof}

\begin{proof}[Proof of Theorem~\ref{lem:query2}:]
\revision{Let $\context$ be a symbolic database over schema $\schema$ and $\query$ be a query. Consider a symbolic relation $R$ and a formula $\formula$ such that $\context \vdash \query \hookrightarrow R$ and $\schema, \context, \vdash \query \leadsto \formula$. If $\formula$ is satisfiable, then for any satisfying interpretation $\interpretation$ of $\formula$, running $\query$ over the concrete database $\interpretation(\context)$ yields the relation $\interpretation(R)$, i.e., 
\[
(\context \vdash \query \hookrightarrow R) \land (\schema, \context \vdash \query \leadsto \formula) \land (\interpretation \models \formula) \Rightarrow \denot{\query}_{\interpretation(\context)} = \interpretation(R)
\]
}
\end{proof}

\begin{proof}[Proof.]
By structural induction on $\query$. 

\begin{enumerate}
\item Base case: $\query = T$. 

$R = \context(T)$ by Figure~\ref{fig:rules-tuples}.
$\denot{T}_{\interpretation(\context)} = \interpretation(\context)(T)$ by Figure~\ref{fig:proof-sematics2}.
Therefore, $\denot{T}_{\interpretation(\context)} = \interpretation(\context)(T) = \interpretation(\context(T)) = \interpretation(R)$.

Thus, Theorem~\ref{lem:query2} for the inductive case $\query = T$ is proved.

\item Inductive case: $\query' = \proj_{L}(\query)$. 

Suppose that $\schema \vdash \query : [a_1, \ldots, a_m]$ by Figure~\ref{fig:rules-schema}; 
$\schema \vdash \filter_L(\query) : [a_1', \ldots, a_l']$ where $l = |L|$ by Figure~\ref{fig:rules-schema};
$\context \vdash \query \hookrightarrow R$ where $R$ has symbolic tuples $[t_1, \ldots, t_n]$ by Figure~\ref{fig:rules-tuples};
$\schema, \context \vdash \query \leadsto \formula$ by Figure~\ref{fig:rules-encode-b};
$\denot{\proj_L(\query)}_{\interpretation(\context)} = \site(\text{hasAggr}(L), [\denot{L}_{\interpretation(\context), \denot{\query}_{\interpretation(\context)}}], \smap(\denot{Q}_{\interpretation(\context)}, \lambda x. \denot{L}_{\interpretation(\context), x}))$ by Figure~\ref{fig:proof-sematics2}.

Let us discuss $\text{hasAggr}(L)$ in two cases.
\begin{enumerate}[label=(\alph*)]
\item \label{proof_query2_proj_query_a}
If $\text{hasAggr}(L)$ is true, then $\context \vdash \proj_{L}(\query) \hookrightarrow R'$ where $R'$ only has one symbolic tuple $[t_1']$ by Figure~\ref{fig:rules-tuples}, $\schema, \context \vdash \proj_{L}(\query) \leadsto \formula \land \land_{i=1}^{l} \formula_{i}' \land (\del(t_1') \leftrightarrow \land_{i=1}^n \del(t_i))$ where $\formula_{i}' = \denot{a_j'}_{\schema, \context, [t_1']} = \denot{a_j'}_{\schema, \context, \vec{t}}$ by Figure~\ref{fig:rules-encode}.

Take $\interpretation$ s.t. $\interpretation \models \formula \land \land_{i=1}^{l} \formula_{i}' \land \neg \del(t_1')$.

By Lemma~\ref{lem:attribute}, we know $\interpretation(\denot{A}_{\schema, \context, \tuples}) = \denot{A}_{\interpretation(\context), \interpretation(\tuples)}$.
Also, by Figure~\ref{fig:proof-sematics2}, we have 
\[
\begin{array}{rcl}
\denot{\proj_L(\query)}_{\interpretation(\context)} 
&=& [\denot{L}_{\interpretation(\context), \denot{\query}_{\interpretation(\context)}}] = [\denot{L}_{\interpretation(\context), \interpretation([t_1, \ldots, t_n])}] \\
&=& [(\text{ToString}(a_1'), \denot{a_1'}_{\interpretation(\context), \interpretation([t_1, \ldots, t_n])}), \ldots, (\text{ToString}(a_l'), \denot{a_l'}_{\interpretation(\context), \interpretation([t_1, \ldots, t_n])})] \\
&=& [(\text{ToString}(a_1'), \denot{a_1'}_{\interpretation(\context), [\interpretation(t_1')]}), \ldots, (\text{ToString}(a_l'), \denot{a_l'}_{\interpretation(\context), [\interpretation(t_1')]})] \\
&=& [\denot{L}_{\interpretation(\context), [\interpretation(t_1')]}] = \interpretation(R')
\end{array}
\]

\item \label{proof_query2_proj_query_b}
If $\text{hasAggr}(L)$ is false, then $\context \vdash \proj_{L}(\query) \hookrightarrow R'$ where $R'$ has symbolic tuples $[t_1', \ldots, t_n']$ by Figure~\ref{fig:rules-tuples},  $\schema, \context \vdash \proj_{L}(\query) \leadsto \formula \land \land_{i=1}^{n} \formula_{i}'$ where $\formula_{i}' = \land_{j=1}^{l} \denot{a_j'}_{\schema, \context, [t_i']} = \denot{a_j'}_{\schema, \context, [t_i]} \land \del(t_i') \leftrightarrow \del(t_i)$ by Figure~\ref{fig:rules-encode}.

Take $\interpretation$ s.t. $\interpretation \models \formula \land \land_{i=1}^{n} \formula_{i}'$.

By Figure~\ref{fig:proof-sematics2}, we have 
\[
\begin{array}{rcl}
\denot{\proj_L(\query)}_{\interpretation(\context)} 
&=& \smap(\denot{\query}_{\interpretation(\context)}, \lambda x. \denot{L}_{\interpretation(\context) ,x}) \\
&=& \smap(\interpretation([t_1, \ldots, t_n]), \lambda x. \denot{L}_{\interpretation(\context) ,x}) \\
&=& [\denot{L}_{\interpretation(\context), [\interpretation(t_1)]}, \ldots, \denot{L}_{\interpretation(\context), [\interpretation(t_n)]}] \\
&=& [\denot{L}_{\interpretation(\context), [\interpretation(t_1')]}, \ldots, \denot{L}_{\interpretation(\context), [\interpretation(t_n')]}] \\
&=& \interpretation(R')
\end{array}
\]

\end{enumerate}

Thus, $\denot{\proj_L(\query)}_{\interpretation(\context)} = \interpretation(R)$ by \ref{proof_query2_proj_query_a} and \ref{proof_query2_proj_query_b} and Theorem~\ref{lem:query2} for the inductive case $\query' = \proj_L(\query)$ is proved.

\item Inductive case: $\query = \filter_{\phi}(\query)$. 

Suppose that $\context \vdash \query \hookrightarrow R$ where $R$ has symbolic tuples $[t_1, \ldots, t_n]$, $\context \vdash \filter_{\phi}(\query) \hookrightarrow R'$ where $R'$ has symbolic tuples $[t_1', \ldots, t_n']$ by Figure~\ref{fig:rules-tuples}; $\schema, \context \vdash \query \leadsto \formula \land \land_{i=1}^n \formula_i'$ where $\formula_i' = (\neg \del(t_i) \land \denot{\phi}_{\schema, \context, [t_i]} = \top) \to t_i' = t \land (\del(t_i) \lor \denot{\phi}_{\schema, \context, [t_i]} \neq \top) \to \del(t_i')$ by Figure~\ref{fig:rules-encode-b}.

Take $\interpretation$ s.t. $\interpretation \models \formula \land \land_{i=1}^n \formula_i'$.

By Lemma~\ref{lem:predicate}, we know $\interpretation(\denot{\phi}_{\schema, \context, [x]}) = \denot{\phi}_{\interpretation(\context), \interpretation([x])}$ and, therefore, 
\[
\begin{array}{rcl}
\denot{\filter_\phi(\query)}_{\interpretation(\context)} 
&=& \sfilter(\denot{\query}_{\interpretation(\context)}, \lambda x. \denot{\phi}_{\interpretation(\context), [x]} = \top) \\
&=& \sfilter(\interpretation([t_1, \ldots, t_n]), \lambda x. \denot{\phi}_{\interpretation(\context), [x]} = \top) \\
\end{array}
\]

For any tuple $t_i \in R$, let us discuss $\del(t_i)$ in two cases.
\begin{enumerate}[label=(\alph*)]
\item \label{proof_query2_filter_query_a}
If $\interpretation(\del)(t_i) = \top$, then we know $\sfilter([\interpretation(t_i)], \lambda x. \denot{\phi}_{\interpretation(\context), [x]} = \top)$ is deleted because of the semantics of $\sfilter$, and $\interpretation(\del)(t_i') = \top$ because of $\del(t_i') \leftrightarrow \del(t_i) \lor \denot{\phi}_{\interpretation(\context), [\interpretation(t_i)]}$.

\item \label{proof_query2_filter_query_b}
If $\interpretation(\del)(t_i) = \bot$, then we know the existence of $\sfilter([\interpretation(t_i)], \lambda x. \denot{\phi}_{\interpretation(\context), [x]} = \top)$ depends on $\denot{\phi}_{\interpretation(\context), [\interpretation(t_i)]} = \top$.
Let us discuss $\denot{\phi}_{\interpretation(\context), [\interpretation(t_i)]}$ in two cases.

\begin{enumerate}[label=(b\arabic*)]
\item 
If $\denot{\phi}_{\interpretation(\context), [\interpretation(t_i)]} = \top$ is true, then $\sfilter([\interpretation(t_i)], \lambda x. \denot{\phi}_{\interpretation(\context), [x]} = \top)$ exists and \\ $\interpretation(\del)(t_i') = \bot$ because $\del(t_i) \lor \denot{\phi}_{\schema, \context, [t_i]} = \top \to \del(t_i')$.
Therefore, $\sfilter([\interpretation(t_i)],$ $\lambda x. \denot{\phi}_{\interpretation(\context), [x]}$ $  = \top) = [\interpretation(t_i')]$.

\item 
If $\denot{\phi}_{\interpretation(\context), [\interpretation(t_i)]} = \top$ is false, then $\sfilter([\interpretation(t_i)], \lambda x. \denot{\phi}_{\interpretation(\context), [x]} = \top)$ is deleted and $\interpretation(\del)$ $(t_i') = \top$ because $\del(t_i) \lor \denot{\phi}_{\schema, \context, [t_i]} = \top \to \del(t_i')$.
Therefore, $\sfilter([\interpretation(t_i)],$ $  \lambda x. \denot{\phi}_{\interpretation(\context), [x]} = \top) = [\interpretation(t_i')]$.

\end{enumerate}

\end{enumerate}

Thus, $\denot{\filter_\phi(\query)}_{\interpretation(\context)} = \interpretation(R)$ by \ref{proof_query2_filter_query_a} and \ref{proof_query2_filter_query_b} and Theorem~\ref{lem:query2} for the inductive case $\query' = \filter_\phi(\query)$ is proved.

\item Inductive case: $\query' = \rename_T(\query)$. 

Suppose that $\context \vdash \query \hookrightarrow R$ where $R$ has symbolic tuples $[t_1, \ldots, t_n]$, $\context \vdash \rename_T(\query) \hookrightarrow R'$ where $R'$ has symbolic tuples $[t_1', \ldots, t_n']$ by Figure~\ref{fig:rules-tuples}; $\schema, \context \vdash \query \leadsto \formula$, $\schema, \context \vdash \formula \land \land_{i=1}^{n} t_i' = t_i$ by Figure~\ref{fig:rules-encode-b}.

Take $\interpretation$ s.t. $\interpretation \models \formula \land \land_{i=1}^{n} t_i' = t_i$.

$\denot{\rename_T(\query)}_{\interpretation(\context)} = \denot{\query}_{\interpretation(\context)[T \mapsto \denot{\query}_{\interpretation(\context)}]} = \denot{\query}_{\interpretation(\context)} = \interpretation([t_1, \ldots, t_n]) = \interpretation(R)$.

Thus, Theorem~\ref{lem:query2} for the inductive case $\query' = \rename_T(\query)$ is proved.

\item Inductive case: $\query'' = \query \times \query'$. 

Suppose that $\schema \vdash \query : [a_1, \ldots, a_p]$, $\schema \vdash \query' : [a_1', \ldots, a_q']$ by Figure~\ref{fig:rules-schema}; $\context \vdash \query \hookrightarrow R$ where $R$ has symbolic tuples $[t_1, \ldots, t_n]$, $\context \vdash \query' \hookrightarrow R'$ where $R'$ has symbolic tuples $[t_1', \ldots, t_m']$, $\context \vdash \query \times \query' \hookrightarrow R''$ where $R''$ has symbolic tuples $[t_{1,1}'', \ldots, t_{n,m}'']$ by Figure~\ref{fig:rules-tuples}; $\schema, \context \vdash \query \leadsto \formula$, $\schema, \context \vdash \query' \leadsto \formula'$, $\schema, \context \vdash \query \times \query' \leadsto \formula \land \formula' \land \land_{i=1}^{n} \land_{j=1}^{m} \formula''_{i,j}$ where $\formula''_{i,j} = (\neg \del(t_i) \land \neg \del(t_j')) \to (\neg \del(t_{i,j}'') \land \land_{k=1}^{p} \denot{a_k}_{\schema, \context, [t_{i,j}'']} = \denot{a_k}_{\schema, \context, [t_i]} \land \land_{k=1}^{q} \denot{a_k}_{\schema, \context, [t_{i,j}'']} = \denot{a_k'}_{\schema, \context, [t_j']}) \land (\del(t_i) \lor \del(t_j') \to \del(t_{i,j}''))$ by Figure~\ref{fig:rules-encode-b}.

Take $\interpretation$ s.t. $\interpretation \models \formula \land \formula' \land \land_{i=1}^{n} \land_{j=1}^{m} \formula''_{i,j}$.

By Figure~\ref{fig:proof-sematics2}, we have 
\[
\begin{array}{rcl}
\denot{\query \times \query'}_{\interpretation(\context)} 
&=& \sfoldl(\lambda xs. \lambda x. \sappend(xs, \smap(\denot{Q_2}_{\interpretation(\context)}, \lambda y. \smerge(x, y))), [], \denot{Q_1}_{\interpretation(\context)}) \\
&=& [ \\
&& \quad [\smerge(\interpretation(t_1), \interpretation(t_1')), \ldots, \smerge(\interpretation(t_1), \interpretation(t_m'))] \\
&& \quad \ldots \\
&& \quad [\smerge(\interpretation(t_n), \interpretation(t_1')), \ldots, \smerge(\interpretation(t_n), \interpretation(t_m'))] \\
&& ] \\
\end{array}
\]

For any tuples $t_i \in R$ and $t_j' \in R'$, let us discuss $\interpretation(\del)(t_i) \lor \interpretation(\del)(t_j')$ in two cases.
\begin{enumerate}[label=(\alph*)]
\item \label{proof_query2_times_query_a}
If $\interpretation(\del)(t_i) \lor \interpretation(\del)(t_j') = \top$, then we know $\smerge(\interpretation(t_i), \interpretation(t_j'))$ is deleted by the semantics of $\smerge$ and $\interpretation(\del)(t_{i,j}'') = \top$ because of $\del(t_i) \lor \del(t_j') \to \del(t_{i,j}'')$.
Therefore, $\smerge(\interpretation(t_i), \interpretation(t_j')) = \interpretation(t_{i,j}'')$.

\item \label{proof_query2_times_query_b}
If $\interpretation(\del)(t_i) \lor \interpretation(\del)(t_j') = \bot$, then we know $\smerge(\interpretation(t_i), \interpretation(t_j')) = \interpretation(t_{i,j}'')$ by the semantics of $\smerge$ and $\interpretation(\del)(t_{i,j}'') = \top$ because of $\formula_{i,j}''$. 
\end{enumerate}

Hence, $\denot{\query \times \query'}_{\interpretation(\context)} = \interpretation(R)$ by \ref{proof_query2_times_query_a} and \ref{proof_query2_times_query_b} and Theorem~\ref{lem:query2} for the inductive case $\query'' = \query \times \query'$ is proved.

\item Inductive case: $\query'' = \query \ijoin_\phi \query'$. 

This case can be proved by the inductive cases $\query'' = \query \times \query'$ and $\query' = \filter_\phi(\query)$.

\item Inductive case: $\query'' = \query \ljoin_\phi \query'$. 

Suppose that $\schema \vdash \query : [a_1, \ldots, a_p]$, $\schema \vdash \query' : [a_1', \ldots, a_q']$ by Figure~\ref{fig:rules-schema}; $\schema \vdash \query \hookrightarrow R$ where $R$ has symbolic tuples $[t_1, \ldots, t_n]$, $\schema \vdash \query' \hookrightarrow R'$ where $R'$ has symbolic tuples $[t_1', \ldots, t_m']$, $\schema \vdash \query \ijoin_\phi \query' \hookrightarrow R$ where $R''$ has symbolic tuples $[t_{1,1}'', \ldots, t_{n, m}'', t_{1, m+1}'', \ldots, t_{n, m+1}'']$ by Figure~\ref{fig:rules-schema}; $\schema, \context \vdash \query \ijoin_\phi \query' \leadsto \formula$, $\schema, \context \vdash \query \ljoin_\phi \query' \leadsto \formula \land \land_{i=1}^{n} \formula_i'$ where $\formula_i' = \land_{k=1}^{p} \denot{a_k}_{\schema, \context, [t_{i,m+1}'']} = \denot{a_k}_{\schema, \context, [t_i]} \land \land_{k=1}^{q} \denot{a_k'}_{\schema, \context, [t_{i,m+1}'']} = \nullv \land (\neg \del(t_i) \land \land_{j=1}^{m} \del(t_{i,j}'') \leftrightarrow \neg \del(t_{i,m+1}''))$ by Figure~\ref{fig:rules-encode-b}.

Take $\interpretation$ s.t. $\interpretation \models \formula \land \land_{i=1}^{n} \formula_i'$.

By Figure~\ref{fig:proof-sematics2}, we have 
\[
\begin{array}{rcl}
\denot{\query \ljoin_\phi \query'}_{\interpretation(\context)}
&=& \sfoldl(\lambda xs. \lambda x. \sappend(xs, \site(|v_1(x)| = 0, v_2(x), v_1(x)), [], \denot{\query}_{\interpretation(\context)})) \\
&=& \sfoldl(\lambda xs. \lambda x. \sappend(xs, \site(|v_1(x)| = 0, v_2(x), v_1(x)), [], \interpretation(\vec{t}_{1:n}))) \\
&=& \sfoldl(\lambda xs. \lambda x. \sappend(xs, \site(|v_1(x)| = 0, v_2(x), v_1(x)), \\
&& \quad [\site(|v_1(\interpretation(t_1))| = 0, v_2(\interpretation(t_1)), v_1(\interpretation(t_1)))], \interpretation(\vec{t}_{2:n}))) \\
&=& \ldots \\
&=& [ \\
&& \quad \site(|v_1(\interpretation(t_1))| = 0, v_2(\interpretation(t_1)), v_1(\interpretation(t_1))) \\
&& \quad \ldots \\
&& \quad \site(|v_1(\interpretation(t_n))| = 0, v_2(\interpretation(t_n)), v_1(\interpretation(t_n))) \\
&& ]
\end{array}
\]

For any entry $\site(|v_1(\interpretation(t_i))| = 0, v_2(\interpretation(t_i)), v_1(\interpretation(t_i))) \in \denot{\query \ljoin_\phi \query'}$, let us discuss $|v_1(\interpretation(t_i))| = 0$ in two cases.

\begin{enumerate}[label=(\alph*)]
\item \label{proof_query2_ljoin_query_a}
If $|v_1(\interpretation(t_i))| = 0$ is true, then $|\denot{[\interpretation(t_i)] \ijoin_\phi \query'}_{\interpretation(\context)}| = 0$ and $\site(|v_1(\interpretation(t_i))| = 0, v_2(\interpretation(t_i)), \\$$v_1(\interpretation(t_i))) = v_2(\interpretation(t_i)) = [\smerge(\interpretation(t_i), T_{\nullv})]$.

Further, if $\interpretation(\del)(t_i) = \top$, then $\site(|v_1(\interpretation(t_i))| = 0, v_2(\interpretation(t_i)), v_1(\interpretation(t_i))) = \smerge$ $(\interpretation(t_i), \\$$T_{\nullv})$ is deleted by the semantics of $\smerge$. Meanwhile, $\interpretation([t_{i,1}'', \ldots, t_{i,m}'', t_{i,m+1}''])$ are all deleted because of $\del(t_{i,j}'') \leftrightarrow \del(t_i) \lor \del(t_j') \lor \denot{\phi}_{\interpretation(\context), [\interpretation(\smerge(t_i, t_j'))]} \neq \top$ for $1 \leq j \leq m$ and $\del(t_{i, m+1}'') \leftrightarrow \del(t_i) \lor \neg \land_{j=1}^m \del(t_{i,j}'')$.
Therefore, $\site(|v_1(\interpretation(t_i))| = 0, v_2(\interpretation(t_i)),\\$$ v_1(\interpretation(t_i))) = \interpretation$ $([t_{i,1}'', \ldots, t_{i,m}'', t_{i,m+1}''])$.

Otherwise, if $\interpretation(\del)(t_i) = \bot$, then $\smerge(\interpretation(t_i), T_{\nullv}) = \interpretation(t_{i, m+1}'')$.
Also, since $\filter_\phi(\interpretation($ $[\smerge(t_i, t_1'), \ldots, \smerge(t_i, t_m')]))$ and $\interpretation([t_{i,1}'', \ldots, t_{i,m}''])$ are deleted, $\site(|v_1(\interpretation(t_i))| = 0, \\$$v_2(\interpretation(t_i)), v_1(\interpretation(t_i))) = [\smerge(\interpretation(t_i), T_\nullv)] = [\interpretation(t_{i,m+1}'')]$.

\item \label{proof_query2_ljoin_query_b}
If $|v_1(\interpretation(t_i))| = 0$ is false, then $|\denot{[\interpretation(t_i)] \ijoin_\phi \query'}_{\interpretation(\context)}| \neq 0$ and $\site(|v_1(\interpretation(t_i))| = 0, v_2(\interpretation(t_i)), \\$$v_1(\interpretation(t_i))) = v_1(\interpretation(t_i)) = \denot{[\interpretation(t_i)] \ijoin_\phi \query'}_{\interpretation(\context)} = \filter_\phi(\interpretation([\smerge(t_i, t_j'),$ $ \ldots, \smerge(t_i, t_j')]))$.

Further, if $\interpretation(\del)(t_i) = \top$, then $\interpretation(\filter_\phi([\smerge(t_i, t_j'), \ldots, \smerge(t_i, t_j')]))$ are deleted by the semantics of $\smerge$.
Meanwhile, $\interpretation([t_{i,1}'', \ldots, t_{i,m}'', t_{i,m+1}''])$ are all deleted because of $\del(t_{i,j}'') \leftrightarrow \del(t_i) \lor \del(t_j') \lor \denot{\phi}_{\interpretation(\context), [\interpretation(\smerge(t_i, t_j'))]} \neq \top$ for $1 \leq j \leq m$ and $\del(t_{i, m+1}'') \leftrightarrow \del(t_i) \lor \neg \land_{j=1}^m \del(t_{i,j}'')$.
Therefore, $\site(|v_1(\interpretation(t_i))| = 0, v_2(\interpretation(t_i)),$ $ v_1(\interpretation(t_i))) = \\$$\interpretation$ $([t_{i,1}'', \ldots, t_{i,m}'', t_{i,m+1}''])$.

Otherwise, if $\interpretation(\del)(t_i) = \bot$, $\interpretation(\del)(t_k') = \bot$ and $\denot{\phi}_{\interpretation(\context), [\interpretation(\smerge(t_i, t_k'))]} = \top$ where $1 \leq k \leq m$, then $\interpretation(t_{i, k}'') = \interpretation(\smerge(t_i, t_k'))$ and $\interpretation(\del)(t_{i, k}'') = \interpretation(\del)(\smerge(t_i, t_k')) = \bot$.

\end{enumerate}

Hence, $\denot{\query \ljoin_\phi \query'}_{\interpretation(\context)} = \interpretation(R)$ by \ref{proof_query2_ljoin_query_a} and \ref{proof_query2_ljoin_query_b} and Theorem~\ref{lem:query2} for the inductive case $\query'' = \query \ljoin_\phi \query'$ is proved.

\item Inductive case: $\query'' = \query \rjoin_\phi \query'$. 

Following a similar method in the inductive case $\query'' = \query \ljoin_\phi \query'$, we prove $\denot{\query \rjoin_\phi \query'}_{\interpretation(\context)} = \interpretation(R)$.

\item Inductive case: $\query'' = \query \fjoin_\phi \query'$. 

Suppose that by $\schema \vdash \query : [a_1, \ldots, a_p]$, $\schema \vdash \query : [a_1', \ldots, a_q']$ by Figure~\ref{fig:rules-schema}; $\schema \vdash \query \hookrightarrow R$ where $R$ has symbolic tuples $[t_1, \ldots, t_n]$, $\schema \vdash \query' \hookrightarrow R'$ where $R'$ has symbolic tuples $[t_1', \ldots, t_m']$, $\schema \vdash \query \fjoin_\phi \query' \hookrightarrow R''$ where $R''$ has symbolic tuples $[t_{1,1}'', \ldots, t_{n+1,m}'', t_{1,m+1}'', \ldots, t_{n,m+1}'']$, by Figure~\ref{fig:rules-tuples}; $\schema, \context \vdash \query \ljoin_\phi \query' \leadsto \formula$, $\schema, \context \vdash \query \fjoin_\phi \query' \leadsto \formula \land \land_{j=1}^{m} \formula_j'$ where $\formula_j' = \land_{k=1}^p \denot{a_k}_{\schema, \context, [t_{n+1,j}'']} = \nullv \land \land_{k=1}^q \denot{a_k'}_{\schema, \context, [t_{n+1,j}'']} = \denot{a_k'}_{\schema, \context, [t_j']} \land (\neg \del(t_j') \land \land_{i=1}^{n} \del(t_{i,j}'') \leftrightarrow \neg \del(t_{n+1,j}''))$ by Figure~\ref{fig:rules-encode-b}.

Take $\interpretation$ s.t. $\interpretation \models \formula \land \land_{j=1}^{m} \formula_j'$.

By the inductive hypothesis, we know $\denot{\query}_{\interpretation(\context)} = \interpretation(R)$,  $\denot{\query'}_{\interpretation(\context)} = \interpretation(R')$ and $\denot{\query \ljoin_\phi \query'}_{\interpretation(\context)} = \interpretation([t_{1,1}'', \ldots, t_{n,m+1}''])$.

For any tuple $t_j' \in \denot{\query'}_{\interpretation(\context)}$, if $\interpretation(\del)(t_j') = \top$, then $|\denot{\query \ijoin_\phi [t_j']}_{\interpretation(\context)}| = 0$ is true and, therefore, $\sfilter([\interpretation(t_j')], \lambda y. |\denot{\query \ijoin_\phi [y]}_{\interpretation(\context)}| = 0)$ is deleted because of the semantics of $\sfilter$ and $\interpretation(\del)(t_{n+1, j}'') = \top$ because of $\neg \del(t_j') \land \land_{i=1}^{n} \del(t_{i,j}'') \leftrightarrow \neg \del(t_{n+1,j}'')$.
Otherwise, $\interpretation(\del)(t_j') = \bot$ and then let us discuss $|\denot{\query \ijoin_\phi [t_j']}_{\interpretation(\context)}| = 0$ in two cases.

\begin{enumerate}[label=(\alph*)]
\item \label{proof_query2_fjoin_query_a}
If $|\denot{\query \ijoin_\phi [t_j']}_{\interpretation(\context)}| = 0$ is true, then $\smap(\sfilter([\interpretation(t_j')], \lambda y. |\denot{\query \ijoin_\phi [y]}_{\interpretation(\context)}| = 0), \\$$\lambda x. \smerge(T_\nullv, x)) = [\smerge$ $(T_\nullv, \interpretation(t_j'))]$, $\interpretation(\del)(t_{i,j}'') = \top$ for $1 \leq i \leq n$, $\interpretation(\del)(t_{n+1,j}'')\\$$ = \bot$ and $\interpretation(t_{n+1,j}'') = \smerge(T_\nullv, \interpretation(t_j'))$ because of $\formula_j'$ and the semantics of $\smerge$.

\item \label{proof_query2_fjoin_query_b}
If $|\denot{\query \ijoin_\phi [t_i]}_{\interpretation(\context)}| = 0$ is false, then $\smap(\sfilter([\interpretation(t_i)], \lambda y. |\denot{\query \ijoin_\phi [y]}_{\interpretation(\context)}| = 0), \lambda x.$ $ \smerge(T_\nullv, x))$ is deleted and $\interpretation(\del)(t_{n+1,j}') = \top$.

\end{enumerate}

Hence, $\denot{\query \fjoin_\phi \query'}_{\interpretation(\context)} = \interpretation(R'')$ by \ref{proof_query2_distinct_query_a} and \ref{proof_query2_distinct_query_b} and Theorem~\ref{lem:query2} for the inductive case $\query'' = \query \fjoin_\phi \query'$ is proved.

\item Inductive case: $\query' = \text{Distinct}(\query)$. 

Suppose that $\schema, \context \vdash \query \hookrightarrow R$ where $R$ has symbolic tuples $[t_1, \ldots, t_n]$, $\schema, \context \vdash \text{Distinct}(\query) \hookrightarrow R'$ where $R'$ has symbolic tuples $[t_1', \ldots, t_n']$ by Figure~\ref{fig:rules-tuples}; $\schema, \context \vdash \query \leadsto \formula$, $\schema, \context \vdash \query \leadsto \formula \land \formula'$ where $\formula' = \text{Dedup}(\vec{t}, \vec{t'})$ by Figure~\ref{fig:rules-encode-b}.

By the inductive hypothesis, we have $\interpretation, \db \models \formula$ and $\denot{\query}_{\interpretation(\context)} = \interpretation(R)$.

Take $\interpretation$ s.t. $\interpretation \models \formula \land \formula'$.

For any tuple $t_i \in \denot{\query}_{\interpretation(\context)}$, if $\interpretation(\del)(t_i) = \top$, then $\interpretation(\del)(t_i') = \top$.
On the other hand, for any existing tuples $t_i$ and $t_j$ in $\denot{\query}_{\interpretation(\context)}$ s.t. $\interpretation(\del)(t_i) = \interpretation(\del)(t_i) = \bot$ and $i \leq j$, let us discuss $\interpretation(t_i) = \interpretation(t_j)$ in two cases.

\begin{enumerate}[label=(\alph*)]
\item \label{proof_query2_distinct_query_a}
If $\interpretation(t_i) = \interpretation(t_j)$ is true, then $\interpretation(t_i') = \interpretation(t_i)$ and $\interpretation(\del)(t_i') = \interpretation(\del)(t_i) = \bot$, and $\sfilter(\interpretation(t_j), \lambda y. \interpretation(t_i) = y)$ is deleted and $\interpretation(\del)(t_j') = \top$ by the semantics of $\sfilter$.

\item \label{proof_query2_distinct_query_b}
If $\interpretation(t_i) = \interpretation(t_j)$ is false, then $\interpretation(t_i') = \interpretation(t_i)$ and $\interpretation(\del)(t_i') = \interpretation(\del)(t_i) = \bot$, and $\interpretation(t_j') = \interpretation(t_j)$ and $\interpretation(\del)(t_j') = \interpretation(\del)(t_j) = \bot$ by the semantics of $\sfilter$.

\end{enumerate}

Hence, $\denot{\text{Distinct}(\query)}_{\interpretation(\context)} = \interpretation(R'')$ by \ref{proof_query2_distinct_query_a} and \ref{proof_query2_distinct_query_b} and Theorem~\ref{lem:query2} for the inductive case $\query' = \text{Distinct}(\query)$ is proved.

\item Inductive case: $\query'' = \query \cap \query'$. 

Suppose that $\schema \vdash \text{Distinct}(\query) \hookrightarrow R$ where $R$ ash symbolic tuples $[t_1, \ldots, t_n]$, $\schema \vdash \query' \hookrightarrow R'$ where $R'$ ash symbolic tuples $[t_1', \ldots, t_m']$, $\schema \vdash \query \cap \query' \hookrightarrow R''$ where $R''$ ash symbolic tuples $[t_1'', \ldots, t_n'']$ by Figure~\ref{fig:rules-tuples}; $\schema, \context \vdash \text{Distinct}(\query) \leadsto \formula$, $\schema, \context \vdash \query' \leadsto \formula'$ , $\schema, \context \vdash \query \cap \query' \leadsto \formula \land \formula' \land \land_{i=1}^{n} \formula_i''$ where $\formula_i'' = t_i \in \vec{t'} \to t_i'' = t_i \land t_i \not \in \vec{t'} \to \del(t_i'')$ by Figure~\ref{fig:rules-encode-b}.

By the inductive hypothesis, we know $\denot{\text{Distinct}(\query)}_{\interpretation(\context)} = \interpretation(R)$ and $\denot{\query'}_{\interpretation(\context)} = \interpretation(R')$.

Take $\interpretation$ s.t. $\interpretation \models \formula \land \formula' \land \land_{i=1}^n \land_{j=1}^m \formula_{i,j}''$.

For any tuple $\interpretation(t_i) \in \denot{\text{Distinct}(\query)}_{\interpretation(\context)}$, if $\interpretation(\del)(t_i) = \top$, then $\interpretation(\del)(t_i') = \top$ because of $t_i \not \in \vec{t'} \to \del(t_i'')$; otherwise, let us discuss $\interpretation(t_i) \in \denot{\query'}_{\interpretation(\context)}$ in two cases.
\begin{enumerate}[label=(\alph*)]
\item \label{proof_query2_cap_query_a}
If $\interpretation(\del)(t_i) = \bot$ and $\interpretation(t_i) \in \denot{\query'}_{\interpretation(\context)}$ is true, then $\interpretation(t_i) = \interpretation(t_i')$ and $\interpretation(\del)(t_i) = \interpretation(\del)(t_i') = \bot$ by the semantics of $\sfilter$. 

\item \label{proof_query2_cap_query_b}
If $\interpretation(\del)(t_i) = \bot$ and $\interpretation(t_i) \in \denot{\query'}_{\interpretation(\context)}$ is false, then $\sfilter([t_i], \lambda x. x \in \denot{\query'}_{\interpretation(\context)})$ is deleted by the semantics of $\sfilter$ and $\interpretation(\del)(t_i') = \top$ because of $t_i \not \in \vec{t'} \to \del(t_i'')$. 

\end{enumerate}

Hence, $\denot{\query \cap \query'}_{\interpretation(\context)} = \interpretation(R'')$ by \ref{proof_query2_cap_query_a} and \ref{proof_query2_cap_query_b} and Theorem~\ref{lem:query2} for the inductive case $\query'' = \query \cap \query'$ is proved.

\item Inductive case: $\query = \query_1 \setminus \query_2$.

Suppose that $\schema \vdash \text{Distinct}(\query) \hookrightarrow R$ where $R$ ash symbolic tuples $[t_1, \ldots, t_n]$, $\schema \vdash \query' \hookrightarrow R'$ where $R'$ ash symbolic tuples $[t_1', \ldots, t_m']$, $\schema \vdash \query \setminus \query' \hookrightarrow R''$ where $R''$ ash symbolic tuples $[t_1'', \ldots, t_n'']$ by Figure~\ref{fig:rules-tuples}; $\schema, \context \vdash \text{Distinct}(\query) \leadsto \formula$, $\schema, \context \vdash \query' \leadsto \formula'$ , $\schema, \context \vdash \query \setminus \query' \leadsto \formula \land \formula' \land \land_{i=1}^{n} \formula_i''$ where $\formula_i'' = t_i \not \in \vec{t'} \to t_i'' = t_i \land t_i \in \vec{t'} \to \del(t_i'')$ by Figure~\ref{fig:rules-encode-b}.

By the inductive hypothesis, we know $\denot{\text{Distinct}(\query)}_{\interpretation(\context)} = \interpretation(R)$ and $\denot{\query'}_{\interpretation(\context)} = \interpretation(R')$.

Take $\interpretation$ s.t. $\interpretation \models \formula \land \formula' \land \land_{i=1}^n \land_{j=1}^m \formula_{i,j}''$.

For any tuple $\interpretation(t_i) \in \denot{\text{Distinct}(\query)}_{\interpretation(\context)}$, if $\interpretation(\del)(t_i) = \top$, then $\interpretation(\del)(t_i') = \top$ because of $t_i \in \vec{t'} \to t_i'' = t_i$; otherwise, let us discuss $\interpretation(t_i) \in \denot{\query'}_{\interpretation(\context)}$ in two cases.
\begin{enumerate}[label=(\alph*)]
\item \label{proof_query2_setmius_query_a}
If $\interpretation(\del)(t_i) = \bot$ and $\interpretation(t_i) \in \denot{\query'}_{\interpretation(\context)}$ is true, then $\sfilter([t_i], \lambda x. x  \not \in \denot{\query'}_{\interpretation(\context)})$ is deleted by the semantics of $\sfilter$ and $\interpretation(\del)(t_i') = \top$ because of $t_i \in \vec{t'} \to \del(t_i'')$. 

\item \label{proof_query2_setmius_query_b}
If $\interpretation(\del)(t_i) = \bot$ and $\interpretation(t_i) \in \denot{\query'}_{\interpretation(\context)}$ is false, then $\interpretation(t_i) = \interpretation(t_i')$ and $\interpretation(\del)(t_i) = \interpretation(\del)(t_i') = \bot$ by the semantics of $\sfilter$. 

\end{enumerate}

Hence, $\denot{\query \setminus \query'}_{\interpretation(\context)} = \interpretation(R'')$ by \ref{proof_query2_setmius_query_a} and \ref{proof_query2_setmius_query_b} and Theorem~\ref{lem:query2} for the inductive case $\query'' = \query \setminus \query'$ is proved.

\item Inductive case: $\query = \query_1 - \query_2$.

Suppose that $\schema \vdash \query \hookrightarrow R$ where $R$ has symbolic tuples $[t_1, \ldots, t_n]$, $\schema \vdash \query' \hookrightarrow R'$ where $R'$ has symbolic tuples $[t_1', \ldots, t_m']$, $\schema \vdash \query - \query' \hookrightarrow R''$ where $R''$ has symbolic tuples $[t_1'', \ldots, t_n'']$ by Figure~\ref{fig:rules-tuples}; $\schema, \context \vdash \query \leadsto \formula$, $\schema, \context \vdash \query' \leadsto \formula'$, $\schema, \context \vdash \query - \query' \leadsto \formula \land \formula' \land \land_{i=1}^n \land_{j=1}^m \formula_{i,j}''$ where $\formula_{i,j}'' = \neg \del(t_i) \land \neg \text{Paired}(i,j) \to t_i'' = t_i \land \del(t_i) \lor \text{Paired}(i,j) \to \del(t_i'')$ by Figure~\ref{fig:rules-encode-c}.

By the inductive hypothesis, we know $\denot{\query}_{\interpretation(\context)} = \interpretation(R)$ and $\denot{\query'}_{\interpretation(\context)} = \interpretation(R')$.

Take $\interpretation$ s.t. $\interpretation \models \formula \land \formula' \land \land_{i=1}^n \land_{j=1}^m \formula_{i,j}''$.

\item Inductive case: $\query'' = \query \cup \query'$. 

Since $\query \cup \query' \Leftrightarrow \text{Distinct}(\query \cplus \query')$, this case can be proved by the inductive cases $\query''= \query \cplus \query'$ and $\query' = \text{Distinct}(\query \cplus \query')$.

\item Inductive case: $\query'' = \query \uplus \query'$. 

Suppose that $\schema \vdash \query \hookrightarrow R$ where $R$ has symbolic tuples $[t_1, \ldots, t_n]$, $\schema \vdash \query' \hookrightarrow R'$ where $R'$ has symbolic tuples $[t_1', \ldots, t_m']$, $\schema \vdash \query \uplus \query' \hookrightarrow R''$ where $R''$ has symbolic tuples $[t_1'', \ldots, t_{n+m}'']$ by Figure~\ref{fig:rules-tuples}; $\schema, \context \vdash \query \leadsto \formula$, $\schema, \context \vdash \query' \leadsto \formula'$, $\schema, \context \vdash \query \uplus \query' \leadsto \formula \land \formula' \land \formula''$ where $\formula'' = \land_{i=1}^n t_i'' = t_i \land \land_{j=n+1}^{n+m} t_j'' = t_{j-n}'$ by Figure~\ref{fig:rules-encode-b}.

By the inductive hypothesis, we know $\denot{\query}_{\interpretation(\context)} = \interpretation(R)$ and $\denot{\query'}_{\interpretation(\context)} = \interpretation(R')$.

Take $\interpretation$ s.t. $\interpretation \models \formula \land \formula' \land \formula''$.

Further, since $t_i'' = t_i$ for $1 \leq i \leq n$ and $t_j'' = t_j$ for $n+1 \leq j \leq n+m$, we have
\[
\begin{array}{rcl}
\denot{\query \uplus \query'}_{\interpretation(\context)}
&=& \sappend(\denot{\query}_{\interpretation(\context)}, \denot{\query'}_{\interpretation(\context)}) \\
&=& \sappend(\interpretation(R), \interpretation(R')) \\
&=& \sappend([\interpretation(t_1), \ldots, \interpretation(t_n)], [\interpretation(t_1'), \ldots, \interpretation(t_m')]) \\
&=& [\interpretation(t_1), \ldots, \interpretation(t_n), \interpretation(t_1'), \ldots, \interpretation(t_m')] \\
&=& \interpretation(R'') \\
\end{array}
\]

Hence, Theorem~\ref{lem:query2} for the inductive case $\query'' = \query \uplus \query'$ is proved.

\item Inductive case: $\query'' = \query \cplus \query'$. 

Since $\query \cplus \query' \Leftrightarrow \query - (\query - \query')$, this case can be proved by the inductive case $\query''= \query - \query'$.

\item Inductive case: $\query' = \text{GroupBy}(\query, \vec{E}, L, \phi)$. 

Suppose that $\context \vdash \query \hookrightarrow R$ where $R$ has symbolic tuples $[t_1, \ldots, t_n]$, $\context \vdash \text{GroupBy}(\query, \vec{E}, L, \phi) \hookrightarrow R$ where $R$ has symbolic tuples $[t_1', \ldots, t_n']$ by Figure~\ref{fig:rules-tuples}; $\schema, \context \vdash \query \leadsto \formula$, $\schema, \context \vdash \text{GroupBy}(\query, \vec{E}, $ $L, \phi) \leadsto \formula \land \land_{i=1}^n \formula_i' \land \land_{i=1}^n \formula_i''$ where $\formula_i' = \Sigma_{j=1}^n \text{If}(g(t_i, j), 0, 1) = \text{If}(\del(t_i), 0, 1) \land \land_{j=1}^{i-1} g(t_i, j) = (\neg \del(t_i) \land g(t_i, j) \land \denot{\vec{E}})_{\schema, \context, [t_i]} = \denot{\vec{E}})_{\schema, \context, [t_ij}$ and $\formula_i'' = (g(t_i, j) \land \denot{\phi}_{\schema, \context, g^{-1}(i)} = \top \to \neg \del(t_i') \land \land_{k=1}^l \denot{a_k'}_{\schema, \context, [t_i']} =\denot{a_k'}_{\schema, \context, g^{-1}(i)}) \land (\neg g(t_i, j) \lor \denot{\phi}_{\schema, \context, g^{-1}(i)} \neq \top \to \del(t_i'))$ by Figure~\ref{fig:rules-encode-b}; $L = [a_1, \ldots a_l]$.

By the inductive hypothesis, we have $\denot{\query}_{\interpretation(\context)} = \interpretation(R)$.

Take $\interpretation$ s.t. $\interpretation \models \formula \land \land_{i=1}^n \formula_i' \land \land_{i=1}^n \formula_i''$.

Let $G_i$ be the $i$-th group of $\denot{\query}_{\interpretation(\context)}$ and $g^{-1}(i)$ be the function takes the index of the group $G_i$ and returns the tuples of the group $G_i$.
Assume there are $m$ non-empty groups in $\denot{\query}_{\interpretation(\context)}$ s.t. $m \leq n$ and the other $n-m$ empty groups are deleted because of no non-deleted tuple.
Therefore, we know that for each group $G_i$, (i) it must have at least one non-deleted tuple, (ii) all tuples $g^{-1}(i)$ in $G_i$ must share the same values over $\vec{E}$ and (iii) $\denot{\vec{E}}_{\interpretation(\context), \interpretation(g^{-1}(i))}$ is unique among those groups because of Figure~\ref{fig:rules-encode-b}.

Further, by the semantics of $\sfilter$ and $\denot{\query}_{\interpretation(\context)} = \interpretation(R)$, we know $\text{Dedup}(\query, \vec{E})$ denotes $m$ unique values over $\vec{E}$.
Therefore, $Gs$ in $\denot{\text{GroupBy}(\query, \vec{E}, L, \phi)}_{\interpretation(\context)}$ is equivalent to the groups $\interpretation([G_1, \ldots, G_m])$ in the formula $\formula_i''$ and $\interpretation(\del)(t_i) = \top$ since $G_i$ is empty group for $m < i \leq n$.
\[
\begin{array}{rcl}
\denot{\text{GroupBy}(\query, \vec{E}, L, \phi)}_{\interpretation(\context)} 
&=& \smap(\sfilter(Gs, \lambda xs. \denot{\phi}_{\interpretation(\context), xs} = \top), \lambda xs. \denot{L}_{\interpretation(\context), xs}) \\
&=& \smap(\sfilter([G_1, \ldots, G_m], \lambda xs. \denot{\phi}_{\interpretation(\context), xs} = \top), \lambda xs. \denot{L}_{\interpretation(\context), xs}) \\
&=& [\interpretation(t_1'), \ldots, \interpretation(t_m')] \\
&=& \sappend([\interpretation(t_1'), \ldots, \interpretation(t_m')], [\interpretation(t_{m+1}), \ldots, \interpretation(t_n)]) \\
&=& [\interpretation(t_1'), \ldots, \interpretation(t_m'), \interpretation(t_{m+1}'), \ldots, \interpretation(t_n')]\\
&=& \interpretation(R')
\end{array}
\]

Hence, Theorem~\ref{lem:query2} for the inductive case $\query' = \text{GroupBy}(\query, \vec{E}, L, \phi)$ is proved.

\item Inductive case: $\query'' = \text{With}(\vec{\query}, \vec{T}, \query')$. 

Suppose that $\schema \vdash \query_i : \attributes_i$ by Figure~\ref{fig:rules-schema}; $\context \vdash \query_i \hookrightarrow R_i$ by Figure~\ref{fig:rules-tuples}; $\schema, \context \vdash \query_i \leadsto \formula_i$ by Figure~\ref{fig:rules-encode-b}.

By the inductive hypothesis, we have $\interpretation(R_i) = \denot{\query_i}_{\interpretation(\context)}$ where $1 \leq i \leq n$ and $n = |\vec{\query}| = |\vec{T}|$; $\context' \vdash \query' \hookrightarrow R'$, $\schema', \context' \vdash \query' \leadsto \formula'$, $\interpretation(R') = \denot{\query'}_{\interpretation(\context')}$ where $\schema' = \schema[T_1 \mapsto \attributes_1, \ldots, T_n \mapsto \attributes_n]$ and $\context' = \context[T_1 \mapsto \tuples_1, \ldots, T_n \mapsto \tuples_n]$.

Take $\interpretation$ s.t. $\interpretation \models \land_{i=1}^n \formula_i \land \formula'$.

By Figure~\ref{fig:proof-sematics2}, we know $\denot{\text{With}(\vec{\query}, \vec{T}, \query')}_{\interpretation(\context')} 
= \denot{\query'}_{\interpretation(\context')} = \interpretation(R')$.


Hence, Theorem~\ref{lem:query2} for the inductive case $\query'' = \text{With}(\vec{\query}, \vec{T}, \query')$ is proved.

\item Inductive case: $\query' = \text{OrderBy}(\query, \vec{E}, b)$. 

Suppose that $\context \vdash \query \hookrightarrow R$ where $R$ has symbolic tuples $[t_1, \ldots, t_n]$, $\context \vdash \text{OrderBy}(\query, \vec{E}, b) \hookrightarrow R'$ where $R'$ has symbolic tuples $[t_1', \ldots, t_n']$ by Figure~\ref{fig:rules-tuples};
$\schema, \context \vdash \query \leadsto \formula$, $\schema, \context \vdash $ $\text{OrderBy}(\query, \vec{E}, b) \leadsto \formula \land \formula' \land \land_{i=1}^n \formula_i''$ where $\formula' = \text{moveDelToEnd}(\vec{t}, \vec{t''})$ and $\formula_i'' = t_i' = t_{\text{find}(i, \vec{E}, b)}''$ by Figure~\ref{fig:rules-encode-c}; $R''$ is the fresh symbolic table where $R''$ has the symbolic tuples $[t_1'', \ldots, t_n'']$ and $t_{\text{indexOf}(\vec{t}, i)} = t_i''$ for any tuple $t_i'' \in R''$.

By the inductive hypothesis, we have $\interpretation(R) = \denot{\query}_{\interpretation(\context)}$. Also, by the semantics of the $\text{moveDelToEnd}$ function, $\interpretation(\del)(t_{\text{indexOf}(\vec{t}, i)}) = \interpretation(\del)(t_i'')$ and $\interpretation(t_{\text{indexOf}(\vec{t}, i)}) = \interpretation(t_i'')$.

Take $\interpretation$ s.t. $\interpretation \models \formula \land \formula' \land \land_{i=1}^n \formula_i''$.


Assume there are $m$ non-deleted symbolic tuples in $R$ where $m \leq n$, then we have, for any tuple $t_i'' \in [t_1'', \ldots, t_m'']$,  $\interpretation(t_i) = \interpretation(t_{\text{indexOf}(\vec{t}, i)}'')$ and $\interpretation(del)(t_i) = \interpretation(del)(t_{\text{indexOf}(\vec{t}, i)}'') = \bot$ and the other tuples $[t_{m+1}'', \ldots, t_n'']$ are deleted (i.e., $\interpretation(\del)(t_i) = \top$ where $m < i \leq n$) by the semantics of the \text{moveDelToEnd} function.
Let $a_i$ denote the $i$-th smallest (biggest) tuple from $\denot{\query}_{\interpretation(\context)} - xs$ (i.e., $a_i = \text{MinTuple}(\vec{E}, b, $ $\denot{\query}_{\interpretation(\context)} - xs)$ where $xs = [a_i, \ldots, a_{i-1}]$), then

\[
\begin{array}{rcl}
\denot{\text{OrderBy}(\query, \vec{E}, b)}_{\interpretation(\context)}
&=& \sfoldl(\lambda xs. \lambda \_. (\sappend(xs, [\text{MinTuple}(\vec{E}, b, \denot{\query}_{\interpretation(\context)} - xs)])), \\
&& \quad [], \denot{\query}_{\interpretation(\context)}) \\
&=& \sfoldl(\lambda xs. \lambda \_. (\sappend(xs, [\text{MinTuple}(\vec{E}, b, \denot{\query}_{\interpretation(\context)} - xs)])), \\
&& \quad [], \interpretation(R)) \\
&& \ldots \\
&=& [\interpretation(a_1), \ldots, \interpretation(a_m)] \\
&=& \interpretation([a_1, \ldots, a_m]) \\
\end{array}
\]

Furthermore, let $t_{\text{indexOf}(\vec{t}, i)}$ be the corresponding tuple of $a_i \in \denot{\query}_{\interpretation(\context)}$,

\[
\begin{array}{rcl}
\denot{\text{OrderBy}$ $(\query, \vec{E}, b)}_{\interpretation(\context)} 
&=& \interpretation([a_1, \ldots, a_m]) \\
&=& \interpretation([t_{\text{indexOf}(\vec{t}, i)}, \ldots, t_{\text{indexOf}(\vec{t}, i)}]) \\
&=& \interpretation([t_1'', \ldots, t_m'']) \\
&=& \interpretation([t_1'', \ldots, t_m'']) + \interpretation([t_{m+1}'', \ldots, t_n'']) \\
&=& \interpretation([t_1'', \ldots, t_n'']) \\
&=& \interpretation([t_{\text{find}(1, \vec{E}, b)}'', \ldots, t_{\text{find}(n, \vec{E}, b)}'']) \\
&=& \interpretation([t_1', \ldots, t_n']) \\
&=& \interpretation(R') \\
\end{array}
\]

Hence, Theorem~\ref{lem:query2} for the inductive case $\query' = \text{OrderBy}(\query, \vec{E}, b)$ is proved.

\end{enumerate}
\end{proof}

\begin{lemma}\label{lem:predicate}
\revision{
Let $\db$ be a database over schema $\schema$ and $xs$ be a tuple list which formulates a predicate $\phi$.
Consider a symbolic database $\context$ over $\schema$ and a list of symbolic tuples $\tuples$ such that $\denot{\phi}_{\schema, \context, \tuples}$ is valid.
For any interpretation $\interpretation$ such that $\interpretation(\context) = \db \land \interpretation(\tuples) = xs$, evaluating $\phi$ over the concrete database $\interpretation(\context)$ and the concrete tuple list $\interpretation(\tuples)$ yields the value of $\interpretation(\denot{\phi}_{\schema, \context, \tuples})$, i.e., 
\[
\interpretation(\context) = \db \land \interpretation(\tuples) = xs \Rightarrow \denot{\phi}_{\interpretation(\context), \interpretation(\tuples)} = \interpretation(\denot{\phi}_{\schema, \context, \tuples})
\]
}
\end{lemma}
\vspace{10pt}

\begin{lemma}\label{lem:proof_predicate}
Suppose $\denot{\phi}_{\db, xs}$ is valid, then $\interpretation(\context) = \db \land \interpretation(\tuples) = xs \Rightarrow \denot{\phi}_{\interpretation(\context), \interpretation(\tuples)} = \interpretation(\denot{\phi}_{\schema, \context, \tuples})$ holds.
\end{lemma}
\begin{proof}[Proof.]
By structural induction on $\phi$. 
\begin{enumerate}

\item Base case: $\phi = \top$.

$\denot{\top}_{\schema, \context, \tuples} = \top$ by Figure~\ref{fig:rules-predicates}. Also, $\denot{\top}_{\interpretation(\context), \interpretation(\tuples)} = \top$ by Figure~\ref{fig:proof-sematics2}. Therefore, $\interpretation(\denot{\top}_{\schema, \context, \tuples}) = \interpretation(\top) = \top = \denot{\top}_{\interpretation(\context), \interpretation(\tuples)}$.

\item Base case: $\phi = \bot$.

$\denot{\bot}_{\schema, \context, \tuples} = \bot$ by Figure~\ref{fig:rules-predicates}. Also, $\denot{\bot}_{\interpretation(\context), \interpretation(\tuples)} = \bot$ by Figure~\ref{fig:proof-sematics2}. Therefore, $\interpretation(\denot{\bot}_{\schema, \context, \tuples}) = \interpretation(\bot) = \bot = \denot{\bot}_{\interpretation(\context), \interpretation(\tuples)}$.

\item Base case: $\phi = \nullv$.

$\denot{\nullv}_{\schema, \context, \tuples} = \nullv$ by Figure~\ref{fig:rules-predicates}. Also, $\denot{\nullv}_{\interpretation(\context), \interpretation(\tuples)} = \nullv$ by Figure~\ref{fig:proof-sematics2}. Therefore, $\interpretation(\denot{\nullv}_{\schema, \context, \tuples}) = \interpretation(\nullv) = \nullv = \denot{\nullv}_{\interpretation(\context), \interpretation(\tuples)}$.

\item Base case: $\phi = A_1 \alllogic A_2$.

$\denot{A_1 \alllogic A_2}_{\schema, \context, \tuples} = \site(v_1 = \nullv \lor v_2 = \nullv, \bot, v_1 \alllogic v_2)$ where $v_1 = \denot{A_1}_{\schema, \context, \tuples}$ and $v_2 = \denot{A_2}_{\schema, \context, \tuples}$ by Figure~\ref{fig:rules-predicates}.
$\denot{A_1 \alllogic A_2}_{\interpretation(\context), \interpretation(\tuples)} = \site(v_1' = \nullv \lor v_2' = \nullv, \bot, v_1' \alllogic v_2')$ where $v_1' = \denot{A_1}_{\interpretation(\context), \interpretation(\tuples)}$ and $v_2' = \denot{A_2}_{\interpretation(\context), \interpretation(\tuples)}$ by Figure~\ref{fig:proof-sematics2}.
By Lemma~\ref{lem:expression}, we have $\denot{A_1}_{\interpretation(\context), \interpretation(\tuples)} = \interpretation(\denot{A_1}_{\schema, \context, \tuples})$ and $\denot{A_2}_{\interpretation(\context), \interpretation(\tuples)} = \interpretation(\denot{A_2}_{\schema, \context, \tuples})$.
Also since $\interpretation(v_1) = \interpretation(\denot{A_1}_{\schema, \context, \tuples}) = \\$$\denot{A_1}_{\interpretation(\context), \interpretation(\tuples)}$ $ = v_1'$ and $\interpretation(v_2) = v_2'$,
\[
\begin{array}{rcl}
\interpretation(\denot{A_1 \alllogic A_2}_{\schema, \context, \tuples})
&=& \interpretation(\site(v_1 = \nullv \lor v_2 = \nullv, \bot, v_1 \alllogic v_2)) \\
&=& \site(\interpretation(v_1) = \nullv \lor \interpretation(v_2) = \nullv, \bot, \interpretation(v_1) \alllogic \interpretation(v_2)) \\
&=& \site(v_1' = \nullv \lor v_2' = \nullv, \bot, v_1' \alllogic v_2') \\
&=& \denot{A_1 \alllogic A_2}_{\interpretation(\context), \interpretation(\tuples)} \\
\end{array}
\]

\item Base case: $\phi = \text{IsNull}(E)$.

$\denot{\text{IsNull}(E)}_{\schema, \context, \tuples} = \site(\denot{E}_{\schema, \context, \tuples} = \nullv, \top, \bot)$ by Figure~\ref{fig:rules-predicates}.
$\denot{\text{IsNull}(E)}_{\interpretation(\context), \interpretation(\tuples)} = \site($ $\denot{E}_{\interpretation(\context), \interpretation(\tuples)} = \nullv, \top, \bot)$ by Figure~\ref{fig:proof-sematics2}.
By Lemma~\ref{lem:expression}, we have $\denot{E}_{\interpretation(\context), \interpretation(\tuples)} = \interpretation(\denot{E}_{\schema, \context, \tuples})$.
\[
\begin{array}{rcl}
\interpretation(\denot{\text{IsNull}(E)}_{\schema, \context, \tuples})
&=& \interpretation(\site(\denot{E}_{\schema, \context, \tuples} = \nullv, \top, \bot)) \\
&=& \site(\interpretation(\denot{E}_{\schema, \context, \tuples}) = \nullv, \top, \bot) \\
&=& \site(\denot{E}_{\interpretation(\context), \interpretation(\tuples)} = \nullv, \top, \bot) \\
&=& \denot{\text{IsNull}(E)}_{\interpretation(\context), \interpretation(\tuples)}
\end{array}
\]

\item Base case: $\phi = \vec{E} \in \vec{v}$.
$\denot{\vec{E} \in \vec{v}}_{\schema, \context, \tuples} = \denot{\lor_{i=1}^{n} \land_{j=1}^{m} E_j = v_{i,j}}_{\schema, \context, \tuples}$ where $\vec{v} = [v_{1,1}, \ldots, v_{n, m}]$ and $m = |\vec{E}|$ by Figure~\ref{fig:rules-predicates}.
$\denot{\vec{E} \in \vec{v}}_{\interpretation(\context), \interpretation(\tuples)} = \lor_{i=1}^{n} \denot{\vec{E}}_{\interpretation(\context), \interpretation(\tuples)} = \denot{v_i}_{\interpretation(\context), \interpretation(\tuples)}$ by Figure~\ref{fig:proof-sematics2}.
By Lemma~\ref{lem:expression}, we have $\denot{E_j}_{\interpretation(\context), \interpretation(\tuples)} = \interpretation(\denot{E_j}_{\schema, \context, \tuples})$ and $\denot{v_{i,j}}_{\interpretation(\context), \interpretation(\tuples)} = \interpretation(\denot{v_{i,j}}_{\schema, \context, \tuples})$.
\[
\begin{array}{rcl}
\interpretation(\denot{\vec{E} \in \vec{v}}_{\schema, \context, \tuples})
&=& \interpretation(\denot{\lor_{i=1}^{n} \land_{j=1}^{m} E_j = v_{i,j}}_{\schema, \context, \tuples}) \\
&=& \lor_{i=1}^{n} \land_{j=1}^{m} \interpretation(\denot{E_j = v_{i,j}}_{\schema, \context, \tuples}) \\
&=& \lor_{i=1}^{n} \land_{j=1}^{m} \interpretation(\denot{E_j}_{\schema, \context, \tuples}) = \interpretation(\denot{v_{i,j}}_{\schema, \context, \tuples}) \\
&=& \lor_{i=1}^{n} \land_{j=1}^{m} \denot{E_j}_{\interpretation(\context), \interpretation(\tuples)} = \denot{v_{i,j}}_{\interpretation(\context), \interpretation(\tuples)} \\
&=& \denot{\vec{E} \in \vec{v}}_{\interpretation(\context), \interpretation(\tuples)} \\
\end{array}
\]

\item Base case: $\phi = \vec{E} \in \query$.

Suppose that $\denot{\query}_{\interpretation(\context)} = \vec{v} = [v_{1,1}, \ldots, v_{n,m}]$ by .

\item Inductive case: $\phi = \phi_1 \land \phi_2$.

$\denot{\phi_1 \land \phi_2}_{\schema, \context, \tuples} = \ite((v_1 = v_2 = \nullv) \land (v_1 = \nullv \land v_2 = \bot) \land (v_1 = \bot \land v_2 = \nullv), \nullv, v_1  = \top \land v_2 = \top)$ where $v_1 = \denot{\phi_1}_{\schema, \context, \tuples}$ and $v_2 = \denot{\phi_2}_{\schema, \context, \tuples}$ by Figure~\ref{fig:rules-predicates}.
$\denot{\phi_1 \land \phi_2}_{\interpretation(\context), \interpretation(\tuples)} = \ite((v_1' = v_2' = \nullv) \land (v_1' = \nullv \land v_2' = \bot) \land (v_1' = \bot \land v_2' = \nullv), \nullv, v_1'  = \top \land v_2' = \top)$ where $v_1' = \denot{\phi_1}_{\interpretation(\context), \interpretation(\tuples)}$ and $v_2' = \denot{\phi_2}_{\interpretation(\context), \interpretation(\tuples)}$ by Figure~\ref{fig:proof-sematics2}.
By Lemma~\ref{lem:predicate}, we have $\denot{\phi_1}_{\interpretation(\context), \interpretation(\tuples)} = \interpretation(\denot{\phi_1}_{\schema, \context, \tuples})$ and $\denot{\phi_2}_{\interpretation(\context), \interpretation(\tuples)} = \interpretation(\denot{\phi_2}_{\schema, \context, \tuples})$.
Also since $\interpretation(v_1) = \interpretation(\denot{\phi_1}_{\schema, \context, \tuples}) = \denot{\phi_1}_{\interpretation(\context), \interpretation(\tuples)} = v_1'$ and $\interpretation(v_2) = v_2'$, 
\[
\begin{array}{rcl}
\interpretation(\denot{\phi_1 \land \phi_2}_{\schema, \context, \tuples}) 
&=& \interpretation(\ite((v_1 = \nullv \land v_2 = \nullv) \land (v_1 = \nullv \land v_2 = \bot) \\
&& \qquad \land (v_1 = \bot \land v_2 = \nullv), \nullv, v_1  = \top \land v_2 = \top)) \\
&=& \ite((\interpretation(v_1) = \nullv \land \interpretation(v_2) = \nullv) \land (\interpretation(v_1) = \nullv \land \interpretation(v_2) = \bot) \\
&& \qquad \land (\interpretation(v_1) = \bot \land \interpretation(v_2) = \nullv), \nullv, \interpretation(v_1)  = \top \land \interpretation(v_2) = \top) \\
&=& \ite((v_1' = \nullv \land v_2' = \nullv) \land (v_1' = \nullv \land v_2' = \bot) \\
&& \qquad \land (v_1' = \bot \land v_2' = \nullv), \nullv, v_1'  = \top \land v_2' = \top) \\
&=& \denot{\phi_1 \land \phi_2}_{\interpretation(\context), \interpretation(\tuples)}
\end{array}
\]

\item Inductive case: $\phi = \phi_1 \lor \phi_2$.

$\denot{\phi_1 \lor \phi_2}_{\schema, \context, \tuples} = \ite((v_1 = v_2 = \nullv) \lor (v_1 = \nullv \land v_2 = \bot) \lor (v_1 = \bot \land v_2 = \nullv), \nullv, v_1  = \top \lor v_2 = \top)$ where $v_1 = \denot{\phi_1}_{\schema, \context, \tuples}$ and $v_2 = \denot{\phi_2}_{\schema, \context, \tuples}$ by Figure~\ref{fig:rules-predicates}.
$\denot{\phi_1 \lor \phi_2}_{\interpretation(\context), \interpretation(\tuples)} = \ite((v_1' = v_2' = \nullv) \lor (v_1' = \nullv \land v_2' = \bot) \lor (v_1' = \bot \land v_2' = \nullv), \nullv, v_1'  = \top \lor v_2' = \top)$ where $v_1' = \denot{\phi_1}_{\interpretation(\context), \interpretation(\tuples)}$ and $v_2' = \denot{\phi_2}_{\interpretation(\context), \interpretation(\tuples)}$ by Figure~\ref{fig:proof-sematics2}.
By Lemma~\ref{lem:predicate}, we have $\denot{\phi_1}_{\interpretation(\context), \interpretation(\tuples)} = \interpretation(\denot{\phi_1}_{\schema, \context, \tuples})$ and $\denot{\phi_2}_{\interpretation(\context), \interpretation(\tuples)} = \interpretation(\denot{\phi_2}_{\schema, \context, \tuples})$.
Also since $\interpretation(v_1) = \interpretation(\denot{\phi_1}_{\schema, \context, \tuples}) = \denot{\phi_1}_{\interpretation(\context), \interpretation(\tuples)} = v_1'$ and $\interpretation(v_2) = v_2'$, 
\[
\begin{array}{rcl}
\interpretation(\denot{\phi_1 \lor \phi_2}_{\schema, \context, \tuples}) 
&=& \interpretation(\ite((v_1 = \nullv \lor v_2 = \nullv) \lor (v_1 = \nullv \land v_2 = \bot) \\
&& \qquad \lor (v_1 = \bot \land v_2 = \nullv), \nullv, v_1  = \top \lor v_2 = \top)) \\
&=& \ite((\interpretation(v_1) = \nullv \lor \interpretation(v_2) = \nullv) \lor (\interpretation(v_1) = \nullv \land \interpretation(v_2) = \bot) \\
&& \qquad \lor (\interpretation(v_1) = \bot \land \interpretation(v_2) = \nullv), \nullv, \interpretation(v_1)  = \top \lor \interpretation(v_2) = \top) \\
&=& \ite((v_1' = \nullv \lor v_2' = \nullv) \lor (v_1' = \nullv \land v_2' = \bot) \\
&& \qquad \lor (v_1' = \bot \land v_2' = \nullv), \nullv, v_1'  = \top \lor v_2' = \top) \\
&=& \denot{\phi_1 \lor \phi_2}_{\interpretation(\context), \interpretation(\tuples)}
\end{array}
\]

\item Inductive case: $\phi = \neg \phi_1$.

$\denot{\neg \phi_1}_{\schema, \context, \tuples} = \site(\denot{\phi_1}_{\schema, \context, \tuples} = \nullv, \nullv, \neg \denot{\phi_1}_{\schema, \context, \tuples})$ because $\denot{\neg \phi}_{\schema, \context, \tuples} = \text{let}~ v = \denot{\phi}_{\schema, \context, \tuples} \\$$~\text{in}~ \ite(v = \nullv, \nullv, \neg v)$ by Figure~\ref{fig:rules-predicates}.
$\denot{\neg \phi_1}_{\interpretation(\context), \interpretation(\tuples)} = \site(\denot{\phi_1}_{\interpretation(\context), \interpretation(\tuples)} = \nullv, \nullv, \\$$ \neg \denot{\phi_1}_{\interpretation(\context), \interpretation(\tuples)})$ because $\denot{\neg \phi}_{\db, xs} = \text{let}~ v = \denot{\phi}_{\db, xs} ~\text{in}~ \ite(v = \nullv, \nullv, \neg v)$ by Figure~\ref{fig:proof-sematics2}.
By Lemma~\ref{lem:predicate}, we have $\denot{\phi_1}_{\interpretation(\context), \interpretation(\tuples)} = \interpretation(\denot{\phi_1}_{\schema, \context, \tuples})$.
Therefore, 
\[
\begin{array}{rcl}
\interpretation(\denot{\neg \phi_1}_{\schema, \context, \tuples}) &=& \site(\interpretation(\denot{\phi_1}_{\schema, \context, \tuples} = \nullv) = \nullv, \interpretation(\neg \denot{\phi_1}_{\schema, \context, \tuples})) \\
&=& \site(\interpretation(\denot{\phi_1}_{\schema, \context, \tuples}) = \nullv, \nullv, \neg \interpretation(\denot{\phi_1}_{\schema, \context, \tuples})) \\
&=& \site(\denot{\phi_1}_{\interpretation(\context), \interpretation(\tuples)} = \nullv, \nullv, \neg \denot{\phi_1}_{\interpretation(\context), \interpretation(\tuples)}) \\
&=& \denot{\neg \phi_1}_{\interpretation(\context), \interpretation(\tuples)}
\end{array}
\]

\end{enumerate}
\end{proof}

\begin{lemma}\label{lem:expression}
\revision{
Let $\db$ be a database over schema $\schema$ and $xs$ be a tuple list which formulates an expression $E$.
Consider a symbolic database $\context$ over $\schema$ and a list of symbolic tuples $\tuples$ such that $\denot{E}_{\schema, \context, \tuples}$.
For any interpretation $\interpretation$ such that $\interpretation(\context) = \db \land \interpretation(\tuples) = xs$, evaluating $E$ over the concrete database $\interpretation(\context)$ and the concrete tuple list $\interpretation(\tuples)$ yields the value of $\interpretation(\denot{E}_{\schema, \context, \tuples})$, i.e., 
\[
\interpretation(\context) = \db \land \interpretation(\tuples) = xs \Rightarrow \denot{E}_{\interpretation(\context), \interpretation(\tuples)} = \interpretation(\denot{E}_{\schema, \context, \tuples})
\]
}
\end{lemma}

\begin{lemma}\label{lem:proof_expression}
Suppose $\denot{E}_{\db, xs} = v$, then $\interpretation(\context) = \db \land \interpretation(\tuples) = xs \Rightarrow \denot{E}_{\interpretation(\context), \interpretation(\tuples)} = \interpretation(\denot{E}_{\schema, \context, \tuples})$ is true iff $\denot{E}_{\interpretation(\context), \interpretation(\tuples)} = v$ and $\interpretation(\denot{E}_{\schema, \context, \tuples}) = v$.
\end{lemma}
\begin{proof}[Proof.]
By structural induction on $E$. 
\begin{enumerate}

\item Base case: $E = a$. 

$\denot{a}_{\schema, \context, \tuples} = \shead(\tuples).a$ by Figure~\ref{fig:rules-expressions}.
$\denot{a}_{\interpretation(\context), \interpretation(\tuples)} = lookup(\shead(\interpretation(\tuples)), a)$ by Figure~\ref{fig:proof-sematics2}. Therefore, 
\[
\begin{array}{rcl}
\interpretation(\denot{a}_{\schema, \context, \tuples}) 
&=& \interpretation(\shead(\tuples).a) = \interpretation(lookup(\shead(\tuples), a)) \\
&=& lookup(\shead(\interpretation(\tuples)), a) = lookup(\shead(xs), a) \\
&=& \denot{a}_{\interpretation(\context), \interpretation(\tuples)}
\end{array}
\]

\item Base case: $E = v$. 

$\denot{v}_{\schema, \context, \tuples} = v$ by Figure~\ref{fig:rules-expressions}.
$\denot{v}_{\interpretation(\context), \interpretation(\tuples)} = v$ by Figure~\ref{fig:proof-sematics2}. Therefore, $\interpretation(\denot{v}_{\schema, \context, \tuples}) = \interpretation(v) = v = \denot{v}_{\db, xs}$.

\item Inductive case: $E = E_1 \allarith E_2$. 

$\denot{E_1 \allarith E_2}_{\schema, \context, \tuples} = \site(v_1 = \nullv \lor v_2 = \nullv, \nullv, v_1 \allarith v_2)$ where $v_1 = \denot{E_1}_{\schema, \context, \tuples}$ and $v_2 = \denot{E_2}_{\schema, \context, \tuples}$ by Figure~\ref{fig:rules-expressions}.
$\denot{E_1 \allarith E_2}_{\interpretation(\context), \interpretation(\tuples)} = \site(v_1' = \nullv \lor v_2' = \nullv, \nullv, v_1' \allarith v_2')$ where $v_1' = \denot{E_1}_{\interpretation(\context), \interpretation(\tuples)}$ and $v_2' = \denot{E_2}_{\interpretation(\context), \interpretation(\tuples)}$ by Figure~\ref{fig:proof-sematics2}.
By inductive hypothesis, we have $\denot{E_1}_{\interpretation(\context), \interpretation(\tuples)} = \interpretation(\denot{E_1}_{\schema, \context, \tuples})$ and $\denot{E_2}_{\interpretation(\context), \interpretation(\tuples)} = \interpretation(\denot{E_2}_{\schema, \context, \tuples})$.
Also since $\interpretation(v_1) = \interpretation(\denot{E_1}_{\schema, \context, \tuples}) = \denot{E_1}_{\interpretation(\context), \interpretation(\tuples)} = v_1'$ and $\interpretation(v_2) = v_2'$, 
\[
\begin{array}{rcl}
\interpretation(\denot{E_1 \allarith E_2}_{\schema, \context, \tuples})
&=& \interpretation(\site(v_1 = \nullv \lor v_2 = \nullv, \nullv, v_1 \allarith v_2)) \\
&=& \site(\interpretation(v_1) = \nullv \lor \interpretation(v_2) = \nullv, \nullv, \interpretation(v_1) \allarith \interpretation(v_2)) \\
&=& \site(v_1' = \nullv \lor v_2' = \nullv, \nullv, v_1' \allarith v_2') \\
&=& \denot{E_1 \allarith E_2}_{\interpretation(\context), \interpretation(\tuples)} \\
\end{array}
\]

\item Inductive case: $E = \text{ITE}(\phi, E_1, E_2)$.

$\denot{\text{ITE}(\phi, E_1, E_2)}_{\schema, \context, \tuples} = \site(\denot{\phi}_{\schema, \context, \tuples} = \top, \denot{E_1}_{\schema, \context, \tuples}, \denot{E_2}_{\schema, \context, \tuples})$ by Figure~\ref{fig:rules-expressions}. \\
$\denot{\text{ITE}(\phi, E_1,$ $ E_2)}_{\interpretation(\context), \interpretation(\tuples)} = \site(\denot{\phi}_{\interpretation(\context), \interpretation(\tuples)} = \top, \denot{E_1}_{\interpretation(\context), \interpretation(\tuples)}, \denot{E_2}_{\interpretation(\context), \interpretation(\tuples)})$ by Figure~\ref{fig:proof-sematics2}.
By inductive hypothesis, we have $\denot{\phi}_{\interpretation(\context), \interpretation(\tuples)} = \interpretation(\denot{\phi}_{\schema, \context, \tuples})$, $\denot{E_1}_{\interpretation(\context), \interpretation(\tuples)} = \interpretation(\denot{E_1}_{\schema, \context, \tuples})$ and $\denot{E_2}_{\interpretation(\context), \interpretation(\tuples)} = \interpretation(\denot{E_2}_{\schema, \context, \tuples})$.
Therefore,
\[
\begin{array}{rcl}
\interpretation(\denot{\text{ITE}(\phi, E_1, E_2)}_{\schema, \context, \tuples})
&=& \interpretation(\site(\denot{\phi}_{\schema, \context, \tuples} = \top, \denot{E_1}_{\schema, \context, \tuples}, \denot{E_2}_{\schema, \context, \tuples})) \\
&=& \site(\interpretation(\denot{\phi}_{\schema, \context, \tuples}) = \top, \interpretation(\denot{E_1}_{\schema, \context, \tuples}), \interpretation(\denot{E_2}_{\schema, \context, \tuples})) \\
&=& \site(\denot{\phi}_{\interpretation(\context), \interpretation(\tuples)} = \top, \denot{E_1}_{\interpretation(\context), \interpretation(\tuples)}, \denot{E_2}_{\interpretation(\context), \interpretation(\tuples)}) \\
&=& \denot{\text{ITE}(\phi, E_1, E_2)}_{\interpretation(\context), \interpretation(\tuples)} \\
\end{array}
\]

\item Inductive case: $E = \text{Case}(\vec{\phi}, \vec{E_1}, E_2)$.

$\denot{\text{Case}(\vec{\phi}, \vec{E_1}, E_2)}_{\schema, \context, \tuples} = \denot{\text{ITE}(\phi_1, E_{1}^{1}, \text{ITE}(\ldots, \text{ITE}(\phi_n, E_{1}^{n}, E_2)))}_{\schema, \context, \tuples}$ where $E_{1}^{i} \in \vec{E_1}$ and $n = |\vec{\phi}| = |\vec{E}|$ by Figure~\ref{fig:rules-expressions}.
$\denot{\text{Case}(\vec{\phi}, \vec{E_1}, E_2)}_{\interpretation(\context), \interpretation(\tuples)} = \sfoldr(\lambda y. \lambda (\phi_i, E_1^i). \denot{\text{ite}(\phi_i, E_1^i, y)}_{\interpretation(\context), \interpretation(\tuples)},$ $ E_2, reverse(zip(\vec{\phi}, \vec{E_1})))$ by Figure~\ref{fig:proof-sematics2}.
By inductive hypothesis, we have $\denot{\phi_i}_{\interpretation(\context), \interpretation(\tuples)} = \interpretation(\denot{\phi_i}_{\schema, \context, \tuples})$, $\denot{E_1^i}_{\interpretation(\context), \interpretation(\tuples)} = \interpretation(\denot{E_1^i}_{\schema, \context, \tuples})$ and $\denot{E_2}_{\interpretation(\context), \interpretation(\tuples)} = \interpretation(\denot{E_2}_{\schema, \context, \tuples})$.

\end{enumerate}
\end{proof}

\begin{lemma}\label{lem:attribute}
\revision{
Let $\db$ be a database over schema $\schema$ and $xs$ be a tuple list which formulates an attribute $A$.
Consider a symbolic database $\context$ over $\schema$ and a list of symbolic tuples $\tuples$ such that $\denot{A}_{\schema, \context, \tuples}$.
For any interpretation $\interpretation$ such that $\interpretation(\context) = \db \land \interpretation(\tuples) = xs$, evaluating $A$ over the concrete database $\interpretation(\context)$ and the concrete tuple list $\interpretation(\tuples)$ yields the value of $\interpretation(\denot{A}_{\schema, \context, \tuples})$, i.e., 
\[
\interpretation(\context) = \db \land \interpretation(\tuples) = xs \Rightarrow \denot{A}_{\interpretation(\context), \interpretation(\tuples)} = \interpretation(\denot{A}_{\schema, \context, \tuples})
\]
}
\end{lemma}

\vspace{10pt}

\begin{proof}[Proof.]\label{lem:proof_attribute}
By structural induction on $A$. 
\begin{enumerate}

\item Inductive case: $A = \text{Cast}(\phi)$.

By the definition of the $\text{Cast}$ function, we know $\text{Cast}(\phi) = \nullv$ iff $\phi = \nullv$, $\text{Cast}(\phi) = 1$ iff $\phi$ is true; otherwise it returns 0. 
Furthermore, by Lemma~\ref{lem:expression}, we know $\interpretation(\denot{A}_{\schema, \context, \tuples}) = \interpretation(\denot{\text{Cast}(\phi)}_{\schema, \context, \tuples}) = \denot{\text{Cast}(\phi)}_{\interpretation(\context), \interpretation(\tuples)} = \denot{v}_{\interpretation(\context), \interpretation(\tuples)} = \denot{E}_{\interpretation(\context), \interpretation(\tuples)} = \denot{A}_{\interpretation(\context), \interpretation(\tuples)}$ where $v \in \{0, 1\}$ if $\phi \in \{\top, \bot\}$; and by Lemma~\ref{lem:proof_predicate} $\interpretation(\denot{A}_{\schema, \context, \tuples}) = \interpretation(\denot{\text{Cast}(\nullv)}_{\schema, \context, \tuples}) = \denot{\text{Cast}(\nullv)}_{\interpretation(\context), \interpretation(\tuples)} = \denot{\nullv}_{\interpretation(\context), \interpretation(\tuples)} = \denot{\phi}_{\interpretation(\context), \interpretation(\tuples)} = \denot{A}_{\interpretation(\context), \interpretation(\tuples)}$ if $\phi = \nullv$.



\item Inductive case: $A = E$.

By Lemma~\ref{lem:expression}, we know $\interpretation(\denot{A}_{\schema, \context, \tuples}) = \interpretation(\denot{E}_{\schema, \context, \tuples}) = \denot{E}_{\interpretation(\context), \interpretation(\tuples)} = \denot{A}_{\interpretation(\context), \interpretation(\tuples)}$.

\item Inductive case: $A = \text{Count}(E)$.

$\denot{\text{Count}(E)}_{\schema, \context, \tuples} = \site(\text{AllNull}(E, \schema, \context, \tuples), \nullv, \Sigma_{t \in \tuples} \text{ite}(\del(t) \lor \denot{E}_{\schema, \context, [t]} = \nullv, 0, 1))$ by Figure~\ref{fig:rules-expressions}. 
$\denot{\text{Count}(E)}_{\interpretation(\context), \interpretation(\tuples)} = \site(\text{AllNull}(E, \interpretation(\context), \interpretation(\tuples)), \nullv, \sfoldl(+, 0, \smap(xs, \lambda y. $ $\site(\denot{E}_{\interpretation(\context), [y]} = \nullv, 0, 1))))$ by Figure~\ref{fig:proof-sematics2}.
By Lemma~\ref{lem:expression}, we have $\denot{E}_{\interpretation(\context), \interpretation(\tuples)} = \\$$ \interpretation(\denot{E}_{\schema, \context, \tuples})$.
By the definition of $\interpretation$, $\interpretation(\del)(t) = \top$ for the symbolic tuple $t \in \tuples \land t \not \in xs$; similarly $\interpretation(\del)(t) = \bot$ for the symbolic tuple $t \in \tuples \land t \in xs$.
\[
\begin{array}{rcl}
\interpretation(\text{AllNull}(E, \schema, \context, \tuples)) 
&=& \interpretation(\land_{t \in \tuples}(\neg \del(t) \to \denot{E}_{\schema, \context, [t]} = \nullv)) \\
&=& \land_{t \in \interpretation(\tuples)}(\neg \interpretation(\del)(t) \to \interpretation(\denot{E}_{\schema, \context, [t]}) = \nullv) \\
&=& \land_{t \in \interpretation(\tuples)}(\neg \bot \to \denot{E}_{\interpretation(\context), \interpretation([t])} = \nullv) \\
&=& \land_{t \in \interpretation(\tuples)}(\denot{E}_{\interpretation(\context), \interpretation([t])} = \nullv) \\
&=& \text{AllNull}(E, \interpretation(\context), \interpretation(\tuples))
\end{array}
\]
By the definition of $\sfoldl$ and $\smap$, $\sfoldl(+, 0, \smap(xs, \lambda y. \site(\denot{E}_{\db, [y]} = \nullv, 0, 1))) = \\ \Sigma_{x \in \interpretation(\tuples)} \site(\denot{E}_{\db, [x]} = \nullv, 0, 1)$.
\[
\begin{array}{rcl}
\interpretation(\denot{\text{Count}(E)}_{\schema, \context, \tuples})
&=& \interpretation(\site(\text{AllNull}(E, \schema, \context, \tuples), \nullv, \\
&& \qquad \Sigma_{t \in \tuples} \text{ite}(\del(t) \lor \denot{E}_{\schema, \context, [t]} = \nullv, 0, 1))) \\
&=& \site(\interpretation(\text{AllNull}(E, \schema, \context, \tuples)), \nullv, \\
&& \qquad \interpretation(\Sigma_{t \in \tuples} \text{ite}(\del(t) \lor \denot{E}_{\schema, \context, [t]} = \nullv, 0, 1))) \\
&=& \site(\text{AllNull}(E, \interpretation(\context), \interpretation(\tuples)), \nullv, \\
&& \qquad \Sigma_{t \in \interpretation(\tuples)} \text{ite}(\denot{E}_{\interpretation(\context), [t]} = \nullv, 0, 1)) \\
&=& \site(\text{AllNull}(E, \interpretation(\context), \interpretation(\tuples)), \nullv, \\
&& \qquad \sfoldl(+, 0, \smap(\interpretation(\tuples), \lambda y. \site(\denot{E}_{\interpretation(\context), [y]} = \nullv, 0, 1))) \\
&=& \denot{\text{Count}(E)}_{\interpretation(\context), \interpretation(\tuples)}
\end{array}
\]

\item Inductive case: $A = \text{Sum}(E)$.

$\denot{\text{Sum}(E)}_{\schema, \context, \tuples} = \site(\text{AllNull}(E, \schema, \context, \tuples), \nullv, \Sigma_{t \in \tuples} \text{ite}(\del(t) \lor \denot{E}_{\schema, \context, [t]} = \nullv, 0, \denot{E}_{\schema, \context, [t]}))$ by Figure~\ref{fig:rules-expressions}.
$\denot{\text{Sum}(E)}_{\interpretation(\context), \interpretation(\tuples)} = \site(\text{AllNull}(E, \interpretation(\context), \interpretation(\tuples)), \nullv, \sfoldl(+, 0, \smap(xs, \lambda y. $ $\site(\denot{E}_{\interpretation(\context), [y]} = \nullv, 0, \denot{E}_{\interpretation(\context), [y]}))))$ by Figure~\ref{fig:proof-sematics2}.
By Lemma~\ref{lem:expression}, we have $\denot{E}_{\interpretation(\context), \interpretation(\tuples)} = \interpretation(\denot{E}_{\schema, \context, \tuples})$.
By the definition of $\interpretation$, we know $\interpretation(\text{AllNull}(E, \schema, \context, \tuples)) = \text{AllNull}(E, \interpretation(\context), \interpretation(\tuples))$.
By the definition of $\sfoldl$ and $\smap$, $\sfoldl(+, 0, \smap(xs, \lambda y. \site(\denot{E}_{\db, [y]} = \nullv, 0, \denot{E}_{\db, [y]}))) = \Sigma_{x \in \interpretation(\tuples)} \site(\denot{E}_{\db, [x]} = \nullv, 0, \denot{E}_{\db, [x]})$.
\[
\begin{array}{rcl}
\interpretation(\denot{\text{Sum}(E)}_{\schema, \context, \tuples})
&=& \interpretation(\site(\text{AllNull}(E, \schema, \context, \tuples), \nullv, \\
&& \qquad \Sigma_{t \in \tuples} \text{ite}(\del(t) \lor \denot{E}_{\schema, \context, [t]} = \nullv, 0, \denot{E}_{\schema, \context, [t]}))) \\
&=& \site(\interpretation(\text{AllNull}(E, \schema, \context, \tuples)), \nullv, \\
&& \qquad \interpretation(\Sigma_{t \in \tuples} \text{ite}(\del(t) \lor \denot{E}_{\schema, \context, [t]} = \nullv, 0, \denot{E}_{\schema, \context, [t]}))) \\
&=& \site(\text{AllNull}(E, \interpretation(\context), \interpretation(\tuples)), \nullv, \\
&& \qquad \Sigma_{t \in \interpretation(\tuples)} \text{ite}(\denot{E}_{\interpretation(\context), [t]} = \nullv, 0, \denot{E}_{\interpretation(\context), [t]})) \\
&=& \site(\text{AllNull}(E, \interpretation(\context), \interpretation(\tuples)), \nullv, \sfoldl(+, 0, \\
&& \qquad \smap(\interpretation(\tuples), \lambda y. \site(\denot{E}_{\interpretation(\context), [y]} = \nullv, 0, \denot{E}_{\interpretation(\context), [y]}))) \\
&=& \denot{\text{Sum}(E)}_{\interpretation(\context), \interpretation(\tuples)}
\end{array}
\]

\item Inductive case: $A = \text{Avg}(E)$.
By the semantics of $\text{Avg}(E)$, we have 
\[
\begin{array}{rcl}
\interpretation(\denot{\text{Avg}(E)}_{\schema, \context, \tuples}) 
&=& \interpretation(\denot{\text{Sum}(E)}_{\schema, \context, \tuples}) / \interpretation(\denot{\text{Count}(E)}_{\schema, \context, \tuples}) \\
&=& \denot{\text{Sum}(E)}_{\interpretation(\context), \interpretation(\tuples)} / \denot{\text{Count}(E)}_{\interpretation(\context), \interpretation(\tuples)} \\
&=& \denot{\text{Avg}(E)}_{\interpretation(\context), \interpretation(\tuples)}
\end{array}
\]

\item Inductive case: $A = \text{Min}(E)$.

$\denot{\text{Min}(E)}_{\schema, \context, \tuples} = \site(\text{AllNull}(E, \schema, \context, \tuples), \nullv, \text{min}_{t \in \tuples} \text{ite}(\del(t) \lor \denot{E}_{\schema, \context, [t]} = \nullv, + \infty, \\$ $\denot{E}_{\schema, \context, [t]}))$ by Figure~\ref{fig:rules-expressions}.
$\denot{\text{Min}(E)}_{\interpretation(\context), \interpretation(\tuples)} = \site(\text{AllNull}(E, \interpretation(\context), \interpretation(\tuples)), \nullv, \sfoldl(\text{min}, $ $+ \infty, \smap(xs, \lambda y. $ $\site(\denot{E}_{\interpretation(\context), [y]} = \nullv, + \infty, \denot{E}_{\interpretation(\context), [y]}))))$ by Figure~\ref{fig:proof-sematics2}.
By Lemma~\ref{lem:expression}, we have $\denot{E}_{\interpretation(\context), \interpretation(\tuples)} = \interpretation(\denot{E}_{\schema, \context, \tuples})$.
By the definition of $\interpretation$, we know $\interpretation(\text{AllNull}(E, \schema, \context, \tuples)) = \text{AllNull}(E, \interpretation(\context), \interpretation(\tuples))$.
By the definition of $\sfoldl$ and $\smap$, $\sfoldl(\text{min}, + \infty, \smap(xs, \lambda y. \site($ $\denot{E}_{\interpretation(\context), [y]} = \nullv, + \infty, \denot{E}_{\interpretation(\context), [y]}))) \\ = \text{min}_{x \in xs}(\site(\denot{E}_{\interpretation(\context), [x]} = \nullv, + \infty, \denot{E}_{\interpretation(\context), [x]}))$.
\[
\begin{array}{rcl}
\interpretation(\denot{\text{Min}(E)}_{\schema, \context, \tuples})
&=& \interpretation(\site(\text{AllNull}(E, \schema, \context, \tuples), \nullv, \\
&& \qquad \text{min}_{t \in \tuples} \text{ite}(\del(t) \lor \denot{E}_{\schema, \context, [t]} = \nullv, + \infty, \denot{E}_{\schema, \context, [t]}))) \\
&=& \site(\interpretation(\text{AllNull}(E, \schema, \context, \tuples)), \nullv, \\
&& \qquad \interpretation(\text{min}_{t \in \tuples} \text{ite}(\del(t) \lor \denot{E}_{\schema, \context, [t]} = \nullv, + \infty, \denot{E}_{\schema, \context, [t]}))) \\
&=& \site(\text{AllNull}(E, \interpretation(\context), \interpretation(\tuples)), \nullv, \\
&& \qquad \Sigma_{t \in \interpretation(\tuples)} \text{ite}(\denot{E}_{\interpretation(\context), [t]} = \nullv, + \infty, \denot{E}_{\schema, \context, [t]})) \\
&=& \site(\text{AllNull}(E, \interpretation(\context), \interpretation(\tuples)), \nullv, \sfoldl(\text{min}, + \infty, \\
&& \qquad  \smap(\interpretation(\tuples), \lambda y. \site(\denot{E}_{\interpretation(\context), [y]} = \nullv, + \infty, \denot{E}_{\schema, \context, [y]}))) \\
&=& \denot{\text{Min}(E)}_{\interpretation(\context), \interpretation(\tuples)}
\end{array}
\]

\item Inductive case: $A = \text{Max}(E)$.

$\denot{\text{Max}(E)}_{\schema, \context, \tuples} = \site(\text{AllNull}(E, \schema, \context, \tuples), \nullv, \text{max}_{t \in \tuples} \text{ite}(\del(t) \lor \denot{E}_{\schema, \context, [t]} = \nullv, - \infty, $ $\denot{E}_{\schema, \context, [t]}))$ by Figure~\ref{fig:rules-expressions}.
$\denot{\text{Max}(E)}_{\interpretation(\context), \interpretation(\tuples)} = \site(\text{AllNull}(E, \interpretation(\context), \interpretation(\tuples)), \nullv, \sfoldl(\text{max},$ $ - \infty, \smap(xs, \lambda y. $ $\site(\denot{E}_{\interpretation(\context), [y]} = \nullv, - \infty, \denot{E}_{\interpretation(\context), [y]}))))$ by Figure~\ref{fig:proof-sematics2}.
By Lemma~\ref{lem:expression}, we have $\denot{E}_{\interpretation(\context), \interpretation(\tuples)} = \interpretation(\denot{E}_{\schema, \context, \tuples})$.
By the definition of $\interpretation$, we know $\interpretation(\text{AllNull}(E, \schema, \context, \tuples)) = \text{AllNull}(E, \interpretation(\context), \interpretation(\tuples))$.
By the definition of $\sfoldl$ and $\smap$, $\sfoldl(\text{max}, - \infty, \smap(xs, \lambda y. \site($ $\denot{E}_{\interpretation(\context), [y]} = \nullv, - \infty, \denot{E}_{\interpretation(\context), [y]}))) = \text{max}_{x \in xs}(\site(\denot{E}_{\interpretation(\context), [x]} = \nullv, - \infty, \denot{E}_{\interpretation(\context), [x]}))$.
\[
\begin{array}{rcl}
\interpretation(\denot{\text{Max}(E)}_{\schema, \context, \tuples})
&=& \interpretation(\site(\text{AllNull}(E, \schema, \context, \tuples), \nullv, \\
&& \qquad \text{max}_{t \in \tuples} \text{ite}(\del(t) \lor \denot{E}_{\schema, \context, [t]} = \nullv, - \infty, \denot{E}_{\schema, \context, [t]}))) \\
&=& \site(\interpretation(\text{AllNull}(E, \schema, \context, \tuples)), \nullv, \\
&& \qquad \interpretation(\text{max}_{t \in \tuples} \text{ite}(\del(t) \lor \denot{E}_{\schema, \context, [t]} = \nullv, - \infty, \denot{E}_{\schema, \context, [t]}))) \\
&=& \site(\text{AllNull}(E, \interpretation(\context), \interpretation(\tuples)), \nullv, \\
&& \qquad \Sigma_{t \in \interpretation(\tuples)} \text{ite}(\denot{E}_{\interpretation(\context), [t]} = \nullv, - \infty, \denot{E}_{\schema, \context, [t]})) \\
&=& \site(\text{AllNull}(E, \interpretation(\context), \interpretation(\tuples)), \nullv, \sfoldl(\text{max}, - \infty, \\
&& \qquad \smap(\interpretation(\tuples), \lambda y. \site(\denot{E}_{\interpretation(\context), [y]} = \nullv, - \infty, \denot{E}_{\schema, \context, [y]}))) \\
&=& \denot{\text{Max}(E)}_{\interpretation(\context), \interpretation(\tuples)}
\end{array}
\]

\item Inductive case: $A = A_1 \allarith A_2$.

By the inductive hypothesis, we have $\interpretation(\denot{A_1}_{\schema, \context, \tuples}) = \denot{A_1}_{\interpretation(\context), \interpretation(\tuples)}$ and $\interpretation(\denot{A_2}_{\schema, \context, \tuples}) = \denot{A_2}_{\interpretation(\context), \interpretation(\tuples)}$.
Therefore, 
\[
\begin{array}{rcl}
\interpretation(\denot{A_1 \allarith A_2}_{\schema, \context, \tuples}) 
&=& \interpretation(\denot{A_1}_{\schema, \context, \tuples}) \allarith \interpretation(\denot{A_2}_{\schema, \context, \tuples}) \\
&=& \denot{A_1}_{\interpretation(\context), \interpretation(\tuples)} \allarith \denot{A_2}_{\interpretation(\context), \interpretation(\tuples)} \\
&=& \denot{A_1 \allarith A_2}_{\interpretation(\context), \interpretation(\tuples)}
\end{array}
\]

\end{enumerate}
\end{proof}

\begin{proof}[Proof of Theorem~\ref{lem:equal}:]
\revision{
Given two relations $R_1 = [t_1, \ldots, t_n]$ and $R_2 = [r_1, \ldots, r_m]$, if the formula $(\ref{eq:bag1}) \land (\ref{eq:bag2})$ is valid, then $R_1$ is equal to $R_2$ under bag semantics. If the formula $(\ref{eq:list1}) \land (\ref{eq:list2})$ is valid, then $R_1$ is equal to $R_2$ under list semantics.
}
\end{proof}

\begin{proof}[Proof.] 

Let us prove this Theorem under bag and list semantics.

\begin{enumerate}
\item Under bag semantics.

By the formula $(\ref{eq:bag1})$, we assume there exist $a$ non-deleted tuples in $R_1$ and $R_2$ where $a \leq \text{min}\{n, m\}$.
Let $[x_1, \ldots, x_b]$ be the distinct tuples of those non-deleted tuples in $R_1$ or $R_2$.
To ensure the formula $(\ref{eq:bag2})$ holds, for any tuple $x_i$, it has the same multiplicity in $R_1$ and $R_2$ (i.e.,  $M(x_i, R_1) = M(x_i, R_2)$) where the function $M$ takes as input a non-delete tuple and the symbolic table, and returns the multiplicity of the tuple.
Therefore, the sum of all distinct tuples' multiplicity is equal to the number of non-deleted tuples, i.e., $a = \Sigma_{i=1}^b M(x_i, R_1) = \Sigma_{i=1}^b M(x_i, R_2)$, and the formula $(\ref{eq:bag2})$ holds for those non-deleted tuples.

Moreover, we know, for any deleted tuples $t_i \in R_1$ and $r_i \in R_2$ where $\del(t_i) = \del(r_i) = \top$, their multiplicities in $R_1$ and $R_2$ are always $0$ and, therefore, the formula $(\ref{eq:bag2})$ holds for them.

Thus, Theorem~\ref{lem:equal} holds under the bag semantics.

\item Under list semantics.

Let us discuss the value of $\text{min}\{n, m\}$ in three cases.
\begin{enumerate}[label=(\alph*)]
\item 
If $n < m$, then $\text{min}\{n, m\} = n$ and $\del(r_i) = \top$ for any $n < i \leq m$  because of the formula $(\ref{eq:list1})$. 
Furthermore, as we know $t_i = r_i$ for any $1 \leq i \leq n$, then $R_1 = [t_1, \ldots, t_n] = [r_1, \ldots, r_n] = \sappend([r_1, \ldots, r_n], [r_{n+1}, \ldots, r_{m}]) = [r_1, \ldots, r_m] = R_2$.

\item 
If $n = m$, then, apparently, we know $\text{min}\{n, m\} = n = m$ and $R_1 = R_2$ by the formula $(\ref{eq:list2})$.

\item 
If $n > m$, then $\text{min}\{n, m\} = m$ and $\del(t_i) = \top$ for any $m < i \leq n$  because of the formula $(\ref{eq:list1})$. 
Furthermore, as we know $t_i = r_i$ for any $1 \leq i \leq m$, then $R_2 = [r_1, \ldots, r_m] = [t_1, \ldots, t_m] = \sappend([t_1, \ldots, t_m], [t_{n+1}, \ldots, t_{m}]) = [t_1, \ldots, t_m] = R_1$.

\end{enumerate}

By the above cases, we can prove that Theorem~\ref{lem:equal} holds under the list semantics.
\end{enumerate}

Hence, Theorem~\ref{lem:equal} is proved by the aforementioned cases.
\end{proof}

\begin{proof}[Proof of Theorem~\ref{thm:verify}:]
\revision{
Given two queries $Q_1, Q_2$ under schema $\schema$, an integrity constraint $\constraint$, a bound $\bound$, if $\textsc{Verify}(\query_1, \query_2, \schema, \constraint, \bound)$ returns $\top$, then \revision{$\query_1 \simeq_{\schema, \constraint, \bound} \query_2$}. Otherwise, if $\textsc{Verify}(\query_1, \query_2, \schema,$ $ \constraint, \bound)$ returns a database $\db$, then \revision{$\db :: \schema \land \denot{Q_1}_{\db} \neq \denot{Q_2}_{\db}$}. 
}
\end{proof}

\begin{proof}[Proof.]

Let us prove Theorem~\ref{thm:verify} in two cases.

\begin{enumerate}
\item 
If $\textsc{Verify}(\query_1, \query_2, \schema, \constraint, \bound)$ returns $\top$, then we cannot find two symbolic relations $R_1$ and $R_2$ s.t. $R_1 \neq R_2$ by Theorem~\ref{lem:equal}.
Let $\interpretation$ be an interpretation s.t. $\db = \interpretation(\context)$.
By the soundness of $\textsc{BuildSymbolicDB}(\schema, N)$ (line 2) of Algorithm~\ref{algo:verify}, we know $\db :: \schema$ and $\forall R \in \db. |R| \leq N$ holds for any database $\db$ instantiated from $\context$.
By Theorem~\ref{lem:ic}, we know $\constraint(\db)$ holds for any database $\db$ instantiated from $\context$.
Further, by Theorem~\ref{lem:query1} and~\ref{lem:query2}, $\denot{\query}_{\interpretation(\context)} = \denot{\query}_{\db} = \interpretation(R)$ where $R$ is a symbolic relation from $\formula_{R}, R \leftarrow \textsc{EncodeQuey}(\schema, \context, \query)$.
Therefore, $R_1 = R_2 \Leftrightarrow \interpretation(R_1) = \interpretation(R_2) \Leftrightarrow \denot{\query_1}_{\db} = \denot{\query_2}_{\db}$ where $\db$ is instantiated from $\context$, and $\query_1 \simeq_{\schema, \constraint, \bound} \query_2$ holds.

\item 
If $\textsc{Verify}(\query_1, \query_2, \schema, \constraint, \bound)$ returns a database $\db$, we know that there exists a database $\db$ s.t. $\formula = \formula_{\constraint}\land \formula_{R_1} \land \formula_{R_2} \land \neg \textsc{Equal}(R_1, R_2)$.  By the soundness of $\text{BuildExample}(\schema, \text{Model}(\formula))$, $\db :: \schema$ is also true.
Further, $\formula_{\constraint} \land \formula_{R_1} \land \formula_{R_2} \land \neg \textsc{Equal}(R_1, R_2) = \top \Rightarrow \neg \textsc{Equal}(R_1, R_2) = \top $ $ \Leftrightarrow \textsc{Equal}(R_1, R_2) = \bot$.
By Theorem~\ref{lem:query1} and~\ref{lem:query2}, there exists an $\interpretation$ s.t. $\interpretation(R) = \denot{\query}_{\interpretation(R)} = \denot{\query}_\db$; and, by Theorem~\ref{lem:equal}, $\textsc{Equal}(R_1, R_2) = \bot \Leftrightarrow \denot{\query_1}_{\interpretation(R_1)} \neq \denot{\query_2}_{\interpretation(R_2)} \Leftrightarrow \denot{\query_1}_{\db} \neq \denot{\query_2}_{\db}$.

\end{enumerate}

Hence, Theorem~\ref{thm:verify} is proved by the aforementioned cases.
\end{proof}

}{}

\end{document}